\documentclass[11pt,onecolumn]{IEEEtran}
\usepackage[caption=false,font=footnotesize]{subfig}
\usepackage{graphicx}
\usepackage{booktabs}
\usepackage{amsmath}
\usepackage{amssymb}
\usepackage{amsthm, amsfonts}
\usepackage{midfloat}
\usepackage{bbm}
\usepackage{color, cite}
\usepackage{threeparttable}
\graphicspath{{figures/}}

\newtheorem{definition}{Definition}

\newtheorem{theorem}{Theorem}
\newtheorem{lemma}{Lemma}

\newtheorem{corollary}{Corollary}
\newtheorem{remark}{Remark}
\newtheorem{proposition}{Proposition}

\begin{document}
\title{Does $\ell_p$-minimization outperform $\ell_1$-minimization?}
\author{Le Zheng$^1$, Arian Maleki$^2$, Haolei Weng$^2$, Xiaodong Wang$^3$, Teng Long$^1$
\thanks{1. School of Information and Electronics, Beijing Institute of Technology, Beijing, China. 2. Department of Statistics, Columbia University, NY, USA. 3. Department of Electrical Engineering, Columbia University, NY, USA. This paper is presented in part at Signal Processing with Adaptive Sparse Structured Representations workshop, and will be presented at IEEE International Symposium on Information Theory.}}
\maketitle
\begin{abstract}

In many application areas ranging from bioinformatics to imaging, we are faced with the following question: can we recover a sparse vector $x_o \in \mathbb{R}^N$ from its undersampled set of noisy observations $y \in \mathbb{R}^n$, $y= A x_o+ w$. The last decade has witnessed a surge of algorithms and theoretical results to address this question. One of the most popular algorithms is the $\ell_p$-regularized least squares given by the following formulation:
\begin{equation*}
\hat x(\gamma,p ) \in \mathop {\arg \min }\limits_x \frac{1}{2}\left\| {y - Ax} \right\|_2^2 + \gamma {\| x \|_p^p},
\end{equation*}
where $p \in [0,1]$. Among these optimization problems, the case $p=1$, also known as LASSO, is the best accepted in practice, for the following two reasons: (i) thanks to the extensive studies performed in the fields of high-dimensional statistics and compressed sensing, we have a clear picture of LASSO's performance. (ii) it is convex and efficient algorithms exist for finding its global minima.  

Unfortunately, neither of the above two properties hold for $0  \leq p<1$. However, they are still appealing because of the following folklores in the high-dimensional statistics: 
(i) $\hat x(\gamma,p )$ is closer to $x_o$ than $\hat{x}(\gamma,1)$. (ii) If we employ iterative methods that aim to converge to a local minima of $ {\arg \min }_x \frac{1}{2}\left\| {y - Ax} \right\|_2^2 + \gamma {\| x \|_p^p}$, then under good initialization, these algorithms converge to a solution that is still closer to $x_o$ than $\hat{x}(\gamma,1)$. In spite of the existence of plenty of empirical results that support these folklore theorems, the theoretical progress to establish them has been very limited.

This paper aims to study the above folklore theorems and establish their scope of validity. Starting with approximate message passing (AMP) algorithm as a heuristic method for solving $\ell_p$-regularized least squares, we study the following questions: (i) what is the impact of initialization on the performance of the algorithm? (ii) when does the algorithm recover the sparse signal $x_o$ under a ``good'' initialization? (iii) when does the algorithm converge to the sparse signal regardless of the initialization? Studying these questions will not only shed light on the second folklore theorem, but also lead us to the answer of the first one, i.e., the performance of the global optima $\hat x(\gamma,p )$. For that purpose, we employ the  replica analysis\footnote{Replica method is a widely accepted heuristic method in statistical physics for analyzing large disordered systems.} to show the connection between the solution of AMP and $\hat{x}(\gamma, p)$ in the asymptotic settings. This enables us to compare the accuracy of $\hat x(\gamma,p )$ and $\hat x(\gamma,1 )$. In particular, we will present an accurate characterization of the phase transition and noise sensitivity of $\ell_p$-regularized least squares for every $0 \leq p \leq 1$. Our results in the noiseless setting confirm that $\ell_p$-regularized least squares (if $\gamma$ is tuned optimally) exhibits the same phase transition for every $0 \leq p< 1$ and this phase transition is much better than that of LASSO. Furthermore, we show that in the noisy setting, there is a major difference between the performance of $\ell_p$-regularized least squares with different values of $p$. For instance, we will show that for very small and very large measurement noises, $p=0$ and $p=1$ outperform the other values of $p$, respectively.

\end{abstract}
\begin{IEEEkeywords}
Compressed sensing, $\ell_p$-regularized least squares, LASSO, non-convex penalties, approximate message passing, state evolution, replica analysis. 
\end{IEEEkeywords}
\IEEEpeerreviewmaketitle

\section{Introduction}
\label{sec:intro}
\subsection{Problem statement}
Recovering a sparse signal ${x_o} \in {\mathbb{R}^N}$ from an undersampled set of random linear measurements $y = A{x_o} + w$ is the main problem of interest in compressed sensing (CS) \cite{maleki2013asymptotic,baraniuk2007compressive}.  Among various schemes proposed for estimating $x_o$, $\ell_p$-regularized least squares (LPLS) has received attention for its proximity to the ``intuitively optimal'' $\ell_0$-minimization. LPLS estimates $x_o$ by solving 
\begin{equation}\label{eq:ell0minimization}
\hat x(\gamma,p ) \in \mathop {\arg \min }\limits_x \frac{1}{2}\left\| {y - Ax} \right\|_2^2 + \gamma {\| x \|_p^p},
\end{equation}
where ${\left\|  \cdot  \right\|_p}$ ($0 \leq p \leq 1$) denotes the $\ell_p$-norm \footnote{The $\ell_p$-norm of a vector $x= [x_1, x_2, \ldots, x_N]^T$ is defined as $\| x\|_p^p \triangleq \sum_{i=1}^N |x_i|^p$.}, $\gamma \in (0, \infty)$ is a fixed number, and  $\gamma {\left\| x \right\|_p^p}$ is a regularizer that promotes sparsity. The convexity of this optimization problem for $p=1$ has made it the most accepted and the best studied scheme among all LPLSs. However, it has always been in the folklore of compressed sensing community that solving \eqref{eq:ell0minimization} for $p<1$ leads to more accurate solutions than the $\ell_1$-regularized least squares, also known as LASSO, since ${\| x \|_p^p}$ models the sparsity better \cite{chartrand2007exact, stojnic2013lifting, chartrand2008restricted, saab2008stable, foucart2009sparsest, saab2010sparse, ge2011note, kabashima2009typical, lai2011unconstrained, sun2012recovery, trzasko2007sparse, aldroubi2011stability, rangan2009asymptotic, mazumder2011sparsenet}. Inspired by this folklore theorem, many researchers have proposed iterative algorithms to obtain a local minima of the non-convex optimization problem \eqref{eq:ell0minimization} with $p \in [0, 1)$\cite{chartrand2007exact, candes2008enhancing}. 

The performance of such schemes is highly affected by their initialization; better initialization increases the chance of converging to the global optima. One popular choice of initialization is the solution of LASSO \cite{candes2008enhancing}. This initialization has been motivated by the following heuristic: The solution of LASSO is closer to the global minima of \eqref{eq:ell0minimization} than a random initialization. Hence it helps the iterative schemes to avoid stationary points that are not the global minima of LPLS. Ignoring the computational issues, one can extend this approach to the following  initialization scheme: Suppose that our goal is to solve \eqref{eq:ell0minimization} for $p=p_0$. Define an increasing sequence of numbers $p_0<p_1< \ldots < p_q=1$ for some $q$. Start with solving LASSO and then use its solution as an initialization for the iterative algorithm that attempts to solves \eqref{eq:ell0minimization} with $p_{q-1}$. Once the algorithm converges, its estimate is employed as an initialization for $p_{q-2}$. The process continues until the algorithm reaches $p_0$.  We call this  approach {\em $p$-continuation}.

Here is a heuristic motivation of the $p$-continuation. Let $\hat{x}(\gamma, p_i)$ denote the global minimizer of  $\frac{1}{2}\left\| {y - Ax} \right\|_2^2 + \gamma {\| x \|_{p_i}^{p_i}}$. Since $p_i$ and $p_{i+1}$ are close, we expect $\hat x(\gamma,p_i )$ and $\hat x(\gamma,p_{i+1})$ to be ``close'' as well. Hence if the algorithm that is solving for $\hat{x}(\gamma, p_i)$ is initialized with $\hat{x}(\gamma, p_{i+1})$, then it may avoid all the local minima and converge to $\hat x(\gamma,p_i )$. Simulation results presented elsewhere confirm the efficiency of such initialization algorithms \cite{pant2014new, mazumder2011sparsenet}.  \\

We can summarize our discussions in the following three {\em folklore theorems of compressed sensing}:
\begin{itemize}
\item[(i)] The global minima of \eqref{eq:ell0minimization} for $p<1$ outperforms the solution of LASSO. Furthermore, smaller values of $p$ lead to more accurate estimates. 
\item[(ii)] There exist iterative algorithms (ITLP) capable of converging to the global minima of \eqref{eq:ell0minimization} under ``good'' initialization.
\item[(iii)] $p$-continuation provides a ``good'' initialization for ITLP. 
\end{itemize}

Our paper aims to evaluate the scope of validity of the above folklore beliefs in the asymptotic settings.\footnote{Parts of our results that are presented in Section \ref{sec:discussion} are based on the replica analysis \cite{rangan2009asymptotic}. Replica method is a non-rigorous but widely accepted technique from statistical physics for studying large disordered systems. Hence the results we will present in Section \ref{sec:discussion} are not fully rigorous. } Toward this goal, we first study a family of message passing algorithms that aim to solve \eqref{eq:ell0minimization}; we characterize the accuracy of the estimates generated from the message passing algorithm under various initializations, including the best initialization obtained by $p$-continuation. We finally connect our results for the message passing algorithm estimates to the analysis of global minima $\hat{x}(\gamma, p)$ of \eqref{eq:ell0minimization} by Replica method. Here is a summary of our results explained informally:

   \begin{itemize}
\item[(i)] If the measurement noise $w$ is zero or small, then the global minima of \eqref{eq:ell0minimization} for $p<1$ (when $\gamma$ is optimally picked) outperforms the solution of LASSO with optimal $\gamma$. Furthermore, all values of $p<1$ have the same performance when $w=0$. When $w$ is small, LPLS with the value of $p$ closer to $0$ has a better performance. However, as the variance of the measurement noise increases beyond a certain level, this folklore theorem is not correct any more. In other words, for large measurement noise, the solution of LASSO outperforms the solution of LPLS for every $0 \leq p<1$. 

\item[(ii)] We introduce approximate message passing algorithms that are capable of converging to the global minima of \eqref{eq:ell0minimization} under ``good'' initialization (in the asymptotic settings). We call these algorithms $\ell_p$-AMP. 
\item[(iii)] The ``performance'' of the message passing algorithm under $p$-continuation is equivalent to the ``performance'' of message passing algorithm for solving \eqref{eq:ell0minimization} with the best value of $p$. As a particular conclusion of this result, we note that $p$-continuation can only slightly improve the phase transition of $\ell_1$-AMP. $p$-continuation is mainly useful when the noise is low and $x_o$ has very few non-zero coefficients.  
\end{itemize}


\vspace{.2cm}

 There has been recent efforts to formally prove some of the above folklore theorems. We briefly review some of these studies and their similarities and differences with our work below. Among the three folklore results we have discussed so far, the first one is the best studied. In particular, many researchers have tried to confirm that at least in the noiseless settings ($w=0$), the global minima of \eqref{eq:ell0minimization} for $p<1$ outperforms the solution of LASSO. Toward this goal, \cite{saab2010sparse, saab2008stable, chartrand2008restricted, foucart2009sparsest, shen2012restricted, davies2009restricted} have employed some popular analysis tools such as the well-known restricted isometry property and derived the conditions under which \eqref{eq:ell0minimization} recovers $x_o$ accurately. We briefly mention the results of \cite{chartrand2008restricted} to emphasize on the strengths and weaknesses of this approach. Let the elements of $A$ be iid $N(0,1)$ and $y=Ax_o$, where $x_o$ is $k$-sparse, meaning it has only $k$ nonzero elements. If $n > C_1(p)k+pC_2(p) k \log \frac{N}{k}$, then the optimization problem 
 \[
 \min_x \|x\|_p \ \ \ {\rm subject \ to } \ \ y=Ax
 \]
 recovers $x_o$ with high probability. Furthermore, $C_1(p)$ and $pC_2(p)$ are increasing functions of $p$. The lower bound derived for the required number of measurements decreases as $p$ decreases. This may be an indication of the fact that smaller values of $p$ lead to better recovery algorithms. However, note that this result only offers a sufficient condition for recovery and hence any conclusion drawn from such results on the strengths of these algorithms may be misleading.\footnote{ Note that even though we have mentioned the results for iid Gaussian matrices, it can be easily extended to many other measurement matrix ensembles.}
 
 To provide more reliable comparison among different algorithms, many researchers have analyzed these algorithms in the asymptotic setting $N \rightarrow \infty$ (while $\epsilon \triangleq k/N$ and $\delta \triangleq n/N$ are fixed) \cite{stojnic2013lifting, rangan2009asymptotic, kabashima2009typical, wang2011performance}. This is the framework that we adopt in our analysis too.  We review these four papers in more details and compare them with our work. Stojnic and Wang et al. \cite{stojnic2013lifting, wang2011performance} consider the noiseless setting and try to characterize the boundary between the success region (in which \eqref{eq:ell0minimization} recovers $x_o$ exactly with probability one) and the failure region. This boundary is known as the phase transition curve (PTC).\footnote{There are some subtle discrepancies between our definition of the phase transition curve and these two papers'. However, our results are more aligned and comparable with the results of \cite{stojnic2013lifting}. } The characterization of PTC in \cite{stojnic2013lifting} is only accurate for the case $p=0$. Also, the analysis of \cite{wang2011performance} is sharp only for $\delta \rightarrow 1$. Our paper derives the exact value of PTC for any value of $0 \leq p <1$ and any value of $\delta$. Furthermore, we present accurate calculation of the risk of $\hat{x}(\gamma, p)$ in the presence of noise and compare the accuracy of $\hat{x} (\gamma,p)$ for different values of $p$.  However, unlike \cite{stojnic2013lifting} and \cite{wang2011performance}, part of our analysis, presented in Section \ref{sec:discussion}, is based on the Replica method and they are not fully rigorous yet. Note that all the results we present for approximate message passing are rigorous and we only employ Replica method to show the connection between the solution of AMP and $\hat{x} (\gamma, p)$. 
 
 Replica method has been employed for studying \eqref{eq:ell0minimization} in \cite{rangan2009asymptotic, kabashima2009typical} to derive the fixed point equations that describe the performance of $\hat{x} (\gamma, p)$ (under the asymptotic settings). These equations are discussed in Section \ref{sec:discussion}. To provide fair comparison of the performance of $\hat{x} (\gamma,p)$ among different $p$, one should analyze the fixed points of these equations under the optimal tuning of the parameter $\gamma$.  Such analysis is missing in both papers. In this paper, by employing the minimax framework, we are able to analyze the fixed points and provide sharp characterization of the phase transition of \eqref{eq:ell0minimization} and its noise sensitivity for the first time. In addition, we present algorithms whose asymptotic behavior can be characterized by the same fixed point equations as the ones derived from Replica method. The minimax framework enables us to analyze the stationary points at which the algorithm may be trapped and derive conditions under which the algorithm can converge to the global minimizer of \eqref{eq:ell0minimization}. 
 
As a final remark, we should emphasize that, to the best of our knowledge, the second and third folklore results have never been studied before and our results may be considered as the first contribution in this direction.

\subsection{Message passing and approximate message passing}\label{sec:firstgaussianity}

One of the main building blocks of our analysis is the approximate message passing (AMP) algorithm.
AMP  is a fast iterative algorithm proposed originally for solving LASSO \cite{donoho2009message}. Starting from $z^0= y$ and $x^0= 0$, the algorithm employs the following iteration:
\begin{eqnarray}
\label{eq:2-4}
{ x^t} &=& \eta_1 ({A^T}{z^{t - 1}} + {x^{t - 1}};{\lambda _t}), \nonumber \\
{z^t} &=& y - A{x^t} + {z^{t - 1}}\frac{1}{\delta }\left\langle {\eta_1' }({A^T}{z^{t - 1}} + {{x}^{t - 1}};{\lambda _t}) \right\rangle,
\end{eqnarray}
where $x^t$ is the estimation of $x_o$ at iteration $t$ and  $\delta  = \frac{n}{N}$. Furthermore, for a vector $u = {\left[ {u_1,...,u_N} \right]^T}$, $\left\langle u \right\rangle  = {{\sum\limits_{i = 1}^N {{u_i}} }}/{N}$. $\eta_1(u; \lambda)$ is the soft thresholding function defined as ${\eta _1}(u;\lambda ) = \left( {\left| u \right| - \lambda } \right) {\rm sign}(u) {\mathbb{I}}\left( {\left| u \right| > \lambda } \right)$ with $\mathbb{I}\left(  \cdot  \right)$ denoting the indicator function. $\lambda$ is called the threshold parameter. ${\eta '_1}$ denotes the derivative of $\eta _1$, i.e., ${\eta '_1}(u;{\lambda}) = \frac{{\partial {\eta _1}(u;{\lambda})}}{{\partial u}}$. When $u$ is a vector, $\eta_1(u;\lambda)$ and $\eta'_1(u;\lambda)$ operate component-wise. In the rest of the paper, we call this algorithm $\ell_1$-AMP. It has been proved that if the entries of $A$ are iid Gaussian, then in the asymptotic settings, the limit of $x^t$ corresponds to the solution of LASSO for a certain value of $\lambda$ \cite{bayati2012lasso, donoho2011noise}. 

First, we extend $\ell_1$-AMP to solve LPLS defined in \eqref{eq:ell0minimization}. Iteration of $\ell_1$-AMP have been derived  from the first order approximation of the $\ell_1$-message passing algorithm \cite{maleki2010approximate} given by
\begin{eqnarray}\label{eq:ell_1messagepassingcomp}
x_{i \rightarrow a}^t &=& \eta_1 \Big(\sum_{b \neq a} A_{bi} z^t_{b \rightarrow i} ; \lambda_t \Big), \nonumber \\
z_{a \rightarrow i}^t &=& y_a - \sum_{j\neq i} A_{aj} x_{j \rightarrow a}^{t-1},
 \end{eqnarray}
 where $x_{i \rightarrow a}^t$ and $z_{a \rightarrow i}^t$ ($i \in \{1,2 \ldots, N\}$ and $a \in \{1,2,\ldots, n\}$) are $2nN$ variables that must be updated at every iteration of the message passing algorithm. Compared to this (full) message passing, AMP is computationally less demanding since it only has to update $n+N$ variables  at each iteration. It is straightforward to replicate the calculations of \cite{maleki2010approximate} for a generic version of LPLS to obtain the following message passing algorithm:
\begin{eqnarray}\label{eq:mpell_p}
x_{i \rightarrow a}^t &=& \eta_p \Big(\sum_{b \neq a} A_{bi} z^t_{b \rightarrow i} ; \lambda_t \Big), \nonumber \\
z_{a \rightarrow i}^t &=& y_a - \sum_{j\neq i} A_{aj} x_{j \rightarrow a}^{t-1}.
 \end{eqnarray}
Here $\eta_p(u; \lambda) \triangleq \mathop {\arg \min }_x \frac{1}{2} \| {u- x} \|_2^2 + \lambda \| x \|_p^p$ is known as the proximal function for $\lambda \|x\|_p^p$. It is worth noting that for $p=1$, $\eta_1(u; \lambda)$ is the soft thresholding function introduced in $\ell_1$-AMP, and for $p=0$, ${\eta _0}(u;\lambda) = u \cdot \mathbb{I}( {\left| u \right| > \sqrt {2\lambda}  } )$ is known as the {\em hard thresholding} function. For the other values of $p \in (0,1)$, $\eta_p(u;\lambda)$ does not have a simple explicit form, but it can be calculated numerically. Figure \ref{fig:ell_pprox} exhibits $\eta_p$ for different values of $p$. Note that all these proximal functions map small values of $u$ to zero and hence promote sparsity. Because of the specific shape of these functions, we may interchangeably call them threshold functions.

Note that iterations of \eqref{eq:mpell_p} are computationally demanding since they update $2nN$ messages at every iteration. Therefore, simplification of this algorithm is vital for practical purposes. One simplification that is proposed in \cite{maleki2010approximate} (and has led to AMP) argues that $z_{b \rightarrow i}^t = z_b^t + \zeta_{b \rightarrow i}^t + O(1/N)$ and $x_{i \rightarrow b}^t = x_i^t + \chi_{i \rightarrow b}^t + O(1/N)$, where $\zeta_{b \rightarrow i}^t, \chi_{i \rightarrow b}^t = O(1/\sqrt{n})$. Under this assumption, one may use a Taylor expansion of $\eta_1$ in \eqref{eq:ell_1messagepassingcomp} and obtain \eqref{eq:2-4}. 

If $\eta_p(\cdot)$ were weakly differentiable, the same simplification could be applied to \eqref{eq:mpell_p}. However, according to Figure \ref{fig:ell_pprox}, $\eta_p(\cdot)$ is discontinuous for $p<1$.  This problem can be resolved by one more approximation of the message passing algorithm. In this process, we not only approximate $x_{i \rightarrow a}^t$ and $z_{a \rightarrow i}^t$, but also approximate $\eta_p(\cdot)$ by a smooth function $\tilde{\eta}_{p,h}$ constructed in the following way. We first decompose $\eta_p(u;\lambda)$ to
\begin{equation}
{\eta _p}(u;\lambda ) = S_p(u;\lambda ) + D_p(u;\lambda ),
\end{equation}
where
\begin{equation}\label{eq:defS_p}
S_p{(u;\lambda )} = \left\{ \begin{gathered}
  {\eta _p}({u};\lambda ) - \eta _p^-(-\tilde{\lambda};\lambda ),{\text{ if }}{u} <  - \tilde \lambda,  \hfill \\
  0,{\text{ if }} - \tilde \lambda  \leqslant {u} \leqslant   \tilde \lambda,  \hfill \\
  {\eta _p}({u};\lambda ) - \eta _p^+(\tilde{\lambda};\lambda ),{\text{ if }}{u} >   \tilde \lambda,  \hfill \\
\end{gathered}  \right.
\hspace{1cm}
{D_p}{(u;\lambda )} = \left\{ \begin{gathered}
    {\eta _p^-}(- \tilde \lambda ;\lambda ),{\text{ if }}{u} <  - \tilde \lambda,  \hfill \\
  0,{\text{ if }} - \tilde \lambda  \leqslant {u} \leqslant   \tilde \lambda,  \hfill \\
  {\eta _p^+}(\tilde \lambda ;\lambda ),{\text{ if }}{u} >  \tilde \lambda.  \hfill \\
\end{gathered}  \right.
\end{equation}
Here $\tilde \lambda$ represents the threshold below which ${\eta _p}(u,\lambda )= 0$. The exact form of $\tilde{\lambda}$ will be derived in Lemma \ref{lem:threshform}.  Furthermore,
\begin{eqnarray}
\eta _p^-(-\tilde{\lambda};\lambda ) &\triangleq& \lim_{u \nearrow -\tilde{\lambda}} \eta _p(u;\lambda ), \nonumber \\
\eta _p^+(\tilde{\lambda};\lambda ) &\triangleq& \lim_{u \searrow \tilde{\lambda}} \eta _p(u;\lambda ), \nonumber 
\end{eqnarray}
where $\nearrow$ and $\searrow$ denote convergence from left and right respectively.  $S_p{(u;\lambda )}$ is a weakly differentiable function, while ${D_p}{(u;\lambda )} $ is not continuous. Let $G_h$ denote the Gaussian kernel with variance $h^2>0$. We construct the following smoothed version of $\eta_p$: 
\begin{equation}\label{eq:smoothetap}
\tilde{\eta}_{p,h} (u;\lambda ) = S_p{(u;\lambda )} +{\tilde{D}_{p,h}}{(u;\lambda )} , 
\end{equation}
 where $\tilde{D}_{p,h}{(u;\lambda )} \triangleq {{D}_p}{(u;\lambda )} * G_h(u)$.\footnote{This smoothing is also proposed for the hard thresholding function in \cite{deledalle2013stein}  for a different purpose. } Here $*$ denotes the convolution operator. If we replace $\eta_p(\cdot)$ with $\tilde{\eta}_{p,h}(\cdot)$ in \eqref{eq:mpell_p}, we obtain a new message passing algorithm:
  \begin{eqnarray}\label{eq:mpell_psmooth}
x_{i \rightarrow a}^t &=& \tilde{\eta}_{p,h} \Big(\sum_{b \neq a} A_{bi} z^t_{b \rightarrow i} ; \lambda_t \Big), \nonumber \\
z_{a \rightarrow i}^t &=& y_a - \sum_{j\neq i} A_{aj} x_{j \rightarrow a}^{t-1},
 \end{eqnarray}
 where $h$ is assumed to be ``small'' to ensure that replacing $\eta_p$ with $\tilde{\eta}_{p,h}$  does not incur major loss to the performance of the message passing algorithm. We discuss practical methods for setting $h$ in the simulation section. Since $\tilde{\eta}_{p,h}(\cdot)$ is smooth, we may apply the approximation technique proposed in \cite{maleki2010approximate} to obtain the following approximate message passing algorithm:
\begin{eqnarray}
\label{eq:ellpamp}
{ x^t} &=& \tilde{\eta}_{p,h} ({A^T}{z^{t - 1}} + {x^{t - 1}};{\lambda _t}), \nonumber \\
{z^t} &=& y - A{x^t} + {z^{t - 1}}\frac{1}{\delta }\left\langle {\tilde{\eta}_{p,h}' }({A^T}{z^{t - 1}} + {{x}^{t - 1}};{\lambda _t}) \right\rangle.
\end{eqnarray}

\begin{figure}[htbp]
  \centering
  \includegraphics[width=3.2in]{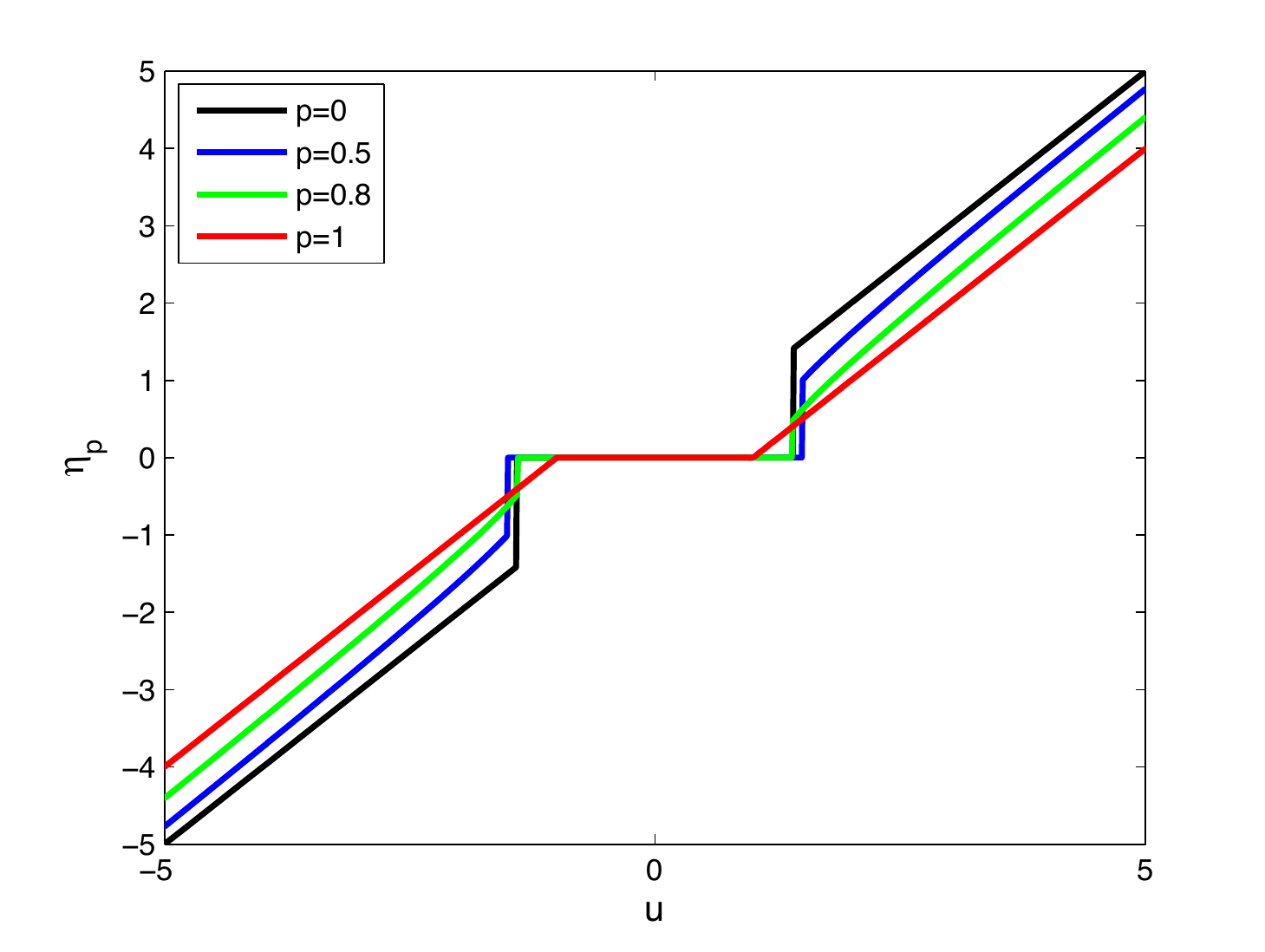}
  \caption{${\eta _p}\left( {u;\lambda } \right)$ for 4 different values of $p$. $\lambda$ is set to $1$.}
  \label{fig:ell_pprox}
\end{figure}
%
We call this algorithm $\ell_p$-AMP. If we define $v^t \triangleq A^T z^t+x^t-x_o$, then we can write $x^{t+1} = \tilde{\eta}_{p,h} (x_o+ v^t; \lambda)$. One of the main features of AMP that has led to its popularity is that for large values of $n$ and $N$, $v^t$ looks like a zero mean iid Gaussian noise. This property has been observed and characterized for different denoisers in \cite{donoho2009message, donoho2011noise, maleki2013asymptotic, rangan2011generalized, donoho2011accurate, schniter2010turbo, metzler2014denoising, maleki2010approximate} and has also been proved for some special cases in \cite{bayati2011dynamics, rangan2011generalized}. Since this key feature plays an important role in our paper, we start by formalizing this statement.

Let $n,N \rightarrow \infty$ while $\delta = \frac{n}{N}$ is fixed. In the rest of this section only, we write the vectors and matrices as $x_o(N), A(N), y(N)$, and $w(N)$ to emphasize dependence on the dimensions of $x_o$. Clearly, matrix $A$ has $\delta N$ rows, but since we assume that $\delta$ is fixed, we do not include $n$ in our notation for $A$.  The same argument is applied to $y(N)$ and $w(N)$. The following definition adopted from \cite{bayati2011dynamics} formalizes the asymptotic setting in which $\ell_p$-AMP is studied. 
\vspace{.3cm}

\begin{definition}\label{def:convseq}
A sequence of instances $\{x_o(N), A(N), w(N)\}$ is called a converging sequence if the following conditions hold:
\begin{itemize}
\item[-] The empirical distribution of $x_o(N) \in \mathbb{R}^N$ converges weakly to a probability measure $p_{X}$ with bounded second moment. Further, $\frac{1}{N} \|x_o(N)\|_2^2$ converges to the second moment of $p_{X}$.
\item[-] The empirical distribution of $w(N) \in \mathbb{R}^n$ ($n = \delta N$) converges weakly to a probability measure $N(0, \sigma_w^2)$. Furthermore, $\frac{1}{n} \|w(N)\|_2^2$ converges to $ \sigma_w^2$.
\item[-] $A_{ij} \sim N(0, 1/n)$.
\end{itemize}
\end{definition}

\vspace{.3cm}

The following theorem not only formalizes the ``Gaussianity'' of $v^t$, but also provides a simple way to characterize its variance. 

\vspace{.3cm}

\begin{theorem}\label{conj:se}
Let $\{x_o(N), A(N), w(N)\}$ denote a converging sequence of instances. Let $x^t(N,h)$ denote the estimates provided by $\ell_p$-AMP according to \eqref{eq:ellpamp}. Let $h_1, h_2, \ldots$ denote a decreasing sequence of numbers that satisfy $h_i>0$ and $h_i \rightarrow 0$ as $i \rightarrow \infty$.  Then,
\[
\lim_{i \rightarrow \infty} \lim_{N \rightarrow \infty}  \frac{\|x^{t+1}(N,h_i)-x_o(N)\|_2^2}{N} \overset{\rm a.s.}{=} \mathbb{E} \left( {{{\left| {\eta_{p} (X + {\sigma _t}Z;{\lambda _t}) - X} \right|}^2}} \right),
\]
where $\sigma_t$ satisfies the following iteration:
\begin{equation}
\label{eq:2-6}
\sigma _{t + 1}^2 = {\sigma_w^2} + \frac{1}{\delta }\mathbb{E}\left( {{{\left| {{\eta}_{p} (X + {\sigma _t}Z;{\lambda _t}) - X} \right|}^2}} \right).
\end{equation}
Here the expected value is with respect to two independent random variables $Z \sim N(0,1)$ and $X\sim p_X$.\footnote{ $\sigma_0^2$ depends on the initialization of the algorithm. If $\ell_p$-AMP is initialized at zero then $\sigma_0^2 = \frac{\mathbb{E} (X^2)}{\delta}$. }
\end{theorem}

The proof of this statement is presented in Section \ref{proof:thm1}. Note that $\sigma_t^2$ only depends on $\sigma_{t-1}^2$ and the selected threshold value at iteration $t-1$. This important feature of AMP will be used later in our paper. $\sigma_t$ and the relation between $\sigma_t$ and $\sigma_{t-1}$ are called {\em state} of $\ell_p$-AMP and {\em state evolution}, respectively.

\subsection{Summary and organization of the paper}
In this paper, we consider $\ell_p$-AMP as a heuristic algorithm for solving $\ell_p$-minimization and analyze its performance through the state evolution. We then use Replica method to connect the $\ell_p$-AMP estimates to the solution of \eqref{eq:ell0minimization}. Our analysis examines the correctness of all folklore theorems discussed in Section \ref{sec:intro}. The remainder of this paper is organized as follows: Section \ref{sec:optAMP} introduces the optimally tuned $\ell_p$-AMP algorithm and the optimal $p$-continuation strategy. Sections \ref{sec:maincontribution} and \ref{sec:analysisnoisy} formally present our main contributions. Section \ref{sec:discussion} discusses our results and their connection with the $\ell_p$-regularized least squares problem defined  in \eqref{eq:ell0minimization}. Section \ref{sec:proof} is devoted to the proof of our main contributions. Section \ref{sec:simulation} demonstrates how we can implement the optimally tuned $\ell_p$-AMP in practice and studies some of the properties of this algorithms. Section \ref{sec:conclusion} concludes the paper.

\section{Optimal $\ell_p$-AMP}\label{sec:optAMP}
\subsection{Roadmap}
The performance of $\ell_p$-AMP depends on the choice of the threshold parameters $\lambda_t$. Any fair comparison between $\ell_p$-AMP for different values of $p$ must take this fact into account. In this section we start by explaining how we set the parameters $\lambda_t$. Then in Section \ref{sec:maincontribution} we analyze $\ell_p$-AMP.

\subsection{Fixed points of state evolution}\label{sec:fpthreshpolicy}
According to the state evolution in \eqref{eq:2-6}, the only difference among different iterations of $\ell_p$-AMP is the standard deviation $\sigma_t$. In the rest of the paper, for notational simplicity, instead of discussing a sequence of threshold values for $\ell_p$-AMP, we consider a \textit{thresholding policy}  that is defined as a function $\lambda( \sigma)$ for $\sigma \geq 0$. Given a thresholding policy, we can run $\ell_p$-AMP in the following way:
\begin{eqnarray}
\label{eq:ellpampthresholdpolicy}
{ x^t} &=& \tilde{\eta}_{p,h} ({A^T}{z^{t - 1}} + {x^{t - 1}};\lambda ({\sigma}_{t-1})), \nonumber \\
{z^t} &=& y - A{x^t} + {z^{t - 1}}\frac{1}{\delta }\left\langle {\tilde{\eta}_{p,h}' }({A^T}{z^{t - 1}} + {{x}^{t - 1}};\lambda({\sigma}_{t-1})) \right\rangle.
\end{eqnarray}
In practice, $\sigma_t^2$ is not known, but can be estimated accurately \cite{maleki2010approximate}. We will mention an estimate of $\sigma_t^2$ in the simulation section.  Note that by making this assumption, we have imposed a constraint on the threshold values. In Section \ref{ssec:discussionwhypolicy} we will show that for the purpose of this paper considering thresholding policies only, does not degrade the performance of $\ell_p$-AMP.

According to Theorem \ref{conj:se}, the performance of $\ell_p$-AMP in \eqref{eq:ellpampthresholdpolicy} (in the limit $h_i \rightarrow 0$) can be predicted by the following state evolution:
\begin{eqnarray}\label{eq:seiterthreshpolicy}
\sigma_{t+1}^2 =  {\sigma^2_w} + \frac{1}{\delta }\mathbb{E}\left( {{{\left| {\eta_p (X + {\sigma_t }Z;{\lambda(\sigma_t)}) - X} \right|}^2}} \right).
\end{eqnarray}

Inspired by the state evolution equation, we define the following function: 
\begin{eqnarray}\label{eq:defininPsi1}
\Psi_{\lambda, p} (\sigma^2) &\triangleq&  \sigma_w^2+ \frac{1}{\delta }\mathbb{E}\left( {{{\left| {\eta_p (X + {\sigma}Z;{\lambda(\sigma)}) - X} \right|}^2}} \right).
\end{eqnarray}

It is straightforward to confirm that the iterations of \eqref{eq:seiterthreshpolicy} converge to a fixed point of $\Psi_{\lambda, p}(\sigma^2)$. There are a few points that we would like to highlight about the fixed points of $\Psi_{\lambda,p}(\sigma^2)$:

\begin{itemize}

\item[(i)] $\Psi_{\lambda, p}(\sigma^2)$ usually has more than one fixed point. If so, the fixed point $\ell_p$-AMP converges to depends on the initialization of the algorithm. This is depicted in Figure \ref{fig:fixedpoints}(a). 

\item[(ii)]Lower fixed points correspond to better recoveries. To see this, consider the two fixed points $\sigma_{f_1}$ and $\sigma_{f_2}$ in Figure \ref{fig:fixedpoints}(a). Call the corresponding estimates of AMP $x^{\infty, 1}$ and $x^{\infty,2}$. According to Theorem \ref{conj:se}, the mean square errors of these two estimates (as $N \rightarrow \infty$ and $h \rightarrow 0$) converge to $\mathbb{E} \left( {{{\left| {\eta_p (X + {\sigma_{f_1}}Z;{\lambda(\sigma_{f_1})}) - X} \right|}^2}} \right) $ and $ \mathbb{E} \left( {{{\left| {\eta_p (X + {\sigma_{f_1}}Z;{\lambda(\sigma_{f_2})}) - X} \right|}^2}} \right)$, respectively. Furthermore, note that since both of them are fixed points we have
\begin{eqnarray*}
\sigma_w^2 + \frac{1}{\delta }\mathbb{E}\left( {{{\left| {\eta_p (X + {\sigma_{f_1}}Z;{\lambda(\sigma_{f_1})}) - X} \right|}^2}} \right)= \sigma_{f_1}^2 < \sigma_{f_2}^2 =  \sigma_w^2 + \frac{1}{\delta }\mathbb{E}\left( {{{\left| {\eta_p (X + {\sigma_{f_2}}Z;{\lambda(\sigma_{f_2})}) - X} \right|}^2}} \right).
\end{eqnarray*}
Hence
\[
\mathbb{E}\left( {{{\left| {\eta_p (X + {\sigma_{f_1}}Z;{\lambda(\sigma_{f_1})}) - X} \right|}^2}} \right) <\mathbb{E}\left( {{{\left| {\eta_p (X + {\sigma_{f_2}}Z;{\lambda(\sigma_{f_2})}) - X} \right|}^2}} \right).
\]
Therefore, the lower fixed points lead to smaller mean square reconstruction errors.

 \item[ (iii)] Two of the fixed points of $\Psi_{\lambda, p}(\sigma^2)$ are of particular interest in this paper: (1) The lowest fixed point: this fixed point indicates the performance one can achieve from $\ell_p$-AMP under the best initialization. As we will discuss later this fixed point is also related to the solution of LPLS. (2) The highest fixed point: the performance $\ell_p$-AMP exhibits under the worst initialization. 
\begin{figure}
\begin{center}
\includegraphics[width= 15cm]{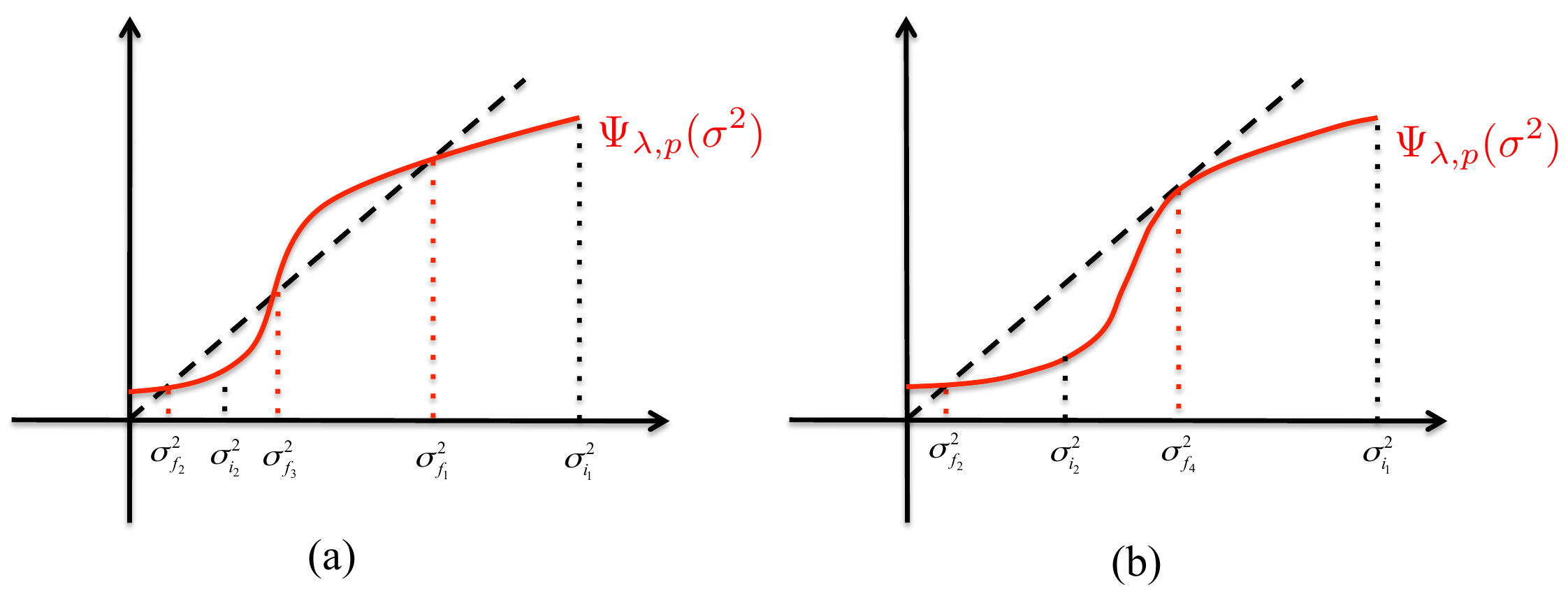}
\caption{The shapes of $\Psi_{\lambda, p} (\sigma^2)$ and its fixed points. (a) If  AMP is initialized at $ \sigma_0^2 =\sigma_{i_1}^2$, then $\lim_{t \rightarrow \infty} \sigma_t^2 = \sigma_{f_1}^2$. However, if $\sigma_0^2 = \sigma_{i_2}^2$, then $\lim_{t \rightarrow \infty} \sigma_t^2 = \sigma_{f_2}^2$. According to Definitions \ref{def:stablefp} and \ref{def:unstablefp}$,\sigma_{f_1}$ and $\sigma_{f_2}$ are stable fixed points, while $\sigma_{f_3}$ is the unstable fixed point. (b)  $\sigma_{f_4}$ is a half-stable fixed point: The algorithm will converge to this fixed point, if it starts in its right neighborhood. Here $\sigma_{f_2}$ is again a stable fixed point. }
\label{fig:fixedpoints}
\end{center}
\end{figure}

\item[(iv)] The shape of $\Psi_{\lambda,p}$ and its fixed points depend on the distribution $p_X$. In this work we study $p_X \in \mathcal{F}_\epsilon$, where $\mathcal{F}_\epsilon$ denotes the set of distributions whose mass at zero is greater than or equal to $1- \epsilon$. In other words, $X \sim p_X$ implies that $P(X\neq 0) \leq \epsilon$. This class of distributions has been studied in many other papers \cite{donoho2009message, rangan2009asymptotic, krzakala2012statistical, maleki2010optimally} and is considered as a good model for exactly sparse signals. 
\end{itemize}

Before discussing the optimal thresholding policy, we should distinguish between three types of fixed points: (i) stable, (ii) unstable, (iii) half-stable. The following definitions can be used for any function of $\sigma^2$, but we introduce them for $\Psi_{\lambda, p}(\sigma^2)$ to avoid introducing new notations.  

\vspace{.3cm}

\begin{definition}\label{def:stablefp}
$\sigma_f$ is called a stable fixed point of $\Psi_{\lambda, p}(\sigma^2)$ if and only if there exists an open interval $I$, with $\sigma_f \in I$, such that for every $\sigma> \sigma_f$ in $I$, $ \Psi_{\lambda,p}(\sigma^2) < \sigma^2$ and for every $\sigma< \sigma_f$ in $I$, $\Psi_{\lambda,p}(\sigma^2)> \sigma^2$. We call $0$ a stable fixed point of $\Psi_{\lambda,p}(\sigma^2)$ if and only if $\Psi_{\lambda, p}(0)=0$ and there exists $\sigma_i>0$ such that for every $0<\sigma <\sigma_i$, $ \Psi_{\lambda,p}(\sigma^2) < \sigma^2$.
\end{definition}

\vspace{.3cm}

In Figure \ref{fig:fixedpoints}(a), both $\sigma_{f_1}$ and $\sigma_{f_2}$ are stable fixed points, while $\sigma_{f_3}$ is not stable. The main feature of a stable fixed point is the following: There exists a neighborhood of $\sigma_f$ in which if we initialize $\ell_p$-AMP, it will converge to $\sigma_f$.\footnote{Note that all the statements we make about AMP are concerned with the asymptotic settings.}

\vspace{.3cm}

\begin{definition}\label{def:unstablefp}
$\sigma_f$ is called an unstable fixed point of $\Psi_{\lambda,p}(\sigma^2)$ if and only if there exists an open interval $I$, with $\sigma_f \in I$ such that for every $\sigma> \sigma_f$ in $I$, $ \Psi_{\lambda,p}(\sigma^2) > \sigma^2$ and for every $\sigma< \sigma_f$  in $I$, $\Psi_{\lambda,p}(\sigma^2)< \sigma^2$. We call $0$ an ustable fixed point of $\Psi_{\lambda,p}(\sigma^2)$ if and only if $\Psi_{\lambda,p}(0)=0$ and there exists $\sigma_i$ such that for every $0<\sigma <\sigma_i$, $ \Psi_{\lambda,p}(\sigma^2) > \sigma^2$.
\end{definition}

\vspace{.2cm}

Note that the state evolution equation will not converge to an unstable fixed point unless it is exactly initialized at that point. Hence, in realistic situations, $\ell_p$-AMP will not converge to unstable fixed points.  

\vspace{.3cm}

\begin{definition}
 A fixed point is called {\em half-stable} if it is neither stable nor unstable.  
 \end{definition}

\vspace{.3cm}

 See Figure \ref{fig:fixedpoints}(b) for an example of a half-stable fixed point. Half stable fixed points occur in very rare situations and for very specific noise levels $\sigma_w^2$.

\subsection{Optimal-$\lambda$ $\ell_p$-AMP}\label{sec:optlambda}

In the last section, we discussed the role of the fixed points of $\Psi_{\lambda,p}(\sigma^2)$ on the performance of $\ell_p$-AMP. Note that the locations of the fixed points of $\Psi_{\lambda,p}(\sigma^2)$ depend on the thresholding policy $\lambda(\sigma)$. Hence, it is important to pick $\lambda(\sigma)$ optimally. Consider the following oracle thresholding policy:
\begin{equation}\label{eq:optlambdadef}
\lambda_*(\sigma) \in \arg \min_{\lambda} \mathbb{E} \left( {{{\left| {\eta_p (X + {\sigma }Z;{\lambda}) - X} \right|}^2}} \right),
\end{equation}
where the expected value is with respect to two independent random variables $X \sim p_X$ and $Z \sim N(0,1)$. $\lambda_*(\sigma)$ is called oracle thresholding policy, since it depends on $p_X$ that is not available in practice. In Section \ref{sec:simulation}, we explain how this thresholding policy can be implemented in practice.  The following lemma is a simple corollary of our definition. 

\vspace{.3cm}

\begin{lemma}\label{lem:optimalityoracle}
For every thresholding policy $\lambda(\sigma)$, we have
\[
\Psi_{\lambda,p}(\sigma^2) \geq \Psi_{\lambda_*,p}(\sigma^2).
\]
Hence, both the lowest and highest stable fixed points of $\Psi_{\lambda_*,p}(\sigma^2)$ are below the corresponding fixed points of $\Psi_{\lambda,p}(\sigma^2)$. 
\end{lemma}

\vspace{.3cm}

The proof of the above lemma is a simple implication of the definition of oracle thresholding policy in \eqref{eq:optlambdadef} and is hence skipped here. According to this lemma, the oracle thresholding policy is an optimal thresholding policy since it leads to the lowest fixed point possible. In the rest of the paper, we call $\lambda_*(\sigma)$ the {\em optimal thresholding policy}. Also, the $\ell_p$-AMP algorithm that employs the optimal thresholding policy is called \textit{optimal-$\lambda$ $\ell_p$}-AMP. The optimal thresholding policy can be calculated numerically. Figure \ref{fig:3} exhibits $\Psi_{\lambda_*,p}(\sigma^2)$ for $p=0, 0.3, 0.6, 0.9, 1$ when the nonzero entries of the sparse vector $x_o$ are $\pm 1$ with probability $0.5$. It turns out that $\Psi_{\lambda_*, p}(\sigma^2)$ has at least one stable fixed point. The following proposition proves this claim. 

\begin{figure}
  \centering
  \includegraphics[width=3.2in]{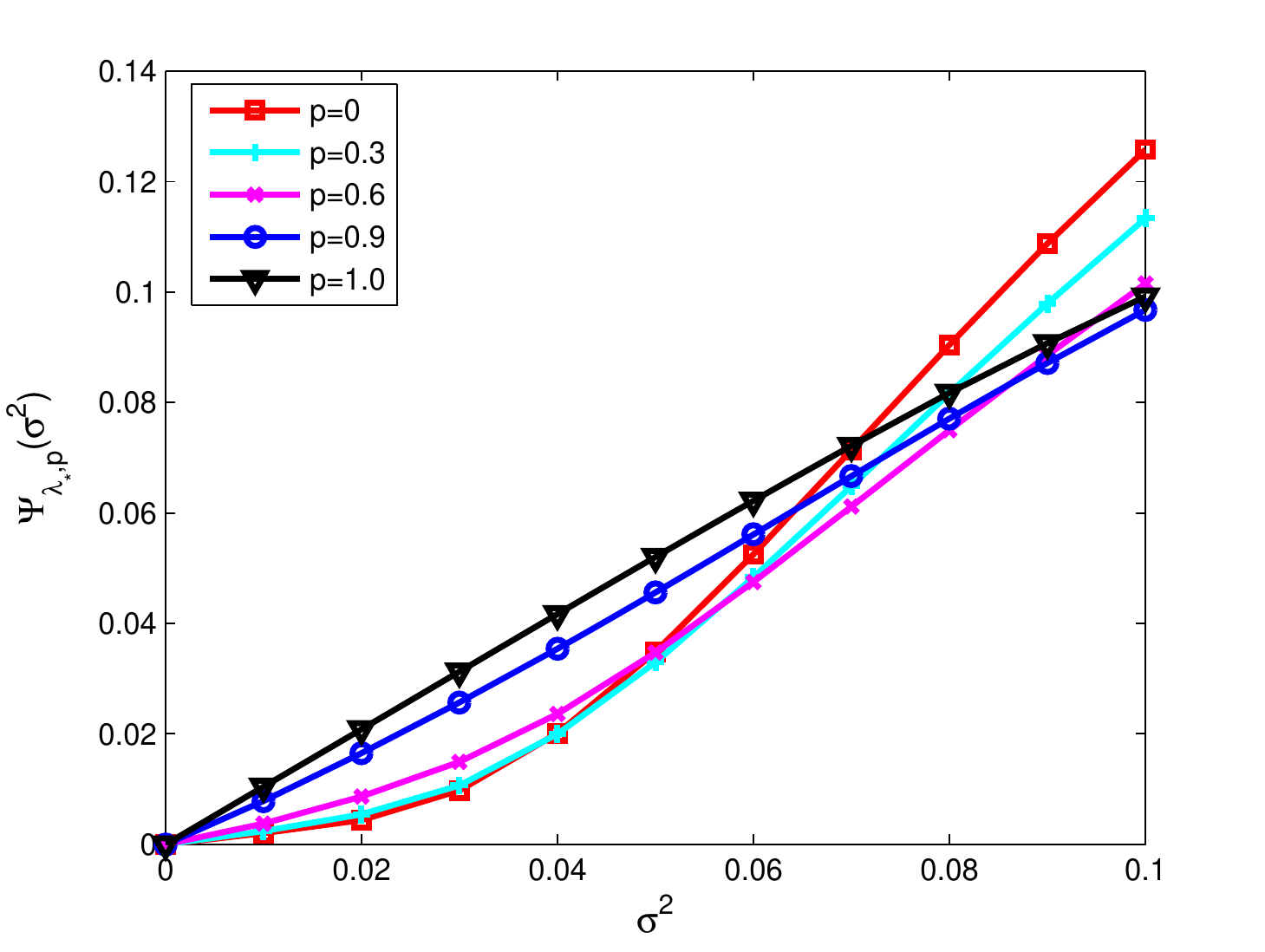}
   \caption{Comparison of $\Psi_{\lambda_*, p} (\sigma^2)$ for $p=0, 0.3, 0.6, 0.9, 1$. The undersampling factor and sparsity are $\delta=0.1$ and $\epsilon=0.02$, respectively.  The non-zero entries are $\pm 1$ equiprobable.  }
   \label{fig:3}
\end{figure}

\vspace{.2cm}

\begin{proposition}\label{proof:existsncestable}
$\Psi_{\lambda_*,p}(\sigma^2)$ have at least one stable fixed point.
\end{proposition}

\vspace{.3cm}
The proof of this statement is presented in Section \ref{sec:proofLemmaexistence}.

\subsection{Optimal-($p,\lambda$) $\ell_p$-AMP}\label{sec:optimalpandlambda}

In the last section, we fixed $p$ and optimized over the threshold parameter $\lambda_t$. However, one can also consider $p \in [0,1]$ as a free parameter that can be tuned at every iteration. This extra degree of freedom, if employed optimally, can potentially improve the performance of $\ell_p$-AMP. To derive the optimal choice of $p$, we first extend the notion of thresholding policy to adaptation policy.  The {\em adaptation policy}  is defined as a tuple $(\lambda(\sigma), p(\sigma))$ where $\lambda: [0, \infty) \rightarrow [0, \infty)$ and $p: [0, \infty) \rightarrow [0,1]$.

Given an adaptation policy, one can run the $\ell_p$-AMP algorithm whose performance in the asymptotic setting can be predicted by the following state evolution equation:
\begin{equation}\label{eq:pcontse}
\sigma_{t+1}^2 =  {\sigma^2_w} + \frac{1}{\delta }\mathbb{E}\left( {{{\left| {\eta_{p(\sigma_t) }(X + {\sigma_t }Z;{\lambda(\sigma_t)}) - X} \right|}^2}} \right).
\end{equation}
Hence the state evolution converges to one of the fixed points of $\Psi_{\lambda(\sigma),p(\sigma)}(\sigma^2)$. Adaptation policy can potentially improve the performance of the $\ell_p$-AMP algorithm. In this paper, we consider the following oracle adaptation policy:
\begin{equation}\label{eq:optadaptpolicy}
(\lambda_*(\sigma),p_*(\sigma) ) \in \arg \min_{\lambda, p} \mathbb{E}\left( {{{\left| {\eta_{p(\sigma)} (X + {\sigma }Z;{\lambda(\sigma)}) - X} \right|}^2}} \right).
\end{equation}
Note that obtaining $(\lambda_*(\sigma),p_*(\sigma) )$ requires the knowledge of $p_X$. We show how a good estimate of $(\lambda_*(\sigma),p_*(\sigma) )$ can be obtained without any knowledge of $p_X$ in Section \ref{sec:simulation}.  The $\ell_p$-AMP algorithm that employs $(\lambda_*(\sigma),p_*(\sigma) )$ is called {\em optimal-$(p,\lambda)$ $\ell_p$}-AMP.
The following lemma  clarifies this terminology:
\vspace{.2cm}

\begin{theorem}\label{thm:optadapt}
For any adaptation policy $(\lambda(\sigma), p(\sigma))$, we have
\[
\Psi_{\lambda_*,p_*} ({\sigma^2}) \leq \Psi_{\lambda, p}  ({\sigma^2}).
\]
\end{theorem}
The proof is a simple implication of \eqref{eq:optadaptpolicy} and is hence skipped. According to Theorem \ref{thm:optadapt}, the oracle adaptation policy is optimal and it outperforms every other adaptation policy. Hence we call it {\em optimal adaptation policy}. Note that in all situations, the optimal-($p,\lambda$) $\ell_p$-AMP outperforms the optimal-$\lambda$ $\ell_p$-AMP (for any $0 \leq p \leq 1$). In the next two sections, we characterize the amount of improvement that is gained from the optimal adaptation policy.

\begin{figure}
\begin{center}
  \includegraphics[width=3.2in]{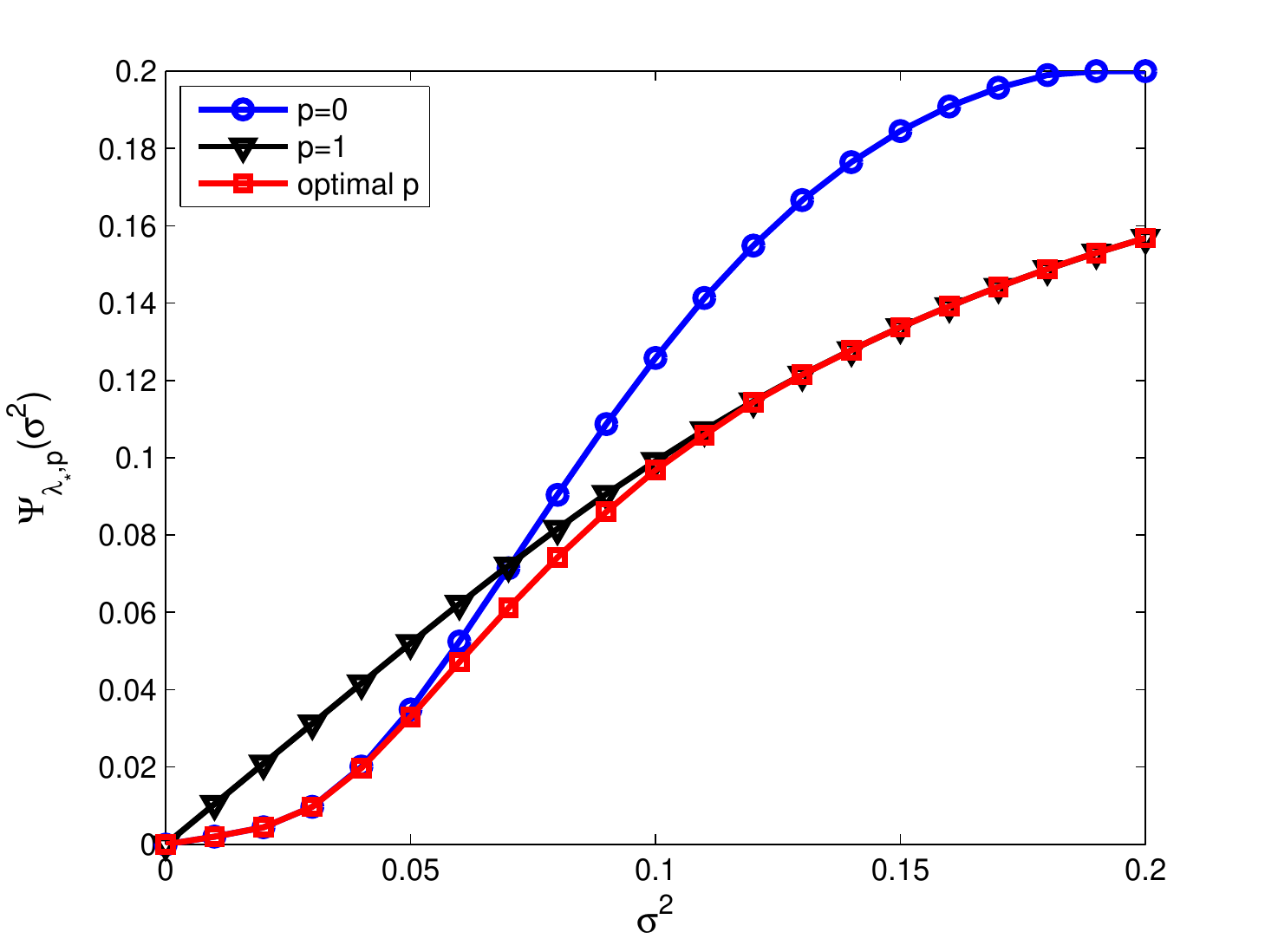}
  \caption{Comparison of (a) $\Psi_{\lambda_*,0}(\sigma^2)$, (b) $\Psi_{\lambda_*,1}(\sigma^2)$ and (c) $\Psi_{\lambda_*,p_*}(\sigma^2)$. $\delta$ and $\epsilon$ are set to $0.1$ and $0.02$, respectively. For small values of $\sigma$, $p=0$ is optimal and for large values of $\sigma$, $p=1$ is optimal. These two observations will be formally addressed in Proposition \ref{prop:optimallassolargesigma}, \ref{lem:riskbehsmallsigma} and \ref{lem:riskbehsmallsigmazero}.}
  \end{center}
\label{fig:5}
  \end{figure}

In this paper, we analyze the performance of $\ell_p$-AMP with optimal thresholding and adaptation policies. We will then employ the Replica method to show the implications of our results for LPLS. 

\subsection{Discussion about thresholding policy and adaptation policy}\label{ssec:discussionwhypolicy}

 Starting with an initialization, one may run $\ell_p$-AMP with thresholds $\lambda_1, \lambda_2, ....$  until the algorithm converges. $\lambda_t$ may depend on not only $\sigma_t$, but also the entire information about $\Psi$. In that sense, it is conceivable that one may pick the threshold in a way that he/she can beat $\ell_p$-AMP with optimal thresholding policy. Suppose that the lowest stable fixed point of $\Psi_{\lambda_*,p}(\sigma^2)$ is denoted with $\sigma_\ell^2$. Also, suppose that $\Psi_{\lambda_*,p}(\sigma^2)$ does not have any unstable fixed point below $\sigma_\ell$. Consider an oracle who runs $\ell_p$-AMP with a good initialization (whatever he/she wants) and picks a converging sequence $\lambda_1, \lambda_2, \ldots$ for the thresholds. Assume that the corresponding sequence of $\sigma_t$ converges to $\sigma_{\infty}$. It is then straightforward to show that no matter what threshold the oracle picks, he/she ends up with $\sigma_{\infty}\geq \sigma_\ell$. Hence, the lowest fixed point of $\Psi_{\lambda_*,p}(\sigma^2)$ specifies the best performance $\ell_p$-AMP offers.\footnote{The optimal thresholding policy has many more optimality properties, if $\eta_p$ satisfies the monotonicity property. For more information about monotonicity and its implications refer to \cite{metzler2014denoising}. We believe $\eta_p$ satisfies the monotonicity property, but have left the mathematical justification of this fact for future research. } 

Similarly, consider an oracle who runs $\ell_p$-AMP with a good choice of $(p_1, \lambda_1), (p_2,\lambda_2), \ldots $. Again we can argue that if $\sigma_t \rightarrow \sigma_{\infty}$, then $\sigma_\infty \geq \sigma_\ell$, where $\sigma_\ell$ denote the lowest fixed point of $\Psi_{\lambda_*,p_*}(\sigma^2)$. Note that considering $p$ as a free parameter and changing it at every iteration can be considered as a generalization of the continuation strategy we discussed in the introduction. Hence, $\sigma_\ell$ reflects the best performance any continuation strategy may achieve.

\section{Our contributions in noiseless settings}\label{sec:maincontribution}

\vspace{.3cm}
Table I summarizes all our contributions and the places they will appear.  This section discusses our main results in the noiseless setting $\sigma_w^2=0$. The discussion of the noisy setting is postponed until Section \ref{sec:analysisnoisy}.  We start with the optimal-$\lambda$ $\ell_p$-AMP. Since there is no measurement noise, the state of this system may converge to $0$, i.e., $\sigma_t \rightarrow 0$ as $t \rightarrow \infty$. If this happens, we say $\ell_p$-AMP has successfully recovered the sparse solution of $y=Ax_o$. Otherwise, we say $\ell_p$-AMP has failed. Depending on the under-determinacy value $\delta$, we may observe three different situations. 
\begin{itemize}
\item[(i)]  $\Psi_{\lambda_*, p}(\sigma^2)$ has only one stable fixed point at zero. In this case, optimal-$\lambda$ $\ell_p$-AMP is successful no matter where it is initialized.
\item[(ii)] $\Psi_{\lambda_*,p} (\sigma^2)$ has more than one stable fixed point, but $\sigma^2 =0$ is still a stable fixed point. In this case,  the performance of optimal-$\lambda$ $\ell_p$-AMP depends on its initialization. However, there exist initializations for which $\ell_p$-AMP is successful.
\item[(iii)] $0$ is not a stable fixed point of $\Psi_{\lambda_*, p}(\sigma^2)$. In such cases, optimal-$\lambda$ $\ell_p$-AMP does not recover the right solution under any initialization.
\end{itemize}

\begin{center}
\begin{table}
\caption{Summary of our findings in both noiseless and noisy settings}
\centering
\hspace{-1.5cm}
\begin{threeparttable}
\begin{tabular}{c|c}
\hline
Noiseless setting $(\sigma_w =0)$& Noisy setting $(\sigma_w>0)$\\
\hline
\shortstack{Phase transition curve for highest fixed point of\\
 optimal-$\lambda$ $\ell_p$-AMP under least favorable distribution: \\ $\overline{M}_p(\epsilon)=\delta$~ ($0\leq p <1$, \textbf{Theorem 3})\\ $M_1(\epsilon)=\delta$~ ($p=1$, \textbf{Proposition 2})  \\ Conclusion: minor improvement of $\ell_p$ over $\ell_1$. \\ ~ }   & \shortstack{ \\ \\ Noise sensitivity for highest fixed point in\\
optimal-$\lambda$ $\ell_p$-AMP under least favorable distribution:\\
 \\  $\sigma^2_h \leq \frac{\sigma^2_w}{1-\frac{\overline{M}_p(\epsilon)}{\delta}}$~($0\leq p \leq 1$, \textbf{Theorem 10}) \\  Conclusion: minor improvement of $\ell_p$ over $\ell_1$. \\ ~ } \\
 \hline
 \shortstack{ Phase transition curve for lowest fixed point in \\ optimal-$\lambda$ $\ell_p$-AMP: \\ $\epsilon =\delta$ ~ ($0\leq p <1$, \textbf{Theorem 4})\\ $M_1(\epsilon)=\delta$~ ($p=1$, \textbf{Proposition 2}) \\ Conclusion: major improvement of $\ell_p$ over $\ell_1$. \\ ~ \\~ } & \shortstack{ \\ \\ Noise sensitivity for lowest fixed point in optimal-$\lambda$ $\ell_p$-AMP: \\  $\lim_{\sigma_w \rightarrow 0}\frac{\sigma^2_{\ell}}{\sigma_w^2} = \frac{1}{1-\frac{\epsilon}{\delta}}$~($0\leq p <1$, \textbf{Theorem 6}) \\ $\lim_{\sigma_w \rightarrow 0}\frac{\sigma^2_{\ell}}{\sigma_w^2} = \frac{1}{1-\frac{M_1(\epsilon)}{\delta}}$~($p=1$, \textbf{Theorem 7})\\   Conclusion: major improvement of $\ell_p$ over $\ell_1$. \\ ~ } \\
 \hline
 \shortstack{Phase transition curve for highest fixed point in \\ optimal-$(p,\lambda)$ $\ell_p$-AMP under least favorable distribution:\\ $\inf_{0\leq p \leq 1} \overline{M}_p(\epsilon)=\delta$~(\textbf{Theorem 5})\\ Conclusion: minor improvement over $\ell_1$. \\ ~ \\ ~ \\ ~ \\ ~ \\ ~ \\ ~}  &  \shortstack{ \\ \\ Noise sensitivity for highest fixed point in \\ optimal-$(p, \lambda)$ $\ell_p$-AMP under least favorable distribution: \\  $\lim_{\sigma_w \rightarrow 0}\frac{\sigma^2_h}{\sigma_w^2} = \frac{1}{1-\frac{\epsilon}{\delta}}$~(\textbf{Theorem 11}) \\   $\lim_{\sigma_w \rightarrow 0}\frac{\sigma^2_h}{\sigma_w^2} = \frac{1}{1-\frac{M_1(\epsilon)}{\delta}}$~($p=1$, \textbf{Theorem 7}) \\ Conclusion: major improvement over $\ell_1$. \\~ } \\
 \hline
\end{tabular}
\begin{tablenotes}
\item $^*$Note that all the results shown for the highest fixed point are sharp for the least favorable distributions (based on the identity $\overline{M}_p(\epsilon)=\underline{M}_p(\epsilon)$ confirmed by our simulations). However, for some specific distribution we may observe major improvement of $\ell_p$ over $\ell_1$. See Figure \ref{fig:ptGaussian} for further information. Also we will present a more accurate analysis for the case when $\sigma_w$ is either very small or very large. Refer to Proposition \ref{prop:optimallassolargesigma} and Theorems  \ref{thm:riskbehsmallsigma} and \ref{thm:lownoisehardthresh1} for further information. 
\end{tablenotes}
\end{threeparttable}
\end{table}
\end{center}

These three cases are summarized in Figure \ref{fig:zerofp}. Our goal is to identify the conditions under which each of these cases happens. The following quantities will play a pivotal role in our results:
\begin{eqnarray}\label{eq:definitionMp}
{\overline{M}_p}(\epsilon) &\triangleq& \mathop {\inf }\limits_{\tau\geq0} \mathop {\sup }\limits_{\mu  \geq 0}  \left[ {(1 - \epsilon )\mathbb{E}{{\left( {{\eta _p}(Z;\tau )} \right)}^2} + \epsilon \mathbb{E}{{\left( {{\eta _p}(\mu  + Z;\tau ) - \mu } \right)}^2}} \right], \nonumber \\
{\underline{M}_p}(\epsilon) &\triangleq& \mathop {\sup }\limits_{\mu  \geq 0} \mathop {\inf }\limits_{\tau\geq0}  \left[ {(1 - \epsilon )\mathbb{E}{{\left( {{\eta _p}(Z;\tau )} \right)}^2} + \epsilon  \mathbb{E}{{\left( {{\eta _p}(\mu  + Z;\tau ) - \mu } \right)}^2}} \right],
\end{eqnarray}
where $\mathbb{E}$ is with respect to $Z \sim N(0,1)$. It is straightforward to confirm that ${\overline{M}_p}(\epsilon) \geq {\underline{M}_p}(\epsilon)$. 
Our next theorem explains the conditions that are required for Case (i) above. \\

\begin{figure}
\begin{center}
\includegraphics[width=14cm]{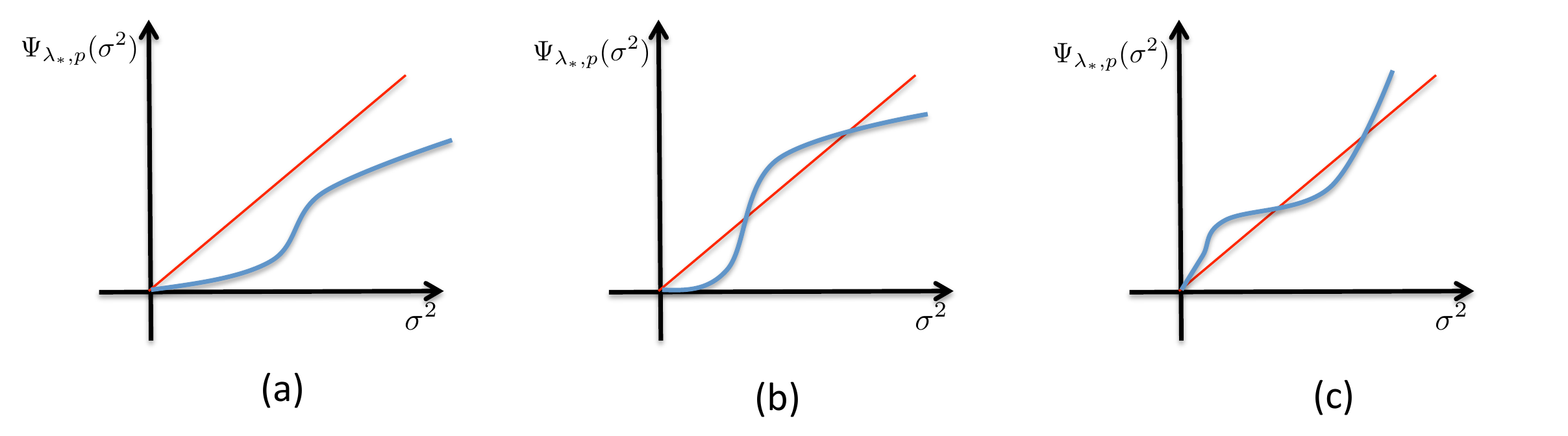}
\caption{Three main cases that may arise in the noiseless setting for $p<1$. (a) $\Psi_{\lambda_*, p}(\sigma^2)$ has only one stable fixed point at zero.  (b) $\Psi_{\lambda_*,p} (\sigma^2)$ has more than one stable fixed point, but $\sigma^2 =0$ is still a stable fixed point.  (c) $0$ is not a stable fixed point of $\Psi_{\lambda_*, p}(\sigma^2)$. }
\label{fig:zerofp}
\end{center}
\end{figure}

\begin{theorem}\label{thm:highestfpnoiseless}
Let $p_X \in \mathcal{F}_\epsilon$. If $\overline{M}_p (\epsilon)< \delta$, then the highest stable fixed point of the optimal-$\lambda$ state evolution happens at zero. In other words, $\Psi_{\lambda_*,p} (\sigma^2)$ has a unique stable fixed point at zero. Furthermore, if $\underline{M}_p(\epsilon)> \delta$, then there exists  a distribution $p_X \in \mathcal{F}_{\epsilon}$ for which $\Psi_{\lambda_*, p}(\sigma^2)$ has more than one stable fixed point.
\end{theorem}

\vspace{.2cm}
The proof of this theorem is summarized in Section \ref{proof:theorem3noiseless}. Note that this theorem is concerned with the minimax framework. In other words, the minimum value of $\delta$ for which $\Psi_{\lambda_*,p}(\sigma^2)$ has a unique fixed point  at zero depends on $p_X \in \mathcal{F}_{\epsilon}$. However, in most applications $p_X$ is not known and we would like to ensure that the algorithm works for any distributions. Theorem \ref{thm:highestfpnoiseless} ensures that under certain conditions, $\Psi_{\lambda_*, p}(\sigma^2)$ has a unique fixed point for any $p_X \in \mathcal{F}_{\epsilon}$. 

\vspace{.2cm}

Based  on Theorem \ref{thm:highestfpnoiseless}, we can discuss the first phase transition behavior of the optimal-$\lambda$ $\ell_p$-AMP algorithm. Let $p_X \in \mathcal{F}_\epsilon$. This phase transition behavior is discussed in the following corollary. 

\vspace{.2cm}

\begin{corollary}\label{cor:noiselessell_p}
For every $0 \leq p<1$ and $\delta$, there exists $\overline{\epsilon}^*_p(\delta)$ such that for every $\epsilon< \overline{\epsilon}^*_p(\delta)$, $\Psi_{\lambda_*,p} (\sigma^2)$ has only one stable fixed point at zero. Furthermore, there exists $\underline{\epsilon}^*_p(\delta)$ such that for every $\epsilon>\underline{\epsilon}^*_p(\delta)$,  $\Psi_{\lambda_*,p} (\sigma^2)$ has more than one stable fixed point for certain distributions in $\mathcal{F}_\epsilon$. 
\end{corollary}
The proof is presented in Section \ref{ssec:cornoiselessproof}.

\vspace{.2cm}

Our numerical results show that $\overline{\epsilon}_p^*(\delta) = \underline{\epsilon}_p^*(\delta)$ holds for every $0 \leq p \leq 1$. If we accept this identity, then Corollary \ref{cor:noiselessell_p} proves the first type of phase transition that we observe in $\ell_p$-AMP; for certain distributions the algorithm switches from having one stable fixed point to more than one stable fixed point at $\overline{\epsilon}_p^*(\delta)$. Figure \ref{fig:localpvs1} exhibits $\overline{\epsilon}_p^*(\delta)$ for several different values of $p$ ($\underline{\epsilon}_p^*(\delta)$ has the same value).  While our simulation results confirm that $\overline{\epsilon}_p^*(\delta) = \underline{\epsilon}_p^*(\delta)$ for every $0\leq p \leq 1$, we have only proved this observation for $p=1$. 

\begin{lemma}\label{lem:ell1pteq}
For $p=1$, we have $\overline{\epsilon}_1^*(\delta) = \underline{\epsilon}_1^*(\delta)$.
\end{lemma}

Proof of this claim will be presented in Section \ref{ssec:prooflemmaell1pteq}. Before we discuss and compare the phase transitions curves that are shown in Figure \ref{fig:localpvs1}, we discuss Case (ii) in which $\Psi_{\lambda_*, p}(\sigma^2)$ does not have a unique stable fixed point, but zero is still a stable fixed point.

\begin{figure}[htbp]
  \centering
  \includegraphics[width=3.2in]{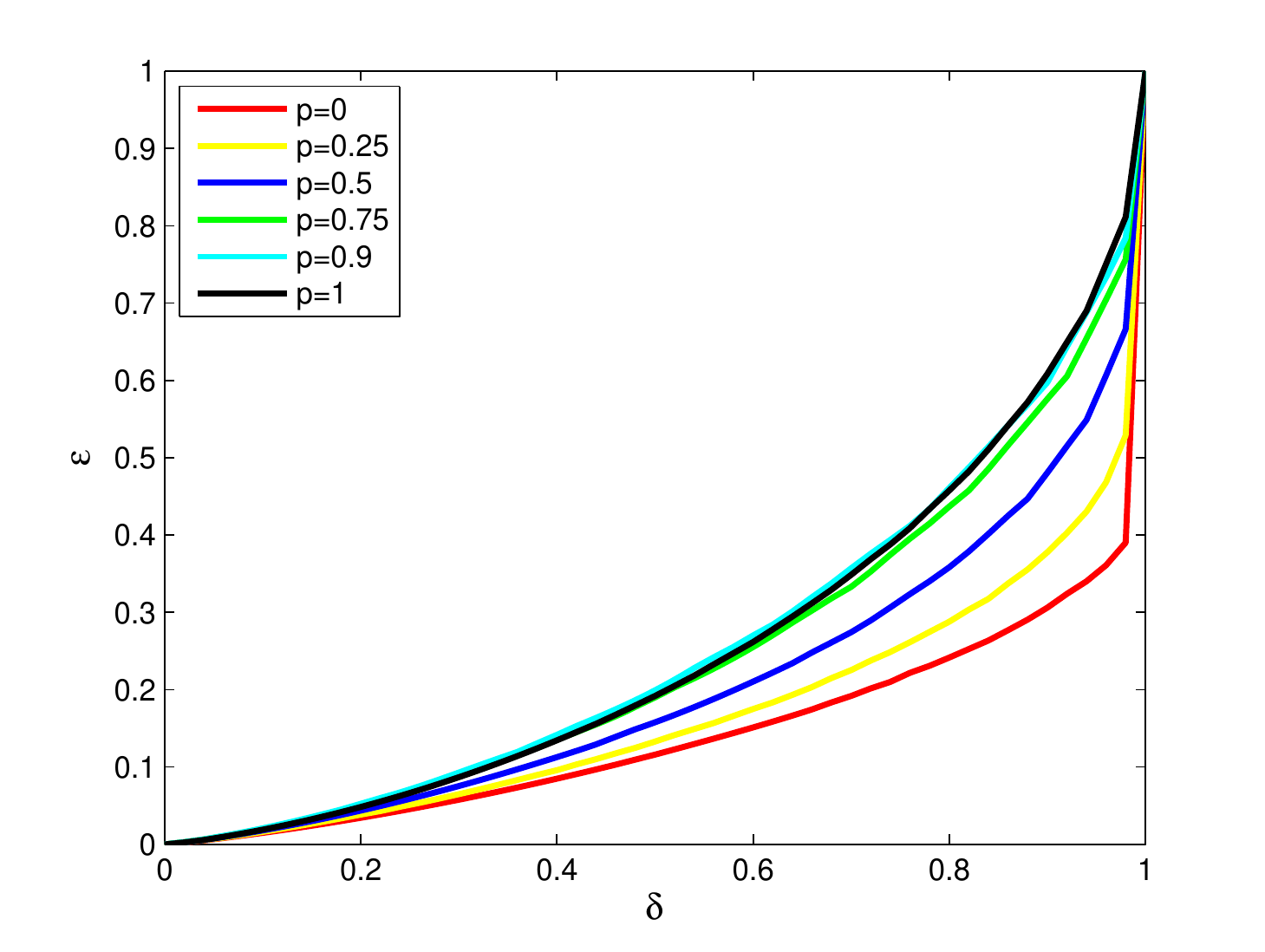}
  \caption{The value of $\overline{\epsilon}_p^*(\delta)$ as a function of $\delta$ for several different values of $p$. As is clear from the figure, the improvement gained from $\ell_p$-AMP  ($p<1$) is minor in the noiseless setting. Also, the values of $p$ that are close to $1$ are the only values of $p$ that can outperform $\ell_1$-AMP. $\ell_0$-AMP performs much worse than $\ell_1$-AMP. Note that  in these phase transition calculations we have assumed that we do not have access to a good initialization for optimal-$\lambda$ $\ell_p$-AMP.  Hence, these phase transitions are concerned with the behavior of $\ell_p$-AMP under the worst initializations. Theorem \ref{lem:noiselesslowfpless1} shows that under good initialization, optimal-$\lambda$ $\ell_p$-AMP outperforms optimal-$\lambda$ $\ell_1$-AMP by  a large margin. }
  \label{fig:localpvs1}
\end{figure}

\vspace{.2cm}

\begin{theorem}\label{lem:noiselesslowfpless1}
Let $p_X$ be an arbitrary distribution in $\mathcal{F}_{\epsilon}$. For any $ 0 \leq p<1$, $0$ is the lowest stable fixed point of $\Psi_{\lambda_*,p}(\sigma^2)$ if and only if $\delta > \epsilon$.
\end{theorem}

\vspace{.2cm}

This theorem is proved in Section \ref{proof sec:thmnoiselesslfp}. There are two main features of this theorem that we would like to emphasize.
\vspace{.2cm}

\begin{remark}
Compared to Theorem \ref{thm:highestfpnoiseless}, this theorem is universal in the sense that the actual distribution that is picked from $\mathcal{F}_{\epsilon}$ does not have any impact on the behavior of the fixed point at $0$. Furthermore, the number of measurements $\delta$  that is required for the stability of this fixed point is the same as the sparsity level $\epsilon$.  
\end{remark}

\vspace{.2cm}

\begin{remark}
As long as $\delta> \epsilon$, zero is a stable fixed point for every value of $p$. As we will see later in Section \ref{sec:discussion} (under the assumptions of Replica method), this fixed point gives the asymptotic results for the global minimizer of \eqref{eq:ell0minimization}. Therefore, for every $0 \leq p<1$, \eqref{eq:ell0minimization} recovers $x_o$ accurately as long as $\delta> \epsilon$. This result seems to be counter-intuitive; if we are concerned with the noiseless settings, all $\ell_p$-minimization algorithms are the same. We will shed some light on this surprising phenomenon in Section \ref{sec:analysisnoisy}, where we consider measurement noise. 
\end{remark}

\vspace{.2cm}

 To provide a fair comparison between optimally tuned $\ell_p$ and $\ell_1$-AMP algorithms, we study the performance of optimally tuned $\ell_1$-AMP in the following theorem. This result is similar to the results  proved in \cite{donoho2011noise}. Since for $p=1$ we have already showed $\overline{M}_1(\epsilon) = \underline{M}_1(\epsilon)$ in the proof of Lemma \ref{lem:ell1pteq}, in the rest of the paper we will use the notation $M_1(\epsilon)$ instead. 

\vspace{.3cm}

\noindent \begin{proposition}\label{prop:ell_1pt} 
Optimal-$\lambda$ $\ell_1$-AMP has a unique stable fixed point. Furthermore, in the noiseless setting and for every $p_X \in \mathcal{F}_\epsilon$, $0$ is the unique stable fixed point of $\Psi_{\lambda_*, 1} (\sigma^2)$ if and only if
\[
\delta >M_1(\epsilon),
\]
where $M_1(\epsilon)$ defined in \eqref{eq:definitionMp} with $p=1$ can be simplified to:
\begin{align}
&M_1(\epsilon)= \inf_{\tau\geq 0} (1- \epsilon)  \mathbb{E}( \eta_1^2(Z;\tau))+ \epsilon(1+ \tau^2). \nonumber
\end{align}
\end{proposition}

\vspace{.2cm}

We present the proof of this proposition in Section \ref{sec:proofprop1}.  Based on this result, we define the phase transition of the optimally tuned $\ell_1$-AMP. Denote
\[
\epsilon_1^*(\delta) \triangleq \sup\{\epsilon : M_1(\epsilon)<\delta\}.
\]

\begin{corollary}\label{cor:ell_1PT} In the noiseless setting, if $\epsilon< \epsilon_1^*(\delta)$, the state evolution of optimal-$\lambda$ $\ell_1$-AMP has only one stable fixed point at zero for every $p_X \in \mathcal{F}_{\epsilon}$. Furthermore, if $\epsilon> \epsilon_1^*(\delta)$, for every $p_X \in \mathcal{F}_\epsilon$  the fixed point at zero becomes unstable and it will have one non-zero stable fixed point.  
\end{corollary}

\vspace{.2cm}

The proof is a simple implication of Proposition \ref{prop:ell_1pt} and  is similar to the proof of Corollary \ref{cor:noiselessell_p}. Hence it is skipped here. We can now compare the performance of optimal-$\lambda$ $\ell_p$-AMP with optimal-$\lambda$ $\ell_1$-AMP. We first emphasize on the following points:

\begin{enumerate}
\item[(i)] Optimal-$\lambda$ $\ell_1$-AMP has only one stable fixed point, while in general optimal-$\lambda$ $\ell_p$-AMP has multiple stable fixed points.

\item[(ii)] In the noiseless setting, $0$ is a fixed point for both optimal-$\lambda$ $\ell_1$-AMP and optimal-$\lambda$ $\ell_p$-AMP. The stability of this fixed point only depends on sparsity level $\epsilon$ and does not depend on the specific choice of $p_X$ that is picked from $\mathcal{F}_{\epsilon}$. The range of the values of $\epsilon$ for which $0$ is a stable fixed point of optimal-$\lambda$ $\ell_p$-AMP is much wider than that of optimal-$\lambda$ $\ell_1$-AMP  as shown in Figure \ref{fig:globalpvs1}. 

\begin{figure}
\begin{center}
\includegraphics[width=3.2in]{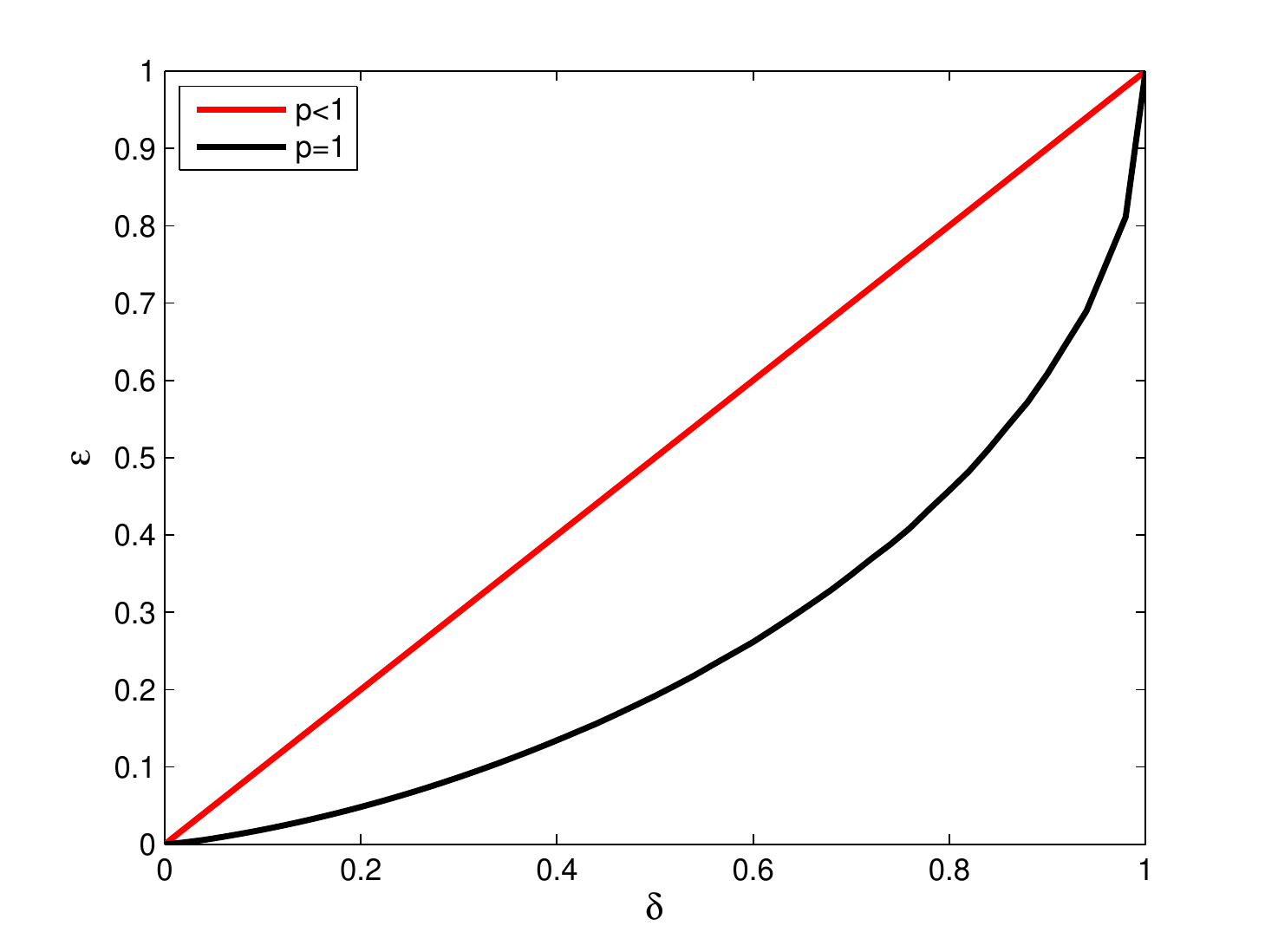}
\caption{Comparison of the best performance of optimal-$\lambda$ $\ell_p$-AMP for $p<1$ (under the best initialization) with optimal-$\lambda$ $\ell_1$-AMP. The phase transition is the same for every $p<1$. According to Replica method, the phase transition of  optimal-$\lambda$ $\ell_p$-AMP corresponds to the phase transition curve of the solution of \eqref{eq:ell0minimization}. }
\label{fig:globalpvs1}
\end{center}
\end{figure}

\item[(iii)] $\Psi_{\lambda_*,p} (\sigma^2)$ may have another stable fixed point in addition to $0$ for $0 \leq p<1$. The value of $\epsilon$ below which $\Psi_{\lambda_*,p} (\sigma^2)$ has only one stable fixed point at zero depends on the distribution $p_X$. Theorem \ref{thm:highestfpnoiseless} characterizes the condition under which for every $p_X \in \mathcal{F}_{\epsilon}$, zero is the unique fixed point. This specifies another phase transition for the $\ell_p$-AMP that we called $\overline{\epsilon}_p^*(\delta)$ (note that in this argument we are assuming the equality of $\overline{\epsilon}_p^*(\delta) = \underline{\epsilon}^*_p(\delta)$). These phase transition curves are exhibited in Figure \ref{fig:localpvs1}. As is clear from the figure, for small values of $p$, the corresponding phase transition curve falls much below the phase transition curve of optimally tuned $\ell_1$-AMP. For $p>0.9$, some improvement can be gained from $\ell_p$-AMP, but the improvement is marginal. 

\end{enumerate}

As is clear from the comparison of the phase transitions in Figure \ref{fig:localpvs1} and Figure \ref{fig:globalpvs1}, a good initialization can lead to major improvement in the performance of $\ell_p$-AMP for $p<1$. According to Folklore Theorem (iii) mentioned in Section \ref{sec:intro}, we expect $p$-continuation to provide such initialization. Hence, we study the performance of the optimal-$(p, \lambda)$ $\ell_p$-AMP. Refer to Section \ref{ssec:discussionwhypolicy} for more information on the connection of the optimal adaptation policy and $p$-continuation that we discussed in the introduction. 

\vspace{.3cm}

\begin{theorem}\label{thm:highestfpnoiselessoptp}
If $\inf_{0\leq p\leq 1} \overline{M}_p(\epsilon) <\delta$, then the highest stable fixed point of optimal-$(p,\lambda)$ $\ell_p$-AMP happens at zero. In other words, $\Psi_{\lambda_*,p_*} (\sigma^2)$ has a unique stable fixed point at zero. Furthermore, if 
$$\mathop {\sup }\limits_{\mu  \geq 0} \inf_{0 \leq p \leq 1} \mathop {\inf }\limits_{\tau\geq0}  \left[ {(1 - \epsilon )\mathbb{E}{{\left( {{\eta _p}(Z;\tau )} \right)}^2} + \epsilon  \mathbb{E}{{\left( {{\eta _p}(\mu  + Z;\tau ) - \mu } \right)}^2}} \right]>\delta,$$
then there exists  a distribution $p_X \in \mathcal{F}_{\epsilon}$ for which $\Psi_{\lambda_*, p_*}$ has an extra stable fixed point in addition to zero.
\end{theorem}

\vspace{.2cm}

The proof of this theorem is very similar to the proof of Theorem \ref{thm:highestfpnoiseless} and hence is skipped here. 
\vspace{.2cm}

\begin{corollary}\label{cor:continuationp}
$\Psi_{\lambda_*, p_*} (\sigma^2)$ has a unique stable fixed point at zero if $\epsilon < \sup_{0\leq p\leq1} \overline{\epsilon}_p^*(\delta)$. Furthermore, there exists $\underline{\epsilon}^{**}(\delta)$ such that if $\epsilon > \underline{\epsilon}^{**}(\delta)$, then for certain distribution $p_X \in \mathcal{F}_{\epsilon}$, $\Psi_{\lambda_*, p_*} (\sigma^2)$ has more than one stable fixed point. 
\end{corollary}

\vspace{.2cm}

The proof of this corollary is straightforward and is skipped. Again our {\em numerical calculations} confirm that
\[
\underline{\epsilon}^{**}(\delta)= \sup_{0\leq p\leq1} \overline{\epsilon}_p^*(\delta). 
\]
Corollary \ref{cor:continuationp} has a simple implication for adaptation policies (and also $p$-continuation). The performance of optimal-$(p, \lambda)$ $\ell_p$-AMP is the same as the performance of optimal-$\lambda$ $\ell_p$-AMP for the best value of $p$. In this sense, the only help that the optimal adaptation policy provides is to automatically find the best value of $p$ for running optimal-$\lambda$ $\ell_p$-AMP.\footnote{Note that in this paper we are only interested in one performance measure of $\ell_p$-AMP algorithms and that is the reconstruction error. Adaptation policy may improve the convergence rate of the algorithm. } Figure \ref{fig:optimaladaptvs1} compares the phase transition of optimal-$(p, \lambda)$ $\ell_p$-AMP with that of optimal-$\lambda$ $\ell_1$-AMP.  As we expected from Theorem \ref{thm:highestfpnoiselessoptp}, the improvement is minor. 
\begin{figure}
\begin{center}
\includegraphics[width=8cm]{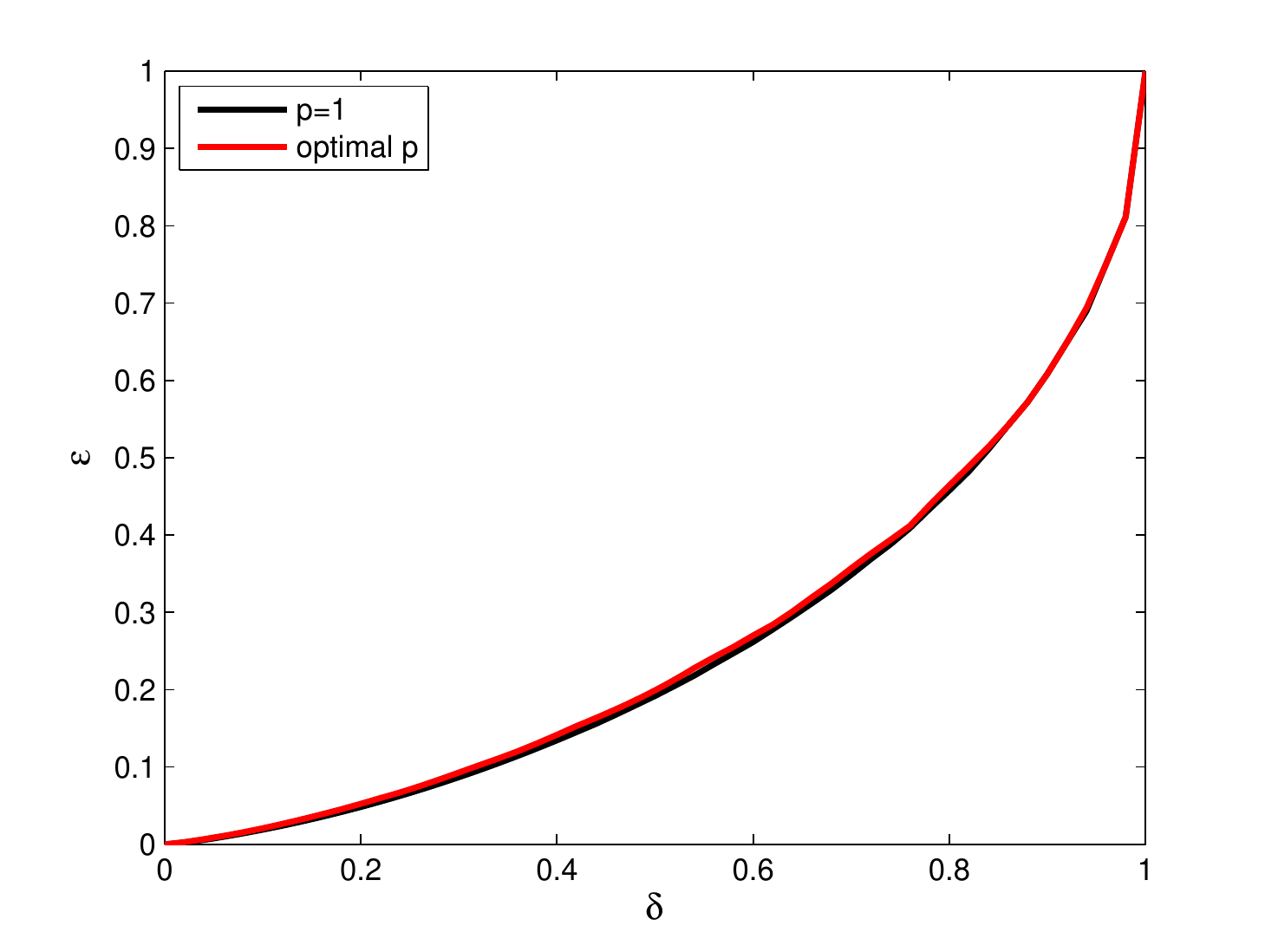}
\caption{Comparison of the phase transition of optimal-$\lambda$ $\ell_1$-AMP and optimal-($p,\lambda$) $\ell_p$-AMP under the minimax framework. The phase transition exhibited for optimal-($p$, $\lambda$) $\ell_p$-AMP is the value of $\epsilon$ at which the number of stable fixed points of optimal-($p$,$\lambda$) $\ell_p$-AMP changes from one to more than one for at least some prior $p_X \in \mathcal{F}_\epsilon$. }
\label{fig:optimaladaptvs1}
\end{center}
\end{figure}
\vspace{.2cm}

The results we have presented so far regarding the highest fixed point of $\ell_p$-AMP are disappointing. It seems that if we do not initialize the algorithm properly (and in practice in most cases we will not be able to do so), then the performance of the algorithm is at best slightly better than $\ell_1$-AMP. However, simulation results presented elsewhere have shown that iterative algorithms that aim to solve LPLS usually outperform LASSO. Such simulation results are not in contradiction with the result we present in this paper. In contrary, they can be explained with the framework we developed in our paper.  Let the distribution of $X$  be denoted by $X \sim (1-\epsilon) \Delta_0 + \epsilon G$, where $\Delta_0$ denotes a point mass at zero and $G$ denotes the distribution of the nonzero elements. According to Proposition \ref{prop:ell_1pt}, the phase transition curve of the optimal-$\lambda$ $\ell_1$-AMP is independent of $G$ and only depends on $\epsilon$. This is not true for the phase transition of $\ell_p$-AMP (the one derived based on the highest fixed point of $\Psi_{ \lambda_*,p} (\sigma^2)$). In fact, the results in Theorem \ref{thm:highestfpnoiseless} are obtained under the {\em least favorable distribution} which is a certain choice of $G$ that leads to the lowest phase transition of $\ell_p$-AMP possible. For other distributions, optimal-$\lambda$ $\ell_p$-AMP can provide a higher phase transition. Figure \ref{fig:ptGaussian} compares the phase transition (based on the highest fixed point) of optimal-$\lambda$ $\ell_p$-AMP with that of optimal-$\lambda$ $\ell_1$-AMP when $G = N(0,1)$. As is clear from this figure, such distributions usually favor $\ell_p$-AMP but not the $\ell_1$-AMP algorithm. Hence, we see that here $p=0.75$ has much higher phase transition than optimal-$\lambda$ $\ell_1$-AMP. 


It is important to note that for different distributions, different values of $p$ provide the best phase transition. However, if we employ optimal-$(p, \lambda)$ $\ell_p$-AMP, it will find the optimal value of $p$ automatically. Hence, even though the continuation strategy does not provide much improvement in the minimax setting, it can in fact offer a huge boost in the performance for practical applications.  

\begin{figure}[htbp]
  \centering
  \includegraphics[width=3.2in]{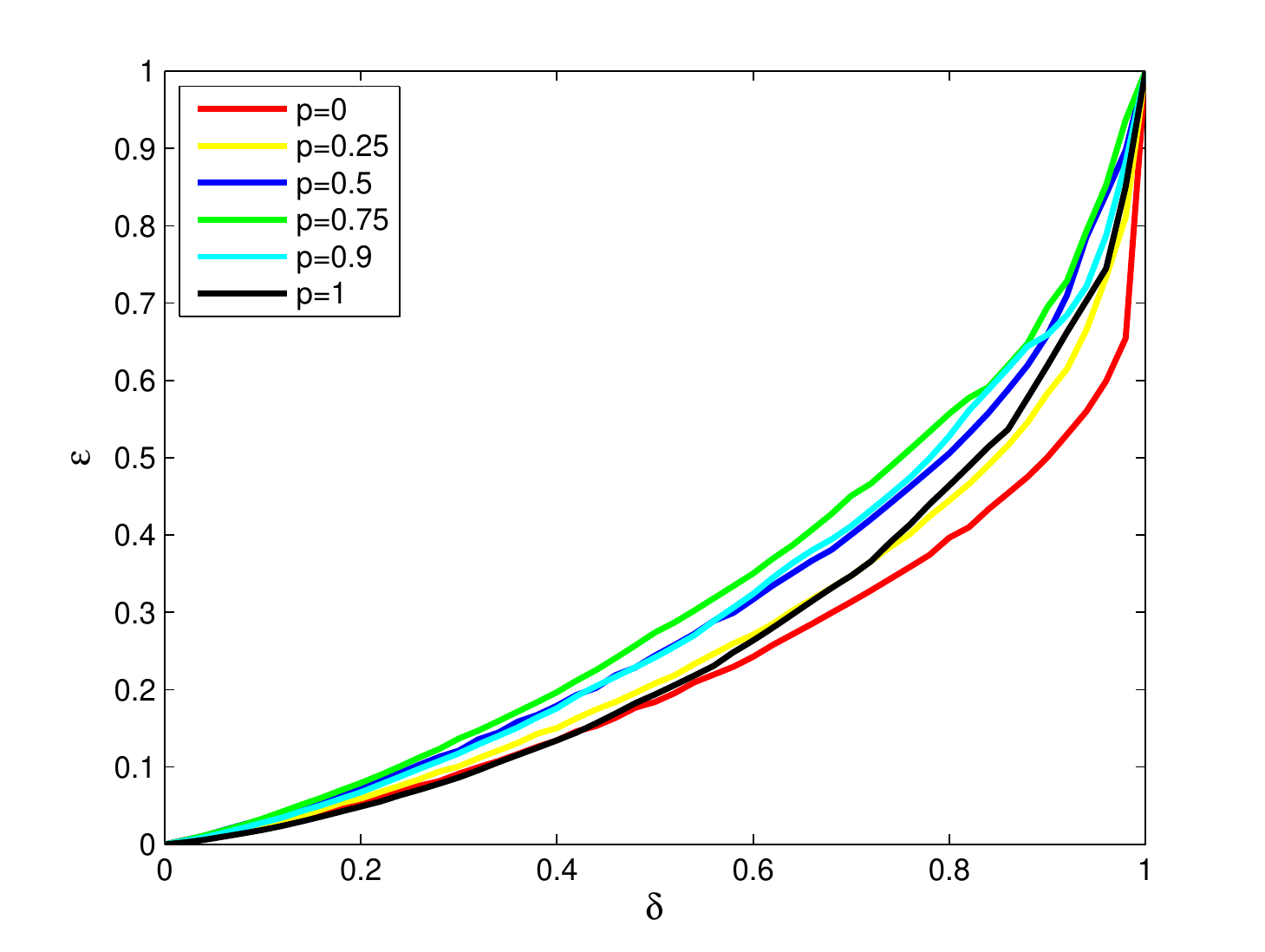}
  \caption{Phase transition curve for $p_X = (1- \epsilon)\Delta_0+ \epsilon G$, where $G$ is the PDF of the standard normal distribution. $p \in \{0, 0.25, 0.5, 0.75, 0.9, 1\}$ are considered in this simulation. These phase transition curves shall be compared with those of Figure \ref{fig:localpvs1}.}
  \label{fig:ptGaussian}
\end{figure}

\section{Our contributions in noisy setting}\label{sec:analysisnoisy}

\subsection{Roadmap}
In this section, we assume that $\sigma_w^2>0$. This implies that the reconstruction error of $\ell_p$-AMP is greater than zero for all $\ell_p$-AMPs. 
We start with analyzing the performance of optimal-$\lambda$ $\ell_p$-AMP. This corresponds to the analysis of the fixed points of $\Psi_{\lambda_*, p} (\sigma^2)$. Generally $\Psi_{\lambda_*, p} (\sigma^2)$ may have more than one stable fixed point. Similar to the last section, we study two of the fixed points of this function: (i) The lowest fixed point that corresponds to the performance of the algorithm under the best initialization, and (ii) the highest fixed point that corresponds to the performance of the algorithm under the worst initializations in Sections \ref{ssec:lfpsectionnoise} and \ref{ssec:hfpsectionnoise} respectively. We have empirically observed that under the initialization that we use, i.e., $x^0=0$, the algorithm converges to the highest fixed point. 

\subsection{Analysis of the lowest fixed point}\label{ssec:lfpsectionnoise}

In this section we study the lowest fixed point of optimally tuned $\ell_p$-AMP. We use the notation $\sigma_\ell$ for the lowest fixed point of $\Psi_{\lambda_*, p} (\sigma^2)$. 
Our first result is concerned with the performance of the algorithm for small amount of noise.

\begin{theorem}\label{thm:noisyoptimallambda}
If $\epsilon < \delta$, then there exists $\sigma^2_0$ such that for every $\sigma_w^2< \sigma^2_0$, $\sigma_{\ell}^2$ is a continuous function of $\sigma_w^2$. Furthermore,
\[
\lim_{\sigma_w^2 \rightarrow 0}  \frac{\sigma_{\ell}^2}{\sigma_w^2}= \frac{1}{1-\frac{\epsilon}{\delta}}.
\]
\end{theorem}

The proof is presented in Section \ref{sec:proofthmnoisyfirst}. It is instructive to compare this result with the corresponding result for the optimal-$\lambda$ $\ell_1$-AMP. \\

\begin{theorem}\label{thm:noisyell_1limit}
If $M_1(\epsilon) < \delta$, then the fixed point of optimal-$\lambda$ $\ell_1$-AMP is unique and satisfies
\[
\lim_{\sigma_w^2 \rightarrow 0} \frac{\sigma_\ell^2}{\sigma_w^2} = \frac{1}{1- M_1(\epsilon)/\delta}.
\]
\end{theorem}
This result can be derived from the results of \cite{donoho2011noise}. But for the sake of completeness and since we are using  a different thresholding policy, we present the proof in Section \ref{sec:proofnoisyell_1first}. \\

\begin{remark}\label{rem:M1}
In Section \ref{ssec:cornoiselessproof}, we show that $M_1(\epsilon) = \inf_\tau  (1-\epsilon) \mathbb{E} (\eta_1(Z; \tau))^2 + \epsilon(1+ \tau^2) > \epsilon$. Hence the performance of the lowest fixed point of optimal-$\lambda$ $\ell_p$-AMP is better than that of optimal-$\lambda$ $\ell_1$-AMP in the limit $\sigma_w^2 \rightarrow 0$. The continuity of $\sigma_\ell^2$ as a function of $\sigma_w^2$ implies that this comparison is still valid, for small values of $\sigma_w^2$. 
 \end{remark}
 
 \vspace{.2cm}
 
What happens as we keep increasing $\sigma_w^2$? Figure \ref{fig:sigma_noisy} that is based on our numerical calculations, answers this question. It compares $\sigma_{\ell}^2$ as a function of $\sigma_w^2$ for several different values of $p$. Two interesting phenomena can be observed in this figure: 
\begin{enumerate}
\item[(i)] Low-noise phenomenon: For small values of $\sigma_w$, the lowest fixed point of $\ell_0$-AMP outperforms the lowest fixed point of all the other $\ell_p$-AMP algorithms. Furthermore smaller values of $p$ seem to have advantage over the larger values of $p$. Note that Theorem \ref{thm:noisyoptimallambda} does not explain this observation. According to this theorem all values of $p=0$ seem to perform similarly. We will present a refinement of Theorem \ref{thm:noisyoptimallambda} in Theorems \ref{thm:riskbehsmallsigma} and \ref{thm:lownoisehardthresh1} that is capable of explaining this phenomenon.

\item[(ii)] High-noise phenomenon: For large values of $\sigma_w$, optimally tuned $\ell_1$-AMP outperforms even the lowest fixed point of $\ell_p$-AMP for every $p<1$. As we mentioned before we will connect the lowest fixed point of $\ell_p$ with the global minimizer of LPLS. This means that LASSO will outperform the global minimizer of $\ell_0$-regularized least squares for large values of noise. This is in contradiction with the first folklore we mentioned in the introduction. Proposition \ref{prop:optimallassolargesigma} will prove this observation. 

\end{enumerate}

\begin{figure}[htbp]
  \centering
  \includegraphics[width=3.2in]{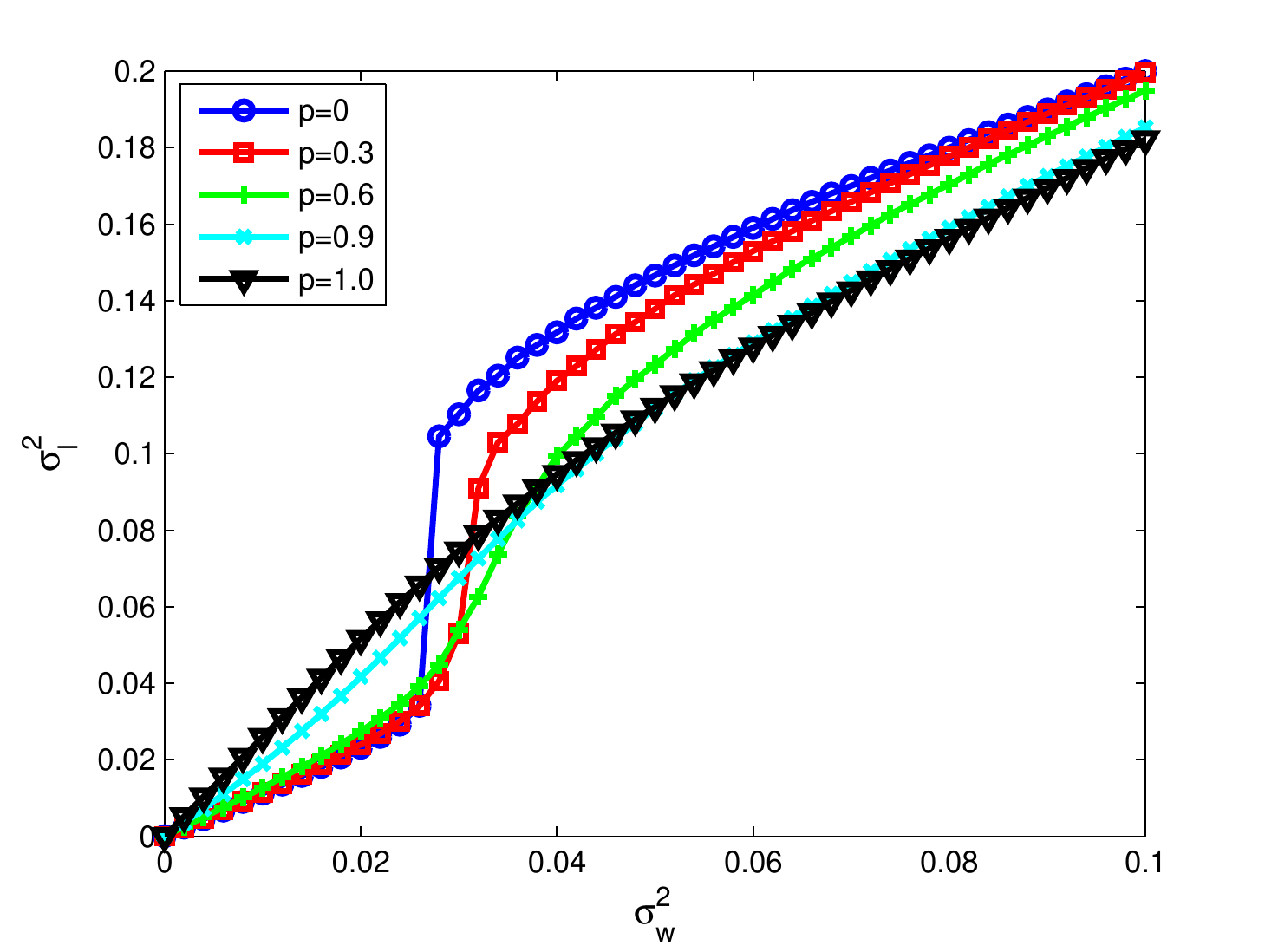}
  \caption{The curve of $\sigma_\ell^2$ as a function of  $\sigma_w^2$ for $p \in\{0, 0.3, 0.6, 0.9, 1\}$. Note that (i) for $p=0$, $\sigma_\ell^2$ is a discontinuous function of $\sigma_w^2$. (ii) For small values of $\sigma_w^2$, $p=0$ provides the smallest $\sigma_\ell^2$, while for large values of $\sigma_w^2$, $p=1$ exhibits the best performance. Here is the set-up for this simulation. $\delta$ and $\epsilon$ are set to $0.1$ and $0.01$ respectively. The non-zero elements of $x_o$ are iid $\pm1$ with probability $0.5$.}
  \label{fig:sigma_noisy}
\end{figure}


%
%
%
%
%
%
Below we justify both the low-noise and high-noise observations. The next two propositions are concerned with low noise phenomenon. According to Theorem \ref{thm:noisyoptimallambda}. We know that $\sigma_\ell^2/ \sigma_w^2 \rightarrow \frac{\delta}{\delta-\epsilon}$.  Hence, in order to see the discrepancy between different values of $p$, we have to explore how $\sigma_\ell^2 - \frac{\delta \sigma_w^2}{\delta-\epsilon}$ behaves for small values of $\sigma_w^2$. Let $X \sim (1-\epsilon)\Delta_0+ \epsilon G$. Let $U$ denote a random variable with distribution $G$. We also use the notation $\mathbb{E}_G(f(U)) \triangleq \int f(u) dG(u)$, and $\mathbb{P}_G(U \in {\cal A}) \triangleq \mathbb{E} (\mathbbm{1} (U \in {\cal A}))$, where $\mathbbm{1}$ denotes the indicator function. 
 
 \begin{theorem}\label{thm:riskbehsmallsigma}
 Suppose $P_G(|U|>\mu)=1$ with $\mu$ being a fixed positive number and $\mathbb{E}_G|U|^2 < \infty$, then for $0<p<1$ and $\epsilon <\delta$,
 \[
 \lim_{\sigma_w \rightarrow 0} \frac{\sigma_\ell^2 - \frac{\delta}{\delta- \epsilon}\sigma_w^2}{\sigma^{4-2p}_w (\log \frac{1}{\sigma_w})^{2-p}} = \frac{\epsilon c_p^{4-2p}p^2\mathbb{E}|U|^{2p-2}\delta^{2-p}}{(4-4p)^{2-p}(\delta-\epsilon)^{3-p}}.
 \]
  \end{theorem}
 The proof of this result can be found in Section \ref{sec:prooflemmasmallnoise}. Before we interpret this result, let us discuss the result for $p=0$ as well. Note that Theorem \ref{thm:riskbehsmallsigma} does not cover $p=0$ case. 
 
 \begin{theorem}\label{thm:lownoisehardthresh1}
 Suppose $\mathbb{E}_G|U|^2 < \infty$ and $P_G(|U|>\mu)=1$, where $\mu=\sup_v \{v : P(|U|>v)=1 \} >0$, then for $p=0$ and $\epsilon< \delta$,
\[
\sigma_\ell^2 = \frac{\delta}{\delta- \epsilon} \sigma_w^2 + o(\phi(\bar{\mu}\sigma^{-1}_w)),
\]
where $\bar{\mu}$ is any constant that is smaller than $\frac{\mu}{2}\sqrt{\frac{\delta-\epsilon}{\delta}}$.
 \end{theorem}
 The proof of this theorem is presented in Section \ref{ssec:prooflownosiehardthresh}. We now discuss how these theorems explain the {\em low-noise phenomenon} in Figure \ref{fig:sigma_noisy}. Suppose that we ignore all the logarithmic terms and study the second dominant term in the expressions of $\sigma_\ell$ that we derived in Theorem \ref{thm:riskbehsmallsigma}. There are two facts we should emphasize here: (i) The second dominant term is proportional to $\sigma_w^{4-2p}$, and is hence smaller for smaller values of $p$. (ii) The second dominant term is positive. If we combine these two facts, we conclude that if $p_1< p_2$, for small enough $\sigma_w$ the lowest fixed point of optimally tuned $\ell_{p_1}$-AMP outperforms optimally tuned $\ell_{p_2}$-AMP, which confirms our observation in Figure \ref{fig:sigma_noisy}.  More interestingly, according to  Theorem \ref{thm:lownoisehardthresh1} for the case $p=0$:
\[
\sigma^2_\ell=\frac{\delta}{\delta -\epsilon}\sigma^2_w+o(\phi(\bar{\mu}\sigma^{-1}_w)),
\]
Here, the second dominant term for $p=0$ decays exponentially faster than the polynomial rate for $p>0$. Hence $\ell_0$-AMP will outperform $\ell_p$-AMP for $p>0$ in low noise regime, which is again consistent with Figure \ref{fig:sigma_noisy}. Another interesting feature of this theorem is its implications for the values of $p$ that are less than $1$, but close to it. Figure \ref{fig:sigma_noisy} shows that their performance is in fact close to that of LASSO. If we look at the first dominant term in Theorem \ref{thm:riskbehsmallsigma}, even $p=0.99$ may seem to outperform LASSO by a large margin. However, note that the order of second dominant term for $p=0.99$ is pretty close to the order of the first dominant term. Hence, any judgement based on the first dominant term in such cases is inaccurate and misleading. This shows the importance of the second dominant term in these cases.

\vspace{.2cm}

So far, we have analyzed the lowest fixed point of $\Psi_{\lambda_*,p}(\sigma^2)$ and have seen that $p<1$ may lead to major improvements over optimal-$\lambda$ $\ell_1$-AMP, if the noise level is not large. Our next goal is to prove the ``high-noise phenomenon", i.e., the fact that for large values of noise, optimally tuned $\ell_1$-AMP outperforms optimally tuned $\ell_p$-AMP for $p<1$. 

 \begin{proposition}\label{prop:optimallassolargesigma}
Suppose $X \sim (1- \epsilon)\Delta_0+ \epsilon \Delta_\mu$ where $\mu$ is a non-zero constant. For any $0\leq p<1$, there exists a threshold $\tilde{\sigma}_w$ such that optimal-$\lambda$ $\ell_1$-{\rm AMP} outperforms the lowest fixed point of optimal-$\lambda$ $\ell_p$-{\rm AMP} for all $\sigma_w > \tilde{\sigma}_w$.
\end{proposition}

\vspace{.2cm}

Proof of this result is presented in Section \ref{sec:prooflargenoiseell_1}. This proposition implies that even if we had access to the best initialization for the optimal-$\lambda$ $\ell_p$-AMP, we should still use optimal-$\lambda$ $\ell_1$-AMP when the measurement noise is large. Note that even though this theorem is concerned with very large values of the measurement noise, as is clear from Figure \ref{fig:sigma_noisy}, even for not so large noise levels, $\ell_1$-AMP outperforms $\ell_p$-AMP.

\subsection{Analysis of the highest fixed point of optimally tuned $\ell_p$-AMP}\label{ssec:hfpsectionnoise}

So far we have analyzed the lowest fixed point of optimally tuned $\ell_p$-AMP. In this section we study its highest fixed point in the presence of noise.

\begin{theorem}\label{thm:hfpnoisy}
Let $\sigma_h$ denote the highest fixed point of the optimal-$\lambda$ $\ell_p$-AMP. If $\overline{M}_p(\epsilon) < \delta$, then 
\begin{equation}\label{eq:noisesensitivity1}
\frac{\sigma_h^2}{\sigma_w^2} \leq \frac{1}{1- \frac{\overline{M}_p(\epsilon)}{\delta}}.
\end{equation}
Furthermore, there exists a distribution $p_X \in \mathcal{F}_{\epsilon}$ and a noise variance $\sigma_w^2$ for which 
\begin{equation}\label{eq:noisesensitivity2}
\frac{\sigma_h^2}{\sigma_w^2} \geq \frac{1}{1- \frac{\underline{M}_p(\epsilon)}{\delta}}.
\end{equation}
\end{theorem}

\vspace{.2cm}
This theorem is proved in Section \ref{sec:proofnoisyhigh}. We again emphasize that our numerical calculations show that $\overline{M}_p(\epsilon) = \underline{M}_p(\epsilon)$. Also, in the proof of Lemma \ref{lem:ell1pteq} we have proved that $\overline{M}_1(\epsilon) = \underline{M}_1(\epsilon)$.  Figure \ref{fig:9} compares $\overline{M}_p(\epsilon)$ for different values of $p$. For most $p<1$, $\overline{M}_p(\epsilon)$ is either larger than $\overline{M}_1(\epsilon)$ or in some cases slightly lower. Hence, as far as the highest fixed point of the $\ell_p$-AMP algorithm on the least favorable signals is concerned, optimal-$\lambda$ $\ell_p$-AMP can offer slight improvements (if any at all) over $\ell_1$-AMP. 

\begin{figure}[htbp]
  \centering
  \includegraphics[width=3.2in]{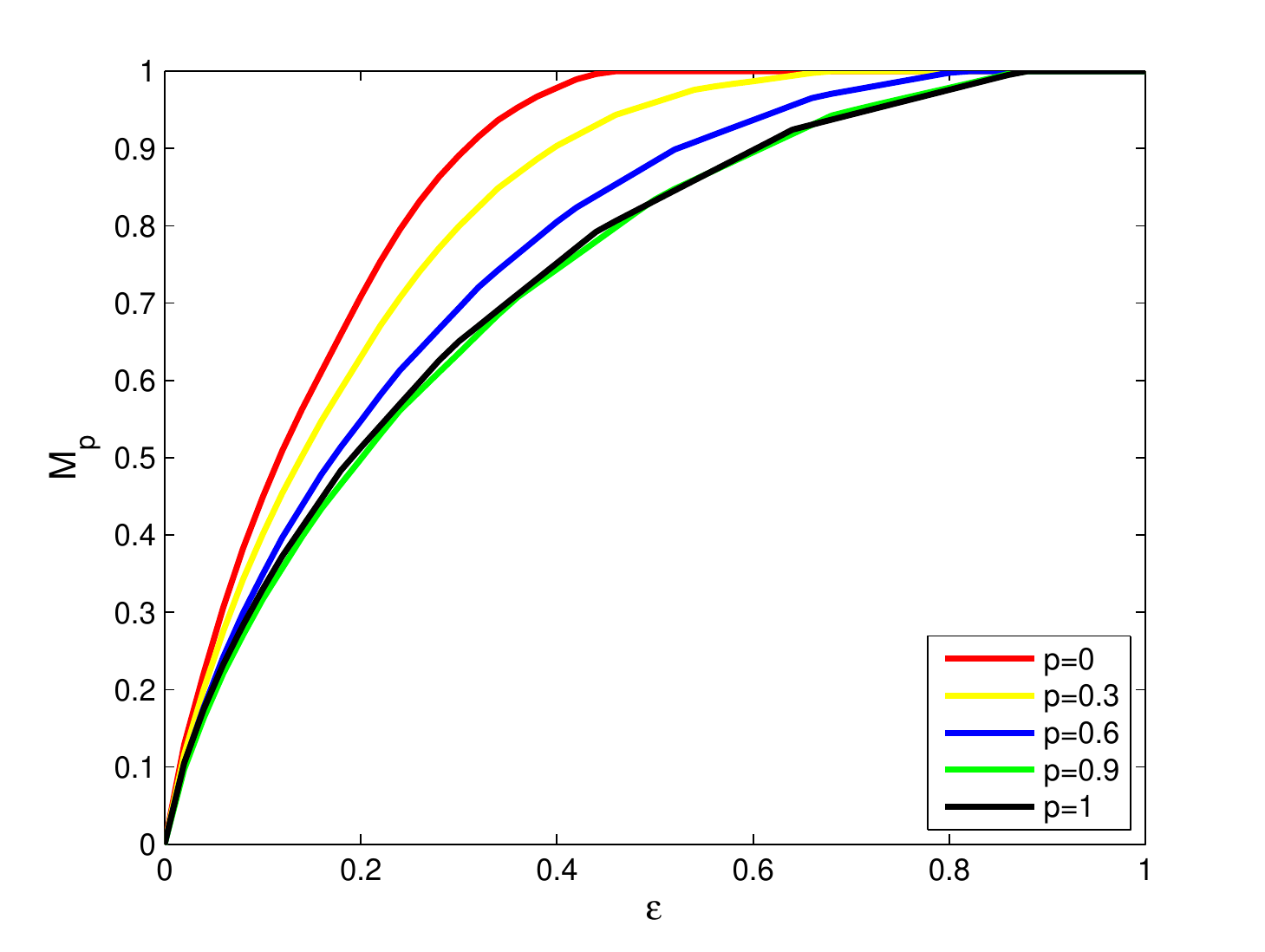}
  \caption{$\overline{M}_p(\epsilon)$ as a function of $\epsilon$ for $p\in \{0, 0.3, 0.6, 0.9, 1\}$. Lower values of $\overline{M}_p(\epsilon)$ lead to better upper bounds for $\sigma_h^2$.   }
  \label{fig:9}
\end{figure}

Again we would like to emphasize that the bound $\frac{1}{1- \frac{\overline{M}_p(\epsilon)}{\delta}}$ is achieved for very specific distributions. If the distribution of $X$ is different from those, optimally tuned $\ell_p$-AMP can achieve major improvement over optimal-$\lambda$ $\ell_1$-AMP. An interesting question that is left for future research is which distributions benefit LPLS more.

As is clear from our discussion, optimal-$\lambda$ $\ell_p$-AMP can outperform $\ell_1$-AMP for small values of noise and if it reaches its lowest fixed point. Also, since in many cases $\ell_p$-AMP has other fixed points, it requires a good initialization to reach its lowest fixed point. Our next goal is to show whether an optimal adaptation policy can resolve the issue of finding a good initialization. As we showed in the last section in the noiseless setting, it does not offer much improvement. However, when the noise is small, this algorithm outperforms optimal-$\lambda$ $\ell_1$-AMP by a large margin. The following theorem confirms this claim.

\vspace{.2cm}

\begin{theorem}\label{thm:improvecontinuation}
Let $\sigma_h$ denote the highest fixed point of the optimal-$(p, \lambda)$ $\ell_p$-AMP. If $\inf_{0\leq p \leq 1 }M_p(\epsilon) < \delta$, then
\[
\lim_{\sigma_w^2 \rightarrow 0} \frac{\sigma_h^2}{\sigma_w^2}= \frac{1}{1-\frac{\epsilon}{\delta}}.
\]
\end{theorem}

The proof of this theorem is presented in Section \ref{sec:proofcontlownoise}. Note that there is a major difference between this theorem and Theorem \ref{thm:noisyoptimallambda}. This result is about the highest fixed point, while Theorem \ref{thm:noisyoptimallambda} evaluates the lowest fixed point. Note that according to this theorem, if the sparsity level of the signal is below the phase transition of optimal-$(p, \lambda)$ $\ell_p$-AMP, then optimal-$(p, \lambda)$ $\ell_p$-AMP offers much better noise sensitivity than that of optimal-$\lambda$ $\ell_1$-AMP (for small values of noise). Note that according to Proposition \ref{prop:optimallassolargesigma}, we expect the noise sensitivity of optimal-($p,\lambda$) $\ell_p$-AMP to be the same as the noise sensitivity of optimal-$\lambda$ $\ell_1$-AMP for large values of noise. This phenomenon can be observed in Figure \ref{fig:continuationstrategyhelps}. As is clear in this figure, for small values of the noise, $p$-continuation leads to substantially better results than the optimal-$\lambda$ $\ell_1$-AMP. 

\begin{figure}
\begin{center}
\includegraphics[width=3.2in]{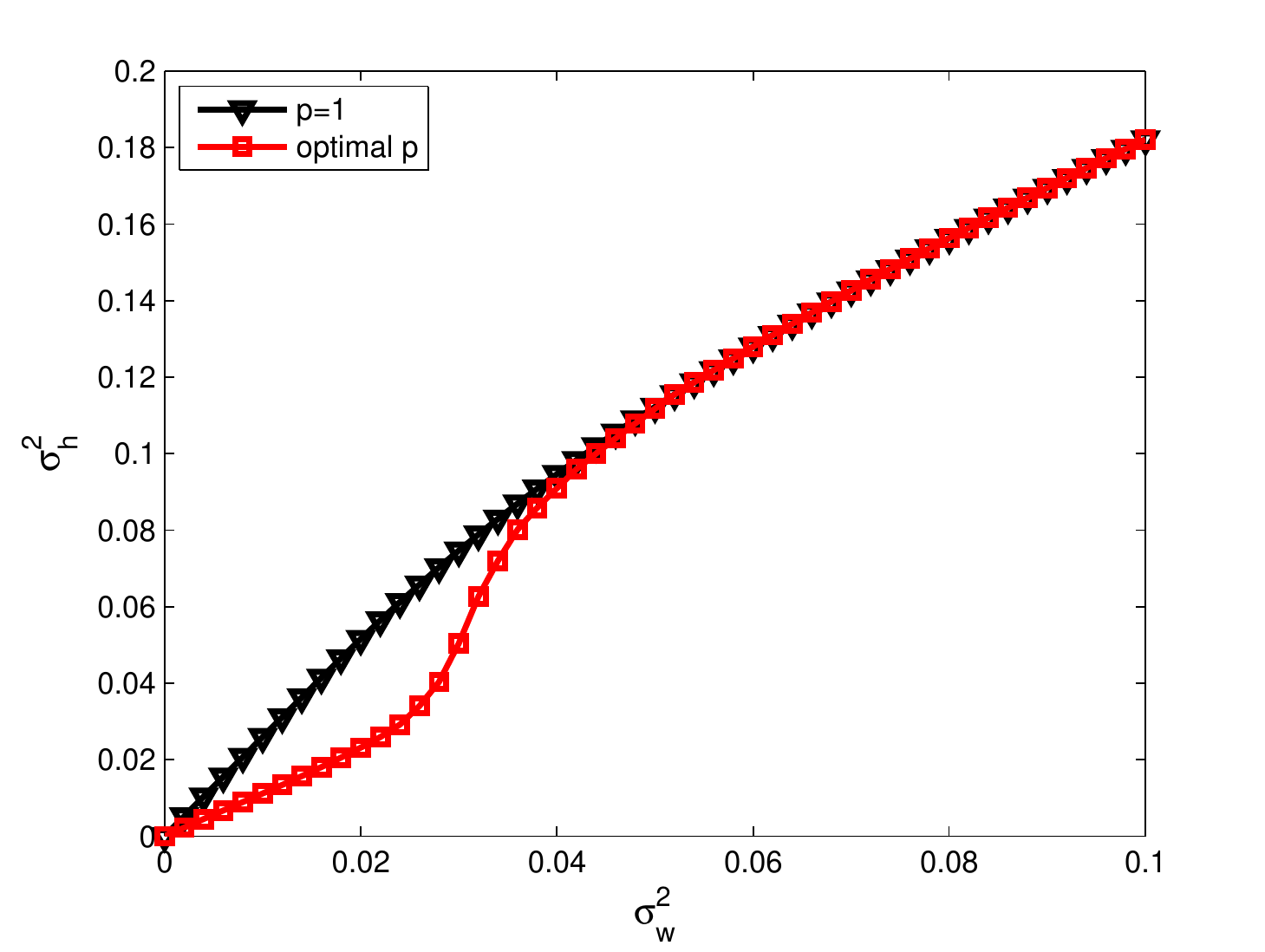}
\caption{Comparison of $\sigma_h^2$ in optimal-$\lambda$ $\ell_1$-AMP and optimal-$(p, \lambda)$ $\ell_p$-AMP. $\delta$ and $\epsilon$ are set to $0.1$ and $0.01$ respectively. The non-zero elements of $x_o$ are iid $\pm1$ with probability $0.5$.}
\label{fig:continuationstrategyhelps}
\end{center}
\end{figure}

\section{Relation with $\ell_p$-norm minimization}\label{sec:discussion}

Replica method is a non-rigorous method invented in statistical physics to study the behavior of large  magnetic and disordered systems. This method has found many applications in science and engineering \cite{mezard1985replicas, castellani2005spin, guo2005randomly, tanaka2002statistical}. In particular, \cite{rangan2009asymptotic } has used this method to analyze the accuracy of $\hat{x}(\gamma, p)$. 
Here we briefly explain the results derived  in \cite{rangan2009asymptotic} and compare them with the results of our paper. Under the replica symmetry assumptions (summarized in Section IV of \cite{rangan2009asymptotic}), as $N \rightarrow \infty$, $(\hat{x}_j(\gamma,p ), x_{0,j})$ converges in distribution to the random vector $(\eta_p (X + \sigma_{eff} Z; \gamma_p) , X)$ where $X \sim p_X$ and $Z \sim N(0,1)$ are independent, and $\sigma_{eff}$ satisfies the following fixed point equation:
\begin{equation}\label{eq:replica}
\sigma_{eff}^2 = \sigma_w^2 + \frac{1}{\delta} \mathbb{E} (\eta_p(X+ \sigma_{eff} Z; \gamma_p) - X)^2,
\end{equation}
where $\gamma_p$ can also be calculated in terms of $\gamma$ and $\sigma_{eff}$, but its particular form is not of interest in this paper. Note that the fixed point of $\ell_p$-AMP satisfies a fixed point equation that is the same as \eqref{eq:replica} (modulo the threshold parameter).  If we pick $\gamma$ optimally in \eqref{eq:replica} (to make $\sigma_{eff}$ or the mean square error of the reconstructed vector as small as possible), then the two fixed point equations that are derived from $\ell_p$-AMP and the Replica method will be exactly the same. This exact correspondence can transform all the results about lowest fixed points we derived for the optimal-$\lambda$ $\ell_p$-AMP to new results for the solution of \eqref{eq:ell0minimization}. For the sake of brevity, we do not repeat all the results here. We qualitatively explain the implications of two of our results:
\begin{enumerate}
\item If $\delta> \epsilon$, then \eqref{eq:ell0minimization} recovers the exact solution in the noiseless setting for any $0 \leq p<1$. This can be derived by combining Theorem \ref{lem:noiselesslowfpless1} with the result of the Replica method described above. Note among all the fixed points of \eqref{eq:replica}, the lowest fixed point corresponds to the minimum free energy \cite{mezard1985replicas} and hence characterizes the asymptotic performance of the global minimizer of \eqref{eq:ell0minimization}.

\item When the noise level, $\sigma_w^2$,  is high, LASSO outperforms LPLS for every $p<1$. This result can be derived by combining the results of Proposition \ref{prop:optimallassolargesigma} and the Replica method result.
\end{enumerate}

\section{Proofs of the main results}\label{sec:proof}

\subsection{Properties of $\eta_p(u, \lambda)$}\label{sec:propell_pprox}

In the proofs of our main results, we employ several properties of the proximal functions ${\eta _p}\left( {u;\lambda } \right)$. This section is devoted to the derivation of these properties. Note that since ${\eta _0}\left( {u;\lambda } \right)$ and ${\eta _1}\left( {u;\lambda } \right)$ have very simple forms,\footnote{$\eta_0(u; \lambda) = u \mathbb{I} (|u| > \sqrt{2 \lambda})$ and $\eta_1(u; \lambda) = (|u|- \lambda) {\rm sign} (u) \mathbb{I} (|u| > \lambda) $. } in some of the results mentioned below these two cases are omitted.  \\

Our first result is concerned with the scale invariance property of ${\eta _p}\left( {u;\lambda } \right)$. This result will be used extensively in the rest of the paper.

\vspace{.2cm}

\begin{lemma}\label{lem:scaleinv}
$\eta_p(u, \lambda)$ has the following scale invariance properties for $0 \leqslant p \leqslant 1$:
\begin{enumerate}
\item[(i)] $\eta_p(-u; \lambda) = - \eta_p(u; \lambda)$.
\item[(ii)] ${\eta_p}(\alpha u;\lambda {\alpha ^{2 - p}}) = \alpha {\eta_p}(u;\lambda )$, for every $\alpha \geq 0$.
\end{enumerate}
\end{lemma}

\vspace{.2cm}

\begin{proof}
First, we prove that $\eta_p(-u; \lambda) = - \eta_p(u; \lambda)$. According to the definition of $\eta _p$, we have
\begin{eqnarray*}
  {\eta _p}( - u;\lambda ) &=& \mathop {\arg \min }\limits_x {\left( { - u - x} \right)^2} + \lambda {\left| x \right|^p} = \mathop {\arg \min }\limits_x {\left( {u - ( - x)} \right)^2} + \lambda {\left| { - x} \right|^p}  \\
   &=&  - \mathop {\arg \min }\limits_x {\left( {u - x} \right)^2} + \lambda {\left| x \right|^p}    =  - {\eta _p}(u;\lambda ).
\end{eqnarray*}

To prove the second part of this lemma, note that it is trivially true when $\alpha =0$. For any $\alpha >0$, we have
\begin{eqnarray}
\label{eq:3-5}
  {\eta_p}(\alpha u;\lambda {\alpha ^{2 - p}}) &=& \mathop {\arg \min }\limits_x {\left( {\alpha u - x} \right)^2} + \lambda {\alpha ^{2 - p}}{\left| x \right|^p} \nonumber  \\
   &=& \mathop {\arg \min }\limits_x {\alpha ^2}{\left( {u - \frac{x}{\alpha }} \right)^2} + \lambda {\alpha ^2}{\left| {\frac{x}{\alpha }} \right|^p} \nonumber  \\
   &=& \alpha \mathop {\arg \min }\limits_x {\left( {u - x} \right)^2} + \lambda {\left| x \right|^p} = \alpha {\eta_p}(u;\lambda ).    \nonumber   
\end{eqnarray}

\end{proof}

\vspace{.2cm}

The next lemma is an auxiliary result that will be used later to derive the main properties of $\eta_p(u; \lambda)$. \\

\begin{lemma} \label{lem:lowerbound}
For $0<p<1$, if $|\eta_p(u; \lambda)| > 0$, then it satisfies
\[
|\eta_p(u; \lambda) | \geq  \zeta^*, 
\]
where $ \zeta^* \triangleq \left( \frac{1}{\lambda p (1-p)} \right)^{\frac{1}{p-2}}$. Furthermore, $\eta_p(u; \lambda) =0$ for every $u$ satisfying
\[
|u| < g(\zeta^*), 
\]
where $g(\zeta) \triangleq \zeta+ \lambda p \zeta^{p-1}$. 
\end{lemma}
\vspace{.2cm}

\begin{proof}
According to Lemma \ref{lem:scaleinv} part (i), we only consider the case $u \geq 0$. If ${\eta _p}(u;\lambda )>0$, since it minimizes
\begin{equation}\label{eq:proxdef2}
{\xi _p}(x,u) \triangleq \frac{1}{2}{(u - x)^2} + \lambda {\left| x \right|^p},
\end{equation}
 it must satisfy
\begin{equation*}
\eta_p(u;\lambda) + \lambda p \eta_p^{p-1}(u;\lambda)=u,
\end{equation*}
which can be written as
\begin{equation*}
g({\eta _p}(u;\lambda ))=u.
\end{equation*}

It is straightforward to check the following facts about $g(\zeta)$: (i) $g(\zeta)$ has a global minimum at $\zeta^*= \left( \frac{1}{\lambda p (1-p)} \right)^{\frac{1}{p-2}}$. (ii) $\mathop {\lim }\limits_{\zeta \to 0} g(\zeta) = \infty$. (iii) $\mathop {\lim }\limits_{\zeta \to \infty } g(\zeta) = \infty$.  (iv) $g(\zeta)$ is a decreasing function below $\zeta^*$ and an increasing function above $\zeta^*$.

According to these properties, three different cases happen for $g(\zeta)=u$:
\begin{itemize}
\item[(i)] If $u< g(\zeta^*)$, then $g(\zeta)=u$ does not have any solution.
\item[(ii)] If $u=g(\zeta^*)$, then $g(\zeta)=u$ has only one solution at $\zeta^*$.
\item[(iii)] If $u> g(\zeta^*)$, then $g(\zeta) =u$ has two solutions; one below $\zeta^*$ and one above $\zeta^*$. Among these two solutions the value of $x$ that minimizes ${\xi _p}(x,u)$ is $x \geq \zeta^*$.
\end{itemize}

This completes the proof of the first part of the lemma. We now prove that for every $u< g(\zeta^*)$, $\eta_p(u;\lambda)=0$. This is due to the fact that the derivative of ${\xi _p}(x,u)$ with respect to $x$ will be always positive for every $x>0$. Hence the minimum must happen at zero. Note that ${\xi _p}(x,u)$ is a continuous function of $x$. 

\end{proof}

\vspace{.2cm}

\begin{lemma}\label{lem:threshform}
For $0 \leqslant p \leqslant 1$, there exists a threshold ${\tilde \lambda _p}$ such that $\forall \left| u \right| < {\tilde \lambda _p}$, ${\eta _p}(u;\lambda ) = 0$ and $|{\eta _p}(u;\lambda )| > 0$ $\forall \left| u \right| > {\tilde \lambda _p}$. Furthermore, we have
\[
{\tilde \lambda _p} = {c_p}{\lambda ^{\frac{1}{{2 - p}}}}
\]
where $c_p=[2(1-p)]^{\frac{1}{2-p}}+p[2(1-p)]^{\frac{p-1}{2-p}}$ . The value of $c_p$ is plotted in Figure \ref{fig:cp} for different values of $p$.

\begin{figure}[htbp]
  \centering
  \includegraphics[height=2.5in]{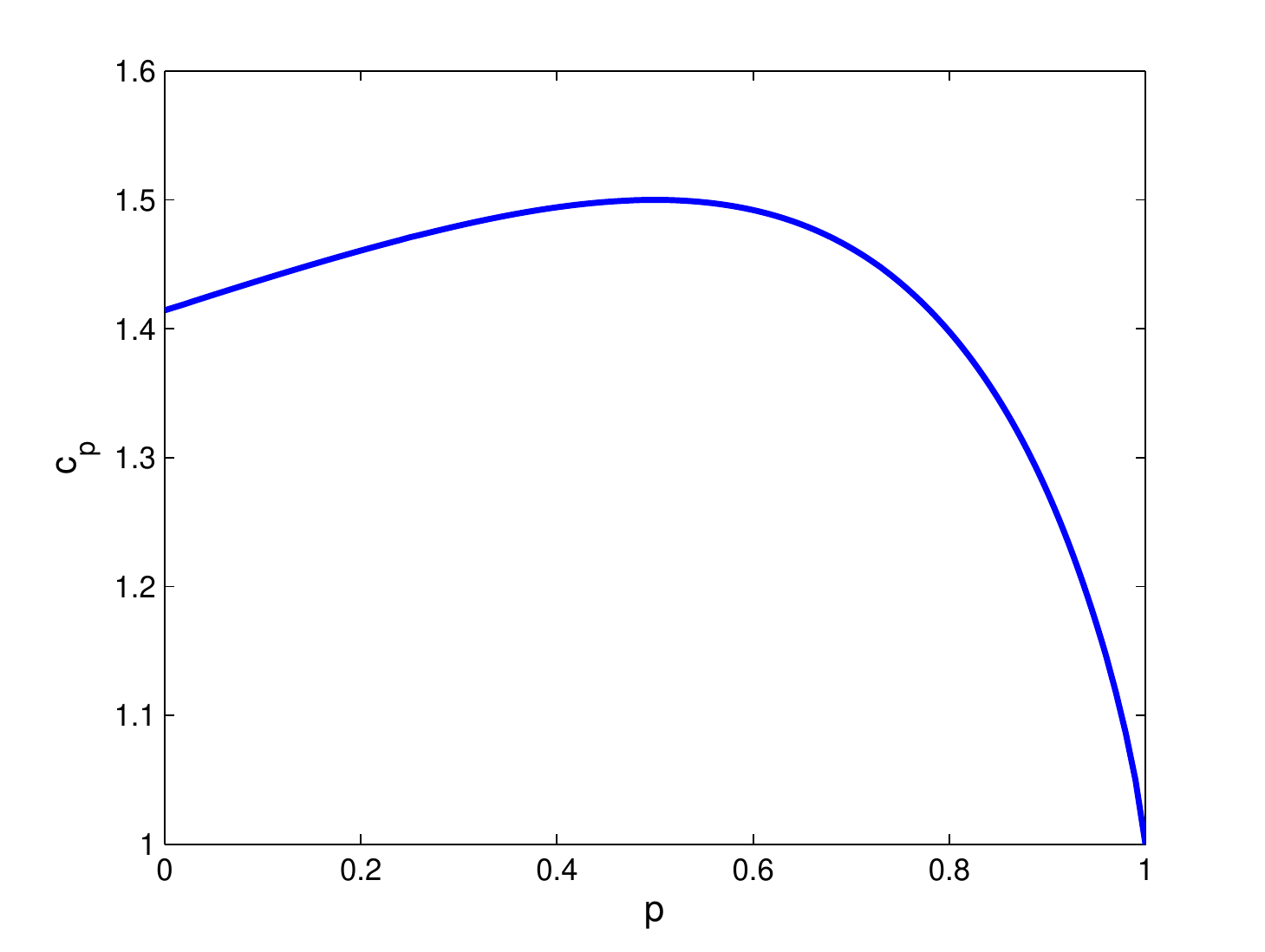}
  \caption{The value of $c_p$ as a function of $p$. Note that the threshold $\tilde{\lambda}_p$ defined in Lemma \ref{lem:threshform} is in the form of $c_p \lambda^{1/(2-p)}$.}
  \label{fig:cp}
\end{figure}
\end{lemma}
\vspace{.2cm}

\begin{proof}
We only consider the case $0<p<1$ and $u \geq 0$. The proof is straightforward for $p=0$ and $p=1$, due to the explicit form of $\eta_p$ in these two cases. Consider the notations ${\xi _p}(x,u)$ and $g(\zeta) $ introduced in the proof of Lemma \ref{lem:lowerbound}. Note that according to the proof of Lemma \ref{lem:lowerbound}, if ${\eta _p}(u;\lambda )>0$, then  $
g({\eta _p}(u;\lambda ))=u$. As the first step, we would like to prove that if $\eta_p(u_0; \lambda) >0$ for $u_0$, then $\eta_p(u; \lambda)$ will be greater than zero for any $u> u_0$. Since $u>0$, it is straightforward to see that $\eta_p(u; \lambda) \geq 0$. Hence, $\eta_p(u; \lambda)$ is either equal to zero or it is the solution of $g(\zeta) = u$ where $\zeta > \zeta^*$ ($\zeta^*$ is defined in Lemma \ref{lem:lowerbound}). Let $\bar{\zeta}> \zeta^*$ denote the solution of $g(\zeta) = u$. Our goal is to show that ${\xi _p}(0,u)-{\xi _p}(\bar{\zeta},u)>0$. Toward this goal, we prove that ${\xi _p}(0,u)-{\xi _p}(\bar{\zeta},u)$ is an increasing function of $u$. 

We have
\begin{eqnarray}
{\xi _p}(0,u)-{\xi _p}(\bar{\zeta},u) &=& \frac{1}{2}u^2 - \frac{1}{2} (\bar{\zeta}- u)^2- \lambda \bar{\zeta}^p \nonumber \\
 &=& \bar{\zeta} u - \frac{1}{2} \bar{\zeta}^2 -  \lambda \bar{\zeta}^p= \bar{\zeta} (\bar{\zeta} + \lambda p \bar{\zeta}^{p-1}) - \frac{1}{2} \bar{\zeta}^2 -  \lambda \bar{\zeta}^p \nonumber\\ 
 &=& \frac{1}{2} \bar{\zeta}^2 + \lambda (p-1) \bar{\zeta}^p. 
\end{eqnarray}
By taking the derivative of this function with respect to $\bar{\zeta}$, it is straightforward to see that ${\xi _p}(0,u)-{\xi _p}(\bar{\zeta},u)$ is an increasing function of $\bar{\zeta}$ for $\bar{\zeta} > \zeta^*$. If we prove that $\bar{\zeta}$ is an increasing function of $u$, then we can conclude that ${\xi _p}(0,u)-{\xi _p}(\bar{\zeta},u)$ is an increasing function of $u$. Note that $g(\bar{\zeta}) =u$. By taking the derivative of both sides with respect to $u$, we obtain
\[
g'(\bar{\zeta}) \frac{\partial \bar{\zeta}}{\partial u} = 1. 
\]   
Again, since $g'(\bar{\zeta}) >0$ for $\bar{\zeta} > \zeta^*$, we conclude that $\bar{\zeta}$ is an increasing function of $u$. Hence we conclude that ${\xi _p}(0,u)-{\xi _p}(\bar{\zeta},u)$ is an increasing function of $u$. If $\eta_p(u_0; \lambda) >0$, we know that ${\xi _p}(0,u_0)-{\xi _p}(\bar{\zeta},u_0)>0$. Since ${\xi _p}(0,u)-{\xi _p}(\bar{\zeta},u)$ is an increasing function of $u$, we have for every $u>u_0$, ${\xi _p}(0,u)-{\xi _p}(\bar{\zeta},u)>0$, which implies that $\eta_p(u;\lambda)= \bar{\zeta} >0$. 

So far we have been able to prove that there exists an interval $[- \tilde{\lambda}_p, \tilde{\lambda}_p]$ such that if $|u| \geq \tilde{\lambda}_p$, $\eta_p (u; \lambda) >0$ and for every $|u| < \tilde{\lambda}_p$, $\eta_p(u; \lambda) =0$. Note that according to Lemma \ref{lem:lowerbound}, $\tilde{\lambda}_p > g(\zeta^*)$. Our next goal is to derive the exact form of $\tilde{\lambda}_p$. For notational simplicity, define $\alpha \triangleq \lambda ^{\frac{1}{2 - p}}$ and $c_p$ as
\begin{eqnarray}\label{c_p:def}
{c_p} \triangleq \sup \left\{ {u:{\eta _p}(u;1) = 0} \right\}.
\end{eqnarray}
We have
\[
{\eta _p}(u;\lambda ) = {\eta _p}(u;{\alpha ^{2 - p}} \cdot 1) = \alpha {\eta _p}\left( {\frac{u}{\alpha };1} \right),
\]
where the second equality is due to Lemma \ref{lem:scaleinv} part (ii). Hence, ${\eta _p}(u;\lambda ) = 0$ if and only if ${\eta _p}\left( {\frac{u}{\alpha };1} \right) = 0$, i.e., ${\eta _p}\left( {u;\lambda } \right) = 0$ if and only if $u/\alpha  < {c_p}$. Therefore, we have ${\tilde \lambda _p} = {c_p}\alpha  = {c_p}{\lambda ^{\frac{1}{{2 - p}}}}$. Finally, we aim to obtain the explicit form of $c_p$. Denote the larger solution of $x+px^{p-1}=u$ by $x^*$. Then note that $\eta_p(u;1)=0$ implies 
\[
\frac{u^2}{2} \leq \frac{1}{2}(x^*-u)^2+(x^*)^{p}=\frac{(x^*)^2}{2}+\frac{u^2}{2}-x^*(x^*+p(x^*)^{p-1})+(x^*)^p,
\]
which yields that $x^* \leq [2(1-p)]^{\frac{1}{2-p}}$. Since $x^*$ is an increasing function of $u$, we know $u \leq [2(1-p)]^{\frac{1}{2-p}} +p [2(1-p)]^{\frac{p-1}{2-p}}$. On the other hand, it is straightforward to show that $\eta_p(u;1)=0$ when $u=[2(1-p)]^{\frac{1}{2-p}} +p [2(1-p)]^{\frac{p-1}{2-p}}$. Combining with the definition of $c_p$ in \eqref{c_p:def} gives us its analytical formula.

\end{proof}

\vspace{.2cm}

\begin{lemma}\label{lem:lowerboundonx}
For $0<p<1$, If $|\eta_p(u;\lambda)|>0$, then $|\eta_p(u;\lambda)|\geq [2(1-p)]^{\frac{1}{2-p}}\lambda^{\frac{1}{2-p}}$
\end{lemma}

\begin{proof}
According to Lemma \ref{lem:threshform} and Lemma \ref{lem:derivativeproperties} part (ii), when $\eta_p(u;\lambda)>0$, we know $\eta_p(u;\lambda)\geq \eta^+_p(c_p\lambda^{\frac{1}{2-p}};\lambda)$. Furthermore, from the proof of Lemma \ref{lem:threshform}, it is straightforward to confirm the following equation,
\begin{eqnarray} \label{pro:one}
c_p\lambda^{\frac{1}{2-p}}=\eta^+_p(c_p\lambda^{\frac{1}{2-p}};\lambda)+\lambda p (\eta^+_p(c_p\lambda^{\frac{1}{2-p}};\lambda))^{p-1},
\end{eqnarray}
where the equation $c_p\lambda^{\frac{1}{2-p}}=x+\lambda p x^{p-1}$ has two roots and $x=\eta^+_p(c_p\lambda^{\frac{1}{2-p}};\lambda)$ is the larger one. 
Dividing both sides of the above equation by $\lambda^{\frac{1}{2-p}}$ gives,
\begin{eqnarray}\label{pro:two}
c_p=\eta^+_p(c_p;1)+p (\eta^+_p(c_p;1))^{p-1}
\end{eqnarray}
According to the explicit form of $c_p$ in Lemma \ref{lem:threshform}, we can obtain $\eta^+_p(c_p;1)=[2(1-p)]^{\frac{1}{2-p}}$.
\end{proof}

\vspace{.3cm}

So far we have studied some of the main properties of $\eta_p(u;\lambda)$. In this paper, we will also work with the derivatives of $\eta_p(u; \lambda)$. Note that the derivative of this function with respect to $u$ exists for every $u$ except at $u= c_p\lambda^{\frac{1}{2-p}}$. Furthermore, its derivative with respect to $\lambda$ exists everywhere except for $\lambda=(u/c_p)^{2-p}$. For notational simplicity, we use the following notations for the partial derivatives:
\begin{eqnarray*}
{\partial _1}{\eta _p}(u;\lambda ) &\triangleq& \frac{{\partial {\eta _p}(u;\lambda )}}{{\partial u}}, \ \ \ \  \ \ \ \ \
\partial _1^2{\eta _p}(u;\lambda ) \triangleq \frac{{{\partial ^2}{\eta _p}(u;\lambda )}}{{\partial {u^2}}}, \nonumber \\
{\partial _2}{\eta _p}(u;\lambda )&\triangleq & \frac{{\partial {\eta _p}(u;\lambda )}}{{\partial \lambda}},  \ \ \ \ \ \ \ \ \
\partial _2^2{\eta _p}(u;\lambda ) \triangleq \frac{{{\partial ^2}{\eta _p}(u;\lambda )}}{{\partial {\lambda^2}}}.
\end{eqnarray*}
 Whenever we use these notations, we refer to the derivative of the function for the values of $u$ and $\lambda$ at which $|{\eta _p}(u;\lambda )|>0$. 

\vspace{.2cm}

\begin{lemma}\label{lem:derivativeproperties}
If ${\eta _p}(u;\lambda ) > 0$, then for $0 < p < 1$ and $\lambda  > 0$, ${\eta _p}(u;\lambda )$ satisfies
\begin{itemize}
\item[      (i)] ${\eta _p}(u;\lambda ) < u$.
\item[      (ii)] $1 < \sup_{\eta_p(u;\lambda)>0}{\partial _1}{\eta _p}(u;\lambda ) < \infty$.
\item[      (iii)] $\partial _1^2{\eta _p}(u;\lambda ) < 0$.
\end{itemize}
Furthermore, since ${\eta _p}(u;\lambda )$ is an odd function, ${\partial _1}{\eta _p}(u;\lambda )$ and $\partial _1^2{\eta _p}(u;\lambda )$ are even and odd functions respectively. Therefore, for ${\eta _p}(u;\lambda ) < 0$, we have (i) ${\eta _p}(u;\lambda ) > u$, (ii) $1 <\sup_{\eta_p(u;\lambda)<0} {\partial _1}{\eta _p}(u;\lambda ) <  \infty$, (iii) $\partial _1^2{\eta _p}(u;\lambda ) > 0$. \\
\end{lemma}

\textit{Proof:} In this proof, we only consider the case ${\eta _p}(u;\lambda ) > 0$. Note that ${\eta _p}(u;\lambda )$ satisfies
\[
{\eta _p}(u;\lambda ) - u + \lambda p\eta _p^{p - 1}(u;\lambda ) = 0
\]
Since ${\eta _p}(u;\lambda ) > 0$, we have ${\eta _p}(u;\lambda ) < u$. Taking the derivative with respect to $u$ from both sides of the equation above, we obtain
\begin{equation}\label{eq:partial1etap}
{\partial _1}{\eta _p}(u;\lambda ) - 1 + \lambda p(p - 1)\eta _p^{p - 2}(u;\lambda ){\partial _1}{\eta _p}(u;\lambda ) = 0.
\end{equation}
Therefore, the derivative of ${\eta _p}(u;\lambda )$ is
\[
{\partial _1}{\eta _p}(u;\lambda ) = \frac{1}{{1 + \lambda p(p - 1)\eta _p^{p - 2}(u;\lambda )}}.
\]
Furthermore, based on Lemma \ref{lem:lowerboundonx}, we have 
 \[
0 > \lambda p(p - 1)\eta _p^{p - 2}(u;\lambda ) \geq \lambda p(p-1)([2(1-p)]^{\frac{1}{2-p}}\lambda^{\frac{1}{2-p}})^{p-2}=p(p-1)a_p^{p-2}=-\frac{p}{2} > -1.
 \]
 Note the inequality above holds for every possible $u$ and $\lambda$ such that $\eta_p(u;\lambda)>0$, which hence shows (ii). 
We now prove the third part of the lemma. By taking another derivative from \eqref{eq:partial1etap} with respect to $u$, we obtain
\[
\partial _1^2{\eta _p}(u;\lambda ) + \lambda p(p - 1)(p - 2)\eta _p^{p - 3}(u;\lambda ){\left( {{\partial _1}{\eta _p}(u;\lambda )} \right)^2} + \lambda p(p - 1)\eta _p^{p - 2}(u;\lambda )\partial _1^2{\eta _p}(u;\lambda ) = 0.
\]
Hence
\[
\partial _1^2{\eta _p}(u;\lambda ) = \frac{{ - \lambda p(p - 1)(p - 2)\eta _p^{p - 3}(u;\lambda ){{\left( {{\partial _1}{\eta _p}(u;\lambda )} \right)}^2}}}{{1 + \lambda p(p - 1)\eta _p^{p - 2}(u;\lambda )}}.
\]
Again by employing Lemma \ref{lem:lowerboundonx}, we can conclude that the second derivative is negative. We may also claim that
\[
\sup_u |\partial _1^2{\eta _p}(u;\lambda )| < \infty. 
\]
 $\hfill \Box$

\vspace{.2cm}

The next lemma is concerned with the properties of $\eta_p(u; \lambda)$ as a function of $\lambda$. 

\vspace{.2cm}

\begin{lemma}\label{lem:derivativeproperties2}
If $\left| {{\eta _p}(u;\lambda )} \right| > 0$ and $0 < p < 1$, we have
\[
{\partial _2}{\eta _p}(u;\lambda ) =  - p{ {\left| {{\eta _p}(u;\lambda )} \right|}^{p - 1}}{\partial _1}{\eta _p}(u;\lambda )\cdot \mbox{sign}(u).
\]
In particular, ${\partial _2}{\eta _p}(u;\lambda ) < 0$ \mbox{when} $u>0$ and ${\partial _2}{\eta _p}(u;\lambda ) > 0$ \mbox{if} $u <0$.
\end{lemma}
\vspace{.2cm}

\textit{Proof.}
We prove the result for the case of ${\eta _p}(u;\lambda ) > 0$. The other case can be proved in exactly the same way. Note that since ${\eta _p}(u;\lambda )>0$, it satisfies
\begin{equation}\label{eq:etapequationmain}
{\eta _p}(u;\lambda ) - u + \lambda p  \eta _p^{p - 1}(u;\lambda ) = 0.
\end{equation}
By taking the derivative of \eqref{eq:etapequationmain} with respect to $u$ we obtain
\begin{equation}\label{eq:partial1implicit}
{\partial _1}{\eta _p}(u;\lambda ) - 1 + \lambda p(p - 1)\eta _p^{p - 2}(u;\lambda ){\partial _1}{\eta _p}(u,\lambda ) = 0.
\end{equation}
By taking the derivative of \eqref{eq:etapequationmain} with respect to $\lambda$ we obtain
\begin{equation}\label{eq:partial2implicit}
{\partial _2}{\eta _p}(u;\lambda ) + p\eta _p^{p - 1}(u;\lambda ) + \lambda p(p - 1)\eta _p^{p - 2}(u;\lambda ){\partial _2}{\eta _p}(u,\lambda ) = 0.
\end{equation}
The final result can be obtained by combining \eqref{eq:partial1implicit} and \eqref{eq:partial2implicit}.
$\hfill \Box$

\vspace{.2cm}

Below we summarize two straightforward corollaries of the above results. These two corollaries enable us to compare $\eta_p$ with $\eta_0$ and $\eta_1$. First note that according to Lemma \ref{lem:threshform}, the threshold at which $\eta_p(u; \lambda)$ switches from zero to a positive number is different for different values of $p$. This makes the comparison of these proximal functions complicated. However, according to Lemma \ref{lem:threshform}, if we set the parameter $\lambda_p \triangleq \left( \frac{\tilde{\lambda}}{c_p} \right)^{2-p}$ with $\tilde{\lambda}$ being a fixed constant, then for every $0\leq p \leq 1$, we have $\eta_p(u; \lambda)= 0$ for $|u|< \tilde{\lambda}$ and $\eta_p(u; \lambda)\neq 0$ for $|u| > \tilde{\lambda}$. Based on this new parametrization, we would like to compare $\eta_p$ with $\eta_0$ and $\eta_1$.

\vspace{.2cm}

\begin{corollary}\label{cor:etap1eta0}
Define $\lambda_p  \triangleq \left( \frac{\tilde{\lambda}}{c_p} \right)^{2-p}$. Then
 $${\eta _p}(u;{\lambda _p}) > {\eta _1}(u;{\lambda _1}), \ \ \  \forall u>\tilde{\lambda}, 0 \leq p<1.$$
\end{corollary}

\vspace{.2cm}

\textit{Proof:} Note that ${\eta _1}(\tilde{\lambda};{\lambda _1}) = 0$ and ${\eta _p}(\tilde{\lambda};{\lambda _p}) > 0$ for every $0 \leq p < 1$. The derivative of the soft thresholding function is ${\partial _1}{\eta _1}(u;{\lambda _1}) = 1$ for $u > \tilde{\lambda}$. According to Lemma 7, the derivative of ${\eta _p}(u;{\lambda _p})$ is ${\partial _1}{\eta _p}(u;{\lambda _p}) > 1$ for $u>\tilde{\lambda}$. Therefore, we have ${\eta _1}(u;{\lambda _1}) < {\eta _p}(u;{\lambda _p})$ when $u>\tilde{\lambda}$ and $0 < p < 1$. It is straightforward to check that the result also holds for $p=0$, i.e., ${\eta _1}(u;{\lambda _1}) < {\eta _0}(u;{\lambda _0})$. $\hfill \Box$

\vspace{.2cm}

\begin{corollary}\label{cor:hardvsellp}
Let $\lambda_p = \left( \frac{\tilde{\lambda}}{c_p} \right)^{2-p}$. We have
$${\eta _p}(u;{\lambda _p}) < {\eta _0}(u;{\lambda _0}), \ \ \ \ \forall u> \tilde{\lambda}, 0<p \leq 1.$$
\end{corollary}
\textit{Proof:} Since $\eta_1(u,\lambda_1)$ admits an explicit form, it is straightforward to verify the result. For $0<p<1$, it is a direction result of Lemma \ref{lem:derivativeproperties} part (i).   $\hfill \Box$

\vspace{.2cm}

Another type of result that we will use in this paper is about the behavior of $\eta_p(u; \lambda)$ and its derivative for large values of $u$. The rest of this section is devoted to such results. 

\vspace{.2cm}

\begin{lemma}\label{lem:asymtoteta_p}
Let $\lambda>0$ and $0 < p < 1$ be two fixed numbers. Then for large value of $u$, we have
\begin{equation*}
\label{eq:3-5}
{\eta _p}(u;\lambda ) = u - \lambda p \  {\rm sign}(u){\left| u \right|^{p - 1}} + o(\left| u \right|^{p - 1}).
\end{equation*}
\end{lemma}

\vspace{.2cm}

\textit{Proof:} For simplicity, we only consider the case $u>0$. First note that Corollary \ref{cor:etap1eta0} shows
\[
\eta_p(u;\lambda_p) > \eta_1(u;\lambda_1) \rightarrow \infty, \mbox{u} \rightarrow \infty.
\]
Moreove, we know for large enough $u$, ${\eta _p}(u;\lambda )$ satisfies
\begin{equation}
\label{eq:3-6}
{\eta _p}(u;\lambda ) - u + \lambda p { {\eta _p^{p - 1}(u;\lambda )}} = 0.
\end{equation}
Define
\begin{equation}\label{eq:definedelta}
 \upsilon_p(u; \lambda) \triangleq u-{\eta _p}(u;\lambda ).
 \end{equation}
 If we plug \eqref{eq:definedelta} in (\ref{eq:3-6}) then we have
\begin{equation*}
\label{eq:3-7}
\upsilon_p(u; \lambda) = \lambda p \eta_p^{p-1}(u; \lambda).
\end{equation*}
Finally,
\begin{eqnarray*}
\lim_{u \rightarrow \infty} \frac{\upsilon_p(u; \lambda) - \lambda p u^{p-1}}{u^{p-1}} = \lambda p \left(\lim_{u \rightarrow \infty} \frac{\eta_p^{p-1}(u; \lambda)}{u^{p-1}}-1\right)= 0.
\end{eqnarray*}
The last equality is due to the fact that $\lim_{u \rightarrow \infty} \frac{\eta_p(u;\lambda)}{ u} - 1 = - \lambda p \lim_{u \rightarrow \infty} \frac{{\eta _p^{p - 1}(u;\lambda )}}{u}=0$.
$\hfill \Box$

\vspace{.5cm}

\begin{lemma}\label{lem:asymptoticbehaviorder1}
Let $ 0 < p<1$ and $\lambda >0$ be two fixed numbers. For large values of $u$ we have
\begin{equation*}
\label{eq:3-9}
{\partial _1}{\eta _p}(u;\lambda ) = 1+ \lambda p(1-p)|u|^{p-2} +o(|u|^{p-2})
\end{equation*}
\end{lemma}

\textit{Proof:} We only consider $u>0$ for simplicity. Taking the derivative of (\ref{eq:3-6}) with respect to $u$ leads to
\begin{equation}
\label{eq:3-11}
{\partial _1}{\eta _p}(u;\lambda ) = \frac{1}{1 + \lambda p(p - 1)\eta _p^{p - 2}(u;\lambda ) }.
\end{equation}
Hence we have
\begin{eqnarray}
\lefteqn{ \lim_{u \rightarrow \infty} \frac{\partial_1 \eta_p(u;\lambda) -1 - \lambda p(1-p) u^{p-2}}{u^{p-2}} } \nonumber \\
&=& \lim_{u \rightarrow \infty} \frac{\lambda p (1-p) u^{p-2} - \lambda p(1-p) \eta_p^{p-2}(u; \lambda)+ \lambda^2 p^2 (1-p)^2 u^{p-2} \eta_p^{p-2}(u; \lambda)  }{u^{p-2}(1+ \lambda p (p-1) \eta_p^{p-2}(u;\lambda)) } =0.
\end{eqnarray}
To obtain the last equality, we have employed the following equalities that are proved in the last lemma:
\begin{eqnarray*}
\lim_{u \rightarrow \infty} \eta_p(u;\lambda)=\infty; \quad \lim_{u \rightarrow \infty} \frac{\eta_p(u;\lambda)}{u} =1.
\end{eqnarray*}

$\hfill \Box$
\vspace{.2cm}


\subsection{Smoothness of state evolution function $\Psi_{\lambda, p} (\sigma^2)$}

In the paper there are many instances at which we require the derivatives of $\Psi_{\lambda_*, p} (\sigma^2)$ or $\Psi_{\lambda_*, p_*} (\sigma^2)$. In this section, we prove all the smoothness properties that are require throughout the paper. For simplicity we define the following notations:
\begin{eqnarray*}
 H_p(\sigma,\lambda)&\triangleq&\mathbb{E}[\eta_p(X+\sigma Z;\lambda)-X]^2,  \\
 \lambda_*(\sigma)&\triangleq&\arg\min_{\lambda \geq 0} H_p(\sigma, \lambda). 
\end{eqnarray*}
Note that 
\[
\Psi_{\lambda_*, p}(\sigma^2)= \frac{1}{\delta}H_p(\sigma,\lambda_*(\sigma)).
\]

\begin{lemma}\label{lemma:partialderivative}
If $\sigma_0>0$ and $\lambda_0>0$, then $\frac{\partial H_p(\sigma,\lambda)}{\partial \sigma}$  exists at $\sigma_0$ and $\lambda_0$ and is equal to
\[
\left. \frac{\partial H_p(\sigma,\lambda)}{\partial \sigma} \right|_{(\sigma_0, \lambda_0)} = \frac{1}{\sigma} \mathbb{E} [(\eta_p(X+ \sigma Z; \lambda) -X)^2(Z^2-1)]. 
\] 
\end{lemma}
\textit{Proof:}
Let $F$ denote the CDF of $X$. Then,
\begin{eqnarray*}
H_p(\sigma,\lambda)&=&\int_{-\infty}^{\infty}\int_{-\infty}^{\infty}(\eta_p(x+\sigma z;\lambda)-x)^2\phi(z)dzdF(x) \\
&=& \frac{1}{\sigma} \int_{-\infty}^{\infty}\int_{-\infty}^{\infty}(\eta_p(z;\lambda)-x)^2\phi((z-x)/\sigma)dzdF(x).
\end{eqnarray*}
Hence our first goal is to show that $ \int_{-\infty}^{\infty}\int_{-\infty}^{\infty}(\eta_p(z;\lambda)-x)^2\phi((z-x)/\sigma)dzdF(x)$ is differentiable and that the derivative may move inside the integral. For the moment we assume that $\sigma> \sigma_0$ and we calculate
\[
\lim_{\sigma \rightarrow \sigma_0} \int_{-\infty}^{\infty}\int_{-\infty}^{\infty}(\eta_p(z;\lambda)-x)^2\frac{\phi((z-x)/\sigma)-\phi((z-x)/\sigma_0) }{\sigma- \sigma_0}dzdF(x).
\]
From mean value theorem we conclude that
\[
\frac{\phi((z-x)/\sigma)-\phi((z-x)/\sigma_0) }{\sigma- \sigma_0} = \frac{|z-x|^2}{ \tilde{\sigma}^3} \phi((z-x)/\tilde{\sigma}),
\]
where $\tilde{\sigma} \in [\sigma_0, \sigma]$. It is straightforward to confirm that 
\begin{eqnarray}
\lefteqn{\int_{-\infty}^{\infty}\int_{-\infty}^{\infty}(\eta_p(z;\lambda)-x)^2\frac{\phi((z-x)/\sigma)-\phi((z-x)/\sigma_0) }{\sigma- \sigma_0}dzdF(x)} \nonumber \\
&=& \int_{-\infty}^{\infty}\int_{-\infty}^{\infty}(\eta_p(z;\lambda)-x)^2\frac{|z-x|^2}{ \tilde{\sigma}^3} \phi((z-x)/\tilde{\sigma}) dzdF(x)  \nonumber \\
&=& \int_{-\infty}^{\infty}\int_{-\infty}^{\infty}(\eta_p(x+z;\lambda)-x)^2\frac{|z|^2}{ \tilde{\sigma}^3} \phi(z/\tilde{\sigma}) dzdF(x) \nonumber \\
&\leq& \int_{-\infty}^{\infty}\int_{-\infty}^{\infty}(2|x|+|z|)^2\frac{|z|^2}{ \tilde{\sigma}^3} \phi(z/\tilde{\sigma}) dzdF(x) \leq \infty. \nonumber
\end{eqnarray}
Hence, the condition of dominated convergence theorem holds and we can switch the integrals and the derivative to obtain
\begin{eqnarray}
\frac{\partial H_p(\sigma, \lambda)}{\partial \sigma}&=&\frac{-1}{\sigma^2}\int_{-\infty}^{\infty}\int_{-\infty}^{\infty}(\eta_p(z;\lambda)-x)^2\phi((z-x)/\sigma)f(x)dzdx + \nonumber \\
&&\frac{1}{\sigma} \int_{-\infty}^{\infty}\int_{-\infty}^{\infty}(\eta_p(z;\lambda)-x)^2\phi((z-x)/\sigma)f(x)\frac{(z-x)^2}{\sigma^3}dzdx \nonumber \\
&=&\frac{1}{\sigma}\int_{-\infty}^{\infty}\int_{-\infty}^{\infty}(\eta_p(x+\sigma z;\lambda)-x)^2(z^2-1)\phi(z)f(x)dzdx \nonumber \\
&=&\frac{1}{\sigma} \mathbb{E}[(\eta_p(X+\sigma Z;\lambda)-X)^2(Z^2-1)].  \label{avoid:one}
\end{eqnarray}

\begin{lemma}\label{dev:cont}
$\frac{\partial H_p(\sigma, \lambda)}{\partial \sigma}$ is a continuous function of $(\lambda, \sigma)$ for any $\lambda>0$ and $\sigma>0$. 
\end{lemma}

\textit{Proof:} Define $J(x, \sigma, \lambda) \triangleq \mathbb{E}[(\eta_p(x+\sigma Z;\lambda)-x)^2(Z^2-1)]$. Lemma \ref{lemma:partialderivative} proves that
\[
\frac{\partial H_p(\sigma, \lambda)}{\partial \sigma}=\frac{1}{\sigma} \mathbb{E}[J(X,\sigma, \lambda)].
\]
We first show that $J(x,\sigma, \lambda)$ is continuous for any $\lambda>0, \sigma>0$, given any fixed $x$. We start by rewriting $J(x,\sigma,\lambda)$:
\begin{eqnarray*}
J(x,\sigma,\lambda)=\underbrace{\mathbb{E}[\eta^2_p(x+\sigma Z;\lambda)(Z^2-1)]}_{\triangleq \tilde{J}(x,\sigma,\lambda)}-2x\underbrace{\mathbb{E}[\eta_p(x+\sigma Z;\lambda)(Z^2-1)]}_{\triangleq\bar{J}(x,\sigma,\lambda)}.
\end{eqnarray*}
Regarding $\tilde{J}(x,\sigma,\lambda)$ we have
\begin{eqnarray}\label{exponential:trick}
\tilde{J}(x,\sigma,\lambda)&\overset{(a)}{=}&\lambda^{\frac{2}{2-p}}\mathbb{E}[\eta^2_p(\lambda^{\frac{1}{p-2}}(x+\sigma Z);1)(Z^2-1)] \nonumber \\
&=&\lambda^{\frac{2}{2-p}} \int_{-\infty}^{\infty}\eta^2_p(\lambda^{\frac{1}{p-2}}(x+\sigma z);1)(z^2-1)\phi(z)dz. \nonumber \\
&=&\frac{\lambda^{\frac{3}{2-p}}}{\sigma} \int_{-\infty}^{\infty} \eta^2_p(z;1)\Big[ \Big(\frac{\lambda^{\frac{1}{2-p}}z-x}{\sigma}\Big)^2 -1\Big]\phi\Big(\frac{\lambda^{\frac{1}{2-p}}z-x}{\sigma}\Big)dz \nonumber \\
&=&\frac{\lambda^{\frac{5}{2-p}}e^{\frac{-x^2}{2\sigma^2}}}{\sigma^3} \int_{-\infty}^{\infty} \eta^2_p(z;1)z^2{\rm exp}\Big(\frac{\lambda^{\frac{2}{2-p}}}{-2\sigma^2}\cdot z^2+\frac{x\lambda^{\frac{1}{2-p}}}{\sigma^2}\cdot z \Big)dz+\nonumber \\
&&\frac{-2x\lambda^{\frac{4}{2-p}}e^{\frac{-x^2}{2\sigma^2}}}{\sigma^2}\int_{-\infty}^{\infty} \eta^2_p(z;1)z{\rm exp}\Big(\frac{\lambda^{\frac{2}{2-p}}}{-2\sigma^2}\cdot z^2+\frac{x\lambda^{\frac{1}{2-p}}}{\sigma^2}\cdot z \Big)dz+ \nonumber \\
&&\frac{(x^2-\sigma^2)\lambda^{\frac{3}{2-p}}e^{\frac{-x^2}{2\sigma^2}}}{\sigma^3} \int_{-\infty}^{\infty} \eta^2_p(z;1){\rm exp}\Big(\frac{\lambda^{\frac{2}{2-p}}}{-2\sigma^2}\cdot z^2+\frac{x\lambda^{\frac{1}{2-p}}}{\sigma^2}\cdot z \Big)dz. 
\end{eqnarray}
We have used Lemma \ref{lem:scaleinv} (ii) to derive $(a)$. Denote $\xi_1= \frac{\lambda^{\frac{2}{2-p}}}{-2\sigma^2}, \xi_2=\frac{x\lambda^{\frac{1}{2-p}}}{\sigma^2}$. Then $p_{(\xi_1,\xi_2)}(z)\triangleq c(\xi_1,\xi_2){\rm exp}(\xi_1 z^2+\xi_2 z)$ defines a two-parameter exponential family with natural parameter space $\{(\xi,\xi_2) \mid \xi_1 <0, \xi_2 \in \mathbb{R})\}$, where $c(\xi_1,\xi_2)$ is the normalization constant. Hence according to Theorem 2.7.1 in \cite{lehmann1986testing}, $ \int_{-\infty}^{\infty}\eta^2_p(z;1)z^2{\rm exp} (\xi_1 z^2 +\xi_2 z )dz$ is continuous with respect to $(\xi_1, \xi_2)$ in the natural parameter space. It further implies that $ \int_{-\infty}^{\infty}\eta^2_p(z;1)z^2{\rm exp}\Big ( \frac{\lambda^{\frac{2}{2-p}}}{-2\sigma^2}\cdot z^2 +\frac{x\lambda^{\frac{1}{2-p}}}{\sigma^2}\cdot z \Big)dz$ is continuous for $\lambda >0, \sigma>0$. Therefore, we can conclude the first term on the right hand side of \eqref{exponential:trick} is continuous. Similar arguments work for the second and third terms. Showing the continuity of $\bar{J}(x,\sigma,\lambda)$ is also similar and is skipped. Now consider any given $\sigma_0>0,\lambda_0>0$. It is straightforward to verify the existance of $c_1, c_2>0$ such that
\begin{eqnarray}\label{uniform:dct}
|J(x,\sigma, \lambda)| \leq \mathbb{E} [(2|x+\sigma Z|^2+2x^2)(Z^2+1)]\leq \mathbb{E} [4\sigma^2 Z^2+6x^2)(Z^2+1)]=c_1x^2+c_2,
\end{eqnarray}
Hence we can apply dominated convergence theorem to obtain
\[
\lim_{\substack{\lambda \rightarrow \lambda_0\\ \sigma \rightarrow \sigma_0}}\mathbb{E}[J(X,\sigma,\lambda)]=\mathbb{E}\lim_{\substack{ \lambda \rightarrow \lambda_0 \\ \sigma \rightarrow \sigma_0}}J(X,\sigma,\lambda)=\mathbb{E}[J(X,\sigma_0,\lambda_0)].
\]

\begin{lemma}\label{dervlambda:cont}
$\frac{\partial H_p(\sigma, \lambda)}{\partial \lambda}$ is a continuous function of $(\lambda, \sigma)$ for any $\lambda>0$ and $\sigma>0$. 
\end{lemma}
The proof is similar to the proof of Lemma \ref{dev:cont} and is hence skipped here.

\begin{lemma}\label{optima:cont}
For a given $\sigma_0>0$, suppose the optimal thresholding value $\lambda_*(\sigma_0)$ satisfies the condition:
\[
\inf_{\lambda \geq 0} H_p(\sigma_0,\lambda) < \inf_{|\lambda-\lambda_*(\sigma_0)|>c} H_p(\sigma_0,\lambda),
\]
for any $c>0$. Then $\lambda_*(\sigma)$ is continuous at $\sigma = \sigma_0$.
\end{lemma} 
\begin{proof}
According to Lemma \ref{lemma:partialderivative} we have
\begin{eqnarray*}
\frac{\partial H_p(\sigma, \lambda)}{\partial \sigma} = \frac{1}{\sigma} \mathbb{E}[(\eta_p(X+\sigma Z;\lambda)-X)^2(Z^2-1)] 
\end{eqnarray*}
Furthermore,
\begin{eqnarray}\label{upper:bound}
\Big |\frac{\partial H_p(\sigma, \lambda)}{\partial \sigma} \Big | \leq \frac{1}{\sigma }\mathbb{E} [(Z^2+1)(2\eta^2_p(X+\sigma Z;\lambda)+2X^2)] \leq \frac{1}{\sigma}\mathbb{E}[(6X^2+4\sigma^2 Z^2)(Z^2+1)]
\end{eqnarray}
Note that the upper bound above does not depend on $\lambda$. This implies that for any given $\sigma_0>0$, there exists a neighborhood $B_{r}(\sigma_0)$ such that the following holds for any $\sigma \in B_r(\sigma_0)$:
\begin{eqnarray*}
\sup_{\lambda} |H_p(\sigma,\lambda)-H_p(\sigma_0,\lambda)| \leq K(\sigma_0) \cdot |\sigma-\sigma_0|,
\end{eqnarray*}
where $K(\sigma_0)$ is a constant depending on $\sigma_0$. We then have
\begin{eqnarray*}
&& H_p(\sigma, \lambda_*(\sigma))-H_p(\sigma_0,\lambda_*(\sigma_0)) \\
&=&[H_p(\sigma,\lambda_*(\sigma))-H_p(\sigma,\lambda_*(\sigma_0))]+[H_p(\sigma,\lambda_*(\sigma_0))-H_p(\sigma_0,\lambda_*(\sigma_0))] \\
&\leq& \sup_{\lambda} |H_p(\sigma,\lambda)-H_p(\sigma_0,\lambda)|  \leq K(\sigma_0) \cdot |\sigma-\sigma_0|
\end{eqnarray*}
On the other hand, 
\begin{eqnarray*}
&& H_p(\sigma, \lambda_*(\sigma))-H_p(\sigma_0,\lambda_*(\sigma_0)) \\
&=&[H_p(\sigma,\lambda_*(\sigma))-H_p(\sigma_0,\lambda_*(\sigma))]+[H_p(\sigma_0,\lambda_*(\sigma))-H_p(\sigma_0,\lambda_*(\sigma_0))]  \\
&\geq&- \sup_{\lambda} |H_p(\sigma,\lambda)-H_p(\sigma_0,\lambda)|  \geq -K(\sigma_0) \cdot |\sigma-\sigma_0|
\end{eqnarray*}
Therefore, we obtain 
\begin{eqnarray}\label{side:one}
|\inf_{\lambda \geq 0}H_p(\sigma,\lambda)-\inf_{\lambda \geq 0}H_p(\sigma_0,\lambda) | \leq K(\sigma_0) \cdot |\sigma-\sigma_0|
\end{eqnarray}
Similarly we can get
\begin{eqnarray}\label{side:two}
|\inf_{|\lambda-\lambda_*(\sigma_0)| \geq \epsilon}H_p(\sigma,\lambda)-\inf_{|\lambda-\lambda_*(\sigma_0)| \geq \epsilon}H_p(\sigma_0,\lambda) | \leq K(\sigma_0) \cdot |\sigma-\sigma_0|
\end{eqnarray}
Now for any given $\varepsilon >0$, by the condition we impose, there exists a constant $d>0$ such that 
\[
\inf_{|\lambda-\lambda_*(\sigma_0)|>\varepsilon} H_p(\sigma_0,\lambda) -\inf_{\lambda \geq 0} H_p(\sigma_0,\lambda) >d
\]
This combined with Equations \eqref{side:one} and \eqref{side:two} yields,
\begin{eqnarray*}
\inf_{|\lambda-\lambda_*(\sigma_0)|>\varepsilon} H_p(\sigma,\lambda) -\inf_{\lambda \geq 0} H_p(\sigma,\lambda)> d-2K(\sigma_0) \cdot | \sigma-\sigma_n| > d/2>0
\end{eqnarray*}
for $\sigma \in B_r(\sigma_0)$ with sufficiently small $r$. It implies that 
\[
|\lambda_*(\sigma)-\lambda_*(\sigma_0)|\leq \varepsilon, \mbox{~for~} \sigma \in B_r(\sigma_0).
\]
This finishes the proof of the continuity.
\end{proof}

\vspace{0.4cm}

\begin{theorem}\label{final:smooth}
Suppose for any $\sigma_0>0$, the global optima $\lambda_*(\sigma_0)$ is isolated \footnote{This assumption turns out to be very mild. Based on our simulations, $H_p(\sigma,\lambda)$, as a function of $\lambda$, has quasi-convex shapes.}, i.e., 
\[
\inf_{\lambda \geq 0} H_p(\sigma_0,\lambda) < \inf_{|\lambda-\lambda_*(\sigma_0)|>c} H_p(\sigma_0,\lambda)
\]
for any $c >0$. Then $\Psi_{\lambda_*,p}(\sigma^2)$ is differentiable with respect to $\sigma$ over $(0,\infty)$ with continuous derivative and
\[
\frac{d \Psi_{\lambda_*,p}(\sigma^2)}{d \sigma}=\frac{\partial H_p(\sigma, \lambda_*(\sigma))}{\partial \sigma}
\]
\end{theorem}
\textit{Proof:} Consider a given $\sigma_0>0$. Then
\begin{eqnarray*}
\frac{d \Psi_{\lambda_*,p}(\sigma_0^2)}{d \sigma}=\lim_{\sigma \rightarrow \sigma_0} \frac{H_p(\sigma,\lambda_*(\sigma))-H_p(\sigma_0,\lambda_*(\sigma_0))}{\sigma-\sigma_0}
\end{eqnarray*}
We first assume $\sigma > \sigma_0$. Note that
\begin{eqnarray*}
 \frac{H_p(\sigma,\lambda_*(\sigma))-H_p(\sigma_0,\lambda_*(\sigma_0))}{\sigma-\sigma_0}&=&\frac{[H_p(\sigma,\lambda_*(\sigma))-H_p(\sigma,\lambda_*(\sigma_0))]+[H_p(\sigma,\lambda_*(\sigma_0))-H_p(\sigma_0,\lambda_*(\sigma_0))]}{\sigma-\sigma_0} \\
 &\leq& \frac{H_p(\sigma,\lambda_*(\sigma_0))-H_p(\sigma_0,\lambda_*(\sigma_0))}{\sigma-\sigma_0}
 \end{eqnarray*}
 Hence we have
 \begin{eqnarray}\label{twoside:one}
 \limsup_{\sigma \rightarrow \sigma^+_0}\frac{H_p(\sigma,\lambda_*(\sigma))-H_p(\sigma_0,\lambda_*(\sigma_0))}{\sigma-\sigma_0} \leq \frac{\partial H_p(\sigma_0, \lambda_*(\sigma_0))}{\partial \sigma}
 \end{eqnarray}
 On the other hand, we have
 \begin{eqnarray*}
  \frac{H_p(\sigma,\lambda_*(\sigma))-H_p(\sigma_0,\lambda_*(\sigma_0))}{\sigma-\sigma_0}&=&\frac{[H_p(\sigma,\lambda_*(\sigma))-H_p(\sigma_0,\lambda_*(\sigma))]+[H_p(\sigma_0,\lambda_*(\sigma))-H_p(\sigma_0,\lambda_*(\sigma_0))]}{\sigma-\sigma_0} \\
 &\geq& \frac{H_p(\sigma,\lambda_*(\sigma))-H_p(\sigma_0,\lambda_*(\sigma))}{\sigma-\sigma_0}=\frac{\partial H_p(\tilde{\sigma}, \lambda_*(\sigma))}{\partial \sigma},
 \end{eqnarray*}
where $\tilde{\sigma}$ is between $\sigma$ and $\sigma_0$. Since we have showed from Lemma \ref{dev:cont} and \ref{optima:cont} that $\lambda_*(\sigma)$ and $\frac{\partial H_p(\sigma, \lambda)}{\partial \sigma}$ are both continuous, we can conclude from the above inequality that
\begin{eqnarray}\label{twoside:two}
 \liminf_{\sigma \rightarrow \sigma^+_0}\frac{H_p(\sigma,\lambda_*(\sigma))-H_p(\sigma_0,\lambda_*(\sigma_0))}{\sigma-\sigma_0} \geq \frac{\partial H_p(\sigma_0, \lambda_*(\sigma_0))}{\partial \sigma}.
\end{eqnarray}
Inequalities \eqref{twoside:one} and \eqref{twoside:two} together show that
\[
 \lim_{\sigma \rightarrow \sigma^+_0}\frac{H_p(\sigma,\lambda_*(\sigma))-H_p(\sigma_0,\lambda_*(\sigma_0))}{\sigma-\sigma_0} =\frac{\partial H_p(\sigma_0, \lambda_*(\sigma_0))}{\partial \sigma}
\]
Similarly, we can prove the same equality when $\sigma \rightarrow \sigma^-_0$. Thus we can obtain that 
\[
\frac{d \Phi_{\lambda_*,p}(\sigma_0^2)}{d \sigma}=\frac{\partial H_p(\sigma_0, \lambda_*(\sigma_0))}{\partial \sigma}
\]
Since $\frac{\partial H_p(\sigma, \lambda)}{\partial \sigma}$ and $\lambda_*(\sigma)$ are both continuous, we know $\frac{\partial H_p(\sigma, \lambda_*(\sigma))}{\partial \sigma}$ is continuous as well.

\vspace{0.4cm}

\begin{theorem}\label{thm:smoothplambdaopt}
Denote $\theta=(\lambda, p)$. Suppose for any $\sigma_0>0$, the global optima $\theta_*(\sigma_0)$ is isolated, i.e., 
\[
\inf_{\lambda \geq 0, 0 \leq p \leq 1} H_p(\sigma_0,\lambda) < \inf_{||\theta-\theta_*(\sigma_0)|| >c} H_p(\sigma_0,\lambda),
\]
for any $c >0$. Then $\Psi_{\lambda_*,p_*}(\sigma^2)$ is differentiable with respect to $\sigma$ over $(0,\infty)$ with continuous derivative and 
\[
\frac{ d \Psi_{\lambda_*,p_*}(\sigma^2)}{d \sigma}=\frac{ \partial H_{p_*(\sigma)}(\sigma,\lambda_*(\sigma))}{\partial \sigma}.
\]
\end{theorem}

\begin{proof}
First note that we can prove $\theta^*(\sigma)$ is continuous over $(0,\infty)$. It follows the same route as the proof of Lemma \ref{optima:cont}. The key observation is that the upper bound on $\frac{\partial H_p(\sigma, \lambda)}{\partial \sigma}$ we showed in \eqref{upper:bound} does not depend on either $p$ or $\lambda$. For the sake of brevity we skip the complete proof. 

The rest of the proof is also very similar to the proof of Theorem \ref{final:smooth}. Note that the key ingredient in the proof of Theorem \ref{final:smooth} is the continuity $\frac{\partial H_p(\sigma, \lambda)}{\partial \sigma}  $ with respect to $(\sigma, \lambda)$. In order to extend that proof to Theorem \ref{thm:smoothplambdaopt}, we should show that $\frac{\partial H_p(\sigma, \lambda)}{\partial \sigma}  $ is continuous with respect to $(\sigma, \lambda, p)$. Recall from \eqref{avoid:one} that
\begin{eqnarray*}
\frac{\partial H_p(\sigma, \lambda)}{\partial \sigma} =\frac{1}{\sigma} \mathbb{E}[(\eta_p(X+\sigma Z;\lambda)-X)^2(Z^2-1)].
\end{eqnarray*}
We can use the same arguments as presented for proving Lemma \ref{lemma:partialderivative} to calculate 
\begin{eqnarray*}
\frac{\partial^2 H_p(\sigma, \lambda)}{\partial \sigma^2}=\frac{1}{\sigma^2}\mathbb{E} [(\eta_p(X+\sigma Z;\lambda)-X)^2(Z^4-5Z^2+2)].
\end{eqnarray*}
Hence, it is straightforward to verify that 
\[
\Big |\frac{\partial^2 H_p(\sigma, \lambda)}{\partial \sigma^2}\Big| \leq \frac{1}{\sigma^2} \mathbb{E} [(6X^2+4\sigma^2 Z^2)(Z^4+5Z^2+2)].
\]
Note that the upper bound above is independent of both $p$ and $\lambda$. Thus according to mean value theorem, $\frac{\partial H_p(\sigma, \lambda)}{\partial \sigma} $ is Lipschitz continuous (with a Lipschitz constant that does not depend on $\lambda$ and $p$)  over $(p, \lambda)$ with respect to $\sigma>0$. If we can further show $ \mathbb{E}[(\eta_p(X+\sigma Z;\lambda)-X)^2(Z^2-1)] $ is continuous with respect to $(p, \lambda)$ for any given $\sigma >0$, we are done. For that purpose, we do the analysis in two steps:
\begin{itemize}
\item  Firstly, we will show $ \mathbb{E}[(\eta_p(X+\sigma Z;\lambda)-X)^2(Z^2-1)] $ is continuous with respect to $p$, for any given $\lambda >0$.
\item We then prove $ \mathbb{E}[(\eta_p(X+\sigma Z;\lambda)-X)^2(Z^2-1)] $ is continuous with respect to $\lambda$ uniformly over $p$.
\end{itemize}
Regarding the first step, note that as $\tilde{p} \rightarrow p$, 
\begin{eqnarray*}
&&(\eta_{\tilde{p}}(X+\sigma Z;\lambda)-X)^2(Z^2-1)\mathbbm{1}(|X+\sigma Z| \neq c_p \lambda^{\frac{1}{2-p}}) \\
& \rightarrow &(\eta_p(X+\sigma Z;\lambda)-X)^2(Z^2-1)\mathbbm{1}(|X+\sigma Z| \neq c_p \lambda^{\frac{1}{2-p}})
\end{eqnarray*}
 Also since $|(\eta_p(X+\sigma Z;\lambda)-X)^2(Z^2-1)|\leq 2(|X+\sigma Z|^2+X^2)(Z^2+1)$, we can apply DCT to conclude it. For the second step, recall the definition in the proof of Lemma \ref{dev:cont}:
 \[ 
 J(x,\lambda, p) =\mathbb{E}[(\eta_p(x+\sigma Z;\lambda)-x)^2(Z^2-1) ]. 
\] 
 We then have $ \mathbb{E}[(\eta_p(X+\sigma Z;\lambda)-X)^2(Z^2-1)] =\mathbb{E}_X J(X,\lambda, p)$. If we can show $\lim_{\lambda \rightarrow \lambda_0}J(X,\lambda, p)= J(X,\lambda_0, p) $ uniformly over $p$, then by \eqref{uniform:dct} we can apply uniform DCT to finish the proof. Hence, what left to prove is $J(x,\lambda, p)$ is uniformly continuous over $p$ with respect to $\lambda>0$, for any given $x$. We rewrite $J(x,\lambda, p)$ as 
\begin{eqnarray*}
J(x, \lambda, p)&=&\int_{-\infty}^{\infty} (\eta_p(x+\sigma z;\lambda)-x)^2 (z^2-1)\phi(z)dz \\
&=&\int_{-\infty}^{\infty} (\lambda^{\frac{1}{2-p}} \eta_p(\lambda^{\frac{1}{p-2}}(x+\sigma z);1)-x)^2 (z^2-1)\phi(z)dz \\
&=& \int_{-\infty}^{\infty}\underbrace{\frac{\lambda^{\frac{1}{2-p}}}{\sigma} (\lambda^{\frac{1}{2-p}}\eta_p(z;1)-x)^2\Big [\Big(\frac{z\lambda^{\frac{1}{2-p}}-x}{\sigma} \Big)^2-1\Big]\phi\Big ( \frac{z\lambda^{\frac{1}{2-p}}-x}{\sigma}   \Big)}_{\triangleq K(x,\lambda,p,z)}dz
\end{eqnarray*}
Since $\sup_{0\leq p \leq 1}|\eta_p(z;1)| \leq |z|$ and 
\[
\sup_{0\leq p \leq 1}|\lambda^{\frac{1}{2-p}} | \leq \max(1, \lambda),  \quad \inf_{0\leq p \leq 1}|\lambda^{\frac{1}{2-p}} | \geq \min(\lambda,\lambda^{1/2}),
\]
 for a given small neighbor $B_{r}(\lambda_0)$, we can easily find an upper bound $L(x,z)$ such that 
\[
\sup_{0\leq p \leq 1}|K(x,\lambda, p, z)| \leq L(x, z)
\]
holds for all $\lambda \in B_r(\lambda_0)$ and $\int_{-\infty}^{\infty}L(x,z)dz < \infty$. Moreover, note that $\lambda^{\frac{1}{2-p}}$ is uniformly continuous at any $\lambda >0$, thus it is easy to see $K(x, \lambda, p,z)$ is uniformly continuous as well. We can then apply uniform DCT again to show $J(x, \lambda, p)$ is uniformly continuous. 

\end{proof}


\subsection{Proof of Theorem \ref{conj:se}}\label{proof:thm1}
According to \eqref{eq:smoothetap}, we have
\begin{equation*}
\tilde{\eta}_{p,h} (u;\lambda) = S_p(u;\lambda) + \tilde{D}_{p,h}(u; \lambda).  
\end{equation*}
Hence,
\begin{equation*}
\tilde{\eta}'_{p,h} (u;\lambda) = S_p'(u;\lambda) + \tilde{D}_{p,h}'(u; \lambda),
\end{equation*}
where $(\cdot)'$ denotes the derivative with respect to the first argument of the function.  Let $\tilde{\lambda}_p$ denote the threshold specified in Lemma \ref{lem:threshform}. 
According to \eqref{eq:defS_p}, the derivative of $S_p(u;\lambda)$ is the same as the derivative of ${\eta}_{p} (u;\lambda)$ for every $|u| > \tilde{\lambda}_p$. Moreover, from Lemma \ref{lem:derivativeproperties} part (ii) we already know that $\sup_{u}|{\eta}'_{p} (u;\lambda)| <\infty$. Hence, our first conclusion is the following:
\begin{equation}\label{eq:ST}
\sup_{u} |S_p'(u;\lambda)| = \sup_u \left| {\eta}'_{p} (u;\lambda)  \right|< \infty. 
\end{equation}
Next we claim that the derivative of $\tilde{D}_{p,h}(u; \lambda)$ with respect to $u$ is bounded as well. To prove this claim, first note that
\begin{equation}
 \tilde{D}_{p,h}(u; \lambda) = {\eta}_p^+(\tilde{\lambda}_p;\lambda )\int \frac{1}{\sqrt{2\pi} h} {\rm e}^{-\frac{(u-s)^2}{2h^2}} \mathbb{I} (s> \tilde{\lambda}_p)ds+  {\eta}_p^-(-\tilde{\lambda}_p;\lambda )\int \frac{1}{\sqrt{2\pi} h} {\rm e}^{-\frac{(u-s)^2}{2h^2}} \mathbb{I} (s<- \tilde{\lambda}_p)ds. \label{inside:limit}
\end{equation}
Therefore, it is straightforward to use the dominated convergence theorem to show that
\[
 \tilde{D}'_{p,h}(u; \lambda) = {\eta}^+_p(\tilde{\lambda}_p;\lambda ) \int \frac{1}{\sqrt{2\pi} h^3} (s-u){\rm e}^{-\frac{(u-s)^2}{2h^2}} \mathbb{I} (s> \tilde{\lambda}_p)ds+{\eta}_p^-(\tilde{\lambda}_p;\lambda ) \int \frac{1}{\sqrt{2\pi} h^3} (s-u){\rm e}^{-\frac{(u-s)^2}{2h^2}} \mathbb{I} (s <- \tilde{\lambda}_p)ds.  
\]
Hence,
\begin{eqnarray}\label{eq:D_p}
| \tilde{D}'_{p,h}(u; \lambda)| \leq 2 {\eta}_p^+(\tilde{\lambda}_p;\lambda )  \int \frac{1}{\sqrt{2\pi} h^3} |u-s|{\rm e}^{-\frac{(u-s)^2}{2h^2}} ds = 4{\eta}_p^+(\tilde{\lambda}_p;\lambda )  \int_{0}^\infty \frac{1}{\sqrt{2\pi} h^3} z {\rm e}^{-\frac{z^2}{2h^2}} dz = \frac{4{\eta}_p^+(\tilde{\lambda}_p;\lambda )}{\sqrt{2\pi} h}. 
\end{eqnarray}
Combining \eqref{eq:ST} and \eqref{eq:D_p} proves that $\sup_{u} |\tilde{\eta}'_{p,h} (u;\lambda)|$ is bounded. Hence, by the mean value theorem we can conclude that $\tilde{\eta}_{p,h} (u;\lambda)$ is Lipschitz continuous. Under the Lipschitz continuity of $\tilde{\eta}_{p,h} (u;\lambda)$, we can employ Theorem 1 of \cite{bayati2012lasso} to show that:
\begin{equation}\label{eq:setemp1}
\lim_{N \rightarrow \infty}  \frac{\|x^{t+1}(N,h)-x_o(N)\|_2^2}{N} \overset{\rm a.s.}{=} \mathbb{E} \left( {{{\left| {\tilde{\eta}_{p,h} (X + {\sigma _{t,h}}Z;{\lambda _t}) - X} \right|}^2}} \right),
\end{equation}
where ${\sigma _{t,h}}$ satisfies the following equation:
\begin{equation}\label{eq:sefirstdraft}
\sigma _{t+1,h}^2 = \sigma_w^2 + \frac{1}{\delta} \mathbb{E} \left( {{{\left| {\tilde{\eta}_{p,h} (X + {\sigma _{t,h}}Z;{\lambda _t}) - X} \right|}^2}} \right).
\end{equation}
It is straightforward to employ \eqref{eq:setemp1} and conclude that 
\[
\lim_{i \rightarrow \infty} \lim_{N \rightarrow \infty}  \frac{\|x^{t+1}(N,h_i)-x_o(N)\|_2^2}{N} \overset{\rm a.s.}{=} \lim_{i \rightarrow \infty} \mathbb{E} \left( {{{\left| {\tilde{\eta}_{p,h_i} (X + {\sigma _{t,h_i}}Z;{\lambda _t}) - X} \right|}^2}} \right).
\]
The last step is to prove that 
\begin{equation}\label{eq:hzero}
\lim_{i \rightarrow \infty} \mathbb{E} \left( {{{\left| {\tilde{\eta}_{p,h_i} (X + {\sigma _{t,h_i}}Z;{\lambda _t}) - X} \right|}^2}} \right) = \mathbb{E} \left( {{{\left| {{\eta}_{p} (X + {\sigma _t}Z;{\lambda _t}) - X} \right|}^2}} \right),
\end{equation}
with $\sigma_t$ satisfying
\begin{equation}\label{eq:sefirstdraft}
\sigma _{t+1}^2 = \sigma_w^2 + \frac{1}{\delta} \mathbb{E} \left( {{{\left| {{\eta}_{p} (X + {\sigma _{t}}Z;{\lambda _t}) - X} \right|}^2}} \right).
\end{equation}
We use an induction on $t$ to prove \eqref{eq:hzero}. 
\begin{enumerate}

\item[(i)] Base of the induction: First note that $\sigma_0 = \sigma_{0,h}$. Hence, we have to prove that 
\begin{equation}\label{eq:inductionbase}
\lim_{i \rightarrow \infty} \mathbb{E} \left( {{{\left| {\tilde{\eta}_{p,h_i} (X + \sigma_0 Z;{\lambda _0}) - X} \right|}^2}} \right) = \mathbb{E} \left( {{{\left| {{\eta}_{p} (X + {\sigma _0}Z;{\lambda _0}) - X} \right|}^2}} \right)
\end{equation}
According to Lemma \ref{lem:derivativeproperties}, we have
\begin{equation*}
|\tilde{\eta}_{p,h} (u;\lambda)| \leq |S_p(u;\lambda) |+ |\tilde{D}_{p,h}(u; \lambda)| \leq |u|+  {\eta}^+_p(\tilde{\lambda}_p;\lambda ).
\end{equation*}
Define $\tilde{\lambda}_{0,p} \triangleq c_p \lambda_0^{1/(2-p)}$, where $c_p$ is the constant we defined in Lemma \ref{lem:threshform}. We have 
\[
|\tilde{\eta}_{p,h} (X + {\sigma _0}Z;{\lambda _0})| \leq |X + \sigma_0 Z| + {\eta}^+_p(\tilde{\lambda}_{0,p};\lambda_0 ). 
\]
Hence,
\[
 {\left| {\tilde{\eta}_{p,h} (X + {\sigma _0}Z;{\lambda _0}) - X} \right|}^2 \leq (|X+ \sigma_0 Z|+|X|+{\eta}^+_p(\tilde{\lambda}_{0,p};\lambda_0 ))^2.
\]
Since $ \mathbb{E}(|X+ \sigma_0 Z|+|X|+ {\eta}^+_p(\tilde{\lambda}_{0,p};\lambda_0 ))^2<\infty$, if we can show 
\begin{equation}
\lim_{h\rightarrow 0+}\tilde{\eta}_{p,h}(u;\lambda)=\eta_{p}(u;\lambda),  \label{inside:limit2}
\end{equation}
then by the dominated convergence theorem we can conclude \eqref{eq:inductionbase}. To show \eqref{inside:limit2}, first notice\begin{eqnarray*}
\int \frac{1}{\sqrt{2\pi} h} {\rm e}^{-\frac{(u-s)^2}{2h^2}} \mathbb{I} (s> \tilde{\lambda}_p)ds=\int_{\tilde{\lambda}_p-u}^{\infty}\frac{1}{\sqrt{2\pi}h}e^{-\frac{z^2}{2h^2}}dz=\int_{\frac{\tilde{\lambda}_p-u}{h}}^{\infty}\frac{1}{\sqrt{2\pi}}e^{-\frac{z^2}{2}}dz.
\end{eqnarray*}
Therefore, it is straightforward to confirm,
\begin{eqnarray*}
\lim_{h\rightarrow 0+}\int \frac{1}{\sqrt{2\pi} h} {\rm e}^{-\frac{(u-s)^2}{2h^2}} \mathbb{I} (s> \tilde{\lambda}_p)ds=
\begin{cases}
1 & \text{if } u>\tilde{\lambda}_p, \\
0 & \text{if } u< \tilde{\lambda}_p.  
\end{cases}
\end{eqnarray*}
Similarly, we can show
\begin{eqnarray*}
\lim_{h\rightarrow 0+}\int \frac{1}{\sqrt{2\pi} h} {\rm e}^{-\frac{(u-s)^2}{2h^2}} \mathbb{I} (s<- \tilde{\lambda}_p)ds=
\begin{cases}
1 & \text{if } u< -\tilde{\lambda}_p, \\
0 & \text{if } u> - \tilde{\lambda}_p.  
\end{cases}
\end{eqnarray*}
Combining the two equalities above with \eqref{inside:limit} proves that $\lim_{h\rightarrow 0+}\tilde{D}_{p,h}(u;\lambda)=D_{p}(u;\lambda)$, which in turn shows $\lim_{h\rightarrow 0+}\tilde{\eta}_{p,h}(u;\lambda)=\eta_{p}(u;\lambda)$. This completes the proof. 

\item[(ii)] Inductive step: Now we assume that \eqref{eq:hzero} is true for iteration $t$ and our goal is to show it for iteration $t+1$. First note that
\begin{equation}
\sigma_{t,h}^2 = \sigma_w^2 + \frac{1}{\delta} \mathbb{E} \left( {{{\left| {\tilde{\eta}_{p,h} (X + {\sigma _{t-1,h}}Z;{\lambda _{t-1}}) - X} \right|}^2}} \right).
\end{equation}
According to the assumption of induction:
 \[
\mathbb{E} \left( {{{\left| {\tilde{\eta}_{p,h_i} (X + {\sigma _{t-1,h_i}}Z;{\lambda _{t-1}}) - X} \right|}^2}} \right) \rightarrow \mathbb{E} \left( {{{\left| {{\eta}_{p} (X + {\sigma _{t-1}}Z;{\lambda _{t-1}}) - X} \right|}^2}} \right). 
\]
Hence, $\sigma_{t,h_i}^2 \rightarrow \sigma_{t}^2$ as $i \rightarrow \infty$. Moreover, note that
\begin{eqnarray*}
\tilde{\eta}_{p,h_i}(X+\sigma_{t,h_i}Z;\lambda_t)&=&S_{p}(X+\sigma_{t,h_i}Z;\lambda_t)+\tilde{D}_{p,h_i}(X+\sigma_{t,h_i}Z;\lambda_t) \\
\eta_{p}(X+\sigma_{t}Z;\lambda_t)&=&S_{p}(X+\sigma_{t}Z;\lambda_t)+D_{p}(X+\sigma_{t}Z;\lambda_t)
\end{eqnarray*}
Since $S_p(u;\lambda)$ is a continuous function of $u$, we have
\[
\lim_{i \rightarrow \infty} S_{p}(X+\sigma_{t,h_i}Z;\lambda_t)=S_{p}(X+\sigma_{t}Z;\lambda_t).
\]
Furthermore, it is not hard to see that the arguments we used to prove $\tilde{D}_{p,h}(u;\lambda)\rightarrow D_p(u;\lambda)$ in step (i) can be applied to show $\tilde{D}_{p,h}(u_h;\lambda)\rightarrow D_p(u;\lambda)$, if $u_h \rightarrow u$, as $h\rightarrow 0+$. Therefore, we can obtain
\[
\lim_{i \rightarrow \infty}  \tilde{D}_{p,h_i}(X+\sigma_{t,h_i}Z;\lambda_t)  = D_{p}(X+\sigma_{t}Z;\lambda_t).
\]
\end{enumerate}
Combining the last two equalities, we have showed that
\begin{eqnarray*}
\lim_{i \rightarrow \infty} \tilde{\eta}_{p,h_i}(X+\sigma_{t,h_i}Z;\lambda_t)=\eta_{p}(X+\sigma_{t}Z;\lambda_t).
\end{eqnarray*}
Since $\sigma_{t,h_i}$ is bounded, we can use similar calculations as in step (i) to bound $|\tilde{\eta}_{p,h_i}(X+\sigma_{t,h_i}Z;\lambda_t)-X|^2$. Hence dominated convergence theorem can be applied to conclude 
\[
\lim_{i\rightarrow \infty} \mathbb{E} \left( \left| \tilde{\eta}_{p,h_i} (X + \sigma _{t,h_i}Z;\lambda _t) - X \right| ^2 \right)=\mathbb{E} \left( \left| \eta_p (X + \sigma _t Z; \lambda _t ) - X \right|^2 \right).
\]


\subsection{Proof of Proposition \ref{proof:existsncestable}} \label{sec:proofLemmaexistence}

We have already proved in Theorem \ref{final:smooth} that $\Psi_{\lambda_*,p} (\sigma^2)$ is a continuous function of $\sigma^2$ (we have in fact proved that it is differentiable). We consider the noiseless setting $\sigma_w^2=0$. The proof for the noisy setting is essentially the same. First note that for the case $\sigma=0$, we have ${\Psi _{{0},p}}(0) = 0$.
Hence, ${\Psi _{{\lambda _*},p}}(0) = 0$. Therefore $\sigma^2 =0$ is a fixed point of $\Psi_{\lambda_*,p}$. If it is a stable fixed point, it will establish the lemma. We assume that it is an unstable fixed point. Then there exists a value of $\sigma$, called $\sigma_u$ for which 
\begin{equation}\label{eq:sigmasmall}
\Psi_{\lambda_*, p} (\sigma_u^2)> \sigma_u^2. 
\end{equation}
Furthermore, we will show that for $\sigma^2>\frac{1}{\delta} [\mathbb{E} (X^2)+1]$ we have 
\begin{equation}\label{eq:sigmainf}
{\Psi _{{\lambda _*},p}}({\sigma ^2})< \sigma^2.
\end{equation} 
Since ${\Psi _{{\lambda _*},p}}({\sigma ^2})$ is continuous, we can combine \eqref{eq:sigmasmall} and \eqref{eq:sigmainf} and conclude the existence of the stable fixed point in the range $[\sigma_u^2, \frac{1}{\delta}(\mathbb{E} (X^2)+ 2)]$. Hence, the only step that is left to prove is \eqref{eq:sigmainf}. Note that from Lemma \ref{lem:derivativeproperties} we have $|\eta_p(X+ \sigma Z; \lambda )-X| \leq |X+ \sigma Z|+|X| \leq 2 |X|  + \sigma |Z|$. Since $\mathbb{E} (2 |X|  + \sigma |Z|)^2$ is bounded, we can employ the dominated convergence theorem to get
\[
\lim_{\lambda \rightarrow \infty} \mathbb{E} (\eta_p(X+ \sigma Z; \lambda)-X)^2= \mathbb{E}  \lim_{\lambda \rightarrow \infty} (\eta_p(X+ \sigma Z; \lambda)-X)^2= \mathbb{E} (X^2). 
\]
Hence, there exists a value of $\lambda_u< \infty$ such that $\mathbb{E} (\eta_p(X+ \sigma Z; \lambda_u)-X)^2 \leq \mathbb{E} (X^2)+ 1$. Therefore,
\[
{\Psi _{{\lambda _*},p}}({\sigma ^2}) < \frac{\mathbb{E} X^2+1}{\delta},
\]
which implies \eqref{eq:sigmainf} and completes the proof.


\subsection{Proof of Theorem \ref{thm:highestfpnoiseless}}  \label{proof:theorem3noiseless}
Let $X \sim (1- \epsilon)\Delta_0+ \epsilon G$, where $G$ is an arbitrary distribution that does not have any mass at zero. Also, let $U$ denote a random variable with distribution $G$ and $\mathbb{E}_U$ be the expectation with respect to $U$. Define
\[
\psi_{\tau,p} (\sigma^2) \triangleq \frac{1}{\delta} \mathbb{E} (\eta_p(X+ \sigma Z; \tau \sigma^{2-p})-X). 
\]
Note that if $\tau = \frac{\lambda(\sigma)}{\sigma ^{2 - p}}$, then
\begin{eqnarray}\label{eq:equivpsiPsi}
\psi_{\tau, p} (\sigma^2) = \Psi _{\lambda, p}(\sigma ^2). 
\end{eqnarray}

We have
\begin{eqnarray*}
\begin{gathered}
  {\psi _{\tau ,p}}({\sigma ^2}) = \frac{1}{\delta }\left[ {(1 - \epsilon ) \mathbb{E} \left({\eta _p^2(\sigma Z;\tau {\sigma ^{2 - p}})} \right) + \epsilon \mathbb{E} {{\left( {{\eta _p}(U + \sigma Z;\tau {\sigma ^{2 - p}}) - U} \right)}^2}} \right] \hfill \\
   \overset{(a)}{=} \frac{{{\sigma ^2}}}{\delta }\left[ {(1 - \epsilon ) \mathbb{E} \left( {\eta _p^2(Z;\tau )} \right) + \epsilon \mathbb{E} {{\left( {{\eta _p}(U/\sigma  + Z;\tau ) - U/\sigma } \right)}^2}} \right] \hfill \\
   = \frac{{{\sigma ^2}}}{\delta }\left[ {(1 - \epsilon ) \mathbb{E} \left( {\eta _p^2(Z;\tau )} \right) + \epsilon \mathbb{E}_U \left[ {{\mathbb{E}_Z}{{\left( {{\eta _p}(U/\sigma  + Z;\tau ) - U/\sigma } \right)}^2}} \right]} \right] \hfill \\
   \leqslant \frac{{{\sigma ^2}}}{\delta }(1 - \epsilon )\mathbb{E}\left( {\eta _p^2(Z;\tau )} \right) + \frac{{{\sigma ^2}}}{\delta }\epsilon \   {\mathop {\sup }\limits_U {\mathbb{E}_Z{\left( {{\eta _p}(U/\sigma  + Z;\tau ) - U/\sigma } \right)}^2}},  \hfill \\
\end{gathered}
\end{eqnarray*}
where Equality (a) is due to Lemma \ref{lem:scaleinv}. It is straightforward to employ \eqref{eq:equivpsiPsi} and derive  
\begin{eqnarray*}
\frac{\Psi _{\lambda_*, p}(\sigma ^2)}{\sigma^2} = \inf_{\tau\geq 0}  \frac{\psi _{\tau ,p}(\sigma^2)}{\sigma^2} \leqslant  \inf_{\tau \geq 0}  \frac{(1 - \epsilon )}{\delta }\mathbb{E}\left[ \eta _p^2(Z;\tau ) \right] + \frac{\epsilon }{\delta }  \sup_U \mathbb{E}_Z \left[ { \left( \eta _p(U/\sigma  + Z;\tau ) - U/\sigma  \right)^2} \right] = \frac{\overline{M}_p(\epsilon)}{\delta}.
\end{eqnarray*}

Hence, if $\overline{M}_p(\epsilon)< \delta$, then the inequality above implies that $\Psi _{\lambda_*, p}(\sigma ^2) <\sigma^2$ for any $\sigma >0$, meaning ${\Psi _{\lambda_*,p}}({\sigma ^2}) $ does not have any fixed point except at zero and that fixed point is stable. Now we prove the second part of the theorem. Suppose that
\[
\underline{M}_p(\epsilon)> \delta.
\]

We would like to show that there exist certain distributions in $\mathcal{F}_{\epsilon}$ for which $\Psi _{\lambda_*, p}(\sigma ^2)$ has a non-zero stable fixed point. Suppose that $X$ has the distribution $(1 - \epsilon ){\Delta _0} + \epsilon {\Delta _1}$, where $\Delta_a$  denotes a point mass at $a$. Then $\frac{{\psi _{\tau,p}(\sigma ^2)}}{\sigma ^2}$ can be written as
\begin{eqnarray*}
\frac{{{\psi _{{\tau },p}}({\sigma ^2})}}{{{\sigma ^2}}} &=& \frac{{(1 - \epsilon )}}{\delta }\mathbb{E}\left[ {\eta _p^2(Z;\tau )} \right] + \frac{\epsilon }{\delta }\mathbb{E}\left[ {{{\left( {{\eta _p}(1/\sigma  + Z;\tau ) - 1/\sigma } \right)}^2}} \right] \nonumber \\
&\geq& \inf_{\tau \geq 0} \frac{{(1 - \epsilon )}}{\delta }\mathbb{E}\left[ {\eta _p^2(Z;\tau )} \right] + \frac{\epsilon }{\delta }\mathbb{E}\left[ {{{\left( {{\eta _p}(1/\sigma  + Z;\tau ) - 1/\sigma } \right)}^2}} \right]
\end{eqnarray*}
For notational simplicity assume that  $\mathop {\sup  }\limits_{\mu \geq 0}  \mathop {\inf } \limits_{\tau \geq 0} \frac{{(1 - \epsilon )}}{\delta }\mathbb{E}\left[ {\eta _p^2(Z;\tau )} \right] + \frac{\epsilon }{\delta }\mathbb{E}\left[ {{{\left( {{\eta _p}(\mu  + Z;\tau ) - \mu } \right)}^2}} \right]$ is achieved at $\mu_*$ and define ${\sigma _0} \triangleq 1/{\mu _*}$.\footnote{If $\mu^*$ is infinite, then we can use the same technique, but we should show that zero is an unstable fixed point. } We then have
\begin{eqnarray}
\frac{{{\psi _{{\tau },p}}(\sigma _0^2)}}{{\sigma _0^2}} \geq \frac{\underline{M}_p(\epsilon)}{\delta}>1.
\end{eqnarray}
This, combined with \eqref{eq:equivpsiPsi}, implies that 
\[
\Psi _{{\lambda_*},p}(\sigma _0^2)>\sigma^2_0. 
\]
Also, according to  \eqref{eq:sigmainf} we know that if $\sigma^2> \frac{1}{\delta}[\mathbb{E} (X^2)+1]$, then 
\[
\Psi _{{\lambda_*},p}(\sigma^2) < \sigma^2
\]
Hence, by the continuity of $\Psi _{{\lambda_*},p}(\sigma^2)$ (proved in Theorem \ref{final:smooth}) we conclude that $\Psi _{\lambda_* ,p}(\sigma ^2)$ has a stable fixed point at some  $\sigma^2> \sigma_0^2$. Therefore, for the distribution $(1 - \epsilon ){\Delta _0} + \epsilon {\Delta _1}$, $\Psi_{\lambda_*,p}(\sigma^2)$ has at least one non-zero stable fixed point. 



\subsection{Proof of Corollary \ref{cor:noiselessell_p}}\label{ssec:cornoiselessproof}

 Define
\[
\bar{\epsilon}_p^*(\delta) \triangleq \inf \{  \epsilon : \overline{M}_p(\epsilon) \geq \delta\}.
\]
First note that it is straightforward to show that $\overline{M}_p(1)=1$. Hence, $ \{  \epsilon : \overline{M}_p(\epsilon) \geq \delta\}$ is not empty. It is clear that for $\epsilon< \bar{\epsilon}_p^*(\delta)$ we have  $\overline{M}_p(\epsilon)< \delta$. Combining this with Theorem \ref{thm:highestfpnoiseless} establishes the first part of our result. For the second part of the corollary, define
\[
\underline{\epsilon}_p^*(\delta) \triangleq \sup  \{  \epsilon : \underline{M}_p(\epsilon)\leq \delta\}.
\]
Since $\underline{M}_p(0)=0$, $\{  \epsilon : \underline{M}_p(\epsilon)\leq \delta\}$ is not empty. Furthermore, if $\epsilon> \underline{\epsilon}_p^*(\delta)$, then  $\underline{M}_p(\epsilon)> \delta$. According to Theorem \ref{thm:highestfpnoiseless}, there exists a distribution for which the recovery of optimally tuned $\ell_p$-AMP is not successful.

\subsection{Proof of Lemma \ref{lem:ell1pteq}} \label{ssec:prooflemmaell1pteq}

For any given $\tau >0$, we have
\begin{eqnarray*}
\frac{d \mathbb{E} (\eta_1(\mu+Z; \tau) - \mu)^2}{d\mu} &=& 2 \mathbb{E}[ (\eta_1(\mu+Z; \tau) - \mu) (\mathbb{I} (|\mu+Z|> \tau) - 1 )] \nonumber \\
&=&2\mu \mathbb{E} (  \mathbb{I} (|\mu+Z|\leq \tau)) > 0,
\end{eqnarray*}
for any $\mu>0$. Hence $ \mathbb{E} (\eta_1(\mu+Z; \tau) - \mu)^2$ is an increasing function of $\mu$ over $[0,\infty)$. This implies that 
\begin{eqnarray*}
\overline{M}_1(\epsilon) &=& \inf_\tau \sup_\mu (1-\epsilon) \mathbb{E} (\eta_1(Z; \tau))^2 + \epsilon\mathbb{E} (\eta_1(\mu+Z; \tau) - \mu)^2 \nonumber \\
&=& \inf_\tau \lim_{\mu \rightarrow \infty}   (1-\epsilon) \mathbb{E} (\eta_1(Z; \tau))^2 + \epsilon\mathbb{E} (\eta_1(\mu+Z; \tau) - \mu)^2 \nonumber \\
&=&  \inf_\tau  (1-\epsilon) \mathbb{E} (\eta_1(Z; \tau))^2 + \epsilon(1+ \tau^2). 
\end{eqnarray*}
The last equality is obtained by dominated convergence theorem (the details can be found in the proof of Theorem \ref{lem:noiselesslowfpless1}). On the other hand, we know
\begin{eqnarray*}
\underline{M}_1(\epsilon) &=&  \sup_\mu \inf_\tau (1-\epsilon) \mathbb{E} (\eta_1(Z; \tau))^2 + \epsilon\mathbb{E} (\eta_1(\mu+Z; \tau) - \mu)^2 \\
&\geq& \lim_{\mu \rightarrow \infty} \inf_\tau (1-\epsilon) \mathbb{E} (\eta_1(Z; \tau))^2 + \epsilon\mathbb{E} (\eta_1(\mu+Z; \tau) - \mu)^2 \\
&\overset{(a)}{=}&\inf_{\beta}  (1-\epsilon) \mathbb{E} (\eta_1(Z; \beta))^2 + \epsilon(1+ \beta^2),
\end{eqnarray*}
where $(a)$ is a direct implication from the proof of Lemma \ref{lem:ratiofinite} (by setting $X=(1-\epsilon)\Delta_0+\epsilon \Delta_1$). Thus, we have showed $\overline{M}_1(\epsilon) \leq \underline{M}_1(\epsilon)$. Moreover, we can easily see $\overline{M}_1(\epsilon) \geq \underline{M}_1(\epsilon)$ from their definitions. So we can conclude $\overline{M}_1(\epsilon) = \underline{M}_1(\epsilon)$. Now we would like to prove that $\overline{M}_1(\epsilon)$ is an increasing function of $\epsilon$. First note that $ (1-\epsilon) \mathbb{E} (\eta_1(Z; \tau))^2 + \epsilon(1+ \tau^2)$ is a strictly convex function of $\tau$ and has a unique global minima, denoted by $\tau_*>0$. Note that the subgradient of $\overline{M}_1(\epsilon)$ with respect to $\epsilon$ is $1+\tau_*^2- \mathbb{E} (\eta_1(Z; \tau_*))^2>0$. Therefore, $\overline{M}_1(\epsilon)$ is a strictly increasing continuous function. Hence, $\overline{\epsilon}_1^*(\delta) = \overline{M}_1^{-1}(\delta) = \underline{M}_1^{-1}(\delta) = \underline{\epsilon}_1^*(\delta)$.

\subsection{Proof of Theorem \ref{lem:noiselesslowfpless1}}\label{proof sec:thmnoiselesslfp}

\subsubsection{Main part}\label{sssec:thm4mainpart}

For any $\sigma>0$ and any thresholding policy $\lambda (\sigma )$, define $\tau(\sigma)\triangleq \frac{\lambda (\sigma )}{\sigma ^{2 - p}}$. Also, let ${\tau _*}(\sigma )$ denote the optimal value of $\tau (\sigma )$ given by  $\tau_*(\sigma) \triangleq \frac{\lambda_* (\sigma )}{\sigma ^{2 - p}}$. In the rest of the proof, we write ${\eta _p}\left( {X + \sigma Z;\lambda (\sigma )} \right)$ as ${\eta _p}\left( {X + \sigma Z;\tau (\sigma ){\sigma ^{2 - p}}} \right)$. This will enable us to employ the scale invariance properties of the proximal function, proved in Lemma \ref{lem:scaleinv}, more efficiently. Since it is easier to work with ${\tau }(\sigma )$, we use the notation 
\begin{equation}\label{eq:defpsismall}
{\psi_{{\tau },p}}({\sigma ^2}) \triangleq \frac{1}{\delta}\mathbb{E} (\eta_p(X + \sigma Z; \tau(\sigma) \sigma^{2-p}) -X)^2.
\end{equation}
Clearly, we have
\[
\psi_{\tau ,p}(\sigma ^2) =\Psi _{\lambda ,p}(\sigma^2), \quad \psi_{\tau_* ,p}(\sigma ^2) =\Psi _{\lambda_* ,p}(\sigma^2).
\]
Note that $\sigma^2=0$ is actually a fixed point of ${\psi _{{\tau _*},p}}({\sigma ^2})$. Furthermore, it is straightforward to see that 0 is a stable fixed point if and only if
\begin{equation}
\label{eq:A-4-1}
{\left. {\frac{{d{\psi _{{\tau _*},p}}({\sigma ^2})}}{{d{\sigma ^2}}}} \right|_{{\sigma ^2} = 0}} < 1.
\end{equation}
Consider a specific thresholding policy $\lambda(\sigma) = \beta \sigma^{2-p}$, where $\beta \geq 0$ is a fixed number and define
\begin{equation}\label{eq:defpsismallbar}
\bar{\psi}_{\beta, p} (\sigma^2)  \triangleq \frac{1}{\delta}\mathbb{E} (\eta_p(X + \sigma Z; \beta \sigma^{2-p}) -X)^2.
\end{equation}
We then have 
\begin{equation}\label{eq:upperbarpsi}
{\left. {\frac{{d{\psi _{{\tau _*},p}}({\sigma ^2})}}{{d{\sigma ^2}}}} \right|_{{\sigma ^2} = 0}} = \mathop {\lim }\limits_{{\sigma ^2} \to 0} \frac{{{\psi _{{\tau _*},p}}({\sigma ^2})}}{{{\sigma ^2}}} \leqslant \mathop {\lim }\limits_{{\sigma ^2} \to 0} \frac{\bar{\psi} _{\beta ,p}(\sigma ^2)}{\sigma ^2},
\end{equation}
where the last inequality is due to the fact that $\lambda_*$ (or $\tau_*$) is the optimal thresholding policy and hence $\psi _{\tau _*,p} \leq \bar{\psi} _{\beta,p}$ for every $\beta \geq 0$ and $\sigma^2$.  Since \eqref{eq:upperbarpsi} holds for every $\beta \geq 0$ we have
\begin{eqnarray}\label{eq:dervativezero1}
\left. {\frac{d\psi _{\tau _*,p}(\sigma ^2)}{d\sigma ^2}} \right|_{{\sigma ^2} = 0} \leq \inf_{\beta \geq 0}  \lim_{\sigma ^2 \to 0} \frac{\bar{\psi} _{\beta ,p}(\sigma ^2)}{\sigma ^2}.
\end{eqnarray}
Let $X \sim (1- \epsilon) \Delta_0+ \epsilon G$, where $G$ is an arbitrary distribution that does not have any point mass at zero. Also, let $U$ denote a random variable with distribution $G$. Then we know
\begin{eqnarray*}
\label{eq:A-4-5}
  \bar{\psi} _{\beta ,p}(\sigma ^2) &=& \frac{1}{\delta }\left[ {(1 - \epsilon )\mathbb{E}{\left( \eta _p(\sigma Z;\beta\sigma ^{2 - p}) \right)^2} + \epsilon \mathbb{E}{\left( {\eta _p(U + \sigma Z;\beta \sigma ^{2 - p}) - U} \right)}^2} \right] \nonumber \\
   &=& \frac{\sigma ^2}{\delta }\left[ {(1 - \epsilon )\mathbb{E}{\left( \eta _p(Z;\beta ) \right)^2} + \epsilon \mathbb{E}{\left( {\eta _p(U/\sigma  + Z;\beta) - U/\sigma } \right)^2}} \right],
\end{eqnarray*}
where the second equality is due to Lemma \ref{lem:scaleinv}. Hence we have
\begin{equation}\label{eq:dervativezero2}
\lim_{\sigma^2 \rightarrow 0} \frac{\bar{\psi} _{\beta ,p}({\sigma ^2})}{\sigma^2} =  \frac{1}{\delta } {(1 - \epsilon )\mathbb{E}{{\left( {{\eta _p}(Z;{\beta })} \right)}^2} +\frac{\epsilon}{\delta}\lim_{\sigma^2 \rightarrow 0} \mathbb{E}{{\left( {{\eta _p}(U/\sigma  + Z;{\beta }) - U/\sigma } \right)}^2}}.
\end{equation}
Our next goal is to show that we can interchange the limit and expectation above. Define ${\upsilon _p}(u;\beta ) \triangleq {\eta _p}(u;\beta) - u$. So we can write
\begin{eqnarray}\label{eq:dervativezero4}
\mathbb{E}{{\left( {{\eta _p}(U/\sigma  + Z;{\beta }) - U/\sigma } \right)}^2} &=& \mathbb{E} (Z+{\upsilon _p}(U/\sigma +Z;\beta )  )^2 \nonumber \\
&=& 1+ \mathbb{E} ({\upsilon _p}(U/\sigma +Z;\beta )  )^2+ 2 \mathbb{E} (Z {\upsilon _p}(U/\sigma +Z;\beta ) ). 
\end{eqnarray}    
From Corollary \ref{cor:etap1eta0} and \ref{cor:hardvsellp}, we know $|{\upsilon _p}(u;\beta)| \leq c_p \beta^{\frac{1}{2-p}}$. So we can get that $({\upsilon _p}(U/\sigma +Z;\beta )  )^2 \leq c_p^2\beta^{\frac{2}{2-p}}$ and $|Z {\upsilon _p}(U/\sigma +Z;\beta ) |  \leq c_p\beta^{\frac{1}{2-p}} |Z|$. We can then employ the dominated convergence theorem to conclude that
\begin{eqnarray}\label{eq:dervativezero3}
\lim_{\sigma^2 \rightarrow 0} \mathbb{E} ({\upsilon _p}(U/\sigma +Z;\beta )  )^2 &=& \mathbb{E}  \lim_{\sigma^2 \rightarrow 0} ({\upsilon _p}(U/\sigma +Z;\beta )  )^2 =0, \nonumber \\
\lim_{\sigma^2 \rightarrow 0} \mathbb{E} (Z {\upsilon _p}(U/\sigma +Z;\beta ) ) &=& \mathbb{E} \lim_{\sigma^2 \rightarrow 0} (Z {\upsilon _p}(U/\sigma +Z;\beta ) ) =0,
\end{eqnarray}
where the second equalities in the two lines above is a straightforward result of Lemma \ref{lem:asymtoteta_p}. Combining \eqref{eq:dervativezero1}, \eqref{eq:dervativezero2}, \eqref{eq:dervativezero4}, and \eqref{eq:dervativezero3} implies that 
\begin{equation}\label{eq:upperboundder0}
\left. {\frac{d\psi _{\tau _*,p}(\sigma ^2)}{d\sigma ^2}} \right|_{\sigma^2 =0} \leq \inf_{\beta \geq 0} \frac{1-\epsilon}{\delta } {\mathbb{E}{{\left( {{\eta _p}(Z;{\beta })} \right)}^2} }+ \frac{\epsilon}{\delta}= \frac{\epsilon}{\delta}. 
\end{equation}
So far we have proved an upper bound for the derivative of $\psi_{\tau_*,p}(\sigma^2)$ at $\sigma=0$. Our next step is to  show that 
\begin{eqnarray*}
\left. {\frac{d\psi _{\tau _*,p}(\sigma ^2)}{d\sigma ^2}} \right|_{\sigma^2 =0} \geq \frac{\epsilon}{\delta}.
\end{eqnarray*}
Note that
\begin{eqnarray}
\lefteqn{\left. {\frac{d\psi _{\tau _*,p}(\sigma ^2)}{d\sigma ^2}} \right|_{\sigma^2 =0} }\nonumber \\
&=&\lim_{\sigma \rightarrow 0} \inf_{\beta \geq 0} \frac{1}{\delta }\left[ {(1 - \epsilon )\mathbb{E}{\left( \eta _p(Z;\beta ) \right)^2} + \epsilon \mathbb{E}{\left( {\eta _p(U/\sigma  + Z;\beta) - U/\sigma } \right)^2}} \right], \nonumber \\
&\geq& \frac{1}{\delta }  \lim_{\sigma \rightarrow 0} \inf_{\beta \geq 0}\left\{ {(1 - \epsilon )\mathbb{E}{\left( \eta _p(Z;\beta ) \right)^2} \mathbb{P} (|U| \geq \mu) + \epsilon \mathbb{E}{\left[ \left( {\eta _p(U/\sigma  + Z;\beta) - U/\sigma } \right)^2\cdot \mathbbm{1}(|U|\geq \mu) \right]}} \right \},  \nonumber \\
&\geq&\frac{ \mathbb{P} (|U| \geq \mu) }{\delta }  \lim_{\sigma \rightarrow 0} \inf_{\beta \geq 0}\left\{ (1 - \epsilon )\mathbb{E}{\left( \eta _p(Z;\beta ) \right)^2} + \epsilon \frac{\mathbb{E}{\left[ \left( {\eta _p(U/\sigma  + Z;\beta) - U/\sigma } \right)^2\cdot \mathbbm{1}(|U|\geq \mu) \right]}}{\mathbb{P} (|U| \geq \mu)} \right \},   \label{eq:lowerboundder0}
\end{eqnarray}
where $\mu$  is an arbitrary positive number that satisfies $\mathbb{P} (|U| \geq \mu) > 0$. Our next step is to prove that
\begin{equation}\label{eq:thm4lb2}
 \lim_{\sigma \rightarrow 0} \inf_{\beta \geq 0}\left\{ (1 - \epsilon )\mathbb{E}{\left( \eta _p(Z;\beta ) \right)^2} + \epsilon \frac{\mathbb{E}{\left[ \left( {\eta _p(U/\sigma  + Z;\beta) - U/\sigma } \right)^2\cdot \mathbbm{1}(|U|\geq \mu) \right]}}{\mathbb{P} (|U| \geq \mu)} \right \} = \epsilon.
\end{equation}
Since this requires more work, we postpone its proof until Section \ref{ssec:thm4aux}, and we discuss how \eqref{eq:lowerboundder0} and \eqref{eq:thm4lb2} finish the proof of Theorem \ref{lem:noiselesslowfpless1}.  By combining \eqref{eq:lowerboundder0} and \eqref{eq:thm4lb2} we obtain
\begin{equation}\label{eq:lowerboundgood1}
 \left. \frac{{d{\psi _{{\tau _*},p}}({\sigma ^2})}}{{d{\sigma ^2}}}  \right|_{\sigma^2=0} \geq \lim_{\mu \rightarrow 0} \frac{\epsilon}{\delta} \mathbb{P} (|U|\geq \mu)= \frac{\epsilon}{\delta}.  
\end{equation}
 Combining \eqref{eq:upperboundder0} and \eqref{eq:lowerboundgood1} proves that
\[
\left. \frac{{d{\psi _{{\tau _*},p}}({\sigma ^2})}}{{d{\sigma ^2}}}  \right|_{\sigma^2=0}= \frac{\epsilon}{\delta}. 
\]
As we discussed before, $0$ is a stable fixed point if and only if
\[
\left. \frac{{d{\psi _{{\tau _*},p}}({\sigma ^2})}}{{d{\sigma ^2}}}  \right|_{\sigma^2=0}= \frac{\epsilon}{\delta} <1. 
\]
The only step that is still unresolved in the proof of Theorem \ref{lem:noiselesslowfpless1} is \eqref{eq:thm4lb2}. Since the proof is different for $0<p<1$ and $p=0$, we prove them in two different sections below, i.e., Section \ref{ssec:thm4aux} and \ref{ssec:thm4auxp0} respectively.


\vspace{.2cm}

\subsubsection{Auxiliary result for $0<p<1$}\label{ssec:thm4aux}
As we discussed before our goal in this section is to prove Equation \eqref{eq:thm4lb2} for every $0<p<1$. Below we prove a stronger result, since this stronger version will be used in other proofs throughout the paper. Define
 \begin{equation*}
 R_p(\tau, \sigma)\triangleq(1-\epsilon)\mathbb{E}\eta_p^2(Z;\tau)+\epsilon \mathbb{E}(\eta_p(U/\sigma+Z;\tau)-U/\sigma)^2,
 \end{equation*}
where $Z \sim N(0,1)$ and $U \sim G$ are independent. Denote the optimal $\tau$ that minimizes $R_p(\tau,\sigma)$ by $\tau_*(\sigma)$. Also, let $X \sim (1-\epsilon)\Delta_0+ \epsilon G$, and define $\mathbb{E}_G(f(U)) \triangleq \int f(u) dG(u)$ and $\mathbb{P}_G(U \in A) \triangleq \mathbb{E}_G (\mathbb{I} (U \in B))$.

\vspace{.2cm}

\begin{proposition}\label{lem:riskbehsmallsigma}
Suppose $\mathbb{P}_G(|U|>\mu)=1$ with $\mu$ being a fixed positive number and $\mathbb{E}_G|U|^2 < \infty$. Then, for $0<p<1$, we have
\begin{eqnarray*}
 R_p(\tau_*(\sigma),\sigma)=\epsilon+\epsilon p^2 \mathbb{E}|U|^{2p-2} (\tau_*(\sigma))^2\sigma^{2-2p}+o((\tau_*(\sigma))^2\sigma^{2-2p}),
\end{eqnarray*}
where the convergence rate of $\tau_*(\sigma)$ can be characterized by 
\begin{eqnarray*}
 \lim_{\sigma \rightarrow 0} \frac{\sigma^{2-2p}}{(\tau_*(\sigma))^{\frac{2p-1}{2-p}}\phi(c_p(\tau_*(\sigma))^{\frac{1}{2-p}})} = \frac{(1-\epsilon)c_p\eta^2_p(c_p;1)}{\epsilon p^2(2-p)\mathbb{E}|U|^{2p-2}}.
\end{eqnarray*}
\end{proposition}

Before we prove this result note that as $\sigma \rightarrow 0$, $ R_p(\tau_*(\sigma),\sigma) \rightarrow \epsilon$, and this implies \eqref{eq:thm4lb2} we required to prove Theorem \ref{lem:noiselesslowfpless1}. In this proposition we go one step further, and characterize the second dominant term as well (in terms of $\sigma$), since it will be used in the proofs of other results later in our paper.  

We prove Proposition \ref{lem:riskbehsmallsigma} in three steps. We first show $\tau_*(\sigma)$ goes off to infinity, but not very fast, as $\sigma \rightarrow 0$. This will be done in Lemma \ref{rough:rate}. Then, we characterize the exact rate of $\tau_*(\sigma)$ in terms of $\sigma \rightarrow 0$. This will be performed in Lemma \ref{sharp:rate}. Finally we use this result to prove Proposition \ref{lem:riskbehsmallsigma}.  
 
\vspace{0.3cm}

\begin{lemma}\label{rough:rate}
Suppose $\mathbb{E}|X|^2 < \infty$, then for $0<p<1$, $\tau_*(\sigma) \rightarrow \infty$ and $\tau_*(\sigma)\sigma^{2-p} \rightarrow 0$, as $\sigma \rightarrow 0$.
\end{lemma}

\vspace{.3cm}

\begin{proof}
If $\tau_*(\sigma)\sigma^{2-p} \nrightarrow 0$, then there exists a sequence $\sigma_k \rightarrow 0$ and a constant $c>0$ such that $\tau_*(\sigma_k)\sigma_k^{2-p}\geq c$, for all $k$. Choose a convergent subsequence $\{\sigma_{k_n}\}$ and denote $\lim_{k_n\rightarrow \infty}\tau_*(\sigma_{k_n})\sigma_{k_n}^{2-p}=\alpha \geq c$ (note $\alpha$ can be $+\infty$). We use Fatou's lemma to get
\begin{eqnarray*}
\liminf_{k_n\rightarrow \infty} \mathbb{E}(\eta_p(X+\sigma_{k_n} Z; \sigma_{k_n}^{2-p}\tau_*(\sigma_{k_n}))-X)^2 &\geq& \mathbb{E} \liminf_{k_n\rightarrow \infty}(\eta_p(X+\sigma_{k_n} Z; \sigma_{k_n}^{2-p}\tau_*(\sigma_{k_n}))-X)^2 \\
&\geq&\mathbb{E} \min((\eta_p(X;\alpha)-X)^2,X^2)>0
\end{eqnarray*}
Hence, we have 
\begin{eqnarray*}
\liminf_{k_n\rightarrow \infty} \mathbb{E}(\eta_p(X/\sigma_{k_n}+Z;\tau_*(\sigma_{k_n}))-X/\sigma_{k_n})^2&=&\lim_{k_n\rightarrow \infty}\frac{1}{\sigma^2_{k_n}} \cdot \liminf_{k_n\rightarrow \infty} \mathbb{E}(\eta_p(X+\sigma_{k_n}Z;\sigma_{k_n}^{2-p}\tau_*(\sigma_{k_n}))-X)^2 \\
&=&+\infty
\end{eqnarray*}
which implies $\liminf_{k_n\rightarrow \infty} R_p(\tau_*(\sigma_{k_n}),\sigma_{k_n})=+\infty$. However, since $\tau_*({\sigma_{k_n}})$ is the optimal thresholding value, we know $R_p(\tau_*(\sigma_{k_n}),\sigma_{k_n}) \leq R_p(0,\sigma_{k_n})=1$, for every $k_n$. This is a contradiction. Similarly, if $\tau_*(\sigma) \nrightarrow \infty$, there exists a sequence $\sigma_k \rightarrow 0$ and a finite constant $\alpha \geq 0$ such that $\tau_*(\sigma_k) \rightarrow \alpha$. By similar arguments as in the previous proof (see \eqref{eq:upperboundder0} for example), we can apply dominated convergence theorem to obtain,
\begin{equation}\label{contra:beta}
\lim_{k\rightarrow \infty} R_p(\tau_*(\sigma_k),\sigma_k)=(1-\epsilon)\mathbb{E}\eta^2_p(Z;\alpha)+\epsilon > \epsilon.
\end{equation}
On the other hand, since $\tau_*(\sigma_k)$ is the optimal thresholding value, we know
\[
\lim_{k\rightarrow \infty} R_p(\tau_*(\sigma_k),\sigma_k) \leq \lim_{k \rightarrow \infty}R_p(\beta,\sigma_k)=(1-\epsilon)\mathbb{E}\eta^2_p(Z;\beta)+\epsilon,
\]
for any finite $\beta$. Letting $\beta \rightarrow \infty$ on both sides of the above inequality yields 
\[
\lim_{k\rightarrow \infty} R_p(\tau_*(\sigma_k),\sigma_k) \leq \epsilon,
\] 
which contradicts \eqref{contra:beta}.
\end{proof}

\vspace{0.4cm}

\begin{lemma}\label{sharp:rate}
Suppose $\mathbb{P}_G(|U|>\mu)=1$ with $\mu$ being a fixed positive number and $\mathbb{E}_G|U|^2 < \infty$, then for every $0<p<1$,
\begin{eqnarray*}
\lim_{\sigma \rightarrow 0} \frac{\sigma^{2-2p}}{(\tau_*(\sigma))^{\frac{2p-1}{2-p}}\phi(c_p(\tau_*(\sigma))^{\frac{1}{2-p}})} = \frac{(1-\epsilon)c_p(\eta^+_p(c_p;1))^2}{\epsilon p^2(2-p)\mathbb{E}|U|^{2p-2}}.
\end{eqnarray*}
\end{lemma}

\vspace{0.3cm}

\begin{proof}
We recall some properties of the proximal operator $\eta_p(u;\lambda)$ that will be used multiple times in the proof. For further information, see the proofs of Lemmas \ref{lem:derivativeproperties} and \ref{lem:derivativeproperties2}.
\begin{enumerate}
\item[(a)] $\frac{\partial \eta_p(u;\lambda)}{\partial \lambda}=\frac{-p\eta^{p-1}_p(u;\lambda)}{1+\lambda p(p-1)\eta^{p-2}_p(u;\lambda)}$, for $u>c_p \lambda^{\frac{1}{2-p}}$
\item[(b)] $\frac{\partial \eta_p(u;\lambda)}{\partial u}=\frac{1}{1+\lambda p(p-1)\eta^{p-2}_p(u;\lambda)}$, for $u>c_p\lambda^{\frac{1}{2-p}}$.
\item[(c)] $u-\eta_p(u;\lambda)=p\lambda \eta^{p-1}_p(u;\lambda)$, for $u>c_p\lambda^{\frac{1}{2-p}}$.
\end{enumerate} 
Let $F(u)$ denote the CDF of $|U|$. We first decompose $R_p(\tau,\sigma)$ to the following terms:
\begin{eqnarray}
R_p(\tau,\sigma)&=& 2(1-\epsilon)\int_{c_p\tau^{\frac{1}{2-p}}}^{\infty}\eta^2_p(z;\tau)\phi(z)dz+\epsilon \int_{\mu}^{\infty}\int_{-u/\sigma+c_p\tau^{\frac{1}{2-p}}}^{\infty}(\eta_p(u/\sigma+z;\tau)-u/\sigma)^2\phi(z)dzdF(u)+   \nonumber \\
&&\epsilon \int_{\mu}^{\infty} \int_{u/\sigma+c_p\tau^{\frac{1}{2-p}}}^{\infty}(\eta_p(-u/\sigma+z;\tau)+u/\sigma)^2\phi(z)dzdF(u)+\epsilon \int_{\mu}^{\infty} \int_{-u/\sigma-c_p\tau^{\frac{1}{2-p}}}^{-u/\sigma+c_p\tau^{\frac{1}{2-p}}}\frac{u^2}{\sigma^2}\phi(z)dzdF(u)   \nonumber  \\
&\triangleq&R_1+R_2+R_3+R_4 \label{leq1:one}
\end{eqnarray}
From the proof of Lemma \ref{rough:rate}, it is straightforward to see that $\tau_*(\sigma)$ is non-zero and finite. Since $\tau_*(\sigma)$ is the optimal thresholding value and$\frac{\partial R_p(\tau, \sigma)}{\partial \tau}$ is differentiable (according to Lemma \ref{dervlambda:cont}), we conclude that $\tau_*(\sigma)$ satisfies $\frac{\partial R_p(\tau_*(\sigma),\sigma)}{\partial \tau}=0$. For notational simplicity, below we use $\tau$ and $\tau_*$  interchangeably. Now we analyze the partial derivative of the four terms in \eqref{leq1:one} separately. For the first term,
\begin{eqnarray}\label{eq:R1formulatotal}
\frac{\partial R_1}{\partial \tau}= \frac{-2(1-\epsilon)c_p\tau^{\frac{p-1}{2-p}}}{2-p}(\eta^+_p(c_p\tau^{\frac{1}{2-p}};\tau))^2\phi(c_p\tau^{\frac{1}{2-p}})-4(1-\epsilon)p\int_{c_p\tau^{\frac{1}{2-p}}}^{\infty}\frac{\eta^p_p(z;\tau)}{1+\tau p(p-1)\eta_p^{p-2}(z;\tau)}\phi(z)dz,
\end{eqnarray}
where we have used property (a). We now compare the order of the two terms on the right hand side of the above equality. According to Lemma \ref{lem:lowerboundonx}, we can conclude that $1+\tau p(p-1)\eta_p^{p-2}(z;\tau)$ is bounded away from zero, for $z\geq c_q\tau^{\frac{1}{2-p}}$.  
Hence, combining with the fact $|\eta_p(z;\tau)|\leq |z|$ (according to Lemma \ref{lem:derivativeproperties}), we know there exists a positive constant $C$ such that
\begin{eqnarray*}
&& \Big |\int_{c_p\tau^{\frac{1}{2-p}}}^{\infty}\frac{\eta^p_p(z;\tau)}{1+\tau p(p-1)\eta_p^{p-2}(z;\tau)}\phi(z)dz\Big |\\
 &\leq& C\int_{c_p\tau^{\frac{1}{2-p}}}^{\infty} z^p\phi(z)dz \overset{(i)}{=} Cc_p^{p-1}\tau^{\frac{p-1}{2-p}}\phi(c_p\tau^{\frac{1}{2-p}})+C\int_{c_p\tau^{\frac{1}{2-p}}}^{\infty}(p-1)z^{p-2}\phi(z)dz  \\
&\leq&Cc_p^{p-1}\tau^{\frac{p-1}{2-p}}\phi(c_p\tau^{\frac{1}{2-p}})+C(p-1)c_p^{p-2}\tau^{-1}\int_{c_p\tau^{\frac{1}{2-p}}}^{\infty}\phi(z)dz \\
&\overset{(ii)}{\leq}& O(\tau^{\frac{p-1}{2-p}}\phi(c_p\tau^{\frac{1}{2-p}}))+O(\tau^{\frac{p-3}{2-p}}\phi(c_q\tau^{\frac{1}{2-p}}))=O(\tau^{\frac{p-1}{2-p}}\phi(c_p\tau^{\frac{1}{2-p}})).
\end{eqnarray*}
To obtain Equality (i) we used integration by parts. To obtain Inequality (ii), we have used $\int_{t}^{\infty}\phi(z)dz\sim \frac{1}{t}\phi(t)$, as $t \rightarrow \infty$. Now we discuss the order of the first term in \eqref{eq:R1formulatotal}. Since according to Lemma \ref{lem:scaleinv}, $\eta^+_p(c_p\tau^{\frac{1}{2-p}};\tau)=\tau^{\frac{1}{2-p}}\eta^+_p(c_p;1)$, we know the first term is of order $\tau^{\frac{p+1}{2-p}}\phi(c_p\tau^{\frac{1}{2-p}})$. Hence, we can conclude that
\begin{equation}
\lim_{\sigma \rightarrow 0} \frac{\partial R_1}{\partial \tau} / (\tau^{\frac{p+1}{2-p}}\phi(c_p\tau^{\frac{1}{2-p}}))=\frac{-2(1-\epsilon)c_p(\eta^+_p(c_p;1))^2}{2-p}.  \label{leq1:two}
\end{equation}
For the last term $R_4$, we can do the following calculations:
\begin{eqnarray*}
\frac{\partial R_4}{\partial \tau}&= &\mathbb{E} \Bigg [\frac{\epsilon U^2c_p\tau^{\frac{p-1}{2-p}}}{\sigma^2(2-p)}(\phi(-U/\sigma+c_p\tau^{\frac{1}{2-p}})+\phi(-U/\sigma-c_p\tau^{\frac{1}{2-p}})) \Bigg ] \\
&{\leq}& \frac{2\epsilon c_p\tau^{\frac{p-1}{2-p}}}{\sigma^2(2-p)} \mathbb{E}[U^2\phi(c_p\tau^{\frac{1}{2-p}}-|U|/\sigma)] {\leq}  \frac{2\epsilon c_p\tau^{\frac{p-1}{2-p}}}{\sigma^2(2-p)}\phi(c_p\tau^{\frac{1}{2-p}}-\mu/\sigma) \mathbb{E}U^2,
\end{eqnarray*}
where the last inequality is based on $|U|/\sigma \geq \mu /\sigma \gg c_p  \tau^{\frac{1}{2-p}}$ from Lemma \ref{rough:rate}. Again using $|\mu|/\sigma \gg c_p\tau^{\frac{1}{2-p}}$, it is straightforward to confirm that 
\begin{eqnarray*}
\lim_{\sigma \rightarrow 0}  \frac{\phi(-\mu/\sigma+c_p\tau^{\frac{1}{2-p}})}{\sigma^2\phi(c_p\tau^{\frac{1}{2-p}})}=0.
\end{eqnarray*}
Therefore, we have
\begin{eqnarray}\label{leq1:three}
\lim_{\sigma \rightarrow 0} \frac{\partial R_4}{\partial \tau}/ (\tau^{\frac{p+1}{2-p}}\phi(c_p\tau^{\frac{1}{2-p}}))=0.
\end{eqnarray}
We now discuss the calculation of $\frac{\partial R_2}{\partial \tau}$. We have
\begin{eqnarray*}
\frac{\partial R_2}{\partial \tau} &=&\frac{-\epsilon c_p}{2-p}\tau^{\frac{p-1}{2-p}}\int_{\mu}^{\infty}(\eta^+_p(c_p\tau^{\frac{1}{2-p}};\tau)-u/\sigma)^2\phi(-u/\sigma+c_p\tau^{\frac{1}{2-p}})dF(u)+  \\
&&2\epsilon \int_{\mu}^{\infty} \int_{-u/\sigma+c_p\tau^{\frac{1}{2-p}}}^{\infty}(\eta_p(u/\sigma+z;\tau)-u/\sigma-z)\partial_2\eta_p(u/\sigma+z;\tau)\phi(z)dz dF(u) \\
&&+ 2\epsilon  \int_{\mu}^{\infty} \int_{-u/\sigma+c_p\tau^{\frac{1}{2-p}}}^{\infty} z\partial_2\eta_p(u/\sigma+z;\tau)\phi(z)dzdF(u)  \\
&=&\frac{-\epsilon c_p}{2-p}\tau^{\frac{p-1}{2-p}}\int_{\mu}^{\infty}(\eta^+_p(c_p\tau^{\frac{1}{2-p}};\tau)-u/\sigma)^2\phi(-u/\sigma+c_p\tau^{\frac{1}{2-p}})dF(u)+ \\
&&2\epsilon \tau p^2 \int_{\mu}^{\infty}\int_{-u/\sigma+c_p\tau^{\frac{1}{2-p}}}^{\infty}\frac{\eta^{2p-2}_p(u/\sigma+z;\tau)}{1+\tau p(p-1)\eta_p^{p-2}(u/\sigma+z;\tau)}\phi(z)dzdF(u) + \\
&& -2\epsilon p \int_{\mu}^{\infty} \int_{-u/\sigma+c_p\tau^{\frac{1}{2-p}}}^{\infty}\frac{\eta^{p-1}_p(u/\sigma+z;\tau)}{1+\tau p(p-1)\eta_p^{p-2}(u/\sigma+z;\tau)}z\phi(z)dzdF(u)   \\
&\triangleq& S_1+S_2+S_3.
\end{eqnarray*}
We have used properties (a) and (c) in the above derivations. We then analyze the above three terms separately. For $S_3$, integration by parts combined with property (b) gives
\begin{eqnarray*}
\hspace{-0.3cm}
S_3&=& \frac{-2\epsilon p(\eta_p^+(c_p\tau^{\frac{1}{2-p}};\tau))^{p-1}}{1+\tau p(p-1)(\eta_p^+(c_p\tau^{\frac{1}{2-p}};\tau))^{p-2}} \int_{\mu}^{\infty} \phi(-u/\sigma+c_p\tau^{\frac{1}{2-p}})dF(u) \\
&&-2\epsilon p(p-1)  \int_{\mu}^{\infty}\int_{-u/\sigma+c_p\tau^{\frac{1}{2-p}}}^{\infty} \frac{\eta^{p-2}_p(u/\sigma+z;\tau)}{(1+\tau p(p-1)\eta_p^{p-2}(u/\sigma+z;\tau))^2}\phi(z)dzdF(u) \\
&&+2\epsilon p^2(p-1)(p-2)\tau \int_{\mu}^{\infty} \int_{-u/\sigma+c_p\tau^{\frac{1}{2-p}}}^{\infty} \frac{\eta^{2p-4}_p(u/\sigma+z;\tau)}{(1+\tau p(p-1)\eta_p^{p-2}(u/\sigma+z;\tau))^3}\phi(z)dzdF(u) \triangleq T_1+T_2+T_3.
\end{eqnarray*}
Choosing a positive constant $0<v<\mu$, note that
\begin{eqnarray}\label{eq:T2bound1}
\frac{T_2}{\sigma^{2-p}}&=&-2\epsilon p(p-1)  \int_{\mu}^{\infty} \int_{-u/\sigma+c_p\tau^{\frac{1}{2-p}}}^{-u/\sigma+v/\sigma} \frac{\sigma^{p-2}\eta^{p-2}_p(u/\sigma+z;\tau)}{(1+\tau p(p-1)\eta_p^{p-2}(u/\sigma+z;\tau))^2}\phi(z)dzdF(u) \nonumber  \\
&&-2\epsilon p(p-1) \mathbb{E}\Bigg [\frac{\mathbbm{1}(Z+|U|/\sigma>v/\sigma)\eta^{p-2}_p(|U|+\sigma Z;\sigma^{2-p}\tau)}{(1+\tau p(p-1)\eta_p^{p-2}(|U|/\sigma+Z;\tau))^2} \Bigg].  
\end{eqnarray}
It is straightforward to check that when $u>c_p \tau^{\frac{1}{2-p}}$, there exists a positive constant $C_0$ such that $1+\tau p(p-1)\eta_p^{p-2}(u;\tau)>C_0>0$. Also since $\eta_p(u;\tau)$ is a non-decreasing function of $u>0$, we can have
\begin{eqnarray*}
\Bigg |\frac{\mathbbm{1}(Z+|U|/\sigma>v/\sigma)\eta^{p-2}_p(|U|+\sigma Z;\sigma^{2-p}\tau)}{(1+\tau p(p-1)\eta_p^{p-2}(|U|/\sigma+Z;\tau))^2}\Bigg | \leq (C_0^2)^{-1}\eta^{p-2}_p(v;\sigma^{2-p}\tau)\leq C_0^{-2}\eta^{p-2}_p(v;1),
\end{eqnarray*}
for sufficiently small $\sigma$ (recall $\eta_p(u;\tau)$ is a non-increasing function of $\tau$ when $\eta_p(u;\tau)>0$). Because $\sigma^{2-p}\tau_*(\sigma)\rightarrow 0$, as $\sigma \rightarrow 0$ from Lemma \ref{rough:rate}, we can easily see 
\begin{eqnarray*}
\lim_{\sigma \rightarrow 0} \frac{\mathbbm{1}(Z+|U|/\sigma>v/\sigma)\eta^{p-2}_p(|U|+\sigma Z;\sigma^{2-p}\tau)}{(1+\tau p(p-1)\eta_p^{p-2}(|U|/\sigma+Z;\tau))^2}=\lim_{\sigma \rightarrow 0}  \frac{\mathbbm{1}(\sigma Z+|U|>v)\eta^{p-2}_p(|U|+\sigma Z;\sigma^{2-p}\tau)}{(1+\sigma^{2-p}\tau p(p-1)\eta_p^{p-2}(|U|+\sigma Z;\sigma^{2-p}\tau))^2}=|U|^{p-2}.
\end{eqnarray*}
We can then use dominated convergence theorem to conclude,
\begin{eqnarray}\label{eq:T2bound2}
\lim_{\sigma \rightarrow 0} \mathbb{E}\frac{\mathbbm{1}(Z+|U|/\sigma>v/\sigma)\eta^{p-2}_p(|U|+\sigma Z;\sigma^{2-p}\tau)}{(1+\tau p(p-1)\eta_p^{p-2}(|U|/\sigma+Z;\tau))^2} =\mathbb{E} |U|^{p-2}.
\end{eqnarray}
Moreover, we can use similar arguments to obtain,
\begin{eqnarray}\label{eq:T2bound3}
&& \Bigg|\int_{\mu}^{\infty} \int_{-u/\sigma+c_p\tau^{\frac{1}{2-p}}}^{-u/\sigma+v/\sigma} \frac{\sigma^{p-2}\eta^{p-2}_p(u/\sigma+z;\tau)}{(1+\tau p(p-1)\eta_p^{p-2}(u/\sigma+z;\tau))^2}\phi(z)dzdF(u) \Bigg | \nonumber \\
&\leq& C_0^{-2}\tau^{-1}(\eta^+_p(c_p;1))^{p-2}\sigma^{p-2}\int_{\mu}^{\infty}\int_{-u/\sigma+c_p\tau^{\frac{1}{2-p}}}^{-u/\sigma+v/\sigma}\phi(z)dzdF(u) \nonumber \\
&\leq & C_0^{-2}\tau^{-1}(\eta^+_p(c_p;1))^{p-2}\sigma^{p-2} (v/\sigma-c_p\tau^{\frac{1}{2-p}} )\phi(-\mu/\sigma+v/\sigma)\rightarrow  0, \mbox{~as~} \sigma \rightarrow 0,
\end{eqnarray}
where the last inequality uses the fact that $\int_{-u/\sigma+c_p\tau^{\frac{1}{2-p}}}^{-u/\sigma+v/\sigma}\phi(z)dz < (v/\sigma-c_p\tau^{\frac{1}{2-p}} )\phi(-u/\sigma+v/\sigma)$ and that $u> \mu$. 
Combining \eqref{eq:T2bound1}, \eqref{eq:T2bound2} and \eqref{eq:T2bound3} we have
\begin{eqnarray*}
\lim_{\sigma \rightarrow 0}\frac{T_2}{\sigma^{2-p}}=-2\epsilon p(p-1)\mathbb{E}|U|^{p-2}.
\end{eqnarray*}
Since $T_3$ and $S_2$ admit similar integral forms as $T_2$'s, we can follow similar calculation steps to derive,
\begin{eqnarray}\label{leq1:four}
\lim_{\sigma \rightarrow 0}\frac{T_3}{\sigma^{4-2p}\tau_*}=2\epsilon p^2(p-1)(p-2)\mathbb{E}|U|^{2p-4}, \quad \lim_{\sigma \rightarrow 0}\frac{S_2}{\sigma^{2-2p}\tau_*}=2\epsilon p^2\mathbb{E}|U|^{2p-2}.
\end{eqnarray}

Furthermore, by applying Lemma \ref{rough:rate}, it is not hard to see
\begin{eqnarray}\label{leq1:five}
\lim_{\sigma \rightarrow 0}\frac{T_1}{\sigma^{2-p}}=0, \quad \lim_{\sigma \rightarrow 0}\frac{S_1}{\sigma^{2-p}}=0.
\end{eqnarray}
Combing the results about $T_1,T_2$ and $T_3$, we have
\begin{eqnarray}\label{leq1:six}
\lim_{\sigma \rightarrow 0}\frac{S_3}{\sigma^{2-p}}=\lim_{\sigma\rightarrow 0}\frac{T_1}{\sigma^{2-p}}+\lim_{\sigma\rightarrow 0} \frac{T_2}{\sigma^{2-p}}+\lim_{\sigma\rightarrow 0}\frac{T_3}{\sigma^{4-2p}\tau_*} \cdot \tau_* \sigma^{2-p}=   -2\epsilon p(p-1)\mathbb{E}|U|^{p-2}.
\end{eqnarray}
Putting \eqref{leq1:four}, \eqref{leq1:five} and \eqref{leq1:six} together, we obtain the order of $\frac{\partial R_2}{\partial \tau}$,
\begin{eqnarray}\label{leq1:eight}
\lim_{\sigma \rightarrow 0}\frac{\partial R_2}{\partial \tau}/ (\sigma^{2-2p}\tau_*)=2\epsilon p^2\mathbb{E}|U|^{2p-2}.
\end{eqnarray}
From Equation \eqref{leq1:one}, we observe that $R_3$ is only different from $R_2$ by a sign of $u$, hence we can follow the same derivation strategy as the one presented for analyzing $\partial R_2/\partial \tau$. We only highlight the differences for calculating $T_2/\sigma^{2-p}$ (we are using the same notations):
\begin{enumerate}
\item $\lim_{\sigma \rightarrow 0} \frac{\mathbbm{1}(Z-|U|/\sigma>v/\sigma)\eta^{p-2}_p(-|U|+\sigma Z;\sigma^{2-p}\tau)}{(1+\tau p(p-1)\eta_p^{p-2}(-|U|/\sigma+Z;\tau))^2}=0$,
\item $\Bigg|\int_{\mu}^{\infty} \int_{u/\sigma+c_p\tau^{\frac{1}{2-p}}}^{u/\sigma+v/\sigma} \frac{\sigma^{p-2}\eta^{p-2}_p(-u/\sigma+z;\tau)}{(1+\tau p(p-1)\eta_p^{p-2}(-u/\sigma+z;\tau))^2}\phi(z)dzdF(u) \Bigg | \leq C_0^{-2}\tau^{-1}\eta^{p-2}_p(c_p;1)\sigma^{p-2} (v/\sigma-c_p\tau^{\frac{1}{2-p}} )\phi(\mu/\sigma+c_p\tau^{\frac{1}{2-p}})=o(1)$.
\end{enumerate}
Therefore, we can conclude $\lim_{\sigma \rightarrow 0}\frac{T_2}{\sigma^{2-p}}=0$. Similar arguments hold for other integral calculations. We finally obtain
\begin{eqnarray}\label{leq1:seven}
\lim_{\sigma \rightarrow 0} \frac{\partial R_3}{\partial \tau} / (\sigma^{2-2p}\tau_*)=0.
\end{eqnarray}
Collecting the results from \eqref{leq1:one}, \eqref{leq1:two}, \eqref{leq1:three}, \eqref{leq1:eight}, and \eqref{leq1:seven}, we achieve
\begin{eqnarray*}
\lim_{\sigma \rightarrow 0}\sigma^{2-2p}\tau_*2\epsilon p^2\mathbb{E}|U|^{2p-2}\cdot  \left[\frac{(\tau_*)^{\frac{p+1}{2-p}}\phi(c_p(\tau_*)^{\frac{1}{2-p}})2(1-\epsilon)c_p(\eta^+_p(c_p;1))^2}{2-p} \right]^{-1}=1.
\end{eqnarray*}
After a simplification, we reach the conclusion 
\[
\lim_{\sigma \rightarrow 0} \frac{ \sigma^{2-2p}}{(\tau_*)^{\frac{2p-1}{2-p}}\phi(c_p(\tau_*)^{\frac{1}{2-p}})}= \frac{(1-\epsilon)c_p(\eta^+_p(c_p;1))^2}{\epsilon p^2(2-p)\mathbb{E}|U|^{2p-2}}.
\]
\end{proof}


\vspace{0.3cm}

\begin{lemma} \label{sharp:rate:risk}
Suppose $\mathbb{P}_G(|U|>\mu)=1$ with $\mu$ being a fixed positive number and $\mathbb{E}_G|U|^2 < \infty$, then for $0<p<1$,
\[
R_p(\tau_*(\sigma),\sigma)=\epsilon+\epsilon p^2 \mathbb{E}|U|^{2p-2} (\tau_*)^2\sigma^{2-2p}+o((\tau_*)^2\sigma^{2-2p}).
\]
\end{lemma}

\vspace{0.3cm}

\textit{Proof:}
We will use the same notation that was introduced in \eqref{leq1:one}, and we analyze $R_1, R_2, R_3$ and $R_4$ separately. Regarding $R_2$, we have
\begin{eqnarray*}
R_2-\epsilon&=&\epsilon \int_{\mu}^{\infty} \int_{-u/\sigma+c_p\tau^{\frac{1}{2-p}}}^{\infty} (\eta_p(u/\sigma+z;\tau)-u/\sigma-z)^2\phi(z)dzdF(u)+ \\
&&\epsilon \int_{\mu}^{\infty} \int_{-u/\sigma+ c_p\tau^{\frac{1}{2-p}}}^{\infty} 2z(\eta_p(u/\sigma+z;\tau)-u/\sigma-z)\phi(z)dzdF(u) -\epsilon \int_{\mu}^{\infty} \int_{-\infty}^{-u/\sigma+c_p\tau^{\frac{1}{2-p}}}z^2\phi(z)dzdF(u) \\
& \triangleq& Q_1+Q_2+Q_3 
\end{eqnarray*}
By property (c) listed in the proof of Lemma \ref{sharp:rate}, we have
\[
Q_1=\epsilon p^2\tau^2\int_{\mu}^{\infty} \int_{-u/\sigma+c_p \tau^{\frac{1}{2-p}}}^{\infty}\eta^{2p-2}_p(u/\sigma+z;\tau)\phi(z)dzdF(u).
\]
Using the same arguments (see the analysis of $T_2$) as in the proof of Lemma \ref{sharp:rate}, it is straightforward to show that
\begin{eqnarray*}
\lim_{\sigma \rightarrow 0} \frac{Q_1}{\tau^2\sigma^{2-2p}}=\epsilon p^2\mathbb{E}|U|^{2p-2}.
\end{eqnarray*}
Regarding $Q_2$, using integration by parts and property (b) given at the beginning of the proof of Lemma \ref{sharp:rate}, we obtain
\begin{eqnarray*}
Q_2&=&2\epsilon(\eta^+_p(c_p\tau^{\frac{1}{2-p}};\tau)-c_p \tau^{\frac{1}{2-p}}) \int_{\mu}^{\infty} \phi(-u/\sigma+c_p\tau^{\frac{1}{2-p}})dF(u)- \\
&&2\epsilon\int_{\mu}^{\infty} \int_{-u/\sigma+c_p\tau^{\frac{1}{2-p}}}^{\infty}\frac{\tau p(p-1)\eta_p^{p-2}(u/\sigma+z;\tau)}{1+\tau p(p-1)\eta_p^{p-2}(u/\sigma+z;\tau)}\phi(z)dzdF(u).
\end{eqnarray*}
We can directly see the first term on the right hand side of the above equation is bounded by $O(\tau^{\frac{1}{2-p}}\phi(\mu/(2\sigma)))$. By using the same technique applied for analyzing $T_2$, we then know the second term is of order $\tau \sigma^{2-p}$. Hence, we have
\begin{eqnarray*}
\lim_{\sigma \rightarrow 0}\frac{Q_2}{\tau \sigma^{2-p}}=2\epsilon p(1-p)\mathbb{E}|U|^{p-2}.
\end{eqnarray*}
We now analyze $Q_3$. A simple integration by parts yields,
\begin{eqnarray*}
Q_3=-\epsilon \left [\int_{\mu}^{\infty}(u/\sigma-c_p\tau^{\frac{1}{2-p}})\phi(u/\sigma-c_p\tau^{\frac{1}{2-p}})dF(u)+\int_{\mu}^{\infty} \int_{u/\sigma-c_p\tau^{\frac{1}{2-p}}}^{\infty}\phi(z)dzdF(u)  \right ].
\end{eqnarray*}
Using the fact that $\int_{t}^{\infty}\phi(z)dz\sim \frac{1}{t}\phi(t)$ and $\mu/\sigma-c_p \tau^{\frac{1}{2-p}}\rightarrow +\infty$, we can derive
\begin{eqnarray*}
&& \int_{\mu}^{\infty}(u/\sigma-c_p\tau^{\frac{1}{2-p}})\phi(u/\sigma-c_p\tau^{\frac{1}{2-p}})dF(u) \leq \int_{\mu}^{\infty}(u/\sigma)\phi(u/(2\sigma))dF(u) \leq \phi(\mu/(2\sigma)) \mathbb{E}|U|, \\
&&\int_{\mu}^{\infty} \int_{u/\sigma-c_p\tau^{\frac{1}{2-p}}}^{\infty}\phi(z)dzdF(u) \leq \int_{\mu/\sigma-c_p\tau^{\frac{1}{2-p}}}^{\infty}\phi(z)dz \leq O(1/(\mu/\sigma-c_p\tau^{\frac{1}{2-p}})\phi(\mu/\sigma-c_p\tau^{\frac{1}{2-p}})).
\end{eqnarray*}
It is then straightforward to confirm that
\begin{eqnarray*}
\lim_{\sigma \rightarrow 0} \frac{Q_3}{\tau \sigma^{2-p}}=0.
\end{eqnarray*}
Combing the results of $Q_1,Q_2$ and $Q_3$, we obtain
\begin{eqnarray}\label{risk:one}
\lim_{\sigma \rightarrow 0 } \frac{R_2-\epsilon}{\tau^2\sigma^{2-2p}}=\epsilon p^2\mathbb{E}|U|^{2p-2}.
\end{eqnarray}
Because of the minor difference between $R_2$ and $R_3$ (the sign of $u$, see more explanations in the proof of Lemma \ref{sharp:rate}), it is not hard to get
\begin{eqnarray}\label{risk:two}
\lim_{\sigma \rightarrow 0}\frac{R_3}{\tau^2 \sigma^{2-2p}}=0.
\end{eqnarray}
Regarding $R_4$, we first derive an upper bound in the following way:
\begin{eqnarray*}
R_4&=&\epsilon \int_{\mu}^{\infty} \int_{-u/\sigma-c_p\tau^{\frac{1}{2-p}}}^{-u/\sigma+c_p\tau^{\frac{1}{2-p}}}\frac{u^2}{\sigma^2}\phi(z)dzdF(u) \\
&\leq& 2\epsilon c_p\tau^{\frac{1}{2-p}}\sigma^{-2} \int_{\mu}^{\infty} u^2\phi(-u/\sigma+c_p\tau^{\frac{1}{2-p}})dF(u) \leq 2\epsilon c_p\tau^{\frac{1}{2-p}}\sigma^{-2}\phi(-\mu/\sigma+c_p\tau^{\frac{1}{2-p}})\mathbb{E}|U|^2.
\end{eqnarray*}
Since $\sigma^{2-p}\tau \rightarrow 0,$ as $\sigma \rightarrow 0$, we have
\begin{eqnarray}\label{risk:three}
\lim_{\sigma \rightarrow 0}\frac{R_4}{\tau^2\sigma^{2-2p}}=0.
\end{eqnarray}
We fianlly analyze $R_1$. A simple integration by parts proves
\begin{eqnarray}
\lefteqn{2(1-\epsilon)\int_{c_p\tau^{\frac{1}{2-p}}}^{\infty}\eta^2_p(z;\tau)\phi(z)dz = -2(1-\epsilon)\int_{c_p\tau^{\frac{1}{2-p}}}^{\infty} \frac{\eta^2_p(z;\tau)}{z}d\phi(z)} \nonumber \\
& =& -2(1-\epsilon) \left[ \frac{\eta^2_p(z;\tau)}{z}\phi(z)\right]_{c_p \tau^{\frac{1}{2-p}}}^\infty + 2(1-\epsilon)\int_{c_p \tau^{\frac{1}{2-p}}}^{\infty} \frac{2z\eta_p(z;\tau)\partial_1\eta_p (z;\tau)- \eta^2_p(z;\tau)}{z^2}\phi(z)dz. \label{r1:two}
\end{eqnarray}
Since $|\eta_p(z; \tau)| \leq |z|$, for the second integral in \eqref{r1:two} we have
\begin{eqnarray}
\lefteqn{\Bigg |\int_{c_p \tau^{\frac{1}{2-p}}}^{\infty} \frac{2z\eta_p(z;\tau)\partial_1\eta_p(z;\tau)- \eta^2_p(z;\tau)}{z^2}\phi(z)dz \Bigg | \leq \int_{c_p \tau^{\frac{1}{2-p}}}^{\infty} \frac{2}{|1+\tau p(p-1)\eta_p^{p-2}(z;\tau)|}\phi(z)dz }  \nonumber \\
&&+ \int_{c_p\tau^{\frac{1}{2-p}}}^{\infty} \phi(z)dz  \overset{(1)}{\leq}  \int_{c_p \tau^{\frac{1}{2-p}}}^{\infty} \frac{2}{C}\phi(z)dz  + \int_{c_p \tau^{\frac{1}{2-p}}}^{\infty} \phi(z)dz \leq (2C^{-1}+1) \int_{c_p\tau^{\frac{1}{2-p}}}^{\infty} \phi(z)dz \nonumber  \\
&\leq& O(\tau^{\frac{1}{p-2}}\phi(c_p \tau^{\frac{1}{2-p}})),
\end{eqnarray}
where $(1)$ is due to Lemma \ref{lem:lowerboundonx}. Hence the dominant term in \eqref{r1:two} is the first term. More specifically, we have
\begin{eqnarray}\label{risk:four}
\lim_{\sigma \rightarrow 0} \frac{R_1}{\tau^{\frac{1}{2-p}}\phi(c_p \tau^{\frac{1}{2-p}})}=\frac{2(1-\epsilon)(\eta^+_p(c_p;1))^2}{c_p}.
\end{eqnarray}
Putting the results from \eqref{risk:one}, \eqref{risk:two}, \eqref{risk:three}, \eqref{risk:four}, and Lemma \ref{sharp:rate}, we can conclude
\begin{eqnarray*}
\lim_{\sigma \rightarrow 0} \frac{R_p(\tau_*,\sigma)-\epsilon}{\tau^2\sigma^{2-2p}}=\epsilon p^2 \mathbb{E}|U|^{2p-2}.
\end{eqnarray*}
$\hfill \Box$


\vspace{.2cm}

\subsubsection{Auxiliary result for $p=0$}\label{ssec:thm4auxp0}
In this previous section we characterized the risk of $ R_p(\tau_*(\sigma),\sigma)$ for every $0<p<1$. The bounds we derived and the analysis we provided are not correct for $p=0$. In this section we derive the corresponding expansion for $p=0$.  Similar to the previous section consider two random variable $X \sim (1-\epsilon)\Delta_0+ \epsilon G$ and $U \sim G$,  and define $\mathbb{E}_G(f(U)) \triangleq \int f(u) dG(u)$ and $\mathbb{P}_G(U \in A) \triangleq \mathbb{E}_G (\mathbb{I} (U \in B))$.

\begin{proposition}\label{lem:riskbehsmallsigmazero}
Suppose $\mathbb{E}|U|^2 < \infty$ and $\mathbb{P}(|U|>\mu)=1$, where $\mu=\sup_v \{v : \mathbb{P}(|U|>v)=1 \} >0$, then for $p=0$,
\begin{eqnarray*}
 R_p(\tau_*(\sigma),\sigma)=\epsilon+o(\phi(\tilde{\mu}\sigma^{-1})),
\end{eqnarray*}
where $\tilde{\mu}$ is any constant that smaller than $\frac{\mu}{2}$.
\end{proposition}

The roadmap of the proof is similar to that of Proposition \ref{lem:riskbehsmallsigma}. We characterize the convergence rate of $\tau_*(\sigma)$ and derive the asymptotic formula for $R_0(\tau_*(\sigma),\sigma)$ in Lemma \ref{exact:rate:zero} and Lemma \ref{exact:riskrate:zero}, respectively. Proposition \ref{lem:riskbehsmallsigmazero} then follows directly by combing the results of these two lemmas. For the sake of brevity, we will skip some calculation details. 

\vspace{0.3cm}

\begin{lemma}\label{exact:rate:zero}
Suppose $\mathbb{E}|U|^2 < \infty$ and $P(|U|>\mu)=1$, where $\mu=\sup_v \{v : P(|U|>v)=1 \} >0$. Then for $p=0$, 
\[
\lim_{\sigma \rightarrow 0} \sqrt{\tau_*(\sigma)} \sigma =\frac{\mu}{2c_0},
\]
where $c_0$ is the constant $c_p$ with $p=0$ introduced in Lemma \ref{lem:threshform}.
\end{lemma}

\vspace{.3cm}

\begin{proof}
By using the same arguments presented in the proof of Lemma \ref{rough:rate}, we can obtain $\tau_*(\sigma) \rightarrow \infty$, as $\sigma \rightarrow 0$. Now we consider an arbitrary convergent sequence $\sigma_k \rightarrow 0$, as $k \rightarrow \infty$, and show $\sqrt{\tau_*(\sigma_k)}\sigma_k \rightarrow \mu/(2c_0)$. Denote $\lim_{k\rightarrow \infty} \sqrt{\tau_*(\sigma_k)}\sigma_k  =\alpha$. For notational simplicity, below we use exchangeably $\tau$ and $\tau_*$. Suppose $\alpha > \mu/c_0$, then by Fatou's lemma, we have
\[
\liminf_{k\rightarrow \infty} \mathbb{E} (\eta_0(U/\sigma_k+Z;\tau_*)-U/\sigma_k)^2 \geq \liminf_{k\rightarrow \infty} \mathbb{E} [\mathbbm{1}(|U+\sigma_kZ| \leq c_0\sqrt{\tau_*}\sigma_k)U^2/\sigma^2_k]=\infty.
\]
On the other hand, $R_0(\tau_*,\sigma_k)\leq R_0(0, \sigma_k)=1$. This is a contradiction. Hence we get $\alpha \leq \mu/c_0$. Next we aim to show $\alpha \leq \mu/(2c_0)$. Due to the explicit formula $\eta_0(u;\tau)=u\mathbbm{1}(|u|>c_0\sqrt{\tau})$, it is straightforward to derive
\begin{eqnarray}
R_0(\tau,\sigma)&=&2(1-\epsilon)\Big [c_0\sqrt{\tau}\phi(c_0\sqrt{\tau})+\int_{c_0\sqrt{\tau}}^{\infty}\phi(z)dz  \Big]+\epsilon \mathbb{E} \Bigg [ \Big(c_0\sqrt{\tau}-\frac{|U|}{\sigma}\Big)\phi \Big(c_0\sqrt{\tau}-\frac{|U|}{\sigma}\Big)+\int_{c_0\sqrt{\tau}-\frac{|U|}{\sigma}}^{\infty}\phi(z)dz \Bigg ] \nonumber \\ 
&&+ \epsilon \mathbb{E} \Bigg [ \Big(c_0\sqrt{\tau}+\frac{|U|}{\sigma}\Big)\phi \Big(c_0\sqrt{\tau}+\frac{|U|}{\sigma}\Big)+\int_{c_0\sqrt{\tau}+\frac{|U|}{\sigma}}^{\infty}\phi(z)dz \Bigg ]+\epsilon \mathbb{E} \int_{-c_0\sqrt{\tau}-\frac{|U|}{\sigma}}^{c_0\sqrt{\tau}-\frac{|U|}{\sigma}}\frac{|U|^2}{\sigma^2}\phi(z)dz, \nonumber \\
&\triangleq& R_1+R_2+R_3+R_4.  \label{pzero:eq:one}
\end{eqnarray}
Moreover, it is straightforward to show that $\tau_*(\sigma_k)$, the optimal thresholding value, is finite and non-zero, and hence we have $\frac{\partial R_0(\tau_*(\sigma_k),\sigma_k)}{\partial \tau}=0$, i.e., 
\begin{equation}
\frac{\partial R_1}{\partial \tau}+\frac{\partial R_2}{\partial \tau}+\frac{\partial R_3}{\partial \tau}+\frac{\partial R_4}{\partial \tau}=0, \label{pzero:eq:two}
\end{equation}
where we know\footnote{The condition $\mathbb{E}|U|^2<\infty$ enables us to apply dominated convergence theorem to exchange the differentiation and expectation in the calculation of the partial derivatives.}
\begin{eqnarray*}
&&\frac{\partial R_1}{\partial \tau}=(\epsilon-1)c_0^3\sqrt{\tau}\phi(c_0\sqrt{\tau}), \quad \frac{\partial R_2}{\partial \tau}= \frac{-\epsilon c_0}{2\sqrt{\tau}}\mathbb{E}\Bigg[ \Big(c_0\sqrt{\tau}-\frac{|U|}{\sigma_k}\Big)^2\phi \Big(c_0\sqrt{\tau}-\frac{|U|}{\sigma_k}\Big)  \Bigg], \\
&&\frac{\partial R_3}{\partial \tau}= \frac{-\epsilon c_0}{2\sqrt{\tau}}\mathbb{E}\Bigg[ \Big(c_0\sqrt{\tau}+\frac{|U|}{\sigma_k}\Big)^2\phi \Big(c_0\sqrt{\tau}+\frac{|U|}{\sigma_k}\Big)  \Bigg], \\
&&\frac{\partial R_4}{\partial \tau} = \frac{\epsilon c_0}{2\sqrt{\tau}\sigma_k^2} \mathbb{E}\Bigg\{ |U|^2\Bigg[ \phi \Big(c_0\sqrt{\tau}-\frac{|U|}{\sigma_k}\Big) +\phi \Big(c_0\sqrt{\tau}+\frac{|U|}{\sigma_k}\Big)  \Bigg] \Bigg \}.
\end{eqnarray*}
A few more algebra calculations yields,
\begin{eqnarray*}
&& \frac{\partial R_2}{\partial \tau}+\frac{\partial R_3}{\partial \tau}+\frac{\partial R_4}{\partial \tau}=\frac{-\epsilon c^2_0}{2\sigma_k}\mathbb{E}\Bigg[ \Big(c_0\sqrt{\tau}\sigma_k-2|U|)\phi \Big(c_0\sqrt{\tau}-\frac{|U|}{\sigma_k}\Big)  \Bigg] + \\
&&\frac{-\epsilon c^2_0}{2\sigma_k}\mathbb{E}\Bigg[ \Big(c_0\sqrt{\tau}\sigma_k+2|U|)\phi \Big(c_0\sqrt{\tau}+\frac{|U|}{\sigma_k}\Big)  \Bigg]  \sim \frac{1}{\sigma_k}\mathbb{E} \Bigg [|U|\phi \Big(c_0\sqrt{\tau}-\frac{|U|}{\sigma_k}\Big) \Bigg],
\end{eqnarray*}
where the notation $\sim$ indicates that they have the same orders in terms of $\sigma_k \rightarrow 0$. Hence, dividing both sides of Equation \eqref{pzero:eq:two} by $\sqrt{\tau}\phi(c_0\sqrt{\tau})$ and letting $k \rightarrow \infty$ shows 
\begin{eqnarray}
0 < \lim_{k\rightarrow \infty} \mathbb{E}\Bigg[ |U|\mbox{exp}\Big(\frac{|U|(|U|-2\sigma_kc_0\sqrt{\tau}) }{-2\sigma^2_k}\Big)\Bigg ]  < \infty.  \label{pzero:eq:three}
\end{eqnarray}
If $\alpha > \mu/(2c_0)$, then we see
\begin{eqnarray*}
\mathbb{E}\Bigg[ |U|\mbox{exp}\Big(\frac{|U|(|U|-2\sigma_kc_0\sqrt{\tau}) }{-2\sigma^2_k}\Big)\Bigg ]  \geq 
\mathbb{E}\Bigg[ |U|\mbox{exp}\Big(\frac{|U|(|U|-2\sigma_kc_0\sqrt{\tau}) }{-2\sigma^2_k}\Big)\cdot \mathbbm{1}(|U|<2\alpha c_0)\Bigg ]  \rightarrow +\infty.
\end{eqnarray*}
We have used Fatou's lemma to obtain the last limit. Obviously the inequality above contradicts \eqref{pzero:eq:three}. Thus we obtain an upper bound $\mu/(2c_0)$ for $\alpha$. Finally we would like to derive $\alpha \geq \mu/(2c_0)$. First note that since $\alpha \leq \mu/(2c_0)$, it is not hard to confirm that when $k$ is large,
\[
\frac{\partial R_4}{\partial \tau} \leq \frac{\epsilon c_0\mathbb{E}|U|^2}{\sqrt{\tau}\sigma^2_k}\phi \Big(c_0\sqrt{\tau}-\frac{\mu}{\sigma_k}\Big)=O\Bigg(\frac{1}{\sqrt{\tau}\sigma_k^2}\phi \Big(c_0\sqrt{\tau}-\frac{\mu}{\sigma_k}\Big) \Bigg).
\]
Based on the inequality above, we can further obtain
\begin{equation}
\Bigg |\frac{\partial R_2}{\partial \tau}+\frac{\partial R_3}{\partial \tau}+\frac{\partial R_4}{\partial \tau}\Bigg | \leq O\Bigg(\frac{1}{\sqrt{\tau}\sigma_k^2}\phi \Big(c_0\sqrt{\tau}-\frac{\mu}{\sigma_k}\Big) \Bigg). \label{pzero:eq:four}
\end{equation}
Now suppose $\alpha < \mu/(2c_0)$, then it follows that
\begin{eqnarray*}
\frac{1}{\sqrt{\tau}\sigma_k^2}\phi \Big(c_0\sqrt{\tau}-\frac{\mu}{\sigma_k}\Big) \cdot \frac{1}{\sqrt{\tau}\phi(c_0\sqrt{\tau})}=\frac{1}{\tau \sigma_k^2}\mbox{exp}\Big(\frac{\mu (\mu-2c_0\sigma_k\sqrt{\tau})}{-2\sigma_k^2} \Big)=o(1).
\end{eqnarray*}
However, this fact combined with \eqref{pzero:eq:four} implies that if we divide Equation \eqref{pzero:eq:two} by $\sqrt{\tau}\phi(c_0\sqrt{\tau})$ and letting $k \rightarrow \infty$, we would get
\[
(\epsilon-1)c_0^3=0,
\]
which is a contradiction. Therefore, we have showed that for an arbitrary convergent sequence $\sigma_k \rightarrow 0$, we have $\sqrt{\tau_*(\sigma_k)}\sigma_k \rightarrow \mu/(2c_0)$, as $k \rightarrow \infty$. This completes the proof.
\end{proof}

\vspace{0.3cm}

\begin{lemma} \label{exact:riskrate:zero}
Suppose $\mathbb{E}|U|^2 < \infty$ and $P(|U|>\mu)=1$, where $\mu=\sup_v \{v : P(|U|>v)=1 \}>0$. Then, for $p=0$
\begin{eqnarray*}
 R_0(\tau_*(\sigma),\sigma)=\epsilon+O(\sqrt{\tau_*}\phi(c_0\sqrt{\tau_*})),
\end{eqnarray*}
\end{lemma}

\vspace{0.3cm}

\begin{proof}
We use the same notations from the proof of Lemma \ref{exact:rate:zero}. Then,
\begin{eqnarray*}
R_0(\tau_*(\sigma),\sigma)-\epsilon =R_1 + (R_2-\epsilon)+R_3+R_4.
\end{eqnarray*}
Using the fact that $\sqrt{\tau_*(\sigma)}\sigma \rightarrow \frac{\mu}{2c_0}$ according to Lemma \ref{exact:rate:zero}, from \eqref{pzero:eq:one} we can easily obtain
\begin{eqnarray*}
R_1 =O(\sqrt{\tau_*}\phi(c_0\sqrt{\tau_*})),~ R_2-\epsilon=O\Bigg(\mathbb{E}\Big[\frac{|U|}{\sigma}\phi \Big(c_0\sqrt{\tau_*}-\frac{|U|}{\sigma}\Big)\Big ]\Bigg),~R_3=O\Bigg(\mathbb{E}\Big[\frac{|U|}{\sigma}\phi \Big(c_0\sqrt{\tau_*}+\frac{|U|}{\sigma}\Big)\Big ]\Bigg)
\end{eqnarray*}
Regarding $R_4$, we have
\begin{eqnarray*}
R_4&=&\epsilon \mathbb{E} \Bigg [\frac{|U|^2}{\sigma^2}\cdot \Big(\int_{\frac{|U|}{\sigma}-c_0\sqrt{\tau_*}}^{\infty}\phi(z)dz-\int_{\frac{|U|}{\sigma}+c_0\sqrt{\tau_*}}^{\infty}\phi(z)dz\Big) \Bigg] \\
&\leq& \epsilon \mathbb{E} \Bigg [\frac{|U|^2}{\sigma^2}\cdot \int_{\frac{|U|}{\sigma}-c_0\sqrt{\tau_*}}^{\infty}\phi(z)dz\Bigg] \leq \epsilon \mathbb{E}\Bigg[ \frac{|U|}{\sigma}\phi \Big(\frac{|U|}{\sigma}-c_0\sqrt{\tau_*}\Big)\frac{|U|}{|U|-c_0\sigma \sqrt{\tau_*}}\Bigg]  \\
&\leq& O\Bigg(\mathbb{E}\Bigg[ \frac{|U|}{\sigma}\phi \Big(\frac{|U|}{\sigma}-c_0\sqrt{\tau_*}\Big) \Bigg] \Bigg).
\end{eqnarray*}
Furthermore, from \eqref{pzero:eq:three} we can see
\begin{eqnarray*}
\mathbb{E}\Bigg[ \frac{|U|}{\sigma}\phi \Big(\frac{|U|}{\sigma}-c_0\sqrt{\tau_*}\Big) \Bigg]  \cdot \frac{1}{\sqrt{\tau_*}\phi(c_0\sqrt{\tau_*})}=O(1).
\end{eqnarray*}
Putting together what we have derived so far shows
\[
R_0(\tau_*(\sigma),\sigma)-\epsilon =O(\sqrt{\tau_*}\phi(c_0\sqrt{\tau_*})).
\]
\end{proof}


\subsection{Proof of Proposition \ref{prop:ell_1pt}}\label{sec:proofprop1}

In order to prove this proposition, we require several preliminary results. Define
\[
\bar{\psi}_{\beta,1}(\sigma^2) \triangleq  \frac{1}{\delta}\mathbb{E} (\eta_1(X + \sigma Z; \beta \sigma) -X)^2,
\]
where $\beta$ denotes a fixed number greater than zero. 
\vspace{.2cm}

\begin{lemma}\label{lem:concavityell_1}\cite{donoho2009supporting}
For every $\beta>0$, $\bar{\psi}_{\beta,1}(\sigma^2)$ is a concave function of $\sigma^2$.
\end{lemma}

\vspace{.2cm}

A simple corollary of this result is that $\bar{\psi}_{\beta,1}(\sigma^2)$ has a unique stable fixed point. Refer to \cite{donoho2009supporting} for more information on this lemma. 

\vspace{.2cm}

\begin{lemma}\label{lem:ratiofinite}
Let $X \sim (1-\epsilon) \Delta_0 + \epsilon G$ with $\epsilon<1$. Let $\lambda_*(\sigma)$ denote the optimal thresholding policy for $\ell_1$-{\rm AMP}. Then,
\[
0 < \lim_{\sigma^2 \rightarrow 0} \frac{\lambda_*(\sigma)}{ \sigma} <\infty. 
\]
\end{lemma}

\textit{Proof:}
We start by assuming $ \lim_{\sigma^2 \rightarrow 0} \frac{\lambda_*(\sigma)}{ \sigma} $ exists. We first show that the limit is not zero. Note that for every $\beta \geq 0$, we have
\begin{eqnarray}\label{eq:ratiofixedthresh}
\frac{\Psi_{\lambda_*,1} (\sigma^2)}{\bar{\psi}_{\beta,1}(\sigma^2)} \leq 1.
\end{eqnarray}
This is due to the fact that $\lambda_*(\sigma)$ is the optimal thresholding policy and outperforms all the other thresholding policies including $\lambda(\sigma) = \beta \sigma$ for a fixed $\beta$. Define $\tau_*(\sigma)\triangleq \frac{\lambda_*(\sigma)}{\sigma}$. Our goal is to show that if $\tau_*(\sigma) \rightarrow 0$ as $\sigma \rightarrow 0$, then the ratio specified in \eqref{eq:ratiofixedthresh} will be larger than $1$ for all the $\beta$ around zero which is in contradiction with \eqref{eq:ratiofixedthresh}. Note that
\begin{eqnarray}\label{eq:proof1}
\lim_{\sigma^2 \rightarrow 0} \frac{\Psi_{\lambda_*,1} (\sigma^2)}{\bar{\psi}_{\beta,1}(\sigma^2)} = \lim_{\sigma^2 \rightarrow 0} \frac{\mathbb{E} (\eta_1(X+ \sigma Z; \tau_*(\sigma) \sigma) - X)^2}{\mathbb{E} (\eta_1(X+ \sigma Z; \beta \sigma) - X)^2}=  \lim_{\sigma^2 \rightarrow 0} \frac{\mathbb{E} (\eta_1(X/\sigma+  Z; \tau_*(\sigma) ) - X/\sigma)^2}{\mathbb{E} (\eta_1(X/\sigma+  Z; \beta ) - X/\sigma)^2},
\end{eqnarray}
Consider a random variable $U \sim G$. Then, \eqref{eq:proof1} can be simplified in the following way:\footnote{We have used dominated convergence theorem in these calculations. It is straightforward to prove that the conditions for this theorem hold. But, for the sake of brevity and since we have studied similar problems in the proof of Theorem \ref{lem:noiselesslowfpless1}, we do not check the conditions here. }
\begin{eqnarray}\label{eq:ratioproofell_1}
\lim_{\sigma^2 \rightarrow 0} \frac{\Psi_{\lambda_*,1} (\sigma^2)}{\bar{\psi}_{\beta,1}(\sigma^2)} &=&  \lim_{\sigma^2 \rightarrow 0} \frac{\epsilon \mathbb{E}  ( \eta_1(U/\sigma+  Z; \tau_*(\sigma) ) - U/\sigma)^2 + (1-\epsilon)\mathbb{E} (\eta_1^2(Z; \tau_*(\sigma)))}{ \epsilon \mathbb{E} (\eta_1(U/\sigma+  Z; \beta ) - U/\sigma)^2 + (1-\epsilon)\mathbb{E} (\eta_1^2(Z; \beta)) } \nonumber \\
&=& \frac{1}{\epsilon(1+ \beta^2) + (1-\epsilon) \mathbb{E} (\eta_1^2(Z; \beta))},
\end{eqnarray} 
where the last equality is due to the assumption that  $\tau_*(\sigma) \rightarrow 0$ as $\sigma \rightarrow 0$. Note that for $\beta=0$, the numerator and denominator will be the same and hence the ratio is equal to one. However,  a simple calculation shows that
\[
\left. \frac{d}{d\beta} \epsilon(1+ \beta^2) + (1-\epsilon) \mathbb{E} (\eta_1^2(Z; \beta)) \right|_{\beta=0} =-4 (1-\epsilon) \int_{0} ^{\infty} z \phi (z) dz <0. 
\]
Hence, for $\beta$ in the neighborhood of zero, the ratio in \eqref{eq:ratioproofell_1} will be greater than $1$, which is in contradiction with \eqref{eq:ratiofixedthresh} and consequently $\tau_*(\sigma) \nrightarrow 0$. 

\vspace{.2cm}

We now discuss the other part of the lemma, i.e., the proof of
\[
 \lim_{\sigma^2 \rightarrow 0} \frac{\lambda_*(\sigma)}{ \sigma}< \infty.
\]
As before, define $\tau_*(\sigma)  \triangleq \frac{\lambda_*(\sigma)}{ \sigma}$ and consider a random variable $U \sim G$. From the derivation in \eqref{eq:ratioproofell_1}, we know that 
\begin{eqnarray}\label{contra:inf}
\infty > \lim_{\sigma^2 \rightarrow 0} \frac{\Psi_{\lambda_*,1} (\sigma^2)}{\sigma^2} &=& \lim_{\sigma^2\rightarrow 0} \epsilon \mathbb{E}  ( \eta_1(U/\sigma+  Z; \tau_*(\sigma) ) - U/\sigma)^2 + (1-\epsilon)\mathbb{E} (\eta_1^2(Z; \tau_*(\sigma))) \nonumber \\
&\geq& \lim_{\sigma^2\rightarrow 0} \epsilon \mathbb{E}  ( \eta_1(U/\sigma+  Z; \tau_*(\sigma) ) - U/\sigma)^2. 
\end{eqnarray}
Suppose $\tau_*(\sigma) \rightarrow \infty$. Since $\eta_1(u;\lambda)=\mbox{sign}(u)(|u|-\lambda)_+$, we can easily see the following,
\[
( \eta_1(U/\sigma+  Z; \tau_*(\sigma) ) - U/\sigma)^2 \geq \min\{(Z-\tau_*(\sigma))^2, (Z+\tau_*(\sigma))^2, U^2/\sigma^2\} \rightarrow +\infty, \mbox{~as~} \sigma \rightarrow 0
\]
Hence, by Fatou's lemma, we conclude
\[
 \lim_{\sigma^2\rightarrow 0} \mathbb{E}  ( \eta_1(U/\sigma+  Z; \tau_*(\sigma) ) - U/\sigma)^2 \rightarrow \infty, 
\]
which contradicts Inequality \eqref{contra:inf}.

So far, we have proved if $ \lim_{\sigma^2 \rightarrow 0} \frac{\lambda_*(\sigma)}{ \sigma} $ exists, it must be a finite non-zero number. Now we consider any convergent sequence $\sigma_n \rightarrow 0$ such that $\lim_{n \rightarrow \infty} \frac{\lambda_*(\sigma_n)}{ \sigma_n}=\alpha$. Then all the arguments presented before work for the sequence. Hence $0<\alpha< \infty$. Similar to \eqref{eq:ratioproofell_1}, we can obtain
\begin{eqnarray}
1 \geq \lim_{n \rightarrow \infty} \frac{\Psi_{\lambda_*,1} (\sigma_n^2)}{\bar{\psi}_{\beta,1}(\sigma_n^2)} = \frac{\epsilon(1+ \alpha^2) + (1-\epsilon) \mathbb{E} (\eta_1^2(Z; \alpha))}{\epsilon(1+ \beta^2) + (1-\epsilon) \mathbb{E} (\eta_1^2(Z; \beta))}, \label{unique:convex}
\end{eqnarray}
for any $\beta \geq 0$. It is straightforward to confirm that $\epsilon(1+ \beta^2) + (1-\epsilon) \mathbb{E} (\eta_1^2(Z; \beta))$, as a function of $\beta$, is strictly convex and has a unique minimizer over $[0,\infty)$. Denote that global optima by $\beta_*$. If we choose $\beta=\beta_*$ in \eqref{unique:convex}, we can immediately conclude $\alpha=\beta_*$. Since we have been discussing an arbitrary convergent sequence, it implies that $\lim_{\sigma^2 \rightarrow 0}\frac{\lambda_*(\sigma)}{ \sigma}=\beta_*$. This completes the proof.

With this background information, we can now prove Proposition \ref{prop:ell_1pt}. 

\vspace{.2cm}

\textit{Proof:}
For simplicity, we only consider the noiseless setting in the proof. The uniqueness of the fixed point in the noisy case follows similar arguments. We start proving the uniqueness by contradiction. Suppose that $\Psi_{\lambda_*,1}$ has two fixed points $0<\sigma_1^2<\sigma_2^2$. Define $\beta_*= \frac{\lambda_*(\sigma_1)}{\sigma_1}$ and consider a new thresholding policy $\lambda(\sigma) = \beta_* \sigma$. According to Lemma \ref{lem:concavityell_1}, we know that $\bar{\psi}_{\beta_*,1}(\sigma^2)$ has only one stable fixed point. That fixed point is clearly $\sigma_1$. Therefore, $\bar{\psi}_{\beta^*,1}(\sigma^2)< \sigma^2$ for every $\sigma^2> \sigma_1^2$. Now since $\Psi_{\lambda_*, 1} (\sigma_2^2) = \sigma_2^2$ and $\sigma^2_2>\sigma^2_1$, we conclude that $\bar{\psi}_{\beta_*,1}(\sigma_2^2)<\Psi_{\lambda_*, 1} (\sigma_2^2)$. This is in contradiction with the fact that $\lambda_*(\sigma)$ is the optimal thresholding policy. Therefore, $\Psi_{\lambda_*,1}$ has at most one fixed point above zero. 

If zero is not a stable fixed point, then according to Proposition \ref{proof:existsncestable} $\Psi_{\lambda_*,1}$ has at least one non-zero stable fixed point, hence it has a unique stable fixed point above zero. Finally, we show that if zero is a stable fixed point, then $\Psi_{\lambda_*,1}$ does not have any other fixed point. Define $\tau_*(\sigma)\triangleq \frac{\lambda_*(\sigma)}{\sigma}$. Note that according to Lemma \ref{lem:ratiofinite},
\[
0 < \lim_{\sigma^2 \rightarrow 0} \tau_*(\sigma) < \infty.
\]
Let $ \lim_{\sigma^2 \rightarrow 0} \tau_*(\sigma) = \beta^*$ and $U \sim G$. Then we have
\begin{eqnarray}\label{eq:ptoofptell11}
\left. \frac{d \Psi_{\lambda_*,1} (\sigma^2)}{d \sigma^2} \right|_{\sigma^2=0} &=& \lim_{\sigma^2 \rightarrow 0} \frac{\Psi_{\lambda_*,1} (\sigma^2)}{\sigma^2}=\frac{1}{\delta} \lim_{\sigma^2 \rightarrow 0} \epsilon \mathbb{E}( \eta_1(U/\sigma +Z; \tau_*(\sigma))  -U/\sigma)^2 + (1-\epsilon) \mathbb{E} \left(\eta^2_1 (Z; \tau_*(\sigma)) \right)\nonumber \\
&=& \frac{\epsilon(1+(\beta^*)^2) + (1-\epsilon)  \mathbb{E} \left(\eta^2_1 (Z; \beta^*) \right)}{\delta} \nonumber \\
&=&\frac{\min_{\beta \geq 0}\epsilon(1+\beta^2) + (1-\epsilon)  \mathbb{E} \left(\eta^2_1 (Z; \beta) \right)}{\delta}. 
\end{eqnarray}
Note that the last two equalities above can be obtained from the arguments in the proof of Lemma \ref{lem:ratiofinite}. If $0$ is a stable fixed point, then 
\[
\left. \frac{d \Psi_{\lambda_*,1} (\sigma^2)}{d \sigma^2} \right|_{\sigma^2=0} <1. 
\]

It is straightforward to confirm that $\left. \frac{d \Psi_{\lambda_*,1} (\sigma^2)}{d \sigma^2} \right|_{\sigma^2=0}$ is the same as the derivative of $\bar{\psi}_{\beta^*,1} (\sigma^2)$ at zero. However, since $\bar{\psi}_{\beta^*,1} (\sigma^2)$ is concave and its derivative at zero is less than $1$, it will not have any other fixed point and $\bar{\psi}_{\beta^*,1} (\sigma^2) < \sigma^2$ for every $\sigma^2>0$. Hence if $\Psi_{\lambda_*,1}(\sigma^2)$ has another fixed point at $\sigma_0^2 >0$, we conclude that
\[
 \bar{\psi}_{\beta^*,1} (\sigma_0^2)  <\Psi_{\lambda_*,1}(\sigma_0^2),
\]
which is in contradiction with the optimality of $\lambda_*$.

It is now straightforward to characterize the phase transition of the optimal-$\lambda$ $\ell_1$-AMP. Note that according to our discussion, $\sigma^2=0$ is the unique fixed point if and only if 
\begin{equation}\label{eq:derivativelessthan1}
 \left. \frac{d \Psi_{\lambda_*,1} (\sigma^2)}{d \sigma^2} \right|_{\sigma^2=0} <1.
\end{equation}
Combining \eqref{eq:ptoofptell11} and \eqref{eq:derivativelessthan1} finishes the proof.

$\hfill \Box$


\subsection{Proof of Theorem \ref{thm:noisyoptimallambda}} \label{sec:proofthmnoisyfirst}

Let $\tilde{\sigma}^2$ denote the smallest value of $\sigma$ at which $\frac{d\Psi _{\lambda _*,p}(\sigma ^2)}{d\sigma ^2}$ is equal to one. If it does not exist, we set $\tilde{\sigma} = \infty$. According to Theorem \ref{lem:noiselesslowfpless1}, the derivative of  ${\Psi _{{\lambda _*},p}}({\sigma ^2})$ at ${\sigma ^2} = 0$ equals to $\frac{\epsilon}{\delta}$. Since $\epsilon  < \delta$, we conclude that $\frac{{d{\Psi _{{\lambda _*},p}}({\sigma ^2})}}{{d{\sigma ^2}}} < 1$ for every $\sigma^2< \tilde{\sigma}^2$. Define $\sigma _0^2 \triangleq {\tilde \sigma ^2} - {\Psi _{{\lambda _*},p}}({\tilde \sigma ^2})$ ($\sigma_0=\infty$ if $\tilde{\sigma}=\infty$). Note that $\sigma_0^2>0$, since ${\Psi _{{\lambda _*},p}}(0) = 0$ and the derivative of ${\Psi _{{\lambda _*},p}}({\sigma ^2})$ is less than one for every $\sigma^2< \tilde{\sigma}^2$. Our next step is to show that for every $\sigma_w^2< \sigma_0^2$, the equation
\[
\sigma^2 = \sigma_w^2 + {\Psi _{{\lambda _*},p}}({\sigma ^2})
\]
has one solution in $[0,\tilde{\sigma}^2]$. Define $\Gamma ({\sigma ^2}) \triangleq {\sigma ^2} - {\Psi _{{\lambda _*},p}}({\sigma ^2}) - \sigma _w^2$. Note that $\Gamma(0)<0$ and $\Gamma(\tilde{\sigma}^2)>0$. Furthermore, it is straightforward to see that the derivative of $\Gamma ({\sigma ^2})$ is positive and hence it is an increasing function. Thus ${\sigma ^2} - {\Psi _{{\lambda _*},p}}({\sigma ^2}) - \sigma _w^2 = 0$ has exactly one solution in the range $[0, \tilde{\sigma}^2]$. This is the lowest fixed point of $\sigma^2 = \sigma_w^2 + {\Psi _{{\lambda _*},p}}({\sigma ^2})$. By employing the implicit function theorem, we conclude that
\begin{eqnarray}\label{eq:derivatvesigmaell}
\frac{{d\sigma _\ell^2}}{{d\sigma _w^2}} = \frac{1}{1 - \left. \frac{d\Psi _{\lambda _*,p}(\sigma ^2)}{d\sigma ^2} \right|_{\sigma^2= \sigma_\ell^2}}.
\end{eqnarray}
Therefore, $\sigma_{\ell}^2$, as a function of $\sigma^2_w$, is differentiable and has finite derivative for any $\sigma_w^2< \sigma_0^2$. According to Theorem \ref{lem:noiselesslowfpless1}, $\epsilon<\delta$ and $\sigma_w^2=0$ leads to $\sigma_{\ell}^2=0$. Hence, the continuity of $\sigma_{\ell}^2$ implies that
\begin{equation}\label{eq:der2sigma}
\lim_{\sigma_w^2 \rightarrow 0} \sigma_{\ell}^2= 0.
\end{equation}
Combining \eqref{eq:derivatvesigmaell} and \eqref{eq:der2sigma} we conclude that
\begin{eqnarray*}
\lim_{\sigma_w^2 \rightarrow 0} \frac{{d\sigma _\ell^2}}{{d\sigma _w^2}} = \lim_{\sigma_w^2 \rightarrow 0} \frac{1}{1 - \left. \frac{d\Psi _{\lambda _*,p}(\sigma ^2)}{d\sigma ^2} \right|_{\sigma= \sigma_\ell} } =  \frac{1}{1 - \left. \frac{d\Psi _{\lambda _*,p}(\sigma ^2)}{d\sigma ^2} \right|_{\sigma= 0} } = \frac{1}{1-\frac{\epsilon}{\delta}},
\end{eqnarray*}
where the last equality is from the proof of Theorem \ref{lem:noiselesslowfpless1}.


\subsection{Proof of Theorem \ref{thm:noisyell_1limit}}\label{sec:proofnoisyell_1first}
This proof is essentially a combination of the results we obtained in the proofs of Theorem \ref{thm:noisyoptimallambda} and Proposition \ref{prop:ell_1pt}. Note that as we proved in Proposition \ref{prop:ell_1pt}, the stable fixed point of 
\[
\sigma^2 = \sigma_w^2+ \Psi_{\lambda_*,1} (\sigma^2) 
\]
is unique and we have used the notation $\sigma_\ell^2$ to refer to this unique fixed point. Moreover, since $M_1(\epsilon) <\delta$, we know $\sigma_{\ell}^2=0$ when $\sigma_w=0$ from Proposition \ref{prop:ell_1pt}. Similar to \eqref{eq:derivatvesigmaell}, we have
\begin{eqnarray}\label{eq:derivatvesigmaell1}
\frac{{d\sigma _\ell^2}}{{d\sigma _w^2}} = \frac{1}{1 -\left. \frac{d\Psi _{\lambda _*,1}(\sigma ^2)}{d\sigma ^2}\right|_{\sigma^2 = \sigma_\ell^2} }.
\end{eqnarray}
 Finally, we already know from Proposition \ref{prop:ell_1pt} that
\begin{eqnarray}\label{eq:derell1atzero}
\left. \frac{\partial \Psi_{\lambda_*,1} (\sigma^2)}{\partial \sigma^2} \right|_{\sigma^2=0} = \inf_{\alpha\geq0} \epsilon(1+\alpha^2) + (1-\epsilon)  \mathbb{E} [\eta^2_1 (Z; \alpha)]=M_1(\epsilon). 
\end{eqnarray}
Using the continuity arguments of $\sigma_{\ell}^2$ as in Theorem \ref{thm:noisyoptimallambda}, combined with \eqref{eq:derivatvesigmaell1} and \eqref{eq:derell1atzero}, completes the proof.


\subsection{Proof of Proposition \ref{prop:optimallassolargesigma} } \label{sec:prooflargenoiseell_1}

Define
\begin{equation*}
r_{\alpha, p} (\sigma^2) \triangleq \mathbb{E} (\eta_p (X+ \sigma Z; (\alpha \sigma/ c_p)^{2-p})-X)^2,
\end{equation*}
where the expected value is with respect to two independent random variables $X \sim (1- \epsilon)\Delta_0+ \epsilon G$ and $Z \sim N(0,1)$. $c_p$ is the constant introduced in Lemma \ref{lem:threshform}. $\alpha$ is a fixed positive number. Note that according to Lemma \ref{lem:threshform}, the thresholding policy $(\alpha \sigma/ c_p)^{2-p}$ that is used in the definition of  $r_{\alpha, p} (\sigma^2)$ ensures that $\eta_p (u; (\alpha \sigma/ c_p)^{2-p})=0$ for $|u|< \alpha \sigma$ and for every $0\leq p \leq 1$.  Furthermore, note that $\frac{1}{\delta}r_{\alpha, p} (\sigma^2)$ is equal to $\Psi_{\bar{\lambda}_{\alpha}(\sigma), p} (\sigma^2)$ for the thresholding policy $\bar{\lambda}_\alpha(\sigma) = (\alpha \sigma/ c_p)^{2-p}$. We start with several lemmas that are important in the proof of Proposition \ref{prop:optimallassolargesigma}. 

\vspace{.2cm}

\begin{lemma}\label{lem:largenoiseell_1versusothers}
For large values of $\sigma$, we have
\[
{r_{\alpha,p }}({\sigma ^2}) \sim {\Gamma _{\alpha,p }}{\sigma ^2},
\]
where $${\Gamma _{\alpha,p }} \triangleq \mathbb{E} \left( {\eta _p^2(Z;{{(\alpha /{c_p})}^{2 - p}})} \right). $$ Furthermore, ${\Gamma _{\alpha,1 }} < {\Gamma _{\alpha,p }}$ for every $0 \leq p < 1$ and $\alpha>0$.
\end{lemma}
\vspace{.2cm}

\textit{Proof:} Let $X$ denote a random variable with distribution $(1- \epsilon)\delta_0+ \epsilon G$ and $U \sim G$ be another random variable. We have

\begin{eqnarray}
  \mathop {\lim }\limits_{{\sigma ^2} \to \infty } \frac{{r_{\alpha,p }}({\sigma ^2})}{\sigma^2} &=& \mathop {\lim }\limits_{{\sigma ^2} \to \infty } \frac{{(1 - \epsilon ) \mathbb{E}\left( {\eta _p^2(\sigma Z;{{(\alpha \sigma /{c_p})}^{2 - p}})} \right) + \epsilon \mathbb{E} {{\left( {{\eta _p}(U + \sigma Z;{{(\alpha \sigma /{c_p})}^{2 - p}}) - U} \right)}^2}}}{{{\sigma ^2}}} \nonumber \\
   &\overset{(a)}{=}& \mathop {\lim }\limits_{{\sigma ^2} \to \infty } \left[ {(1 - \epsilon ) \mathbb{E}\left( {\eta _p^2(Z;{{(\alpha /{c_p})}^{2 - p}})} \right) + \epsilon \mathbb{E}{{\left( {{\eta _p}(U/\sigma  + Z;{{(\alpha /{c_p})}^{2 - p}}) - U/\sigma } \right)}^2}} \right] \nonumber \\
   &\overset{(b)}{=}& (1 - \epsilon ) \mathbb{E} \left( {\eta _p^2(Z;{{(\alpha /{c_p})}^{2 - p}})} \right) + \epsilon \mathbb{E} \left( {\eta _p^2(Z;{{(\alpha /{c_p})}^{2 - p}})} \right) \nonumber  \\
   &=& \mathbb{E} \left( {\eta _p^2(Z;{{(\alpha /{c_p})}^{2 - p}})} \right), \nonumber
\end{eqnarray}
where Equality (a) is according to Lemma \ref{lem:scaleinv}. To obtain Equality (b), we have assumed that the limit and expectation are interchangeable.  The proof is similar to the proof we presented in Section \ref{proof sec:thmnoiselesslfp} and hence skipped. Furthermore, according to Corollary \ref{cor:etap1eta0},
\begin{eqnarray*}
&& \eta _p^2(u;{(\alpha /{c_p})^{2 - p}}) > \eta _1^2(u;(\alpha /{c_1})) \ \ \ \forall |u|>\alpha, 0 \leqslant p < 1, \\
&& \eta _p^2(u;{(\alpha /{c_p})^{2 - p}}) = \eta _1^2(u;(\alpha /{c_1})) \ \ \ \forall |u| <\alpha, 0 \leqslant p < 1.
\end{eqnarray*}
Hence,
\[
{\Gamma _{\alpha,1 }} < {\Gamma _{\alpha,p }}.
\]
$\hfill \Box$

\vspace{.2cm}
 We can employ this lemma to obtain the following result for the performance of the $\ell_p$-AMP with thresholding policy  $\bar{\lambda}_\alpha(\sigma) = (\alpha \sigma/ c_p)^{2-p}$. Although this lemma is not useful in our proof of Proposition \ref{prop:optimallassolargesigma}, since this is an interesting application of the above lemma, we include it here.

\vspace{.3cm}

\begin{corollary}\label{lem:largenoiselasso}
Suppose $\Gamma_{\alpha, p} < \delta$. Let $\sigma_{\ell}^2$ denote the lowest fixed point of $\ell_p$-AMP with thresholding policy  $\bar{\lambda}_\alpha(\sigma) = (\alpha \sigma/ c_p)^{2-p}$, where $c_p$ is the constant introduced in Lemma \ref{lem:threshform}, and $\alpha$ is a fixed number. For large values of $\sigma_w^2$ we have
\[
\frac{\sigma_\ell^2}{\sigma_w^2} = \frac{1}{1-\frac{1}{\delta} \Gamma_{\alpha,p}} + o(1),
\]
\end{corollary}
\vspace{.3cm}
\textit{Proof:}
First note that, according to the state evolution equation, $\sigma_{\ell}^2$ satisfies 
\begin{eqnarray}\label{eq:largesigmans1}
\sigma_{\ell}^2 = \sigma_w^2 + \frac{1}{\delta} r_{\alpha,p}(\sigma_\ell^2) \geq \sigma^2_{w}.
\end{eqnarray}
Thus $\sigma_{\ell}^2 \rightarrow \infty$, as $\sigma_w^2 \rightarrow \infty$. Dividing both sides of the equation in \eqref{eq:largesigmans1} by $\sigma_{\ell}^2$, combined with the result of Lemma \ref{lem:largenoiseell_1versusothers}, we have
\[
\lim_{\sigma_w^2 \rightarrow \infty} \frac{\sigma^2_w}{\sigma^2_{\ell}}=\lim_{\sigma_\ell^2 \rightarrow \infty} \frac{\sigma^2_w}{\sigma^2_{\ell}}=\lim_{\sigma_\ell^2 \rightarrow \infty}1-\frac{1}{\delta} \cdot \frac{ r_{\alpha,p}(\sigma_\ell^2)}{\sigma^2_{\ell}}=1- \frac{1}{\delta}\Gamma_{\alpha,p}.
\]
$\hfill \Box$

We state another corollary of Lemma \ref{lem:largenoiseell_1versusothers} that is important in our proof. Let $\alpha_{p}^{\rm opt}(\sigma^2)$ denote the value of $\alpha$ that minimizes $r_{\alpha, p}(\sigma^2)$. If as $\alpha \rightarrow \infty$, $r_{\alpha, p}(\sigma^2)$ reaches its infimum,  we set $\alpha_{p}^{\rm opt}(\sigma^2)$ to infinity. 

\vspace{.2cm}

\begin{corollary}\label{cor:optimalinf}
For every value of $p$, $\alpha^{\rm opt}_p(\sigma^2) \rightarrow \infty$ as $\sigma^2 \rightarrow \infty$.  
\end{corollary}
\vspace{.2cm}

\textit{Proof:}
Suppose this is not true. Then there exists a sequence $\sigma_n \rightarrow \infty$, as $n\rightarrow \infty$ such that $\alpha^{\rm opt}_p(\sigma_n^2) \rightarrow \alpha^* <\infty$. From the proof of Lemma \ref{lem:largenoiseell_1versusothers}, it is straightforward to confirm that,
\begin{equation}
\lim_{n\rightarrow \infty} \frac{r_{\alpha^{\rm opt}_p, p} (\sigma_n^2) }{\sigma_n^2} = \Gamma_{\alpha^*,p} > 0. \label{opt:one}
\end{equation}
However, since $\alpha^{\rm opt}_p(\sigma_n^2)$ is the optimal thresholding value, we know
\[
r_{\alpha^{\rm opt}_p, p}(\sigma_n^2) \leq \lim_{\alpha \rightarrow \infty} r_{\alpha, p} (\sigma_n^2) = \epsilon \mu^2, 
\]
which is in contradiction with \eqref{opt:one}.
$\hfill \Box$

\vspace{.2cm}

\begin{lemma}\label{lem:ell1isbetterlargesig}
Recall the notation $\eta_p^+(\cdot, \cdot)$ in \eqref{eq:defS_p}. If $\eta_p^+(\alpha; (\alpha/c_p)^{2-p}) > \frac{2 \mu}{\sigma}$ and $\alpha< \infty$, then we have
\[
r_{\alpha, 1}(\sigma^2) < r_{\alpha, p} (\sigma^2).
\]
\end{lemma}

\vspace{.2cm}

\textit{Proof:}
Define $\mu_\sigma \triangleq \mu/\sigma$, $\alpha_{c,p} \triangleq (\alpha /c_p)^{2-p}$,  $\xi_p(\mu_\sigma+ z; \alpha_{c,p})= \eta_p(\mu_\sigma+ z; \alpha_{c,p})- \eta_1(\mu_\sigma+ z; \alpha_{c,1})$. Note that $\xi_p(\mu_\sigma+ z; \alpha_{c,p})=0$ for $|\mu_\sigma+ z|<\alpha$, $\xi_p(\mu_\sigma+ z; \alpha_{c,p})>0$ for $\mu_\sigma+ z>\alpha$, and $\xi_p(\mu_\sigma+ z; \alpha_{c,p})<0$ for $\mu_\sigma+ z<-\alpha$. We have
\begin{eqnarray}\label{eq:riskell_pforproof}
\frac{r_{\alpha, p} (\sigma^2) }{\sigma^2}&=& (1-\epsilon) \mathbb{E} \eta^2_p (Z; \alpha_{c,p}) + \epsilon \mathbb{E} (\eta_p(\mu_\sigma+Z ; \alpha_{c,p}) - \mu_\sigma )^2 \nonumber \\
&=&  (1-\epsilon) \mathbb{E} \eta^2_p (Z; \alpha_{c,p}) + \epsilon \mu_\sigma^2+ \epsilon \mathbb{E} [(\eta_p(\mu_\sigma+Z; \alpha_{c,p}) -2 \mu_\sigma) \eta_p(\mu_\sigma+Z; \alpha_{c,p})] \nonumber \\
&=& (1-\epsilon) \mathbb{E} \eta^2_p (Z; \alpha_{c,p}) + \epsilon \mu_\sigma^2+\epsilon \mathbb{E} [(\eta_1(\mu_\sigma+Z; \alpha_{c,1}) -2 \mu_\sigma) \eta_1(\mu_\sigma+Z; \alpha_{c,1})] \nonumber \\
&&+\epsilon \mathbb{E} [(2 \eta_1(\mu_\sigma+Z; \alpha_{c,1}) - 2 \mu_\sigma) (\xi_p(\mu_\sigma+ Z; \alpha_{c,p}))] + \epsilon\mathbb{E} (\xi_p(\mu_\sigma+ Z; \alpha_{c,p}))^2,
\end{eqnarray}
where the first equality is due to Lemma \ref{lem:scaleinv}. Note that  
\begin{eqnarray}\label{eq:riskell_1forproof}
\frac{r_{\alpha, 1} (\sigma^2)}{\sigma^2}  = (1-\epsilon) \mathbb{E} \eta^2_1 (Z; \alpha_{c,1}) + \epsilon \mu_\sigma^2+\epsilon \mathbb{E} [(\eta_1(\mu_\sigma+Z; \alpha_{c,1}) -2 \mu_\sigma) \eta_1(\mu_\sigma+Z; \alpha_{c,1})].
\end{eqnarray}
According to Corollary \ref{cor:etap1eta0}, we have $\eta^2_1 (Z; \alpha_{c,1}) \leq  \eta^2_p (Z; \alpha_{c,p})$ for every $Z$. Hence
\begin{equation}\label{eq:riskcomparisonellpell1}
 (1-\epsilon) \mathbb{E} \eta^2_p (Z; \alpha_{c,p}) \geq (1-\epsilon) \mathbb{E} \eta^2_1 (Z; \alpha_{c,1}).
\end{equation}
Combining \eqref{eq:riskell_pforproof}, \eqref{eq:riskell_1forproof}, and \eqref{eq:riskcomparisonellpell1}, we conclude that if we prove 
\begin{equation}\label{eq:laststepcomparerisks}
\mathbb{E} [(2 \eta_1(\mu_\sigma+Z; \alpha_{c,1}) - 2 \mu_\sigma) (\xi_p(\mu_\sigma+ Z; \alpha_{c,p}))] + \mathbb{E} (\xi_p(\mu_\sigma+ Z; \alpha_{c,p}))^2 > 0,
\end{equation}
then 
\[
r_{\alpha, 1}(\sigma^2) < r_{\alpha, p} (\sigma^2).
\]
Hence in the rest of the proof we focus on showing \eqref{eq:laststepcomparerisks}. First note that by employing corollary \ref{cor:etap1eta0}, it is straightforward to conclude
\begin{equation}\label{eq:finalstepproofcompell1ellp1}
\eta_1(\mu_\sigma+Z; \alpha_{c,1})\cdot \xi_p(\mu_\sigma+ Z; \alpha_{c,p}) \geq 0. 
\end{equation}
Furthermore, according to Lemma \ref{lem:derivativeproperties}, if $ |\eta_p(u; \lambda)|>0$, then $\partial _p \eta_p(u; \lambda) \geq 1$. Since $\partial _1 \eta_1(u; \lambda) = 1$, we conclude that if $|\mu_\sigma+ Z| > \alpha$, then 
\[
|\xi_p(\mu_\sigma+ Z; \alpha_{c,p})|\geq \eta_p^+(\alpha; \alpha_{c,p}). 
\]
Hence, if $|\mu_\sigma+ Z| > \alpha$
\begin{equation}
 (\xi_p(\mu_\sigma+ Z; \alpha_{c,p}))^2-2 \mu_\sigma \xi_p(\mu_\sigma+ Z; \alpha_{c,p})\geq |\xi_p(\mu_\sigma+ Z; \alpha_{c,p})| \cdot (\eta_p^+(\alpha; \alpha_{c,p})- 2\mu_\sigma)> 0.
\end{equation}
This completes the proof of our lemma. 
$\hfill \Box$

\vspace{.2cm}

The following lemma enables us to complete the proof. 

\vspace{.2cm}

\begin{lemma}[\cite{mousavi2013asymptotic}, Proposition 3.8]\label{lem:alispaper}
If $r_{\alpha, 1}(\sigma^2)$ denotes the risk of the soft thresholding function, then for large values of $\alpha$, $\frac{\partial r_{\alpha, 1}(\sigma^2)}{\partial \alpha}>0$.
\end{lemma}

\vspace{.2cm}

See \cite{mousavi2013asymptotic} for the exact value of $\alpha$ above which the derivative is positive.

\vspace{.2cm}

\begin{corollary}
For every value of $\sigma$, we have
\[
\inf_{\alpha\geq 0} r_{\alpha, 1}(\sigma^2) < \epsilon \mu^2. 
\]
\end{corollary}
The proof is a straightforward combination of Lemma \ref{lem:alispaper} and the fact that $\lim_{\alpha \rightarrow \infty} r_{\alpha, p} (\sigma^2) = \epsilon \mu^2$. 

\vspace{.2cm}

Now we return to the proof of Proposition \ref{prop:optimallassolargesigma}. We only mention the sketch of the proof since the details are straightforward. According to Corollary \ref{cor:optimalinf}, for large values of $\sigma$, $\alpha_{p}^{\rm opt}$ is large, hence $\alpha_{c,p}^{\rm opt}=(\alpha_{p}^{\rm opt}/c_p)^{2-p}$ is large as well. Hence by Lemma \ref{lem:lowerboundonx}, we can assume that $\eta_p^+({\alpha}_{p}^{\rm opt}; \alpha_{c,p}^{\rm opt}) > \frac{2 \mu}{\sigma}$ for large $\sigma$. Suppose that $\alpha_p^{\rm opt}< \infty$. Then according to Lemma \ref{lem:ell1isbetterlargesig}, we have
\[
\inf_\alpha r_{\alpha, 1}(\sigma^2) \leq  r_{\alpha_p^{\rm opt}, 1}(\sigma^2) < r_{\alpha_p^{\rm opt}, p} (\sigma^2).
\]
If $\alpha_p^{\rm opt}= \infty$, then 
\begin{equation*}
\inf_\alpha r_{\alpha, 1}(\sigma^2) <  r_{\alpha_p^{\rm opt}, 1}(\sigma^2) = r_{\alpha_p^{\rm opt}, p} (\sigma^2),
\end{equation*}
where the first inequality is due to Lemma \ref{lem:alispaper}. In any case, we have showed that,
\begin{equation}\label{eq:finalresullargenoise1}
r_{\alpha_1^{\rm opt}, 1} (\sigma^2) < r_{\alpha_p^{\rm opt}, p} (\sigma^2),
\end{equation}
for large values of $\sigma$. The last step of the proof is to connect this result with the fixed points of the state evolution equation.

Note that the fixed points of the state evolution must satisfy
\begin{equation*}
\sigma^2 = \sigma_w^2 + \frac{1}{\delta} r_{\alpha_p^{\rm opt}, p}(\sigma^2). 
\end{equation*}
Therefore, as $\sigma_w \rightarrow \infty$, the lowest fixed point $\sigma_{\ell}^2$ goes off to $\infty$. On the other hand, Inequality \eqref{eq:finalresullargenoise1} implies that the function $\sigma_w^2 + \frac{1}{\delta} r_{\alpha_p^{\rm opt}, p}(\sigma^2)$ is above $\sigma_w^2 + \frac{1}{\delta} r_{\alpha_1^{\rm opt}, 1}(\sigma^2)$ over a range $(\bar{\sigma}^2, \infty)$. Hence, we can increase $\sigma_w$ to make sure $\sigma^2_{\ell}$ of both $\ell_1$-AMP and $\ell_p$-AMP fall into that range. Then clearly the lowest fixed point of $\ell_1$-AMP is smaller than that of $\ell_p$-AMP.

\subsection{Proof of Theorem \ref{thm:riskbehsmallsigma}} \label{sec:prooflemmasmallnoise}

We first remind the reader the definition
 \begin{equation*}
 R_p(\tau, \sigma)\triangleq(1-\epsilon)\mathbb{E}\eta_p^2(Z;\tau)+\epsilon \mathbb{E}(\eta_p(U/\sigma+Z;\tau)-U/\sigma)^2,
 \end{equation*}
where $Z \sim N(0,1)$ and $U \sim G$ are independent. Also let $\tau_*(\sigma)$ denote the optimal $\tau$ that minimizes $R_p(\tau,\sigma)$. In Proposition \ref{lem:riskbehsmallsigma} we proved that as $\sigma \rightarrow 0$

\begin{eqnarray*}
 R_p(\tau_*(\sigma),\sigma)=\epsilon+\epsilon p^2 \mathbb{E}|U|^{2p-2} (\tau_*(\sigma))^2\sigma^{2-2p}+o((\tau_*(\sigma))^2\sigma^{2-2p}),
\end{eqnarray*}
where the convergence rate of $\tau_*(\sigma)$ can be characterized by 
\begin{eqnarray}\label{eq:tau*calc1}
 \lim_{\sigma \rightarrow 0} \frac{\sigma^{2-2p}}{(\tau_*(\sigma))^{\frac{2p-1}{2-p}}\phi(c_p(\tau_*(\sigma))^{\frac{1}{2-p}})} = \frac{(1-\epsilon)c_p\eta^2_p(c_p;1)}{\epsilon p^2(2-p)\mathbb{E}|U|^{2p-2}}.
\end{eqnarray}

Here we aim to analyze the lowest fixed point of the optimally tuned $\ell_p$-AMP, $\sigma_{\ell}$, that satisfies
\begin{eqnarray}\label{lownoise:ana}
\sigma_{\ell}^2=\sigma_w^2+ \frac{\sigma^2_{\ell}}{\delta}R_p(\tau_*(\sigma_{\ell}),\sigma_{\ell}).
\end{eqnarray}

We focus on the regime where the sparsity level $\epsilon$ is below the phase transition of the lowest fixed point, i.e. $\delta > \epsilon$. In Theorem \ref{thm:noisyoptimallambda} we have already proved that 
\begin{eqnarray}
\lim_{\sigma_w \rightarrow 0} \frac{\sigma_{\ell}^2}{\sigma_w^2}=\frac{\delta}{\delta-\epsilon}. \label{lownoise:prop}
\end{eqnarray}
Now we first characterize the following limit:
\begin{eqnarray}\label{eq:prooffinln1}
&& \lim_{\sigma_w \rightarrow 0}\frac{\sigma_\ell^2-\frac{\delta }{\delta - \epsilon}\sigma_w^2}{\sigma^{4-2p}_w (\tau_*(\sigma_{\ell}))^2}\overset{(a)}{=}\lim_{\sigma_w \rightarrow 0} \frac{\sigma_{\ell}^2-\frac{\delta }{\delta - \epsilon}(\sigma_{\ell}^2-\frac{\sigma^2_{\ell}}{\delta}R_p(\tau_*(\sigma_{\ell}),\sigma_{\ell}))}{\sigma^{4-2p}_w (\tau_*(\sigma_{\ell}))^2} \nonumber \\
&=&\frac{1}{\delta-\epsilon}\cdot  \lim_{\sigma_{\ell}\rightarrow 0} \frac{R_p(\tau_*(\sigma_{\ell}),\sigma_{\ell})-\epsilon}{(\tau_*(\sigma_{\ell}))^2\sigma^{2-2p}_{\ell}}\cdot \lim_{\sigma_{\ell}\rightarrow 0} \frac{\sigma^{4-2p}_{\ell}}{\sigma^{4-2p}_w} \overset{(b)}{=}\frac{\epsilon p^2\mathbb{E}|B|^{2p-2}\delta^{2-p}}{(\delta-\epsilon)^{3-p}}.
\end{eqnarray}
We have used  \eqref{lownoise:ana} to obtain (a); the derivation of (b) is due to \eqref{lownoise:prop} and Proposition \ref{lem:riskbehsmallsigma}. Our next step is to show that 
\[
 \lim_{\sigma_w \rightarrow 0}\frac{\sigma_\ell^2-\frac{\delta }{\delta - \epsilon}\sigma_w^2}{\sigma^{4-2p}_w (\tau_*(\sigma_{w}))^2}= \frac{\epsilon p^2\mathbb{E}|B|^{2p-2}\delta^{2-p}}{(\delta-\epsilon)^{3-p}}.
\]

First note that we have proved in Lemma \ref{rough:rate}, $\tau_*(\sigma_w) \rightarrow \infty$ as $\sigma_w \rightarrow 0$. We can use \eqref{eq:tau*calc1} to prove that $\tau_*(\sigma_\ell) / \tau_*(\sigma_w) \rightarrow 1$ as $\sigma_w \rightarrow 0$ in the following way. According to Lemma \ref{rough:rate}, $\tau_*(\sigma_w), \tau_*(\sigma_\ell) \rightarrow \infty$. Furthermore, by employing \eqref{eq:tau*calc1} we obtain

\begin{eqnarray*}
 \lim_{\sigma_w \rightarrow 0} \frac{\sigma_\ell^{2-2p}}{(\tau_*(\sigma_\ell))^{\frac{2p-1}{2-p}}\phi(c_p(\tau_*(\sigma_\ell))^{\frac{1}{2-p}})} \frac{(\tau_*(\sigma_w))^{\frac{2p-1}{2-p}}\phi(c_p(\tau_*(\sigma_w))^{\frac{1}{2-p}})}{\sigma_w^{2-2p}}=1. 
\end{eqnarray*}
By applying \eqref{lownoise:prop} we reach   
\begin{eqnarray*}
 \lim_{\sigma_w \rightarrow 0} \frac{(\tau_*(\sigma_w))^{\frac{2p-1}{2-p}}\phi(c_p(\tau_*(\sigma_w))^{\frac{1}{2-p}})}{(\tau_*(\sigma_\ell))^{\frac{2p-1}{2-p}}\phi(c_p(\tau_*(\sigma_\ell))^{\frac{1}{2-p}})} =\Big (1- \frac{\epsilon}{\delta} \Big)^{1-p},
\end{eqnarray*}
which implies
\begin{eqnarray*}
\lim_{\sigma_w \rightarrow 0} \frac{2p-1}{2-p} \log \frac{\tau_*(\sigma_w)}{\tau_*(\sigma_\ell)}-\frac{c_p^2}{2} (\tau_*(\sigma_w)^{2/(2-p)}- \tau_*(\sigma_\ell)^{2/(2-p)} ) = (1-p) \log(1- \epsilon/\delta). 
\end{eqnarray*}
If we combine this with the fact that $\tau_*(\sigma_\ell) \rightarrow \infty$ and $\tau_*(\sigma_w) \rightarrow \infty$, then we conclude that 
\begin{eqnarray*}
\lim_{\sigma_w \rightarrow 0} \frac{ \frac{2p-1}{2-p} \log \frac{\tau_*(\sigma_w)}{\tau_*(\sigma_\ell)}-\frac{c_p^2}{2} (\tau_*(\sigma_w)^{2/(2-p)}- \tau_*(\sigma_\ell)^{2/(2-p)} )}{\tau_*(\sigma_w)^{2/(2-p)}} = 0,
\end{eqnarray*}
that in turn implies that 
\begin{equation}\label{eq:prooffinln2}
\lim_{\sigma_w \rightarrow 0} \frac{\tau_*(\sigma_w)}{\tau_*(\sigma_\ell)} =1. 
\end{equation}
Combining \eqref{eq:prooffinln1} and \eqref{eq:prooffinln2} proves that
\[
 \lim_{\sigma_w \rightarrow 0}\frac{\sigma_\ell^2-\frac{\delta }{\delta - \epsilon}\sigma_w^2}{\sigma^{4-2p}_w (\tau_*(\sigma_{w}))^2}= \frac{\epsilon p^2\mathbb{E}|B|^{2p-2}\delta^{2-p}}{(\delta-\epsilon)^{3-p}},  
\]
where $\tau_*(\sigma_w)$ satisfies \eqref{eq:tau*calc1}. It is then straightforward to use \eqref{eq:tau*calc1} and the fact that $\tau_*(\sigma_w) \rightarrow \infty$ to show that
\[
 \lim_{\sigma_w \rightarrow 0} \frac{(\tau_*(\sigma_w))^{\frac{2}{2-p}}}{\log \frac{1}{\sigma_w}} = \frac{4 (1-p)}{c_p^2}. 
 \]

\subsection{Proof of Theorem \ref{thm:lownoisehardthresh1}} \label{ssec:prooflownosiehardthresh}

We first remind the reader the definition 
 \begin{equation*}
 R_0(\tau, \sigma)\triangleq(1-\epsilon)\mathbb{E}\eta_0^2(Z;\tau)+\epsilon \mathbb{E}(\eta_0(U/\sigma+Z;\tau)-U/\sigma)^2,
 \end{equation*}
where $Z \sim N(0,1)$ and $U \sim G$ are independent. Also let $\tau_*(\sigma)$ denote the optimal $\tau$ that minimizes $R_0(\tau,\sigma)$. In Proposition \ref{lem:riskbehsmallsigmazero} we proved that as $\sigma \rightarrow 0$ if $\mu=\sup_v \{v : P(|B|>v)=1 \} >0$, then for $p=0$,
\begin{eqnarray*}
 R_0(\tau_*(\sigma),\sigma)=\epsilon+o(\phi(\tilde{\mu}\sigma^{-1})),
\end{eqnarray*}
where $\tilde{\mu}$ is any constant that smaller than $\frac{\mu}{2}$.

The rest of the proof is similar to the proof of Theorem \ref{sec:prooflemmasmallnoise} and is hence skipped.

\subsection{Proof of Theorem \ref{thm:hfpnoisy}}\label{sec:proofnoisyhigh}
Let $X$ and $U$ denote two independent random variables with distributions $(1- \epsilon)\Delta_0+ \epsilon G$ and $G$, respectively. As before, our only assumption on $G$ is that it does not have any point mass at zero. Note that $\sigma_h^2$ satisfies the following fixed point equation:
\begin{eqnarray*}
\sigma_h^2 &=& \sigma_w^2 + \frac{1}{\delta}\inf_{\tau \geq 0}  \Big [ (1-\epsilon) \mathbb{E}  \eta^2_p(\sigma_h Z ; \tau)+ \epsilon \mathbb{E} ( \eta_p(U+ \sigma_h Z ; \tau) - U)^2  \Big ]\nonumber \\
&=& \sigma_w^2 +  \frac{\sigma_h^2}{\delta} \inf_{\tau \geq 0} \Big [(1-\epsilon) \mathbb{E}  \eta_p^2( Z ; \tau \sigma_h ^{p-2})+ \epsilon \mathbb{E}_U (\mathbb{E}_Z ( \eta_p(U/\sigma_h+ Z ; \tau \sigma_h^{p-2}) - U/\sigma_h)^2) \Big ] \nonumber \\
&\leq & \sigma_w^2 + \frac{\sigma_h^2}{\delta} \inf_{\tau \geq 0} \Big [(1-\epsilon) \mathbb{E}  \eta_p^2( Z ; \tau \sigma_h ^{p-2})+ \epsilon \sup_{\mu \geq 0} \mathbb{E}( \eta_p(\mu/\sigma_h+ Z ; \tau \sigma_h^{p-2}) - \mu /\sigma_h)^2 \Big ] \nonumber \\
&\leq& \sigma_w^2  + \frac{\sigma_h^2}{\delta}\inf_{\tau \geq 0}  \Big [(1-\epsilon) \mathbb{E}  \eta_p^2( Z ; \tau \sigma_h ^{p-2})+ \epsilon \sup_{\mu \geq 0} \mathbb{E}( \eta_p(\mu+ Z ; \tau \sigma_h^{p-2}) - \mu)^2 \Big ] \nonumber \\
&=& \sigma_w^2  + \frac{\sigma_h^2}{\delta}\inf_{\tau \geq 0}  \Big [(1-\epsilon) \mathbb{E}  \eta_p^2( Z ; \tau)+ \epsilon \sup_{\mu \geq 0} \mathbb{E}( \eta_p(\mu+ Z ; \tau ) - \mu)^2 \Big ] \nonumber \\
&= & \sigma_w^2 +  \frac{\sigma_h^2}{\delta}\overline{M}_p(\epsilon).
\end{eqnarray*}

The proof of the second part of the theorem consider the following definitions:
\begin{eqnarray}
\tau_{*,\mu} &\triangleq& \arg\min_\tau \mathbb{E}\left( \epsilon (\eta_p (\mu+ z;\tau)- \mu)^2 + (1-\epsilon) \mathbb{E} (\eta_p(z; \tau))^2 \right), \nonumber \\
\mu_* &\triangleq& \arg\max_\mu \mathbb{E}\left( \epsilon (\eta_p (\mu+ z;\tau_{*,\mu})- \mu)^2 + (1-\epsilon) \mathbb{E} (\eta_p(z; \tau_{*,\mu}))^2 \right).
\end{eqnarray}
Note for notational simplicity we have assumed the maximas and minimas are achieved. Also define
\[
(\sigma_h^*)^2 \triangleq \frac{\sigma_w^2}{1- \frac{\underline{M}_p(\epsilon)}{\delta}}. 
\]
 We consider a distribution $G$ that has a  point mass at $\mu^* \sigma_h^*$. For this distribution we have
\begin{eqnarray}
 \Psi_{\tau_*, p}((\sigma_h^*)^2) &=& \sigma_w^2  + \frac{(\sigma^*_h)^2}{\delta}\inf_{\tau \geq 0}  \Big [(1-\epsilon) \mathbb{E}  \eta_p^2( Z ; \tau (\sigma^*_h)^{p-2})+ \epsilon  \mathbb{E}( \eta_p(\mu_*+ Z ; \tau (\sigma^*_h)^{p-2}) - \mu_*)^2 \Big ] \nonumber \\
&=& \sigma_w^2  + \frac{(\sigma_h^*)^2}{\delta}  \Big [(1-\epsilon) \mathbb{E}  \eta_p^2( Z ; \tau_{*,{\mu_*}})+ \epsilon \mathbb{E}( \eta_p(\mu_*+ Z ; \tau_{*,{\mu_*}} ) - \mu_*)^2 \Big ] \nonumber\\
&=& \sigma_w^2 + \frac{\underline{M}_p (\epsilon)}{\delta}(\sigma_h^*)^2 = (\sigma_h^*)^2. 
\end{eqnarray}
Hence, $\sigma_h^*$ is a fixed point of the function. If it is an unstable fixed point we can use the argument presented for Proposition \ref{proof:existsncestable} to show that there is another stable fixed point above $(\sigma_h^*)^2$.

\subsection{Proof of Theorem \ref{thm:improvecontinuation}}\label{sec:proofcontlownoise}

Let $X \sim (1-\epsilon) \Delta_0+ \epsilon G$, where $G$ is an arbitrary distribution that does not have any mass at zero. Also let $U \sim G$ denote a random variable. Then we have
\begin{eqnarray*}
\sigma_h^2 = \sigma_w^2 +\frac{1}{\delta} \inf_{0 \leq p \leq 1, \lambda \geq 0} \mathbb{E} (\eta_p(X+ \sigma_h Z; \lambda) - X)^2  \overset{(a)}{\leq} \sigma_w^2 + \frac{M_p(\epsilon)\sigma_h^2}{\delta},
 \end{eqnarray*}
 where (a) follows similar arguments as in Section \ref{sec:proofnoisyhigh}. Hence, we conclude that
\[
\sigma_h^2 \leq \frac{\sigma_w^2}{1-M_p(\epsilon)/\delta}. 
\]
This implies that if $\sigma_w^2 \rightarrow  0$, then $\sigma_h^2 \rightarrow 0$. Moreover, it is straightforward to see that 
\begin{eqnarray}\label{need:one}
\lim_{\sigma_w^2 \rightarrow 0} \frac{\sigma_h^2}{\sigma_w^2} = \frac{1}{1-  \lim_{\sigma_h \rightarrow 0}\frac{ \Psi_{\lambda_*, p_*} (\sigma_h^2)}{ \sigma_h^2}}. 
\end{eqnarray}

From the proof of Theorem \ref{lem:noiselesslowfpless1}, we know that for every fixed $p<1$,
\[
\lim_{\sigma^2 \rightarrow 0} \frac{\Psi_{\lambda_*, p} (\sigma^2)}{\sigma^2} = \frac{\epsilon}{\delta}.  
\]
Hence, it is straightforward to show that 
\begin{eqnarray}\label{need:two}
\limsup_{\sigma^2 \rightarrow 0} \frac{ \Psi_{\lambda_*, p_*} (\sigma^2)}{ \sigma^2}  \leq \limsup_{\sigma^2 \rightarrow 0} \frac{ \Psi_{\lambda_*, 0} (\sigma^2)}{ \sigma^2} = \frac{\epsilon}{\delta}. 
\end{eqnarray}

Define $\Gamma(\sigma^2) \triangleq \frac{1}{\delta}\mathbb{E} (\mathbb{E}(X | X+\sigma Z) -X)^2$. Since $\mathbb{E} (X | X+ \sigma Z)$ is the minimum mean square error estimator we have
\[
\Psi_{\lambda_*, p_*}(\sigma) \geq \Gamma(\sigma^2).
\]
Hence,
\begin{eqnarray}\label{need:three}
\liminf_{\sigma \rightarrow 0} \frac{\Psi_{\lambda_*, p_*}(\sigma^2)}{\sigma^2} \geq \lim_{\sigma \rightarrow 0} \frac{\Gamma(\sigma^2)}{\sigma^2}= \frac{\epsilon}{\delta},
\end{eqnarray}
where the last equality is a combination of Theorems 5 and 8 in \cite{wu2011mmse}. Combing \eqref{need:one}, \eqref{need:two} and \eqref{need:three} together finishes the proof. Note that \eqref{need:two} and \eqref{need:three} together shows $\left. \frac{d \Psi_{\lambda_*,p_*}(\sigma^2)}{d\sigma^2} \right|_{\sigma^2=0}=\frac{\epsilon}{\delta}$. Then by using similar arguments as in Theorem \ref{thm:noisyoptimallambda}, we can show $\sigma_h$ is the unique stable fixed point when $\sigma_w$ is small enough. Hence implicit function theorem can be applied to claim the continuity of $\sigma_h$, as a function of $\sigma_w$ in a neighborhood of 0.

\section{Optimal $\ell_p$-AMP in practice}\label{sec:simulation}
\subsection{Stein Unbiased Risk Estimate}\label{ssec:simSURE}
The optimal $\ell_p$-AMP algorithm introduced in Section \ref{sec:optlambda} employs the following thresholding policy:
\[
\lambda_*(\sigma) \in \arg \min_{\lambda \geq 0} \mathbb{E} \left( {{{\left| {\eta_{p} (X + {\sigma }Z;{\lambda}) - X} \right|}^2}} \right),
\]
where the expected value is with respect to both $Z\sim N(0,1)$ and $X \sim p_X$. Note that $\lambda_*$ is a function of both $\sigma$ and $p_X$. While it is possible to provide a good estimate of $\sigma$, coming up with a good estimate of $p_X$ is very challenging, if not impossible, in many applications. The question we would like to answer in this section is whether we can provide an accurate estimate of $\lambda_*(\sigma)$ without any knowledge of $p_X$ in practice. Similar questions can be asked regarding the optimal choice of $p$, introduced in Section \ref{sec:optimalpandlambda}, or even the optimal choice of $h$ introduced in \eqref{eq:ellpamp}. We answer these questions in this section. Our approach is motivated by Stein Unbiased Risk Estimate (SURE), that we briefly summarize here. Let $x_o \in \mathbb{R}^N$  and suppose that we observe $\tilde{x} = x_o +\rho$ with $\rho \sim N(0, \sigma^2 I)$. To estimate $x_o$ we employ a denoiser $\mathcal{D}: \mathbb{R}^N \rightarrow \mathbb{R}^N$. Can we estimate the risk of this denoiser, i.e., $$r_{\mathcal{D}}  \triangleq \mathbb{E} \|\mathcal{D}(\tilde{x}) - x_o\|_2^2?$$ 
If the answer is affirmative, then the risk estimate can be employed for tuning the free parameters of the denoiser or in comparing different denoisers. Note that the main challenge for estimating the risk, $r_{\mathcal{D}}$, is that $x_o$ is not known. The following theorem due to Stein provides a simple way to find an unbiased estimate of $r_{\mathcal{D}}$.

\vspace{.2cm}

\begin{lemma} {\rm \cite{stein1981estimation}} \label{lem:stein}
Let $\mathcal{D}(\tilde{x})$ denote the denoiser. If $\mathcal{D}$ is weakly differentiable, then
\begin{equation}\label{eq:stein}
\mathbb{E} \| \mathcal{D}(\tilde{x}) - x_o\|^2/N = \mathbb{E} \|\mathcal{D}(\tilde{x}) - \tilde{x}\|_2^2/N+  \sigma^2 + 2 \sigma^2  \mathbb{E} (\mathbf{1}^T (\nabla \mathcal{D}(\tilde{x})-\mathbf{1}))/N,
\end{equation}
where $\nabla \mathcal{D}(\tilde{x})=(\frac{\partial \mathcal{D}_1(\tilde{x})}{\partial \tilde{x}_1},\ldots, \frac{\partial \mathcal{D}_N(\tilde{x})}{\partial \tilde{x}_N})^T$ and $\mathbf{1}$ is an all one vector.
\end{lemma}

\vspace{.2cm}

Note that the terms inside the expectation on the right hand side of \eqref{eq:stein} do not depend on $x_o$. This enables us to provide the following estimate of the risk function:
\[
\hat{r}_{\mathcal{D}} =  \| \mathcal{D}(\tilde{x}) - \tilde{x}\|_2^2/N+  \sigma^2 + 2 \sigma^2 (\mathbf{1}^T (\nabla \mathcal{D}(\tilde{x})-\mathbf{1}))/N.
\]  
According to Lemma \ref{lem:stein}, we have $r_{\mathcal{D}} = \mathbb{E} (\hat{r}_{\mathcal{D}})$. Hence, $\hat{r}_{\mathcal{D}}$ provides an unbiased estimate of $r_{\mathcal{D}}$. SURE has been used elsewhere for model selection \cite{hastie2009elements}.  Our next goal is to employ the idea of SURE for the $\ell_p$-AMP algorithm.

\subsection{$\ell_p$-AMP and SURE}\label{sec:ampsure}

Can SURE be employed to estimate $\lambda_*(\sigma)$ or $p_*(\sigma)$ for the optimal $\ell_p$-AMP? As we discussed in Section \ref{ssec:simSURE}, SURE can be used for denoising problems in which the noise is Gaussian. Also as we discussed in Section \ref{sec:firstgaussianity}, if we define $v^t \triangleq A^T z^t+x^t-x_o$, then we can write $x^t+ A^Tz^t = x_o+ v^t$, where $v^t$ resembles iid Gaussian random vector in the asymptotic settings. Hence, at the intuitive level, we should be able to use SURE to estimate the optimal parameters of $\ell_p$-AMP.  This intuition is in fact valid and we formalize it below. We only consider the estimation of $\lambda_*(\sigma)$. But, the approach can be extended to the tuning of the other parameters as well. \\

Consider the iterations of $\ell_p$-AMP with the optimal thresholding policy, $\lambda_*(\sigma)$. If the algorithm starts at $\sigma_t = \sigma_0$, then the first threshold is $\lambda_*(\sigma_0)$. Note that $\lambda_*(\sigma_0)$ is the value of $\lambda$ that minimizes $ \lim_{N \rightarrow \infty} \frac{1}{N}\| \tilde \eta_{p,h}(x_o+v^0; \lambda) -x_o \|_2^2$. 
Once we run $\ell_p$-AMP with this threshold, the standard deviation of the next iteration will be $\sigma_1$ and hence the next threshold will be $\lambda(\sigma_1)$, and again this is the value of $\lambda$ that minimizes $ \lim_{N \rightarrow \infty} \frac{1}{N}\| \tilde \eta_{p,h}(x_o+v^1; \lambda) -x_o \|_2^2$. This discussion reveals two main properties of the optimal thresholding policy:

\begin{enumerate}
\item[(i)] We do not have to estimate the entire function $\lambda_*(\sigma)$. We only need to estimate it at the values of $\sigma_0, \sigma_1, \sigma_2, \ldots$ that are actually observed in the $\ell_p$-AMP algorithm. 

\item[(ii)] At iteration $t$, $\lambda_*(\sigma_t)$ is the value of $\lambda$ that minimizes the risk $\lim_{N \rightarrow \infty} \frac{1}{N}\| \tilde \eta_{p,h}(x_o+v^t; \lambda) -x_o \|_2^2$. 
\end{enumerate}

These two conclusions imply that if at iteration $t$ we find the value of $\lambda$ that minimizes $\lim_{N \rightarrow \infty} \frac{1}{N}\| \tilde \eta_{p,h}(x_o+v^t; \lambda) -x_o \|_2^2$, then the resulting $\ell_p$-AMP algorithm will perform the same as optimal-$\lambda$ $\ell_p$-AMP. Hence, the problem of finding the optimal thresholding policy for $\ell_p$-AMP is simplified to the problem of tuning the parameters of $\ell_p$-AMP at a single iteration (without taking the other iterations into account).  If the iterations of $\ell_p$-AMP are given by:
\begin{eqnarray}
\label{eq:estimation}
{ x^t} &=& \tilde{\eta}_{p,h} ({A^T}{z^{t - 1}} + {x^{t - 1}};{\lambda _t}), \nonumber \\
{z^t} &=& y - A{x^t} + {z^{t - 1}}\frac{1}{\delta }\left\langle {\tilde{\eta}_{p,h}' }({A^T}{z^{t - 1}} + {{x}^{t - 1}};{\lambda _t}) \right\rangle.
\end{eqnarray}
The optimal value of $\lambda_t$ at iteration $t$ is the value of $\lambda$ that minimizes
\[
r_{p,h}^t (\lambda)\triangleq \lim_{N \rightarrow \infty} \frac{1}{N} \| \tilde \eta_{p,h}(x_o+v^t; \lambda)- x_o \|_2^2.
\]
Since we know that $v^t$ is almost Gaussian, inspired by Stein Unbiased Risk Estimate (SURE), we consider the following empirical estimate of the risk at iteration $t$:
\begin{eqnarray}
\hat{r}_{p,h, \lambda}^t = \frac{1}{N}  \| \tilde \eta_{p,h}(x_o+ v^t; \lambda)-x_o-v^t\|^2- \sigma_t^2 + \frac{2\sigma_t^2}{N} {\rm div} (\tilde \eta_{p,h}({x_o+v^t} ; \lambda)),
\end{eqnarray}
where ${\rm div}$ denotes the divergence of $\tilde \eta_{p,h}$ and is defined as ${\rm div} (\tilde \eta_{p,h}({x_o+v^t} ; \lambda))\triangleq \sum_{i=1}^N \frac{\partial (\tilde \eta_{p,h}({x_{o,i}+v_i^t} ; \lambda))}{\partial x_{o,i}}$. The following theorem confirms  that in the asymptotic setting $\hat{r}_{p,h, \lambda}^t $ provides an accurate estimate of the risk function,i.e.,
\[
\lim_{N \rightarrow \infty} \hat{r}_{p,h, \lambda}^t  \overset{a.s.}{=} r_{p,h}^t (\lambda).
\]

\vspace{.2cm}

\begin{theorem}\label{thm:empiricalriskconvergence}
Let $\{x_o(N), A(N), w(N)\}$ denote a converging sequence. Let the parameters $\lambda_1, \lambda_2, \ldots, \lambda_{t-1}$ denote the threshold parameters of $\ell_p$-AMP for the first $t-1$ iterations and $x^t(N)$ and $z^t(N)$ denote the estimates of $\ell_p$-AMP according to \eqref{eq:estimation}. Then,
\begin{eqnarray}
 \lim_{N \rightarrow \infty} \hat{r}_{p,h, \lambda}^t  \overset{\rm a.s.}{=} \mathbb{E} (\tilde \eta_{p,h}(X+ \sigma_t Z; \lambda) - X)^2.
\end{eqnarray}
where $\sigma_t$ satisfies the following iteration:
\begin{equation}
\sigma _{t }^2 = {\sigma_w ^2} + \frac{1}{\delta }\mathbb{E}\left( {{{\left| {\tilde{\eta}_{p,h} (X + {\sigma _{t-1}}Z;{\lambda _{t-1}}) - X} \right|}^2}} \right).
\end{equation}
Here the expected value is with respect to two independent random variables $Z \sim N(0,1)$ and $X\sim p_X$. $\sigma_0^2$ depends on the initialization of the algorithm. \end{theorem}
\vspace{.2cm}

The proof of this result can be found in the Appendix. According to the above theorem the empirical risk provides an accurate estimate of $\mathbb{E} (\tilde \eta_{p,h}(X+ \sigma_t Z; \lambda) - X)^2$ for large values of $N$. Hence one can estimate the optimal value of $\lambda_t$ and $p_t$ in the following way:
\[
(\hat{\lambda}_t, \hat{p}_t) \in \arg\min_{\lambda, p} \hat{r}_{p,h, \lambda}^t. 
\]
Note that the empirical risk can be even employed for finding the optimal value of the parameter $h$. However, to reduce the computational complexity, we set $h$ automatically. This approach will be explained in the next section.

\subsection{Simulation Result}

In this section, we would like to compare our asymptotic results with the simulations that are performed at finite values of $N$. As we will present later, it turns out that our asymptotic results provide accurate predictions of the performance of the algorithm even for not too large sample sizes, such as $N=5000$. Furthermore, we present the result of the tuning approach we  proposed for the $\ell_p$-AMP algorithm and we show that the tuning approach we proposed based on SURE is in fact accurate even in medium problem sizes. 

\subsubsection{State Evolution versus $\ell_p$-AMP}

\begin{figure}
  \centering
  \subfloat[][]{\includegraphics[width=2.5in]{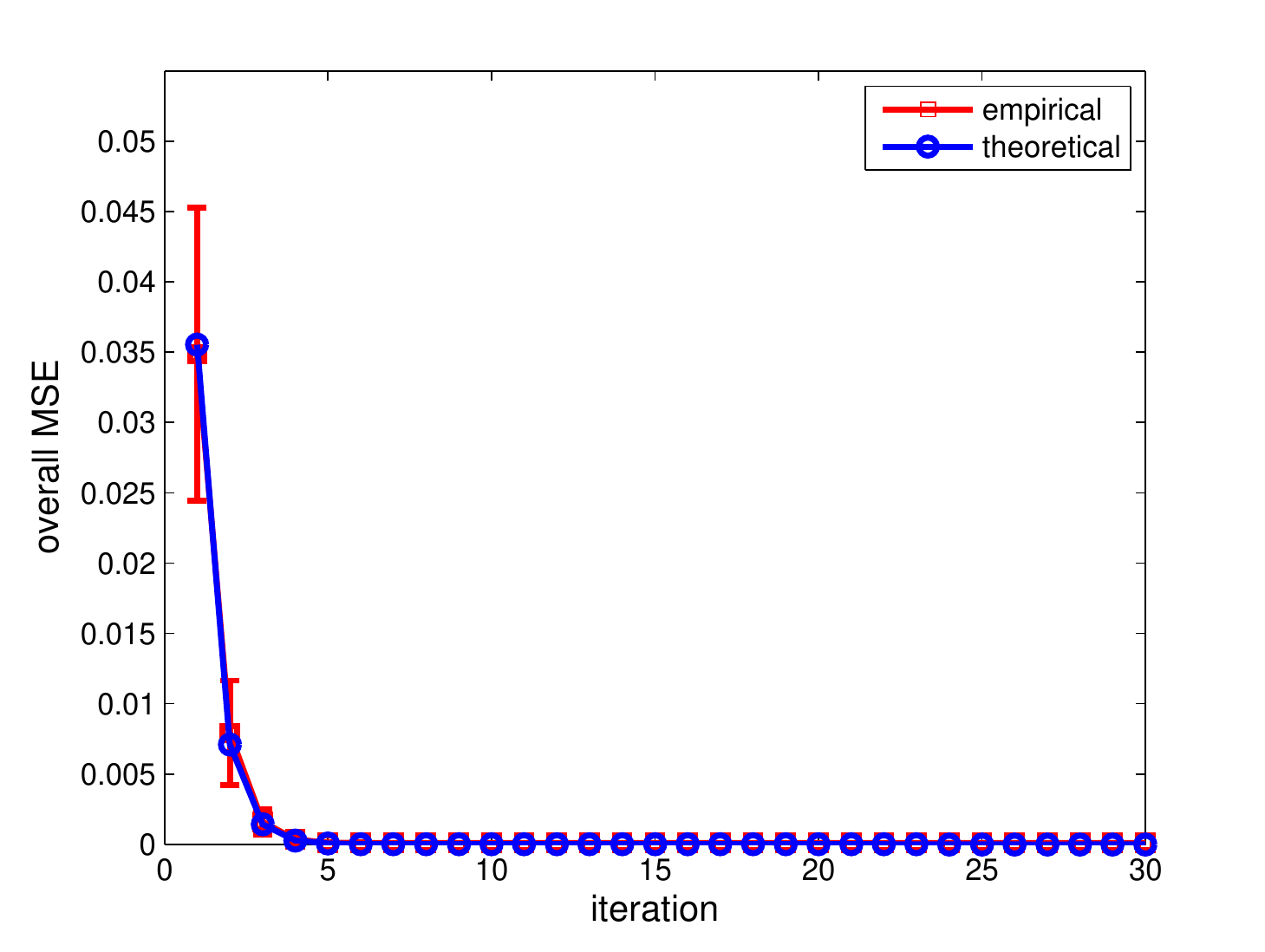}\label{fig:sim1_1a}}
  \subfloat[][]{\includegraphics[width=2.5in]{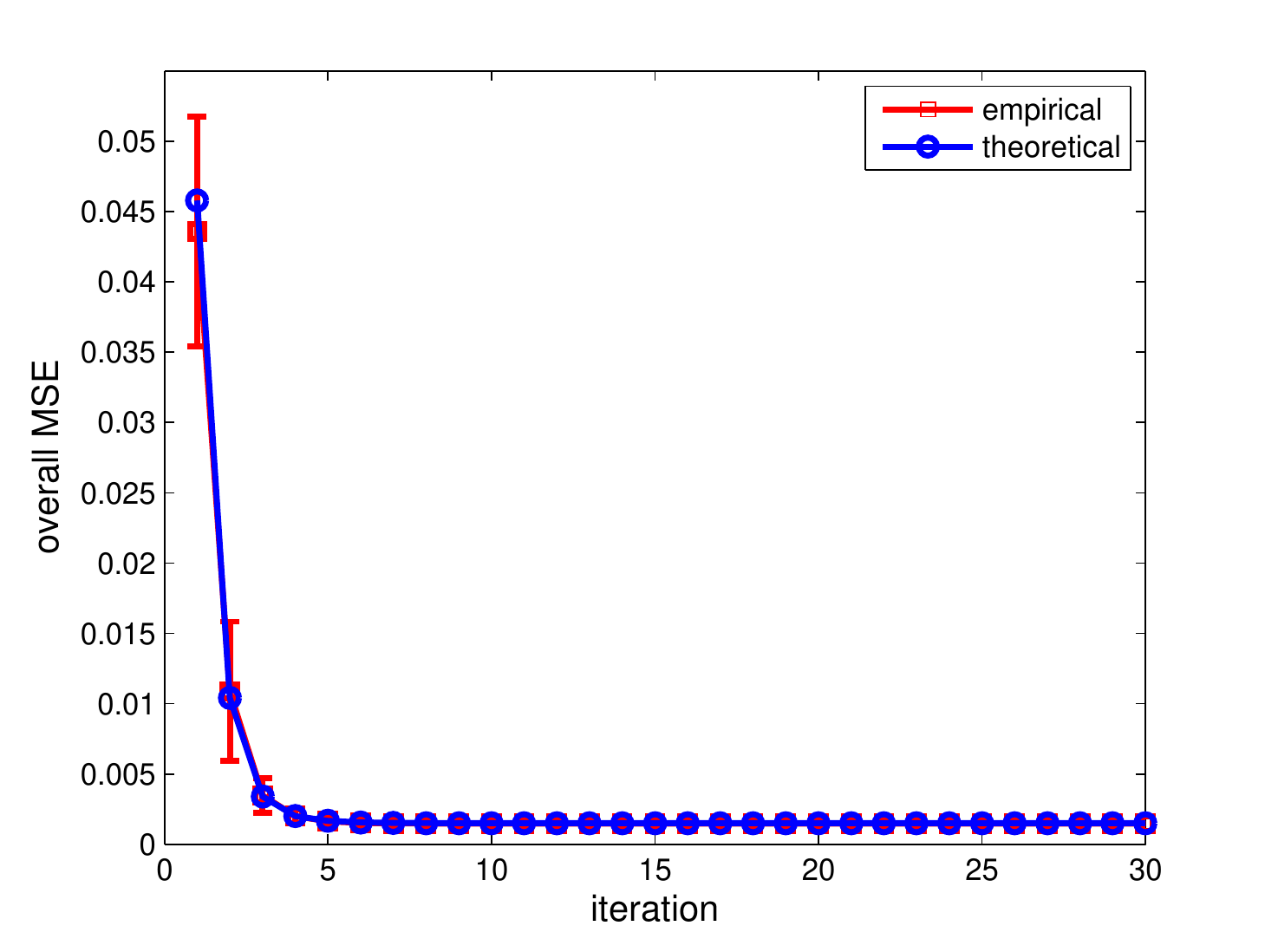}\label{fig:sim1_1b}}

  \subfloat[][]{\includegraphics[width=2.5in]{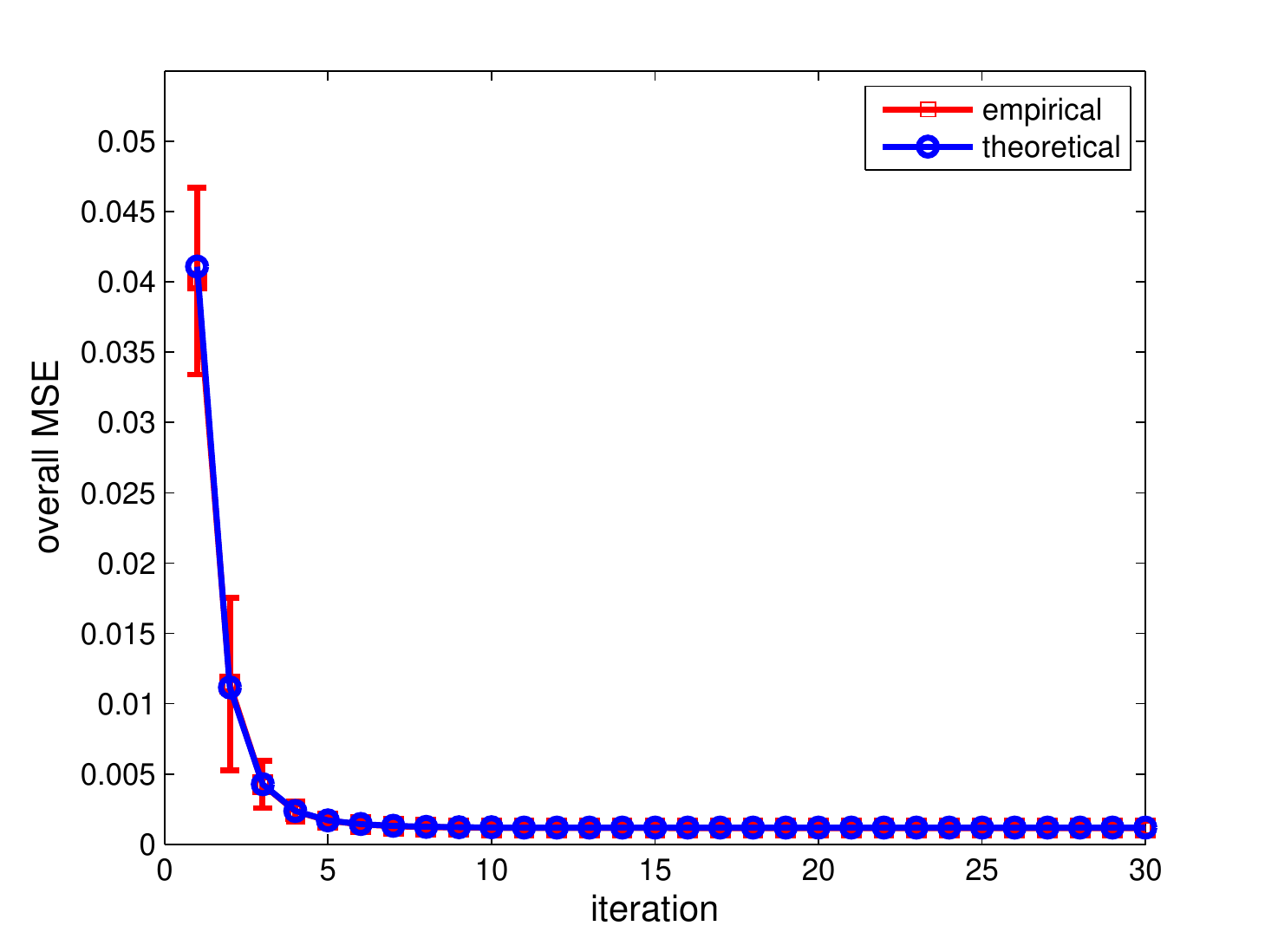}\label{fig:sim1_1c}}
  \subfloat[][]{\includegraphics[width=2.5in]{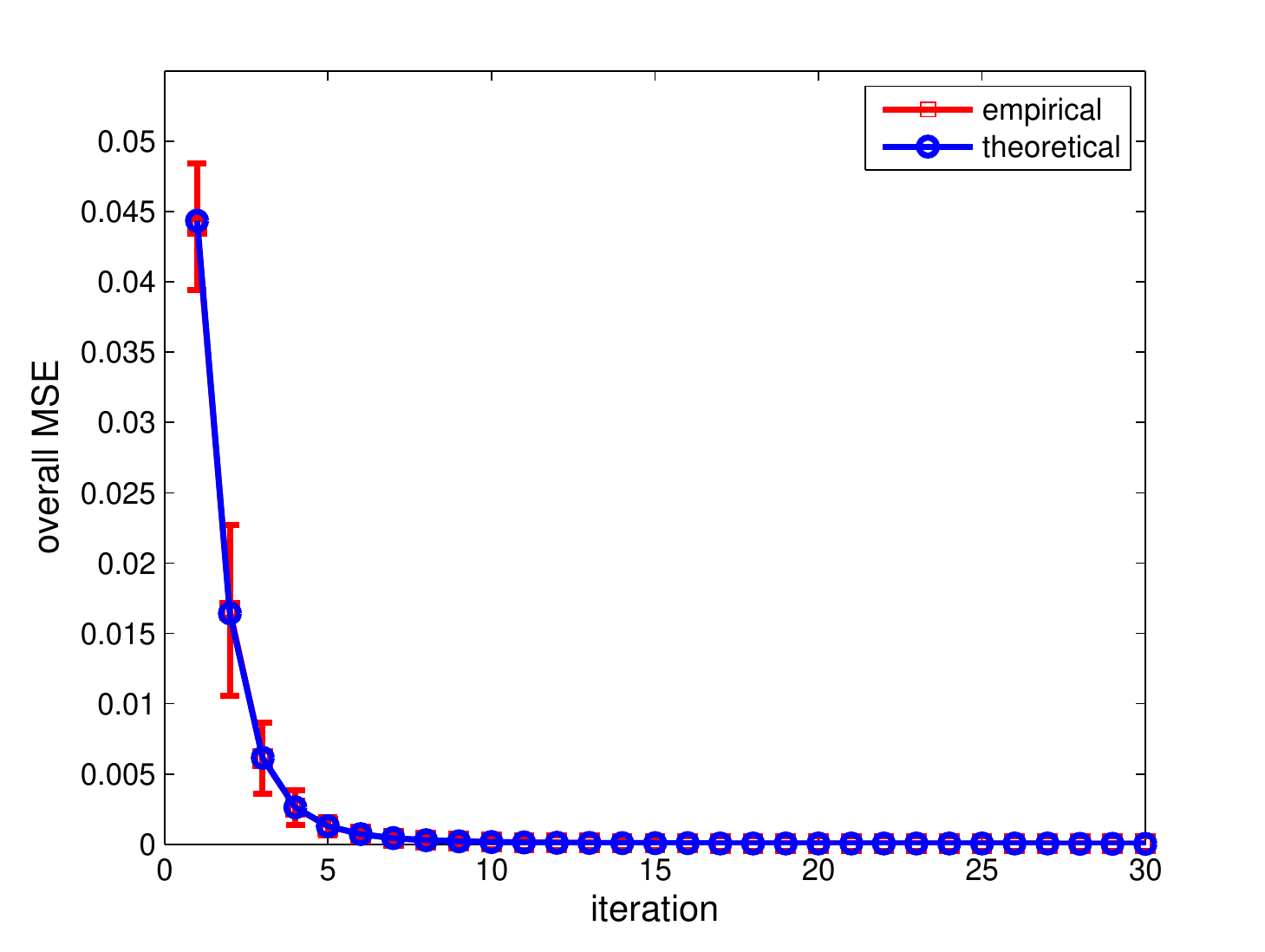}\label{fig:sim1_1d}}
  \caption{Comparison of the theoretical prediction and the Monte Carlo simulation result for (a) $p=0$, (b) $p=0.3$, (c) $p=0.5$ and (d) $p=0.8$. In the four cases, parameter $\tau$ is  set as 0.2, 0.3, 0.5 and 1 in Figures (a), (b), (c), and (d) respectively. The nonzero elements of the sparse vector are $\pm 1$ equiprobable in this scenario. In Figures (b) and (c), the error of the recovery does not converge to 0 under the parameter setting of the scenario, but SE still provides an accurate prediction.}
  \label{fig:sim1_1}
\end{figure}

\begin{figure}
  \centering
  \subfloat[][]{\includegraphics[width=2.5in]{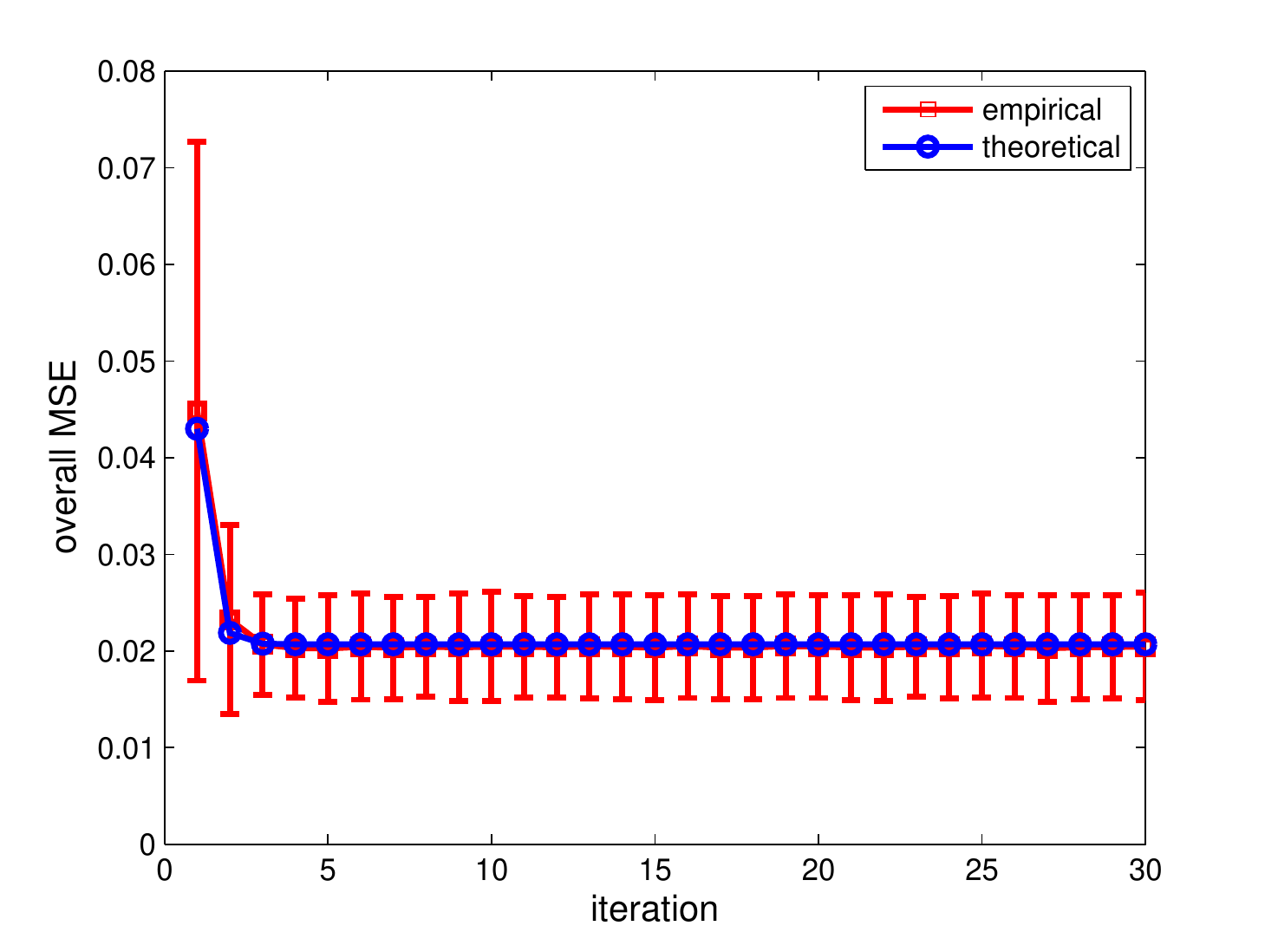}\label{fig:sim1_2a}}
  \subfloat[][]{\includegraphics[width=2.5in]{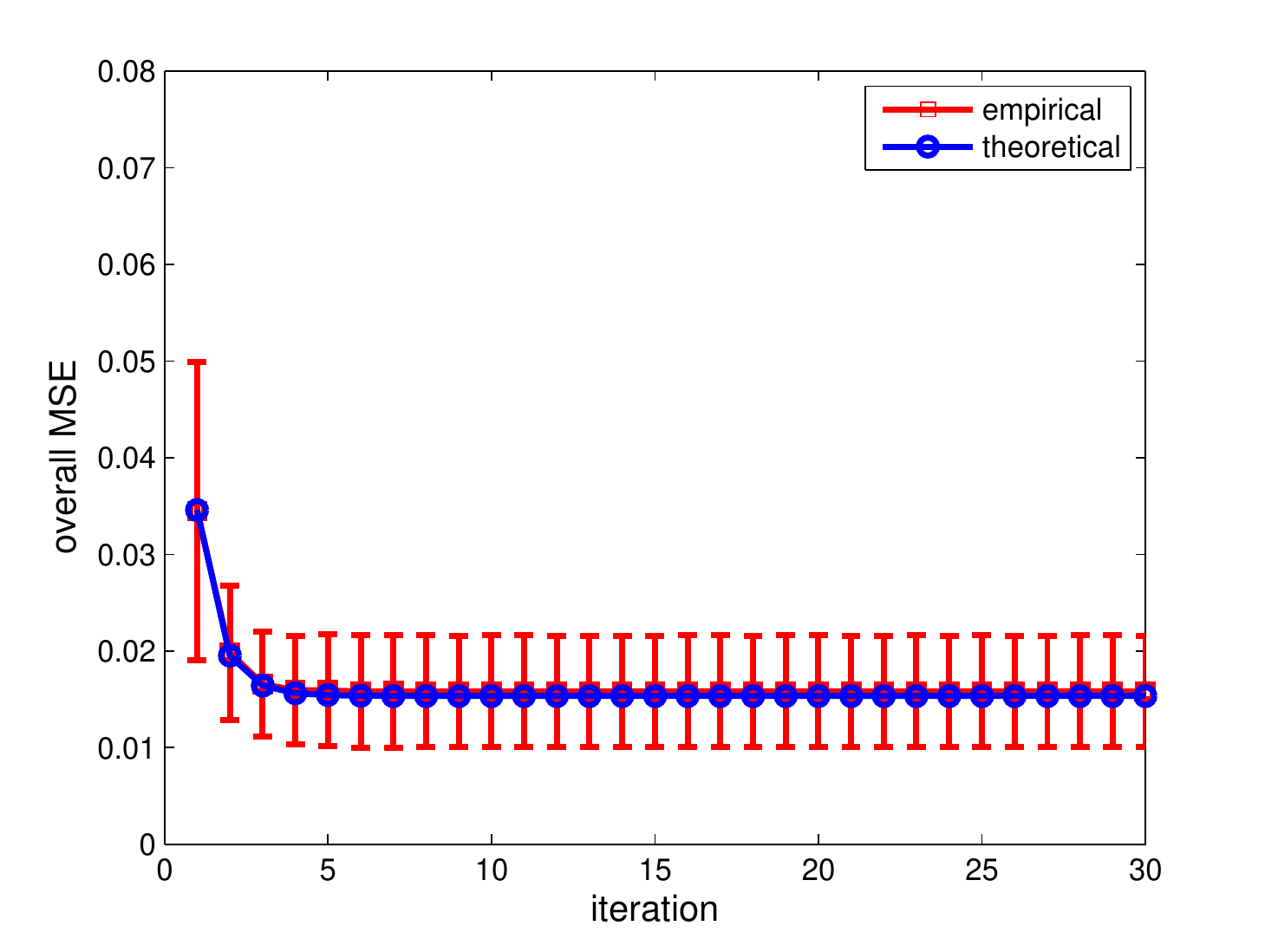}\label{fig:sim1_2b}}

  \subfloat[][]{\includegraphics[width=2.5in]{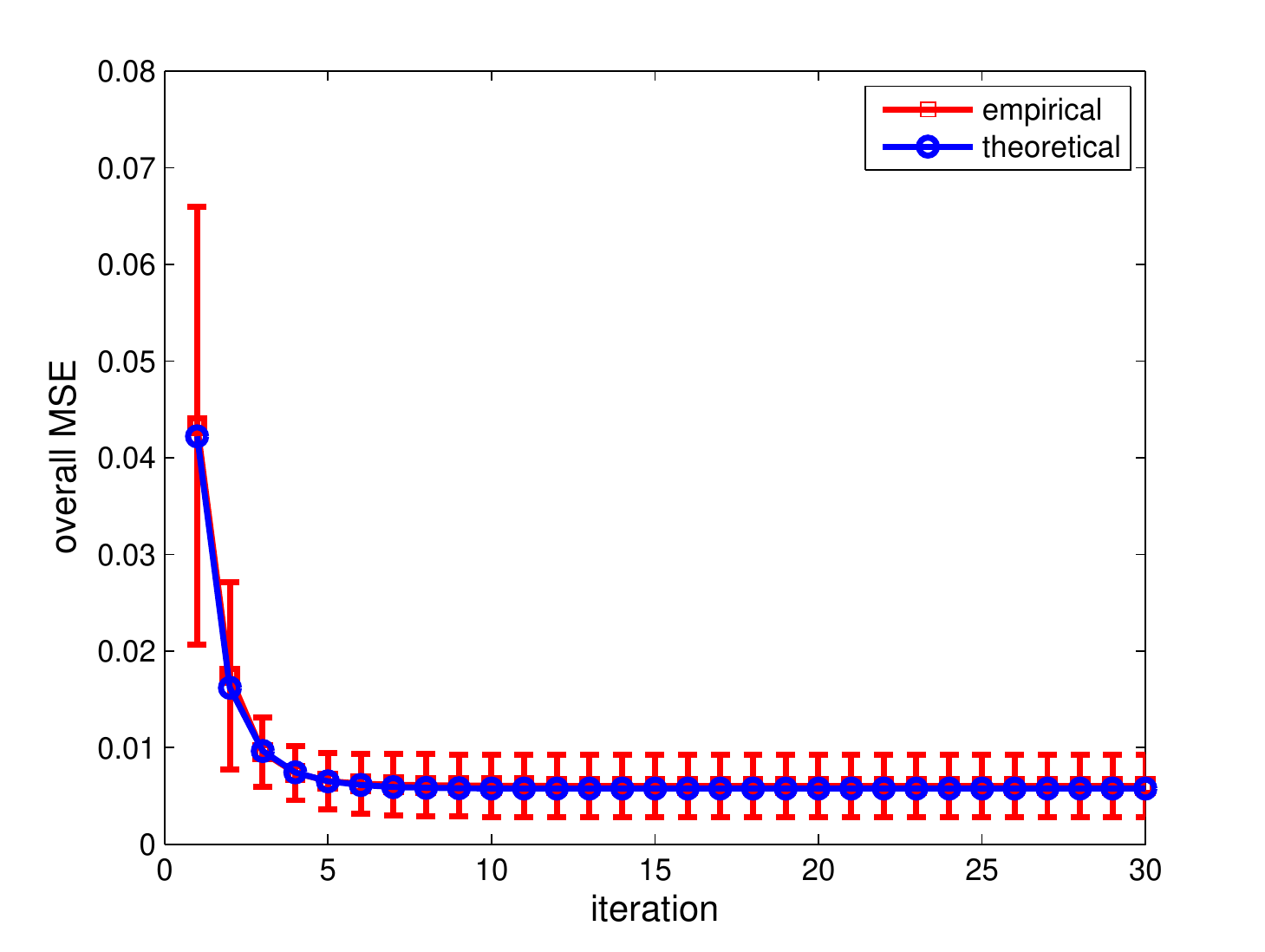}\label{fig:sim1_2c}}
  \subfloat[][]{\includegraphics[width=2.5in]{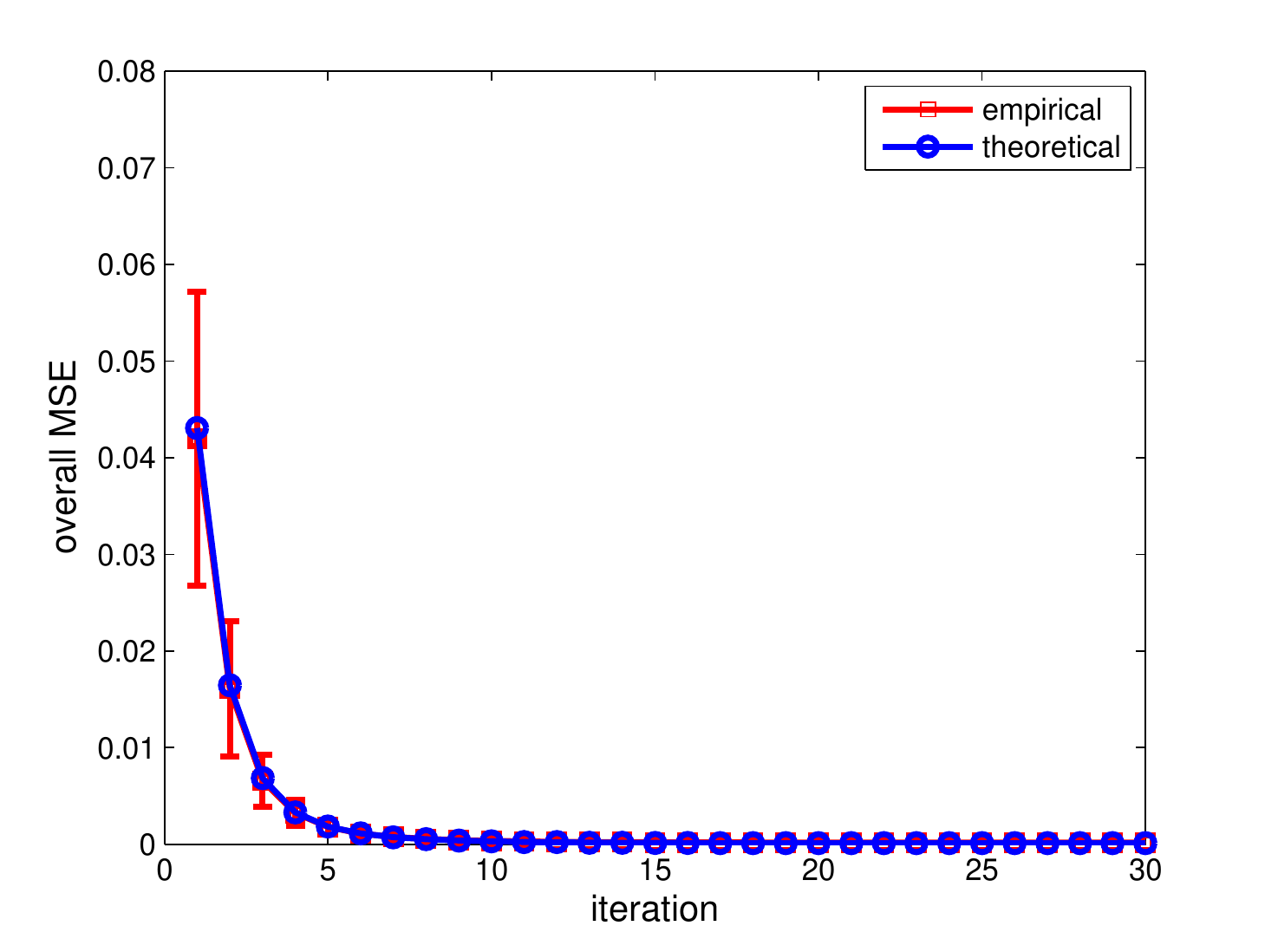}\label{fig:sim1_2d}}
  \caption{Comparison of the theoretical prediction and the Monte Carlo simulation result for (a) $p=0$, (b) $p=0.3$, (c) $p=0.5$ and (d) $p=0.8$. In the four cases, parameter $\tau$ is respectively set as 0.18, 0.4, 0.5 and 1. The nonzero elements of the sparse vector are Gaussian distributed in this scenario. In Figures (a), (b) and (c), the error of the recovery does not converge to 0 under the parameter setting of the scenario, but SE still gives a reasonably accurate prediction to the performance of the $\ell_p$-AMP algorithm.}
  \label{fig:sim1_2}
\end{figure}

In this section, the predictions given by the state evolution are compared with the performance of Monte Carlo simulations. For Monte Carlo simulations, the dimension of the sparse vector $x_0$ is set to $N=5000$ which is relatively large. The measurement matrix $A$ is iid Gaussian distributed and the number of measurements is set to $n=1000$, i.e., $\delta=0.2$. There are 40 nonzero elements in $x_0$. We run $\ell_p$-AMP for $T=30$ iterations. We set the thresholding policy to $\lambda(\sigma)=\tau \sigma^p$ where $\tau$ is a fixed number. The value of $\tau$ may differ in different simulations and will be mentioned below each figure. The empirical MSE reported in the figures is the average of 100 Monte Carlo simulations. Finally $h$ is set to $h=\sigma_t/N^{1/3}$ at iteration $t$. We have not optimized over the parameter $h$. We have empirically noticed that $\ell_p$-AMP with this choice of $h$ has a good performance. Accurate analysis of the effect of $h$ on the performance of $\ell_p$-AMP is left for a future research. 

Figures \ref{fig:sim1_1} and \ref{fig:sim1_2} compare the result of Monte Carlo simulation with the SE when the nonzero elements of the sparse vector are $\pm 1$ equiprobable and Gaussian respectively. The bars show 95\% confidence intervals. As can be seen from the figures, the empirical results are reasonably close to the theoretical result that we obtained from the state evolution.

\vspace{.2cm}

\subsubsection{Optimal tuning of $\lambda$}\label{sec:noisefree}
Our goal in this section is to show the accuracy of the parameter selection technique we proposed in Section \ref{sec:ampsure} for finite sample sizes. As we discussed in Section \ref{sec:ampsure}, for the optimal tuning of $\lambda_t$ we can employ the following estimate:
\begin{eqnarray*}
\hat{\lambda }_t \in \arg\min_{\lambda \geq 0} \frac{1}{N}  \| \tilde{\eta}_{p,h}({x}^t+A^T z^t; \lambda)-{x}^t-A^T z^t\|^2- \sigma_t^2 + \frac{2\sigma_t^2}{N} {\rm div} (\tilde{\eta}_{p,h}({x^t+A^Tz^t} ; \lambda)).
\end{eqnarray*}
This requires an estimate of $\sigma_t$ at every iteration. It is straightforward to use the results of \cite{bayati2011dynamics} and prove that as $N \rightarrow \infty$, $\|z^t\|_2^2/n \rightarrow \sigma_t^2$. Hence, in our simulations we will employ the estimate $\|z^t\|_2^2/n$. 

In this section, we would like to show that even in moderately large sample sizes $N = 5000$, $\hat{\lambda}_t$ is close enough to $\lambda_t^*$ introduced in \eqref{eq:optlambdadef}. Here is our simulation settings. The dimension of the sparse vector $x_o$ is set to be $N=5000$. The elements of the measurement matrix $A$ are iid Gaussian. The dimension of measurements is set to be $n=1000$, i.e., $\delta=0.2$. $x_o$ has only 40 nonzero elements. The nonzero elements of the sparse vector are $\pm 1$ equiprobable. The variance of the measurement noise is set to be $\sigma_w^2=0.01$. We have run the simulations for $p=0, 0.3, 0.5, 0.8$. In each iteration, $\lambda_t=\tau_t \sigma_t^p$ is optimized by using the empirical risk function.

The risk and its SURE estimation in the third iteration are plotted against $\tau_t$ in Figure \ref{fig:sim2_1}. As can be seen from the figure, the SURE estimate is reasonably close to the true value of the risk function. 

\begin{figure}
  \centering
  \subfloat[][]{\includegraphics[width=2.5in]{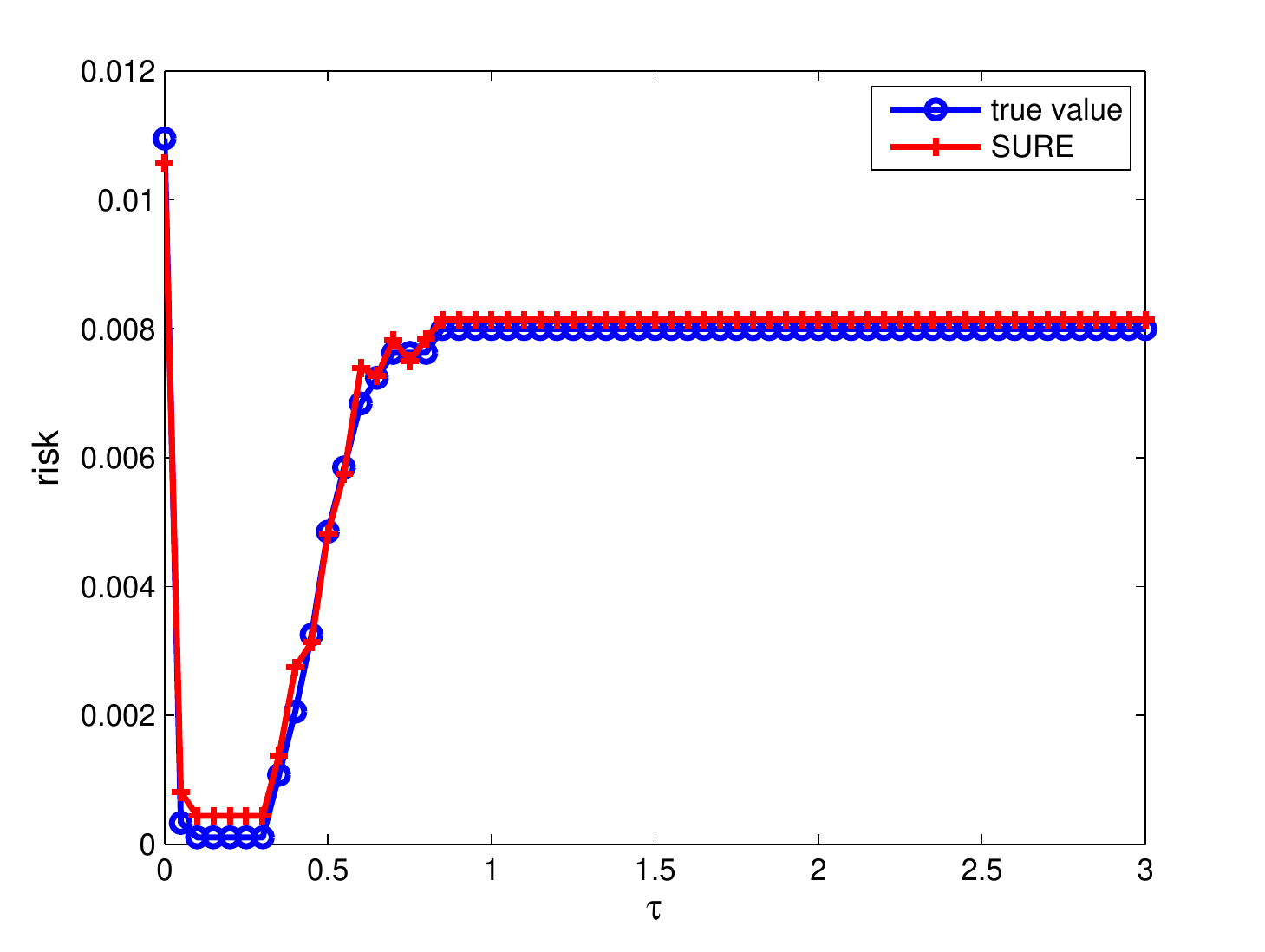}}
  \subfloat[][]{\includegraphics[width=2.5in]{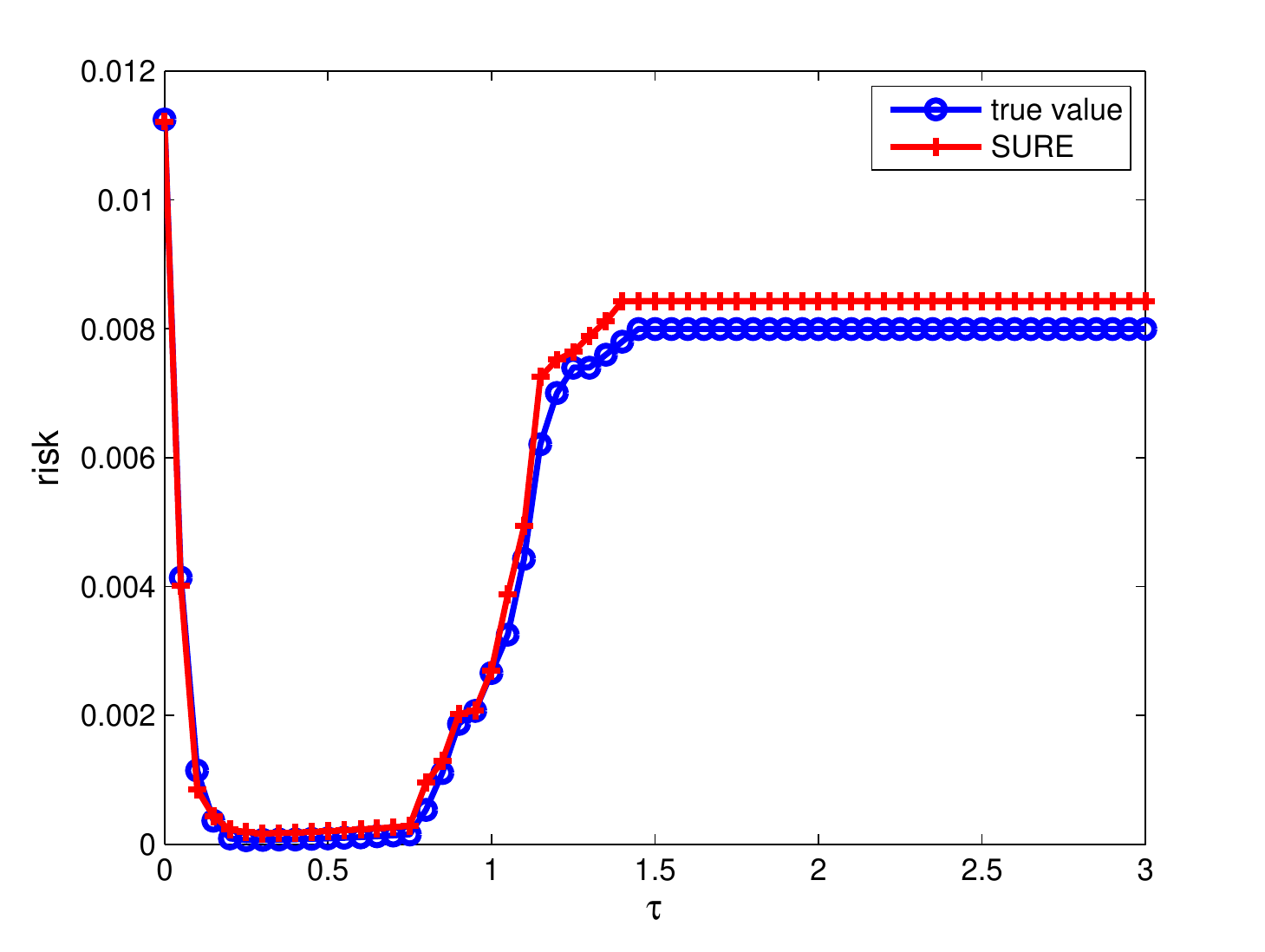}}

  \subfloat[][]{\includegraphics[width=2.5in]{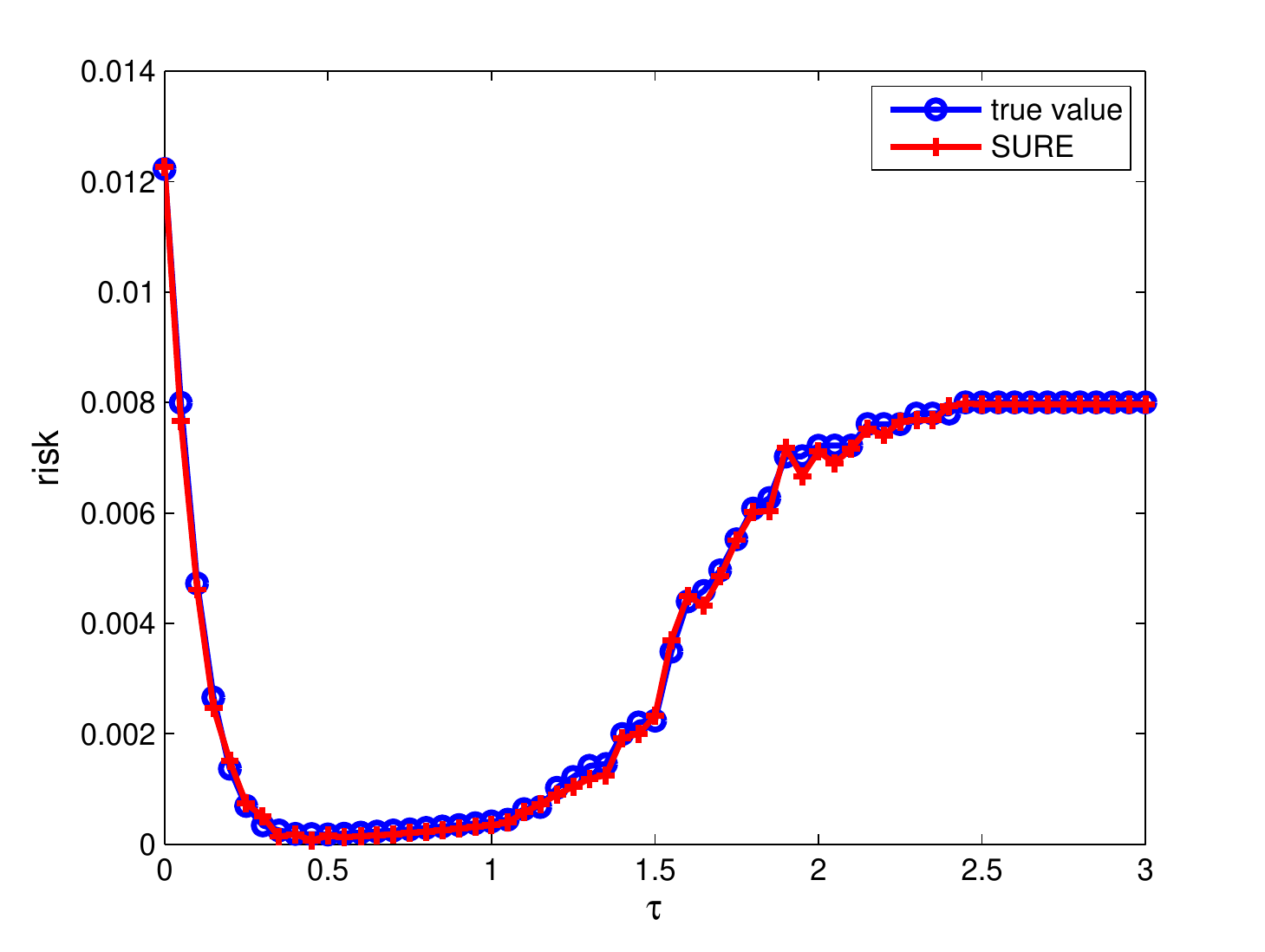}}
  \subfloat[][]{\includegraphics[width=2.5in]{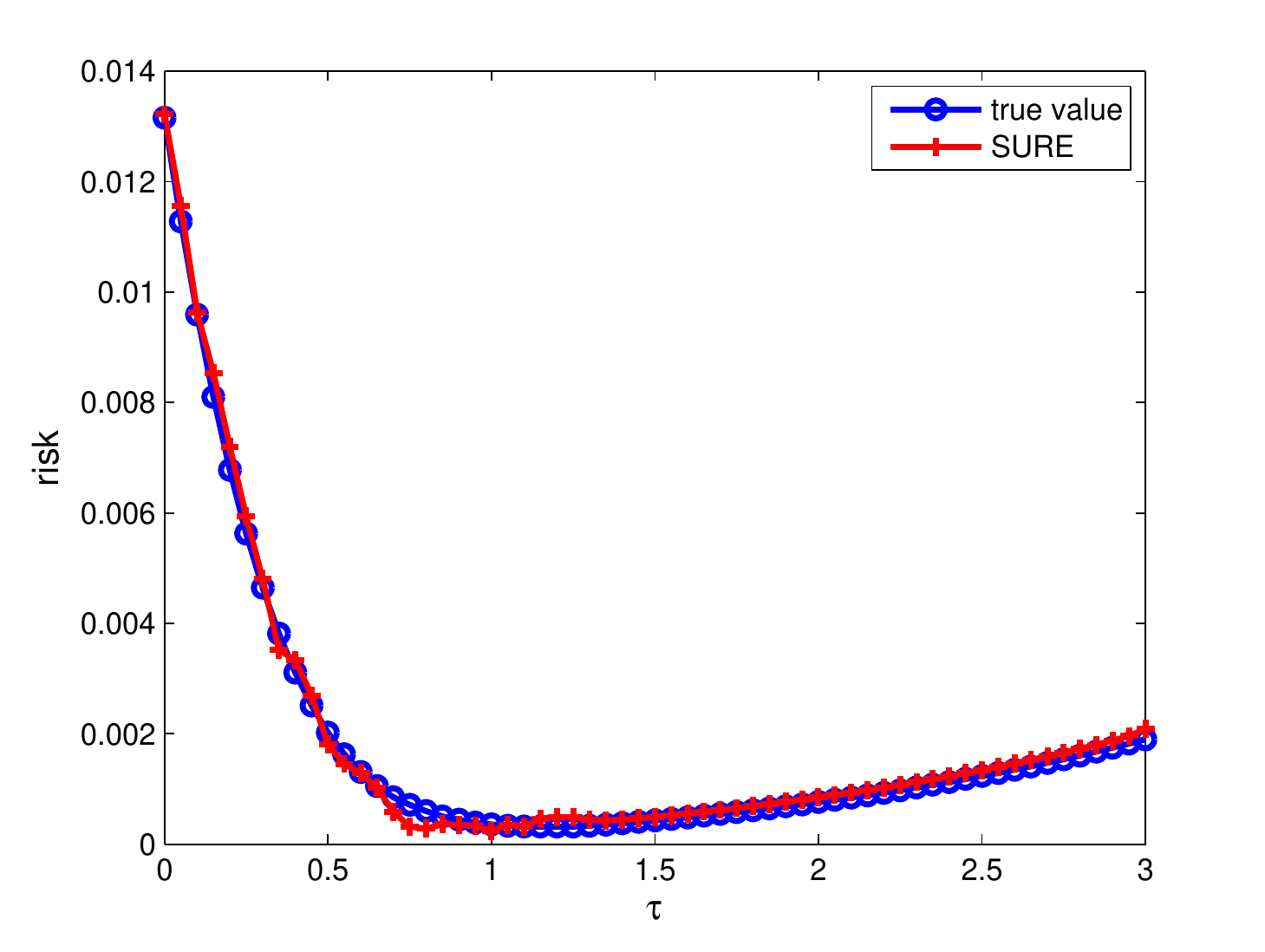}}

  \caption{The actual risk function and its unbiased estimate (SURE) as a function of $\tau_t$. We set (a) $p=0$, (b) $p=0.3$, (c) $p=0.5$ and (d) $p=0.8$. The figure is the result for the third iteration. The power of measurement noise is set to be $\sigma_w^2=0.01$.}
  \label{fig:sim2_1}
\end{figure}

\section{Conclusion and future work}
\label{sec:conclusion}
We have studied the performance of both $\ell_p$-regularized least squares (LPLS) problem and approximate message passing (AMP) that aims to solve LPLS. Employing the state evolution framework, we have derived conditions under which $\ell_p$-AMP for $0 \leq p<1$  outperforms $\ell_1$-AMP. It turns out that in the noiseless setting if the algorithm is initialized properly, it can outperform $\ell_1$-AMP by a large margin. We applied the Replica method to connect our results to LPLS. It turns out that, in the noiseless regime, the phase transitions of LPLS are exactly the same for every $p<1$. We also studied the performance of these algorithm in the presence of the measurement noise. We showed that for small values of measurement noise, $p=0$ outperforms the other values of $p$. However, when the measurement noise is large, $p=1$ outperforms the other values of $p$. 

There are many questions that we have left for future research. For instance, extensions of this approach to other types of structure, such as group-sparsity or low-rankness, is an open direction that needs to be explored in the future. Such research may shed some light on the benefit of non-convex penalties for these popular structures. Finally, it is not yet clear if it is possible to design algorithms that can find the global minima of LPLS in the asymptotic settings. As we have shown in this paper, below certain sparsity level, $\epsilon^*(\delta)$, message passing algorithms may recover the global minima of LPLS. But, beyond this level they may  be trapped at a fixed point that is different from the global minima of LPLS. Unfortunately, $\epsilon^*(\delta)$ is much below the actual phase transition of LPLS for $p<1$. Whether we can find an algorithm that is better than message passing for these problems is a major open question that may have major impact in the field of compressed sensing.

\appendix

\section{Proof of Theorem \ref{thm:empiricalriskconvergence}}\label{app:proof}
Our main  objective in this section is to prove Theorem \ref{thm:empiricalriskconvergence}. We start with the following lemmas that will be used later in the proof. 

\vspace{.2cm}

\begin{lemma}\label{lem:forsimplelip}
Let $f: \mathbb{R} \rightarrow \mathbb{R}$ denote a differentiable function with bounded derivative, i.e., $|f'(u)|<M$ for every $u$. Then,
\begin{equation*}
|f(s_1+\vartheta_1) - f(s_0+\vartheta_0)|\leq \sqrt{2}M \sqrt{(s_1-s_0)^2+ (\vartheta_1-\vartheta_0)^2}. 
\end{equation*}
\end{lemma}
\textit{Proof:}
According to the mean value theorem, we have
\begin{eqnarray}\label{eq:differentiabletoLip}
|f(s_1+\vartheta_1)-f(s_0+\vartheta_0)| = \left. \left[\frac{\partial f(s+\vartheta)}{\partial s} , \frac{\partial f(s+\vartheta)}{\partial \vartheta}  \right] \right|_{(s^*,\vartheta^*)} [s_1-s_0, \vartheta_1-\vartheta_0]^T, 
\end{eqnarray}
where $(s^*,\vartheta^*)$ is a point on the line that connects $(s_0,\vartheta_0)$ and $(s_1,\vartheta_1)$. By applying Cauchy-Schwartz inequality to \eqref{eq:differentiabletoLip} and using the fact that  $\left. \frac{\partial f(s+\vartheta)}{\partial s}\right|_{(s^*,\vartheta^*)} = f'(s^*+\vartheta^*)$ and  $\left.\frac{\partial f(s+\vartheta)}{\partial \vartheta}\right|_{(s^*,\vartheta^*)} = f'(s^*+\vartheta^*)$, we can finish the proof. 
$\hfill \Box$

\vspace{.2cm}

We now turn to the proof of Theorem \ref{thm:empiricalriskconvergence}. The proof employs Theorem 2 of \cite{bayati2011dynamics}. This theorem confirms that under the conditions we presented in Theorem \ref{thm:empiricalriskconvergence}, we have
\[
\lim_{N \rightarrow \infty} \frac{1}{N} \sum_{i=1}^N J(v_i^t, x_{o,i}) \overset{a.s.}{=} \mathbb{E} J(\sigma_t Z, X_o),
\]
where $Z \sim N(0,1)$ and $X \sim p_X$ are two independent random variables, and $J$ is a Lipschitz function of $v_i^t$ and $x_{o,i}$.\footnote{In fact, the result of Theorem 2 of \cite{bayati2011dynamics} considers more general pseudo-Lipschitz function $J$.} Note that 
\begin{equation}\label{eq:empiricalrisk1}
\hat{r}^t_{p,h, \lambda} = \frac{1}{N} \|\tilde \eta_{p,h}(x_o+v^t; \lambda) -x_o-v^t\|_2^2 - \sigma_t^2 + \frac{2 \sigma_t^2}{N} {\rm div} (\tilde \eta_{p,h}(x_o+v^t; \lambda)).  
\end{equation}
If we prove that both $ |\tilde \eta_{p,h} (x_{o,i}+v^t_i;\lambda) -x_{o,i}-v_i^t|^2$ and $\tilde \eta_{p,h}'(x_{o,i}+ v_i^t ; \lambda)$ are Lipschitz functions of $(x_{o,i},v_i^t)$, then we can characterize the limit of $\hat{r}^t_{p,h, \lambda}$. Hence as the next step, we prove these two quantities are Lipschitz.  We proved in Section \ref{proof:thm1} that  $\tilde \eta_{p,h}(u; \lambda)$ is a differentiable function of $u$ with $\sup_{u} |\tilde \eta_{p,h}'(u; \lambda)|$ bounded. Furthermore, it is straightforward to show that $|\tilde \eta_{p,h}(u; \lambda)-u|$ is bounded. Hence, $|\tilde \eta_{p,h}(u; \lambda)-u|^2$ has a bounded derivative. If we combine this fact with  Lemma \ref{lem:forsimplelip}, we know that $|\tilde \eta_{p,h} (x_{o,i}+v^t_i; \lambda) -x_{o,i}-v^t_i|^2$ is a Lipschitz function of $(x_{o,i}, v^t_i)$.  Hence, by Theorem 2 of \cite{bayati2011dynamics} we conclude that
\begin{equation}\label{eq:empiricalrisk2}
\lim_{N \rightarrow \infty} \frac{1}{N} \sum_{i=1}^N |\tilde \eta_{p,h} (x_{o,i}+v_i ; \lambda) -x_{o,i}-v_i|^2 \overset{a.s.}{=} \mathbb{E} (\tilde \eta_{p,h}(X+\sigma_t Z; \lambda)-X-\sigma_tZ)^2.
\end{equation}
Moreover, according to the proof of Theorem \ref{conj:se} that we presented in Section \ref{proof:thm1}, we know $\tilde{\eta}'_{p,h}(\cdot;\lambda)$ is bounded and has finite discontinuity points. We can then apply the arguments for proving Equation (4.11) in \cite{bayati2012lasso} to conclude that
\begin{equation}\label{eq:empiricalrisk3}
\lim_{N \rightarrow \infty} \frac{2 \sigma_t^2}{N} {\rm div} (\tilde \eta_{p,h}(x_o+v^t; \lambda))  \overset{a.s.}{=}  2 \sigma_t^2  \mathbb{E} (\tilde \eta'_{p,h} (X + \sigma_t Z; \lambda))
\end{equation}
Hence, if we combine \eqref{eq:empiricalrisk1}, \eqref{eq:empiricalrisk2}, and \eqref{eq:empiricalrisk3} we obtain
\[
\lim_{N \rightarrow \infty} \hat{r}^t_{p,h, \lambda}\overset{a.s.}{=}  \mathbb{E} (\tilde \eta_{p,h}(X+\sigma_t Z; \lambda)-X-\sigma_tZ)^2- \sigma_t^2 + 2 \sigma_t^2 \mathbb{E} (\tilde \eta'_{p,h} (X + \sigma_t Z; \lambda)). 
\]
Finally, Lemma \ref{lem:stein} (with $N=1$) shows that
\[
 \mathbb{E} (\tilde \eta_{p,h}(X+\sigma_t Z; \lambda)-X-\sigma_tZ)^2- \sigma_t^2 + 2 \sigma_t^2  \mathbb{E} (\tilde \eta'_{p,h} (X + \sigma_t Z; \lambda)) = \mathbb{E}  (\tilde \eta_{p,h}(X+\sigma_t Z; \lambda)-X)^2 .
\]

\bibliographystyle{IEEEtran}
\bibliography{database} 

\begin{thebibliography}{10}
\providecommand{\url}[1]{#1}
\csname url@samestyle\endcsname
\providecommand{\newblock}{\relax}
\providecommand{\bibinfo}[2]{#2}
\providecommand{\BIBentrySTDinterwordspacing}{\spaceskip=0pt\relax}
\providecommand{\BIBentryALTinterwordstretchfactor}{4}
\providecommand{\BIBentryALTinterwordspacing}{\spaceskip=\fontdimen2\font plus
\BIBentryALTinterwordstretchfactor\fontdimen3\font minus
  \fontdimen4\font\relax}
\providecommand{\BIBforeignlanguage}[2]{{%
\expandafter\ifx\csname l@#1\endcsname\relax
\typeout{** WARNING: IEEEtran.bst: No hyphenation pattern has been}%
\typeout{** loaded for the language `#1'. Using the pattern for}%
\typeout{** the default language instead.}%
\else
\language=\csname l@#1\endcsname
\fi
#2}}
\providecommand{\BIBdecl}{\relax}
\BIBdecl

\bibitem{maleki2013asymptotic}
A.~Maleki, L.~Anitori, Z.~Yang, and R.~Baraniuk, ``Asymptotic analysis of
  complex {LASSO} via complex approximate message passing ({CAMP}),''
  \emph{IEEE Transactions on Information Theory}, vol.~59, no.~7, pp.
  4290--4308, 2013.

\bibitem{baraniuk2007compressive}
R.~Baraniuk, ``Compressive sensing,'' \emph{IEEE signal processing magazine},
  vol.~24, no.~4, 2007.

\bibitem{chartrand2007exact}
R.~Chartrand, ``Exact reconstruction of sparse signals via nonconvex
  minimization,'' \emph{Signal Processing Letters}, vol.~14, no.~10, pp.
  707--710, 2007.

\bibitem{stojnic2013lifting}
M.~Stojnic, ``Lifting $\ell_q$-optimization thresholds,'' \emph{arXiv preprint
  arXiv:1306.3976}, 2013.

\bibitem{chartrand2008restricted}
R.~Chartrand and V.~Staneva, ``Restricted isometry properties and nonconvex
  compressive sensing,'' \emph{Inverse Problems}, vol.~24, no.~3, p. 035020,
  2008.

\bibitem{saab2008stable}
R.~Saab, R.~Chartrand, and O.~Yilmaz, ``Stable sparse approximations via
  nonconvex optimization,'' in \emph{Proc. IEEE International Conference on
  Acoustics, Speech and Signal Processing}, 2008, pp. 3885--3888.

\bibitem{foucart2009sparsest}
S.~Foucart and M.-J. Lai, ``Sparsest solutions of underdetermined linear
  systems via $\ell_q$-minimization for $0< q<1$,'' \emph{Applied and
  Computational Harmonic Analysis}, vol.~26, no.~3, pp. 395--407, 2009.

\bibitem{saab2010sparse}
R.~Saab and {\"O}.~Y{\i}lmaz, ``Sparse recovery by non-convex
  optimization--instance optimality,'' \emph{Applied and Computational Harmonic
  Analysis}, vol.~29, no.~1, pp. 30--48, 2010.

\bibitem{ge2011note}
D.~Ge, X.~Jiang, and Y.~Ye, ``A note on the complexity of $l_p$-minimization,''
  \emph{Mathematical programming}, vol. 129, no.~2, pp. 285--299, 2011.

\bibitem{kabashima2009typical}
Y.~Kabashima, T.~Wadayama, and T.~Tanaka, ``A typical reconstruction limit for
  compressed sensing based on lp-norm minimization,'' \emph{Journal of
  Statistical Mechanics: Theory and Experiment}, vol. 2009, no.~09, p. L09003,
  2009.

\bibitem{lai2011unconstrained}
M.-J. Lai and J.~Wang, ``An unconstrained $\ell_q$ minimization with $0 \leq q
  \leq 1$ for sparse solution of underdetermined linear systems,'' \emph{SIAM
  Journal on Optimization}, vol.~21, no.~1, pp. 82--101, 2011.

\bibitem{sun2012recovery}
Q.~Sun, ``Recovery of sparsest signals via $\ell_q$-minimization,''
  \emph{Applied and Computational Harmonic Analysis}, vol.~32, no.~3, pp.
  329--341, 2012.

\bibitem{trzasko2007sparse}
J.~Trzasko, A.~Manduca, and E.~Borisch, ``Sparse {MRI} reconstruction via
  multiscale $\ell_0$-continuation,'' in \emph{Proc. of IEEE Workshop on
  Statistical Signal Processing}, 2007, pp. 176--180.

\bibitem{aldroubi2011stability}
A.~Aldroubi, X.~Chen, and A.~Powell, ``Stability and robustness of lq
  minimization using null space property,'' in \emph{Proc. of Sampling Theory
  and Its Applications}, 2011.

\bibitem{rangan2009asymptotic}
S.~Rangan, V.~Goyal, and A.~K. Fletcher, ``Asymptotic analysis of map
  estimation via the replica method and compressed sensing,'' in \emph{Proc.
  Advances in Neural Information Processing Systems}, 2009, pp. 1545--1553.

\bibitem{mazumder2011sparsenet}
R.~Mazumder, J.~H. Friedman, and T.~Hastie, ``Sparsenet: Coordinate descent
  with nonconvex penalties,'' \emph{Journal of the American Statistical
  Association}, vol. 106, no. 495, 2011.

\bibitem{candes2008enhancing}
E.~J. Candes, M.~B. Wakin, and S.~P. Boyd, ``Enhancing sparsity by reweighted
  $\ell_1$ minimization,'' \emph{Journal of Fourier analysis and applications},
  vol.~14, no. 5-6, pp. 877--905, 2008.

\bibitem{pant2014new}
J.~K. Pant, W.-S. Lu, and A.~Antoniou, ``New improved algorithms for
  compressive sensing based on p norm.'' \emph{IEEE Transactions on Circuits
  and Systems}, vol.~61, no.~3, pp. 198--202, 2014.

\bibitem{shen2012restricted}
Y.~Shen and S.~Li, ``Restricted p--isometry property and its application for
  nonconvex compressive sensing,'' \emph{Advances in Computational
  Mathematics}, vol.~37, no.~3, pp. 441--452, 2012.

\bibitem{davies2009restricted}
M.~E. Davies and R.~Gribonval, ``Restricted isometry constants where
  $\ell_p$-sparse recovery can fail for $0 < p \leq 1$,'' \emph{IEEE
  Transactions on Information Theory}, vol.~55, no.~5, pp. 2203--2214, 2009.

\bibitem{wang2011performance}
M.~Wang, W.~Xu, and A.~Tang, ``On the performance of sparse recovery via
  $\ell_p$-minimization,'' \emph{IEEE Transactions on Information Theory},
  vol.~57, no.~11, pp. 7255--7278, 2011.

\bibitem{donoho2009message}
D.~L. Donoho, A.~Maleki, and A.~Montanari, ``Message-passing algorithms for
  compressed sensing,'' \emph{Proceedings of the National Academy of Sciences},
  vol. 106, no.~45, pp. 18\,914--18\,919, 2009.

\bibitem{bayati2012lasso}
M.~Bayati and A.~Montanari, ``The lasso risk for gaussian matrices,''
  \emph{IEEE Transactions on Information Theory}, vol.~58, no.~4, pp.
  1997--2017, 2012.

\bibitem{donoho2011noise}
D.~L. Donoho, A.~Maleki, and A.~Montanari, ``The noise-sensitivity phase
  transition in compressed sensing,'' \emph{IEEE Transactions on Information
  Theory}, vol.~57, no.~10, pp. 6920--6941, 2011.

\bibitem{maleki2010approximate}
A.~Maleki, \emph{Approximate message passing algorithms for compressed
  sensing}.\hskip 1em plus 0.5em minus 0.4em\relax Stanford University, 2010.

\bibitem{deledalle2013stein}
C.-A. Deledalle, G.~Peyr{\'e}, and J.~Fadili, ``Stein consistent risk estimator
  (score) for hard thresholding,'' \emph{arXiv preprint arXiv:1301.5874}, 2013.

\bibitem{rangan2011generalized}
S.~Rangan, ``Generalized approximate message passing for estimation with random
  linear mixing,'' in \emph{Proc. of IEEE International Symposium on
  Information Theory}, 2011, pp. 2168--2172.

\bibitem{donoho2011accurate}
D.~Donoho, I.~Johnstone, and A.~Montanari, ``Accurate prediction of phase
  transitions in compressed sensing via a connection to minimax denoising,''
  \emph{arXiv preprint arXiv:1111.1041}, 2011.

\bibitem{schniter2010turbo}
P.~Schniter, ``Turbo reconstruction of structured sparse signals,'' in
  \emph{Proc. of Annual Conference on Information Sciences and Systems}, 2010,
  pp. 1--6.

\bibitem{metzler2014denoising}
C.~A. Metzler, A.~Maleki, and R.~G. Baraniuk, ``From denoising to compressed
  sensing,'' \emph{arXiv preprint arXiv:1406.4175}, 2014.

\bibitem{bayati2011dynamics}
M.~Bayati and A.~Montanari, ``The dynamics of message passing on dense graphs,
  with applications to compressed sensing,'' \emph{IEEE Transactions on
  Information Theory}, vol.~57, no.~2, pp. 764--785, 2011.

\bibitem{krzakala2012statistical}
F.~Krzakala, M.~M{\'e}zard, F.~Sausset, Y.~Sun, and L.~Zdeborov{\'a},
  ``Statistical-physics-based reconstruction in compressed sensing,''
  \emph{Physical Review X}, vol.~2, no.~2, p. 021005, 2012.

\bibitem{maleki2010optimally}
A.~Maleki and D.~L. Donoho, ``Optimally tuned iterative reconstruction
  algorithms for compressed sensing,'' \emph{IEEE Journal of Selected Topics in
  Signal Processing}, vol.~4, no.~2, pp. 330--341, 2010.

\bibitem{mezard1985replicas}
M.~M{\'e}zard and G.~Parisi, ``Replicas and optimization,'' \emph{Journal de
  Physique Lettres}, vol.~46, no.~17, pp. 771--778, 1985.

\bibitem{castellani2005spin}
T.~Castellani and A.~Cavagna, ``Spin-glass theory for pedestrians,''
  \emph{Journal of Statistical Mechanics: Theory and Experiment}, vol. 2005,
  no.~05, p. P05012, 2005.

\bibitem{guo2005randomly}
D.~Guo and S.~Verd{\'u}, ``Randomly spread {CDMA}: Asymptotics via statistical
  physics,'' \emph{IEEE Transactions on Information Theory}, vol.~51, no.~6,
  pp. 1983--2010, 2005.

\bibitem{tanaka2002statistical}
T.~Tanaka, ``A statistical-mechanics approach to large-system analysis of
  {CDMA} multiuser detectors,'' \emph{IEEE Transactions on Information Theory},
  vol.~48, no.~11, pp. 2888--2910, 2002.

\bibitem{lehmann1986testing}
E.~L. Lehmann and J.~P. Romano, \emph{Testing statistical hypotheses}.\hskip
  1em plus 0.5em minus 0.4em\relax Wiley New York et al, 2005.

\bibitem{mousavi2015consistent}
A.~Mousavi, A.~Maleki, and R.~G. Baraniuk, ``Consistent parameter estiamtion
  for lasso and approximate message passing,'' \emph{arXiv preprint
  arXiv:1511.01017}, 2015.

\bibitem{weng2016overcoming}
H.~Weng, A.~Maleki, and L.~Zheng, ``Overcoming the limitations of phase
  transition by higher order analysis of regularization techniques,''
  \emph{arXiv preprint arXiv:1603.07377}, 2016.

\bibitem{donoho2009supporting}
D.~L. Donoho, A.~Maleki, and A.~Montanari, ``Supporting information to:
  Message-passing algorithms for compressed sensing,'' \emph{Proceedings of
  National Academy of Sciences}, 2009.

\bibitem{mousavi2013asymptotic}
A.~Mousavi, A.~Maleki, and R.~Baraniuk, ``Asymptotic analysis of {LASSO}'s
  solution path with implications for approximate message passing,''
  \emph{arXiv preprint arXiv:1309.5979}, 2013.

\bibitem{wu2011mmse}
Y.~Wu and S.~Verd{\'u}, ``Mmse dimension,'' \emph{Information Theory, IEEE
  Transactions on}, vol.~57, no.~8, pp. 4857--4879, 2011.

\bibitem{stein1981estimation}
C.~Stein, ``Estimation of the mean of a multivariate normal distribution,''
  \emph{The Annals of Statistics}, pp. 1135--1151, 1981.

\bibitem{hastie2009elements}
T.~Hastie, R.~Tibshirani, and J.~Friedman, \emph{The elements of statistical
  learning}.\hskip 1em plus 0.5em minus 0.4em\relax Springer, 2009, vol.~2,
  no.~1.

\end{thebibliography}
\end{document}